\documentclass[a4paper,12pt]{article}
\usepackage[title]{appendix}
\usepackage{a4wide}
\usepackage{amsmath}
\usepackage{mathtools}
\usepackage{amssymb}
\usepackage{amsthm}
\usepackage{amsfonts,bm}
\usepackage{scrextend}
\usepackage{subfig}
\usepackage{graphicx}
\usepackage{wrapfig}
\usepackage{tikz}
\usepackage{xcolor}
\usepackage{placeins}
\usepackage{floatrow}
\usepackage{subfig}
\usepackage{url}
\usepackage{hyperref}
\usepackage{algpseudocode}
\usepackage{algorithm}
\usepackage{array}
\usepackage{eufrak}
\usepackage{yfonts}
\usepackage{mathrsfs}
\usepackage{enumerate}
\usepackage[perpage]{footmisc}
\usepackage{authblk}
\usetikzlibrary{external}
\newtheorem{theorem}{Theorem}[section]
\newtheorem{corollary}[theorem]{Corollary}
\newtheorem{lemma}[theorem]{Lemma}
\newtheorem{proposition}[theorem]{Proposition}
\numberwithin{equation}{section}

\title{A stochastic modeling framework for single cell migration: coupling contractility and focal adhesions}
\author{Aydar Uatay}
\affil{\footnotesize Technische Universit\"at Kaiserslautern, Felix-Klein-Zentrum f\"ur Mathematik\protect\\ Paul-Ehrlich-Str.31, 67663 Kaiserslautern, Germany\protect\\(uatay at mathematik.uni-kl.de)}
\begin{document}
\maketitle
\begin{abstract}
The interaction of the actin cytoskeleton with cell-substrate adhesions is necessary for cell migration. While the trajectories of motile cells have a  stochastic character, investigations of cell motility mechanisms rarely elaborate on the origins of the observed randomness. Here, guided by a few fundamental attributes of cell motility, we construct a minimal stochastic cell migration model from ground-up. The resulting model couples a deterministic actomyosin contractility mechanism with stochastic cell-substrate adhesion kinetics, and yields a well-defined piecewise deterministic process. Numerical simulations reproduce several experimentally observed results, including anomalous diffusion, tactic migration, and contact guidance. This work provides a basis for the development of cell-cell collision and population migration models.
\end{abstract}

\section{Introduction}
Cell migration is essential for embryogenesis, wound healing, immune surveillance, and progression of diseases, such as cancer metastasis. During embryogenesis, coordinated collective migration to particular sites is essential to the development of an organism. Tissue repair and immune response rely on directed movement of cells following external cues, which are produced after the tissue layer is damaged and infected with pathogens. Likewise, cancer cells migrate away from the tumor into surrounding tissue and distant organs to form metastases, which is the leading cause of death among cancer patients.

Cell motility is a cyclical process, involving morphological changes of the cell body and adhesive contacts with the underlying substrate \cite{lauffenburger1996cell}. The cycle can be divided into three steps: protrusion of the leading edge, assembly of adhesions to the substrate at the front and disassembly in the rear, and contraction of the cell body, thereby producing locomotion \cite{Abercrombie129}. This type of crawling movement requires transmission of contractile forces, generated within the cell cytoskeleton, through substrate adhesions. Such mechanical interactions between various cellular structures is attained by integrating numerous signaling molecules, the most prominent of which are Rho GTPases \cite{Raftopoulou2004}, \cite{ridley2003cell}. While there are other modes of motility relying on, for example, flagellar activity (e.g. spermatozoa or E.coli) or rolling in the bloodstream (e.g. leukocytes), here the focus is on the crawling type of movement.   
       
Due to cell motility being a highly complex and not fully elucidated process, mathematical modeling of the corresponding phenomena is a challenging task. Numerous approaches have been developed, each being able to capture a certain aspect of the process (see \cite{holmes2012comparison}, \cite{ziebert2016computational} for extensive reviews on whole cell motility models and \cite{mogilner2009mathematics} for a review on modeling of its critical components). For example, free boundary and phase-field models of steadily migrating cells in \cite{Camley2017crawling}, \cite{oelz2010cell} (and its extension in {\cite{manhart2015extended}}), \cite{preziosi2011multiphase}, \cite{Rubinstein2005multiscale}, \cite{shao2010computational} are able to reproduce cell morphology as a result of mechanical and biochemical interactions. Models of cell migration in \cite{COPOS20172672}, \cite{zhu2016comparison} explored emergence of various motility modes due to mechanical coupling of intracellular components and the substrate. In their hybrid motility model, Mar\'{e}e et al. \cite{maree2006polarization}, \cite{Athanasius2012how} explored the mechanochemical interaction in detail by considering the Rho GTPase signaling circuit. A common feature of these continuum models is that they are not able to reproduce experimentally observed random paths of migrating cells. Stochastic motility models of eukaryotic cells have been proposed, for example, in \cite{arrieumerlou2005local}, \cite{dickinson1993stochastic}, \cite{groh2010stochastic}, \cite{satulovsky2008exploring}, \cite{tranquillo1988stochastic}. There stochasticity is driven by a Gaussian process, although there is evidence that the paths of the migrating cells follow a non-Gaussian process \cite{dieterich2008anomalous}, \cite{liang2008persistent}, \cite{SELMECZI2005912}. Another way to include randomness is based on velocity jump process, as proposed by \cite{othmer1988models} in the context of cell motility. Extending the model and including cell-substrate interactions, population migration models have been proposed in \cite{engwer2015glioma}, \cite{kelkel2012multiscale}. Here, however, our focus is solely on single cell migration. 

In this paper and the follow-up works, we strike a middle ground. By proposing a minimal cell representation including a few cellular structures essential for cell motility, we aim at reproducing stochastic migratory paths in various experimental settings. In our model, stochasticity arises as a result of mechanochemical coupling between the cell cytoskeleton and the substrate through adhesive contacts. Specifically, the random events under consideration are the (de)adhesion events of the cell migration cycle, whereas deterministic locomotion and contraction occur between arrival of the events. Here, we do not make any prior assumptions about the distributions that the events and their arrival times follow. Rather, we consider, in detail, the major determinants of adhesion dynamics and derive the complete stochastic description. This is in contrast to most mathematically well-posed models on stochastic cell migration, where it is assumed that the motion follows a Gaussian process. Thereby, the construction of our cell motility model will result in a piecewise deterministic process \cite{davis84}, \cite{davis93}. Our simulations reproduce experimental observations, such as superdiffusive scaling of the squared displacement \cite{dieterich2008anomalous}, \cite{liang2008persistent}, \cite{liu2015confinement}, biased migration in the presence of an external cue, contact guidance \cite{ROG17}, and directed movement due to asymmetric contractility (and in the absence of guidance cues) \cite{verkhovsky1999self}, \cite{yam2007actin}. Thus, our approach illustrates how detailed treatment of adhesion cluster dynamics can translate into stochastic cell motility description in a mathematically consistent manner. Moreover, due to the minimal character of the model and our bottom-up approach, this work serves as a basis for modeling cell-cell collisions in \cite{uatay2019mathematical} and population migration, which will be presented in follow-up papers.                

This article is organized as follows. In Section \ref{section: The Model} we briefly describe the cell motility cycle and the relevant agents. We then introduce the minimal cell representation and describe the deterministic cell motion between the cycle steps. In Section \ref{section: jump process} we construct a stochastic model of adhesion events, which signify transitions of cycle stages. In Section \ref{section: PDMP} we combine the deterministic and stochastic components of the migration cycle to obtain a well-defined piecewise deterministic Markov process. In Section \ref{section: adhesion kinetics} we specify the kinetics of adhesion events. Numerical simulations are performed in Section \ref{section: numerical simulations}. A discussion of the results and future outlook are presented in Section \ref{section: discussion}.

\section{The cell motility model}\label{section: The Model}
The cell migration cycle begins with protrusion of the leading edge as a result of actin polymerization (Figure \ref{fig: four cycle}). The polymerization process in lamellipodia is mediated by the Arp2/3 complex, which acts downstream of signaling pathways responsible for cell polarization \cite{ridley2003cell}. Next, the protrusions are stabilized due to formation of focal adhesions (FAs) in the lamellae (regions behind the lamellipodia), which link the actin cytoskeleton to the extracellular matrix (ECM). An FA is a multiprotein integrin-based adhesion cluster, which matures in a Rho GTPase dependent manner \cite{Raftopoulou2004}. Furthermore, FA maturation depends on the applied tension and occurs concomitantly with actomyosin bundle formation \cite{gardel2010mechanical}. These bundles, called stress fibers (SFs), generate contractile forces due to non-muscle myosin II motors. Due to increased tension at cell rear, FAs rupture. Finally, deadhesion leads to cell body translocation due to cytoskeletal contraction.

In order to construct the mathematical model, we make the following assumptions. First, FA unbinding leads to remodeling of the SF configuration (and of the entire cytoskeleton) and the cell movement, whereas assembly of new FAs leads to restructuring only. Second, FA events need not occur in the order described above. Several adhesions might be assembled (disassembled) before deadhesion (adhesion) occurs. Note also that while the contractile machinery is important, the dynamic instability of adhesions is what drives the migratory process, for stable FAs prevent retraction. Thus, we consider only interactions of SFs and FAs. Moreover, we do not consider the actin polymerization process and simplify the migration cycle to two steps: after FA assembly occurs, a cell does not move, but reconfigures SFs; after disassembly, a cell does both. Neglecting the polymerization process and the reduction to binding/unbinding events can be justified by the fact that one of the major consequences of the leading edge protrusions is promotion of FA assembly. Because the repolarization of migrating cells occurs frequently as an outcome of intricate biochemical activity, then, in order to keep the model tractable, we do not explicitly model cell polarity. Instead, (de)adhesion frequency is indicative of (rear)front. 
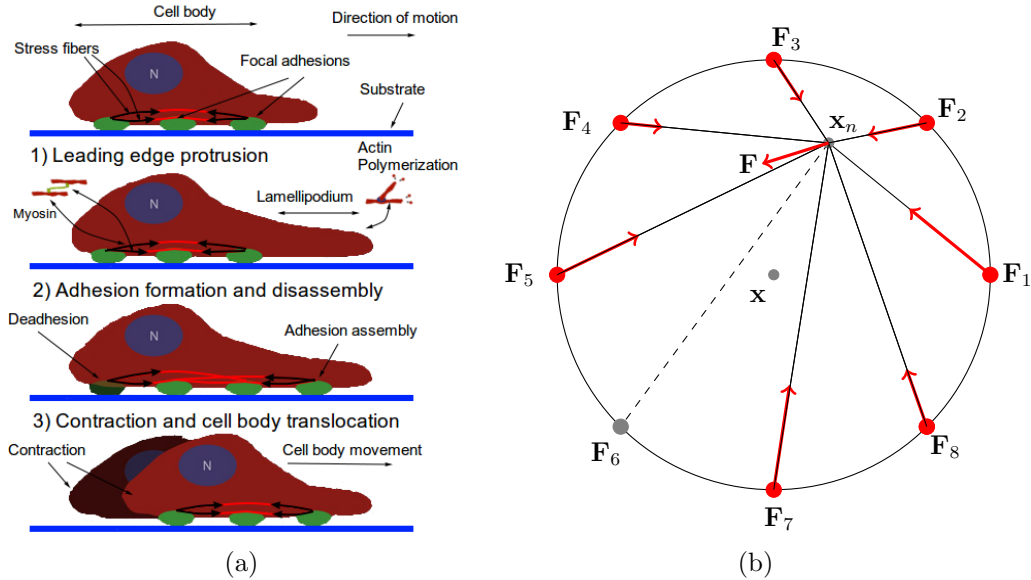
\begin{figure}[H]
\subfloat[]{
\includegraphics[width=0.37\textwidth]{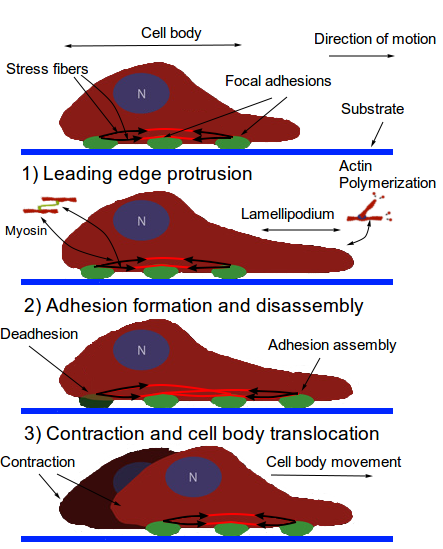}
\label{fig: four cycle}
}
\subfloat[]{
\begin{tikzpicture}[scale=0.95,transform shape,every node/.style={scale=0.9}]
\draw (0.76,1.84) -- (3,0);
\draw (0.76,1.84) -- (0,3);
\draw (0.76,1.84) -- (-3,0);
\draw (0.76,1.84) -- (0,-3);
\draw (0.76,1.84) -- (2.12,2.12);
\draw (0.76,1.84) -- (-2.12,2.12);
\draw (0.76,1.84) -- (2.12,-2.12);
\draw [dashed] (0.76,1.84) -- (-2.12,-2.12);
\draw (0,0) circle [radius=3cm];
\filldraw [gray] (0.76,1.84) circle [radius=2pt];

\filldraw [red] (3,0) circle [radius=3pt];
\filldraw [red] (0,3) circle [radius=3pt];
\filldraw [red] (0,-3) circle [radius=3pt];
\filldraw [red] (-3,0) circle [radius=3pt];
\filldraw [red] (2.12,2.12) circle [radius=3pt];
\filldraw [red] (-2.12,2.12) circle [radius=3pt];
\filldraw [red] (2.12,-2.12) circle [radius=3pt];
\filldraw [red] (2.12,-2.12) circle [radius=3pt];
\filldraw [gray] (-2.12,-2.12) circle [radius=3pt];

\filldraw [gray] (0,0) circle [radius=2pt];
\draw (-0.3,-0.3) node[text width = 2pt] {$\mathbf{x}$};

\draw [very thick, red][->](3,0) -- (1.88,0.92);
\draw (3.2,0) node[text width = 2pt] {$\mathbf{F}_1$};
\draw [very thick, red][->](2.12,2.12) -- (1.3,1.95);
\draw (2.3,2.3) node[text width = 2pt] {$\mathbf{F}_2$};
\draw [very thick, red][->](0,3) -- (0.38,2.42);
\draw (0,3.3) node[text width = 2pt] {$\mathbf{F}_3$};
\draw [very thick, red][->](-2.12,2.12) -- (-1.54,2.06);
\draw (-2.9,2.12) node[text width = 2pt] {$\mathbf{F}_4$};
\draw [very thick, red][->](-3,0) -- (-1.87,0.55);
\draw (-3.7,0) node[text width = 2pt] {$\mathbf{F}_5$};
\draw [very thick, red][->](0,-3) -- (0.22,-1.54);
\draw (-2.5,-2.5) node[text width = 2pt] {$\mathbf{F}_6$};
\draw (-0.1,-3.4) node[text width = 2pt] {$\mathbf{F}_7$};
\draw [very thick, red][->](2.12,-2.12) -- (1.84,-1.32);
\draw (2.2,-2.4) node[text width = 2pt] {$\mathbf{F}_8$};

\draw (0.76,1.84) -- (0,3);
\draw (0.76,1.84) -- (-3,0);
\draw (0.76,1.84) -- (0,-3);
\draw (0.76,1.84) -- (2.12,2.12);
\draw (0.76,1.84) -- (-2.12,2.12);
\draw (0.76,1.84) -- (2.12,-2.12);

\draw (0.8,2.11) node[text width = 2pt] {$\mathbf{x}_n$};

\draw[very thick, red] [->](0.76,1.84) -- (-0.16,1.55);
\draw (-0.45,1.55) node[text width = 2pt] {\textbf{F}};
\end{tikzpicture}

\label{fig: four}
}

\caption{(a) Schematic diagram of the cell migration cycle and the implicated cellular structures. Actin polymerization at the front pushes the membrane allowing protrusions to form.  Then, adhesions assemble at the front and disassemble at the rear. Finally, deadhesion and cell contraction produce locomotion, pulling the body forward. The black arrows overlaying the stress fibers show the inwardly directed contractile forces.
(b) Schematic representation of a cell with $M=8$ focal adhesions. Solid black lines represent stress fibers while red bullets represent focal adhesions. Red arrows indicate the direction and magnitude of applied traction force $\mathbf{F}_i$, $i=1,\ldots,8$. The dashed line and the corresponding gray circle represent an absent stress fiber and unbound focal adhesion, respectively. The central red arrow indicates the net force $\mathbf{F}$ on $\mathbf{x}_n$.}

\end{figure}

\subsection{The cell representation}

Consider the situation in Figure \ref{fig: four}. The disk represents a cell. Let the radius be $R_{cell}$ and let the position of the center at time $t$ be $\mathbf{x}(t)\in\mathbb{R}^2$. Suppose there are $M$ equally spaced adhesion sites $\mathbf{x}_i(t)\in R_{cell}\mathbb{S}^1$, $i=1,\ldots M$ on a cell circumference with constant relative distance (see Remark in Section {\ref{section: kinematics summary}}). Let $\mathbf{Y}(t)\in \left\lbrace 0,1\right\rbrace^M$ be a vector of focal adhesion states at time $t$, i.e. $Y_{i}(t) = 0,1$ correspond to unbound and bound FA at node $i$, respectively.

Since the traction stresses are oriented inward, transmitted to ECM by FAs, and generated by contractile SFs, then the FAs on the circumference must be one of the ends of SFs. Suppose the other end of all SFs at time $t$ is at the position $\mathbf{x}_n(t)\in\Omega_{cell}:=\left\lbrace(x,y)\in\mathbb{R}^2\phantom{.}|\phantom{.}x^2+y^2\leq R^2_{cell} \right\rbrace$ (in a cell's reference frame with origin at $\mathbf{x}$), i.e. all SFs are connected at $\mathbf{x}_n$. Since stress fibers behave like Hookean springs on extension, but readily buckle under compression \cite{Murrell2015}, then, inspired by Guthardt Torres et al.  \cite{TBS12}, the force $\mathbf{F}_i$ at focal adhesion $i$ is given by:
\begin{align}\label{eq: Fi}
\mathbf{F}_i = 
\begin{cases}
\left(T_i + EA\frac{L_i-L_0}{L_0}\right)\mathbf{e}_i,\phantom{aabc} L_0<L_i\\
T_i\mathbf{e}_i,\phantom{asdasdasdaabcsaaa}L_c\leq L_i\leq L_0\\
\frac{L_i-L_c+\delta}{\delta}T_i\mathbf{e}_i, \phantom{Asdaasssda} L_c-\delta\leq L_i<L_c\\
0 \phantom{sdaasdAADSADAdA} L_i<L_c-\delta, 
\end{cases} 
\end{align}
where $T_i$ is the magnitude of the contractile force due to myosin motors, $EA$ is the one-dimensional Young's modulus, $L_0$ and $L_c$ are, respectively, rest and critical lengths, $L_i = \lVert \mathbf{x}_n-\mathbf{x}_i\rVert$, $\mathbf{e}_i = \frac{\mathbf{x}_n-\mathbf{x}_i}{L_i}$ is the unit vector along the $i^{\text{th}}$ SF, and $\delta$ is a small positive constant. The first case in (\ref{eq: Fi}) is due to the Hookean behavior of SFs upon extension and myosin tension generation. Furthermore, stress fiber laser ablation experiments \cite{kassianidou2017geometry}, \cite{Kumar2006}, \cite{russell2009sarcomere} revealed that the initial instantaneous response (elastic behavior due to the SF length dependence in the first case) is followed by slower contraction due to myosin activity (force dependence on $T_i$) in the remaining portion of the fiber. Combined with stress fiber buckling, we obtain the second case in (\ref{eq: Fi}). Deguchi et al. \cite{deguchi2006stress} also found that SF contraction ceased after reaching a certain critical length. This implies that $F_i=0$ when $L_i<L_c-\delta$. For technical reasons we assume $F_i$ is piecewise continuous - hence the last cases in (\ref{eq: Fi}). We also assumed for simplicity that myosin generated force $T_i$ may vary between SFs, but is otherwise constant. 

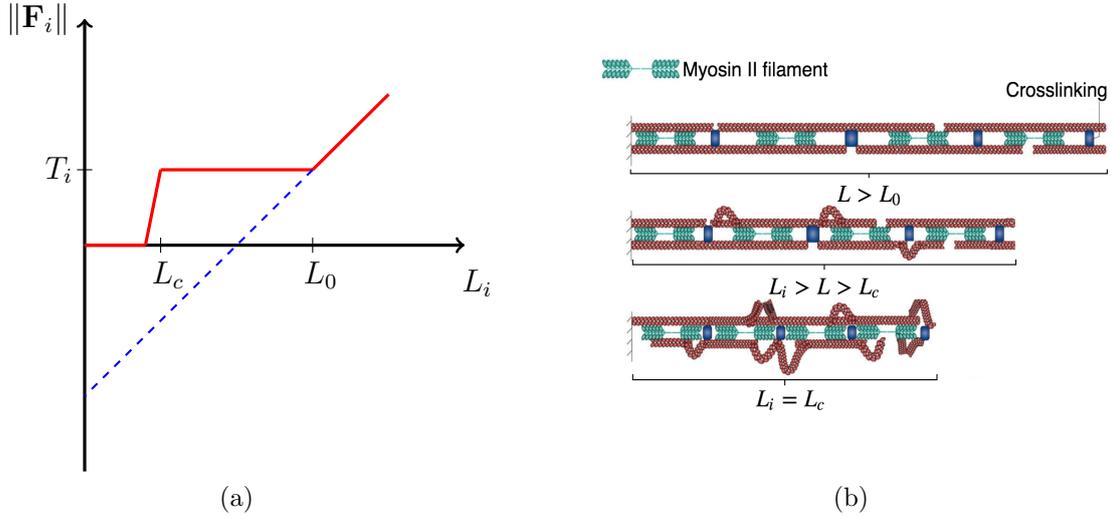
\begin{figure}
\subfloat[]
{
\begin{tikzpicture}
\draw [->][very thick] (0,0) -- (5,0);
\draw [->][very thick] (0,0) -- (0,3);
\draw [very thick] (0,0) -- (0,-3);
\draw (5,-0.5) node[text width = 2pt] {$L_i$};
\draw (-1.0,3) node[text width = 2pt] {$\lVert\mathbf{F}_i\rVert$};

\draw [red,very thick] (4,2) -- (3,1);
\draw [red,very thick] (3,1) -- (1,1);
\draw [red,very thick] (1,1) -- (0.8,0);
\draw [red,very thick] (0.8,0) -- (0,0);

\draw [blue,thick,dashed] (3,1)--(0,-2);

\draw (3,0.1) -- (3,-0.1);
\draw (2.9,-0.4) node[text width = 1pt] {$L_0$};
\draw (1,0.1) -- (1,-0.1);
\draw (0.9,-0.4) node[text width = 1pt] {$L_c$};

\draw (0.1,1) -- (-0.1,1);
\draw (-0.5,1) node[text width = 1pt] {$T_i$};
\end{tikzpicture}
}
\hspace{1cm}
\subfloat[]
{
	\includegraphics[width=7cm,height=6cm]{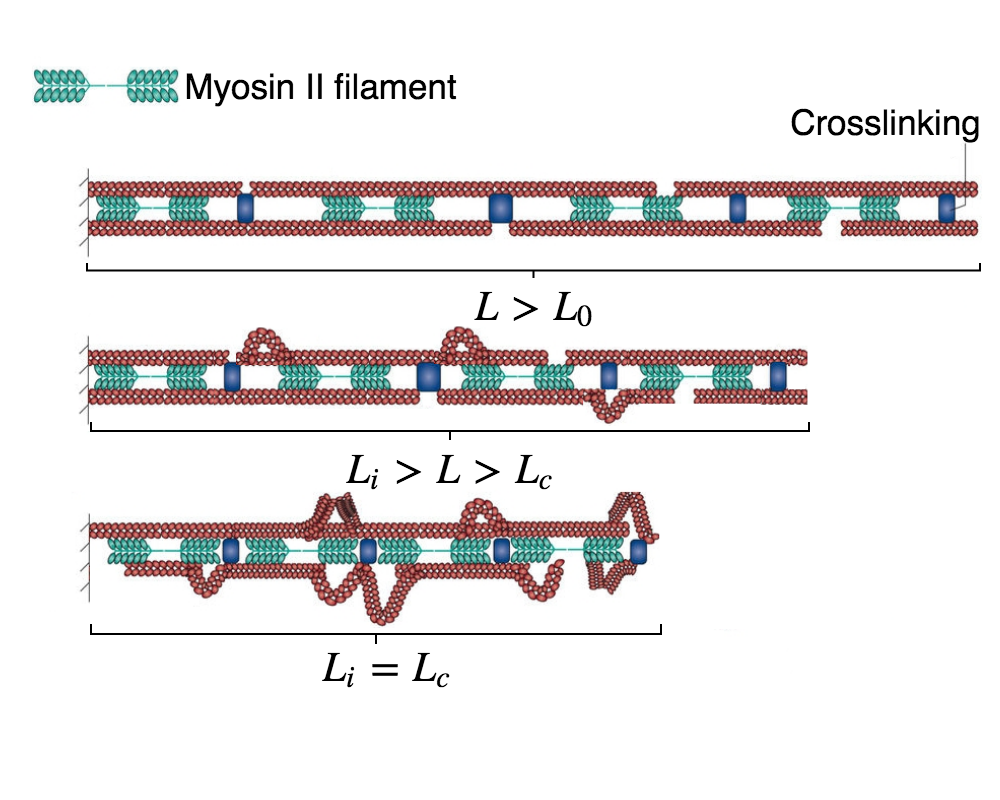}
}
\caption{(a) Magnitude of $\mathbf{F}_i$ in red. The blue dashed line corresponds to the profile of $F_i$ if we were to treat the fiber as a Hookean spring with constant $EA/L_0$. (b) Schematic representation of stress fiber contraction. As the fiber contracts below the rest length $L_0$, buckling occurs. As myosin mediated contraction causes the fiber to contract below the rest length $L_0$, buckling occurs due to lack of resistance to compression. Below the critical length $L_c$, the fiber ceases to contract due to vanishing interfilament distance. Modified from \cite{Murrell2015}.}
\label{figure: stress fiber contraction 1}
\end{figure}
Note that since $\mathbf{x}_i(t)\in R_{cell}\mathbb{S}^1$ and FA sites are equally spaced, then in polar coordinates we have:
\begin{align*}
\mathbf{x}_i(t) = R_{cell}(cos(\theta_i),sin(\theta_i))^T, && \theta_i(t):=\theta_1(t)+(i-1)\frac{2\pi}{M},
\end{align*}
and so $\mathbf{F}_i = \mathbf{F}_i(\mathbf{x}_n,\theta_1)$.

Since the force at $\mathbf{x}_n$ due to $i^{\text{th}}$ SF is $-\mathbf{F}_i$, then the net force at $\mathbf{x}_n$ is
\begin{align}\label{eq: total F}
\mathbf{F}(\mathbf{x}_n,\theta_1,\mathbf{Y}):=-\sum_{i=1}^{M}Y_{i}\mathbf{F}_i(\mathbf{x}_n,\theta_1),
\end{align} 
then, assuming negligible inertial effects (due to the viscous nature of cytoplasm) and constant $\mathbf{Y}$:
\begin{align}\label{eq: 1}
\beta_{cell}\dot{\mathbf{x}}_n = \mathbf{F}(\mathbf{x}_n,\theta_1,\mathbf{Y}),
\end{align} 
where $\beta_{cell}$ is the drag coefficient in the cytoplasm. 

The representation of a cell in such a way is justified for the following reasons:
\begin{itemize}
\item The traction stresses are largely applied on the cell periphery and their magnitude decays rapidly towards the center \cite{MKS12}, \cite{SG12}. Thus, the cell body SF ends are at or near mechanical equilibrium. Since contractile forces are generated by SFs, then a cell body SF end must be balanced by all other SFs (due to the equilibrium). Hence, it is reasonable to have a single connecting node of radial SFs which is either at mechanical equilibrium (for stationary cells) or tends to it. 
\item Paul et al. \cite{PHSS08} demonstrated in their active cable network model, combined with force application originating from nuclear region on FAs by star-like SF arrangement, results in cells acquiring morphologies typical for motile cells. Since the distribution of forces applied on FAs determines which one ruptures, then it also influences the motion of a cell (due to retraction). Since we are primarily interested in cell migration it is justified to assume that this architecture represents a realistic situation. Furthermore, Oakes et al. \cite{Oakes109595} found that modeling SFs embedded in contractile networks, where only SFs actively contract, yields a behavior mimicking the experimental results - the cytoskeletal flow was directed along the stress fibers. In the same study, the authors concluded that it is appropriate to treat an SF as a 1D viscoelastic contractile element, which also justifies neglecting inertia in Equation \eqref{eq: total F}. 
\item Since motile cells assume a wide variety of cell shapes and continuously remodel their actin cytoskeleton, one can view this representation as a cell shape normalization (it is implicitly assumed that a cell volume remains constant). That is, Figure \ref{fig: four} depicts a cell and forces applied on FAs normalized to a circle. M{\"o}hl et al. \cite{MKS12} applied the shape normalization technique to a timelapse series data of migrating keratinocytes and demonstrated that this allows consistent analysis of FA dynamics, actin flow and traction forces. In view of their results, a particular cell traction force map and FA configuration normalized to a circle can be effectively captured by our representation. 
\end{itemize}

Note that $\mathbf{x}_n(t)\in\Omega_{cell}$ for $t>0$, proved below.
\begin{proposition}\label{proposition: inside cell}
Let $\mathbf{x}_n\in \partial\Omega_{cell}$, $\theta_1\in[0,2\pi)$, and $\mathbf{Y}\in\left\lbrace 0,1\right\rbrace ^M$ be arbitrary. Let $\mathbf{n}$ be outward unit normal at $\mathbf{x}_n$. Then $\mathbf{F}(\mathbf{x}_n,\theta_1,\mathbf{Y})\cdot\mathbf{n}\leq 0$ with equality sign if and only if $\mathbf{F}(\mathbf{x}_n,\theta_1,\mathbf{Y}) = 0$.
\end{proposition}
\begin{proof}
The proposition above is obviously true, and it follows from the fact that
\begin{align*}
Y_{i}(-\mathbf{F}_i(\mathbf{x}_n,\theta_1))\cdot\mathbf{n} = Y_{i}\lVert\mathbf{F}_i(\mathbf{x}_n,\theta_1))\rVert(-\mathbf{e}_i\cdot\mathbf{n})=\frac{Y_{i}\lVert\mathbf{F}_i(\mathbf{x}_n,\theta_1))\rVert}{L_i}\left(\mathbf{x}_i-\mathbf{x}_n\right)\cdot\mathbf{n}\leq0,
\end{align*} 
since $\mathbf{x}_i\cdot\mathbf{n}<0$ and $\mathbf{x}_n\cdot\mathbf{n}>0$. 
\end{proof}

\begin{corollary}\label{corollay: xn exit}
Let $\mathbf{x}_n(0)\in\Omega_{cell}$ be arbitrary and let $\theta_1(t)$ be given. Suppose $\mathbf{x}_n\in C^1([0,\infty))$ is a solution of (\ref{eq: 1}). Then, $\mathbf{x}_n(t)\in\Omega_{cell}$  $\forall t>0$ and $\forall \mathbf{Y}(t)\in\left\lbrace 0,1\right\rbrace ^M$.
\end{corollary}
\begin{proof}
Due to (\ref{eq: 1}) it suffices to show that $\forall \mathbf{x}_n\in\partial\Omega_{cell}$ we have $\mathbf{F}(\mathbf{x}_n,\theta_1,\mathbf{Y})\cdot\mathbf{n}\leq 0$, which follows from Proposition \ref{proposition: inside cell}.
\end{proof}
\FloatBarrier
\subsection{The cell migration cycle}
Recall that during the migration cycle, deadhesion leads to cell body translocation, while adhesion binding does not. In both cases actomyosin contractility leads to reconfiguration of the cytoskeleton. Here we show how our cell representation can describe the reconfiguration and cell body motion following binding and unbinding events.

Without loss of generality assume that an event occurred at $t=0$. Let $\tau>0$ be the time of the next adhesion event, be it binding or unbinding. Let $\mathbf{Y}(0)\in\left\lbrace 0,1\right\rbrace ^M$, $\mathbf{x}(0)\in \mathbb{R}^2$, and $\mathbf{x}_n(0)\in\Omega_{cell}$ be arbitrary. Then, $\mathbf{Y}(t)=const.$ for $t\in[0,\tau)$. 

We assume $\theta_1(t=0) = 0$. Since the FA sites are equally spaced, it is sufficient to consider $\theta_1(t)$ only.

\subsubsection{Focal adhesion binding}
Following an FA binding, we suppose that a cell becomes stationary (i.e. the cell centroid remains constant). However, a newly formed FA and the associated SF lead to cytoskeletal reshaping. Thus, we have the following system of ODEs for $t\in[0,\tau)$:
\begin{align}\label{eq: binding ODEs}
\dot{\mathbf{x}} &=0\nonumber\\
\dot{\mathbf{x}}_n &= \beta_{cell}^{-1}\mathbf{F}(\mathbf{x}_n,\theta_1,\mathbf{Y})\nonumber\\
\dot{\theta}_1&=0.
\end{align}
\subsubsection{Focal adhesion unbinding}

Following an unbinding event, cytoskeletal contraction leads to cell body movement. Due to the circular geometry, the contractile forces induce both rotational and translational motion.

Note that the bound focal adhesions are able to slide for short distances \cite{MKS12}. Oakes et al. \cite{Oakes109595} found that the cytoskeleton behaves like an elastic solid on timescales up to one hour. Provided the time $\tau$ between adhesion events is small enough, the following is justified.

\begin{wrapfigure}[23]{r}{0.4\textwidth}
\vspace{-1cm}
\begin{tikzpicture}
\draw [->,thick, blue] (1cm,3.6cm)arc[start angle=67.5, end angle=112.5, radius=3cm];
\draw (-1.65,3.5) node[text width = 2pt] {$\dot{\theta}_1$};

\draw [->,thick, blue] (1.14,2.77) -- (0.48,3.04);
\draw (-0.12,3.23) node[text width = 2pt] {$F_{\varphi}\hat{\boldsymbol{\varphi}}$};

\draw (0,0) circle [radius=3cm];
\filldraw [gray] (0.76,1.84) circle [radius=2pt];
\filldraw [gray] (0,0) circle [radius=2pt];
\draw [->](0.76,1.84) -- (1.14,2.77);
\draw (0,0) -- (0.76,1.84);
\draw [->](0.76,1.84) -- (-0.16,2.22);
\draw (0.76,1.84) -- (0.85,2.07) -- (0.62,2.16) -- (0.52, 1.93);

\draw (1.25, 2.4) node[text width = 2pt] {$\hat{\mathbf{r}}$};
\draw (-0.5,2.3) node[text width = 2pt] {$\hat{\boldsymbol{\varphi}}$};
\draw (0.9,1.5) node[text width = 2pt] {$\mathbf{x}_n$};

\draw[thick, red] [->](0.76,1.84) -- (-0.16,1.45);
\draw (-0.45,1.45) node[text width = 2pt] {\textbf{F}};
\draw [dashed] (-0.16,1.45) -- (0.48,1.18);
\draw [dashed] (-0.16,1.45) -- (0.1,2.11);
\draw[thick, blue] [->](0.76,1.84) -- (0.48,1.18); 
\draw[thick, blue] [->](0.76,1.84) -- (0.1,2.11); 

\draw (-0.4,0.25) node[text width = 2pt] {$\mathbf{x}$};
\draw (0,0) -- (2.12, -2.12);
%\draw (3mm,0)arc[start angle=0, end angle=67.5, radius=3mm];
%\draw (0.4,0.25) node[text width = 2pt] {$\theta$};
\draw (1,-0.75) node[text width = 2pt] {$R_{cell}$};
\draw (3.1,0) node[text width = 2pt] {$\theta_1$};
\filldraw [red] (3,0) circle [radius = 1.33pt];

\draw[thick, blue] [->](0,0) -- (-0.28,-0.66);
\draw (-0.39,-0.92) node[text width = 2pt] {$F_r\hat{\mathbf{r}}$};

\draw [dashed] (0,0) circle [radius=1.99cm];
\end{tikzpicture}
%\end{center}
%\vspace{35mm}
\caption{Force diagram showing transmission of internally generated contractile forces into translational and rotational motion. $\hat{\mathbf{r}}$ and $\hat{\boldsymbol{\varphi}}$ are radial and angular unit vectors, respectively. $\dot{\theta}_1$ is the angular velocity, $\mathbf{F}$ is a net contractile force, $F_r$ and $F_{\varphi}$ are radial and tangential components of $\mathbf{F}$, $\mathbf{x}$ and $R_{cell}$ are cell center and radius, respectively.}
\label{fig: five}
%\vspace{-5mm}
\end{wrapfigure}
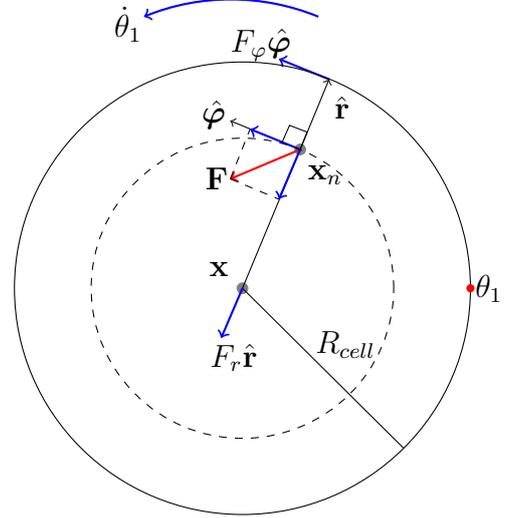
The force $\mathbf{F}$ along the radial vector $\hat{\mathbf{r}}(\mathbf{x}_n)$ is acting on the cell center, thereby inducing translational motion (see Figure \ref{fig: five}). On the other hand, the rotational motion is produced due to $\mathbf{F}$ acting along the tangential vector $\hat{\boldsymbol{\varphi}}(\mathbf{x}_n)$.
The radial and tangential components of the force $\mathbf{F}$ are given by:
\begin{align*}
F_r &:= \mathbf{F}(\mathbf{x}_n,\theta_1,\mathbf{Y})\cdot\hat{\mathbf{r}}(\mathbf{x}_n)\\
F_{\varphi} &:= \mathbf{F}(\mathbf{x}_n,\theta_1,\mathbf{Y})\cdot\hat{\boldsymbol{\varphi}}(\mathbf{x}_n),
\end{align*}
where $\mathbf{x}_n=(x_{n,1},x_{n,2})$ and 
\begin{align*}
\hat{\mathbf{r}}(\mathbf{x}_n) = \frac{\mathbf{x}_n}{\lVert\mathbf{x}_n\rVert}, && \hat{\boldsymbol{\varphi}}(\mathbf{x}_n) = \left(-\frac{x_{n,2}}{\lVert\mathbf{x}_n\rVert},\frac{x_{n,1}}{\lVert\mathbf{x}_n\rVert}\right)^T.
\end{align*}
Note that the characteristic Reynolds number $Re$ is given by 
\begin{align*}
Re = \frac{\rho\cdot s\cdot L}{\nu} \sim10^{-6}-10^{-4},
\end{align*} 
where we assumed the surrounding fluid is water (with corresponding values for density $\rho$ and viscosity $\nu$, and that characteristic cell speed $s$ and size $L$ are $0.1-1\mu m/s$, $L=10-50\mu m$, respectively. Thus, neglecting inertia, we have:
\begin{align}\label{eq: 3}
\dot{\mathbf{x}} &=\beta_{ECM}^{-1}\mathbf{F}(\mathbf{x}_n,\theta_1,\mathbf{Y})\cdot\hat{\mathbf{r}}(\mathbf{x}_n)\hat{\mathbf{r}}(\mathbf{x}_n)\nonumber \\ 
\dot{\mathbf{x}}_n &= \beta_{cell}^{-1}\mathbf{F}(\mathbf{x}_n,\theta_1,\mathbf{Y})\nonumber\\
\dot{\theta}_1 &=  \beta_{rot}^{-1}\lVert\mathbf{x}_n\rVert\mathbf{F}(\mathbf{x}_n,\theta_1,\mathbf{Y})\cdot\hat{\boldsymbol{\varphi}}(\mathbf{x}_n),
\end{align} 
where $\beta_{ECM}$ and $\beta_{rot}$ are, respectively, translational and rotational drag coefficients in the ECM (see Appendix \ref{appendix: equations of motion} for derivations). 

\FloatBarrier 
\subsection{Specification of kinematics}\label{section: kinematics summary}

It is convenient to transform the system above into nondimensional form. In order to do so, we define the following scales:
\begin{itemize}
\item The spatial and cell length scales are defined by cell radii $R_{cell}$.
\item The time scale is set by FA disassociation rate $k^0_{off}$, since FA unbinding of leads to locomotion.
\item The force scale is defined by the characteristic force $F_b$.
\end{itemize}
The constants are to be specified later. Whence we define the new variables:
\begin{align*}
\widetilde{\mathbf{x}} &:= \frac{\mathbf{x}}{R_{cell}}, &\widetilde{\mathbf{x}}_n &:= \frac{\mathbf{x}_n}{R_{cell}}, &\tilde{t} &:= k^0_{off}t
\end{align*}
and transform $\mathbf{F}_i$ from (\ref{eq: Fi}):
\begin{align}\label{eq: new F_i}
\widetilde{\mathbf{F}}_i := \frac{\mathbf{F}_i}{F_b}=
\begin{cases}
\left(\widetilde{T}_i + \widetilde{EA}\frac{\widetilde{L}_i-\widetilde{L}_0}{\widetilde{L}_0}\right)\widetilde{\mathbf{e}}_i,\phantom{abc} \widetilde{L}_0<\widetilde{L}_i\\
\widetilde{T}_i\widetilde{\mathbf{e}}_i,\phantom{asd1sdasdaabcsaa}\widetilde{L}_c\leq \widetilde{L}_i\leq \widetilde{L}_0\\
\frac{\widetilde{L}_i-\widetilde{L}_c+\widetilde{\delta}}{\widetilde{\delta}}\widetilde{T}_i\widetilde{\mathbf{e}}_i, \phantom{sdaasdssda} \widetilde{L}_c-\widetilde{\delta}\leq \widetilde{L}_i<\widetilde{L}_c\\
0 \phantom{sdaasdAADSADAdA} \widetilde{L}_i<\widetilde{L}_c-\widetilde{\delta},
\end{cases}
\end{align}
where
\begin{align*}
\widetilde{L}_i&= \frac{L_i}{R_{cell}}, &\widetilde{L}_0 &= \frac{L_0}{R_{cell}}, &\widetilde{L}_c&= \frac{L_0}{R_{cell}}, &\widetilde{\delta} &= \frac{\delta}{R_{cell}},\\ 
\widetilde{T}_i &= \frac{T_i}{F_b}, &\widetilde{EA} &=\frac{EA}{F_b} &\widetilde{\mathbf{e}}_i&= \frac{\widetilde{\mathbf{x}}_n-\widetilde{\mathbf{x}}_i}{\widetilde{L}_i}, & \widetilde{\mathbf{x}}_i&=\frac{\mathbf{x}_i}{R_{cell}}.
%\tilde{L}^2_i &= \frac{L^2_i}{R^2_{cell}} = \tilde{r}^2+1-2\tilde{r}cos(\theta-\theta_i) &\tilde{\mathbf{e}}_i &= \frac{\tilde{r}\hat{r}(\theta)-\hat{r}(\theta_i)}{\tilde{L}_i}
\end{align*}
Note that we have $\widetilde{\mathbf{x}}_n\in\widetilde{\Omega}_{cell} := \left\lbrace (x,y)\in\mathbb{R}^2\phantom{,}|\phantom{,}x^2+y^2\leq1\right\rbrace$ and $\widetilde{\mathbf{x}}_i \in\mathbb{S}^1$.

Let 
\begin{align*}
\widetilde{\mathbf{F}} &:= \mathbf{F}/F_b, &\widetilde{\beta}_{cell} &:= \frac{k^0_{off}R_{cell}}{F_b}\beta_{cell}, &\widetilde{\beta}_{ECM} &:= \frac{k^0_{off}R_{cell}}{F_b}\beta_{ECM}, &\widetilde{\beta}_{rot} &:= \frac{k^0_{off}}{R_{cell}F_b}\beta_{rot}. 
\end{align*}
To complete the specification of cell kinematics between adhesion events, we introduce another discrete variable $\mu(t)\in\left\lbrace 0,1\right\rbrace$:
\begin{align*}
\mu=
\begin{cases*}
1, \text{if the last event was unbinding}\\
0, \text{if the last event was binding}.
\end{cases*}
\end{align*}
Then, plugging in (\ref{eq: binding ODEs}), (\ref{eq: 3}) the rescaled quantities and dropping tildes, it follows that the system evolves according to the following ODE system between the FA events:
\begin{align}\label{eq: system ODE}
\dot{\mathbf{x}} &=\mu\beta_{ECM}^{-1}\mathbf{F}(\mathbf{x}_n,\theta_1,\mathbf{Y})\cdot\hat{\mathbf{r}}(\mathbf{x}_n)\hat{\mathbf{r}}(\mathbf{x}_n)\nonumber \\ 
\dot{\mathbf{x}}_n &= \beta_{cell}^{-1}\mathbf{F}(\mathbf{x}_n,\theta_1,\mathbf{Y})\nonumber\\
\dot{\theta}_1 &=  \mu\beta_{rot}^{-1}\lVert\mathbf{x}_n\rVert\mathbf{F}(\mathbf{x}_n,\theta_1,\mathbf{Y})\cdot\hat{\boldsymbol{\varphi}}(\mathbf{x}_n)
\end{align}
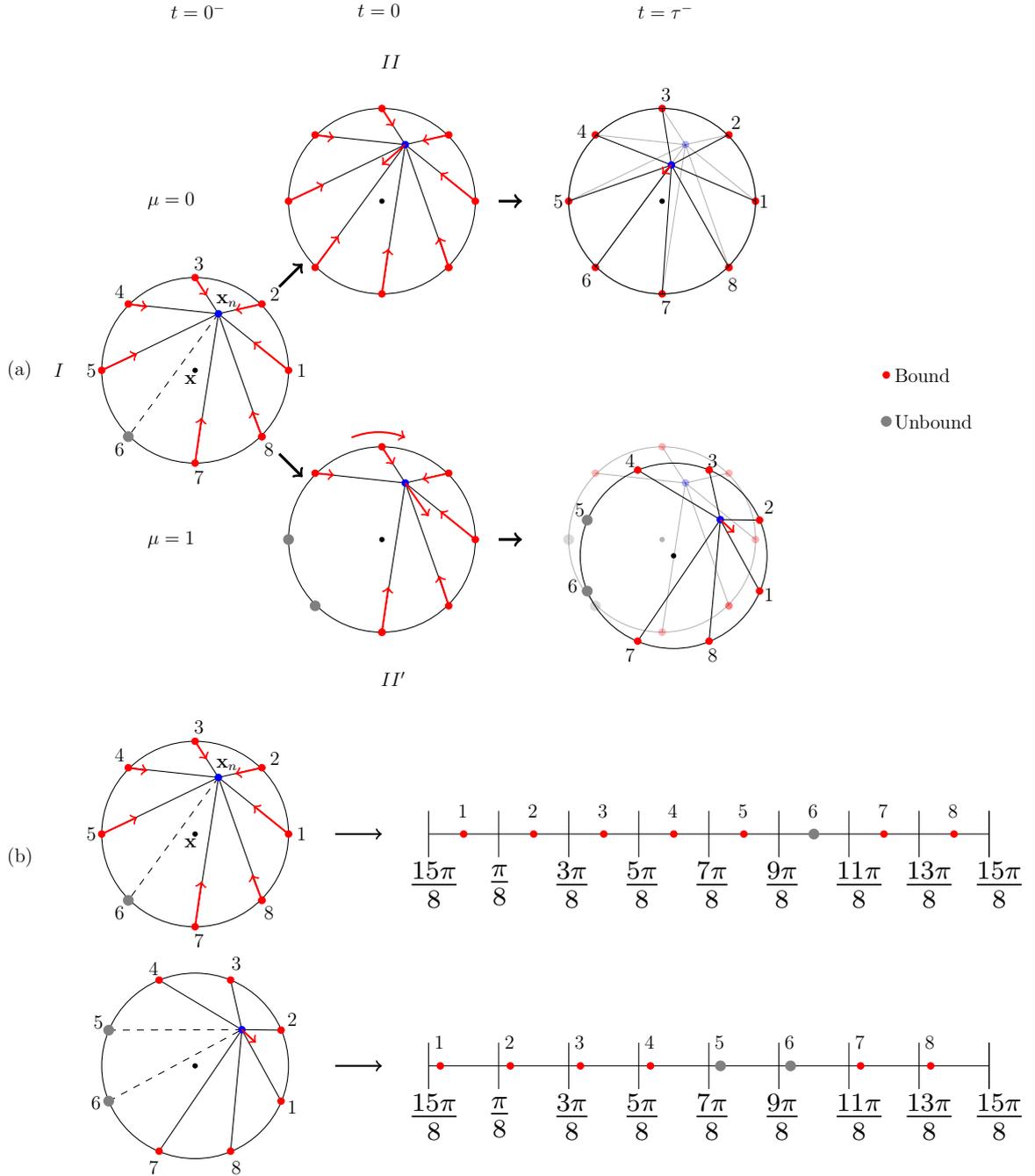
\begin{figure}[H]
\begin{tikzpicture}[scale=0.7,transform shape,every node/.style={scale=1}]
\draw (0.5,1.22) -- (2,0);
\draw (0.5,1.22) -- (0,2);
\draw (0.5,1.22) -- (-2,0);
\draw (0.5,1.22) -- (0,-2);
\draw (0.5,1.22) -- (1.41,1.43);
\draw (0.5,1.22) -- (-1.43,1.43);
\draw (0.5,1.22) -- (1.43,-1.43);
%\draw (0.5,1.22) -- (-1.43,-1.43);
\draw [dashed] (0.5,1.22)-- (-1.43,-1.43);
\draw (0,0) circle [radius=2cm];
\filldraw [blue] (0.5,1.22) circle [radius=2pt];

\filldraw [red] (2,0) circle [radius=2pt];
\filldraw [red] (0,2) circle [radius=2pt];
\filldraw [red] (0,-2) circle [radius=2pt];
\filldraw [red] (-2,0) circle [radius=2pt];
\filldraw [red] (1.43,1.43) circle [radius=2pt];
\filldraw [red] (-1.43,1.43) circle [radius=2pt];
\filldraw [red] (1.43,-1.43) circle [radius=2pt];
\filldraw [red] (1.43,-1.43) circle [radius=2pt];
\filldraw [gray] (-1.43,-1.43) circle [radius=3pt];

\filldraw [black] (0,0) circle [radius=1.33pt];
\draw (-0.2,-0.2) node[text width = 1.33pt] {$\mathbf{x}$};
\draw [thick, red][->](2,0) -- (1.25,0.61);
\draw [thick, red][->](1.43,1.43) -- (0.86,1.3);
\draw [thick, red][->](0,2) -- (0.25,1.61);
\draw [thick, red][->](-1.43,1.43) -- (-1.02,1.37);
\draw [thick, red][->](-2,0) -- (-1.24,0.36);
\draw [thick, red][->](0,-2) -- (0.14,-1.02);
%\draw [thick, red][->](-1.43,-1.43) -- (-0.93,-0.75);

\draw [thick, red][->](1.43,-1.43) -- (1.22,-0.88);

\draw (0.48,1.5) node[text width = 1.33pt] {$\mathbf{x}_n$};

\draw [very thick][->] (1.8,1.8) -- (2.3,2.3);
\draw [very thick][->] (1.8,-1.8) -- (2.3,-2.3);

\draw (2.2,0) node[text width = 1.33pt] {$1$};
\draw (1.63,1.63) node[text width = 1.33pt] {$2$};
\draw (0,2.3) node[text width = 1.33pt] {$3$};
\draw (-1.7,1.7) node[text width = 1.33pt] {$4$};
\draw (-2.3,0) node[text width = 1.33pt] {$5$};
\draw (-1.7,-1.7) node[text width = 1.33pt] {$6$};
\draw (0,-2.3) node[text width = 1.33pt] {$7$};
\draw (1.5,-1.7) node[text width = 1.33pt] {$8$};

\draw (-3,0) node[text width = 1.33pt] {$I$};
\draw (-4,0) node[text width = 1.33pt] {(a)};
% % %II
\draw [yshift=-10pt](4.5,5.22) -- (6,4);
\draw [yshift=-10pt](4.5,5.22) -- (4,6);
\draw [yshift=-10pt](4.5,5.22) -- (2,4);
\draw [yshift=-10pt](4.5,5.22) -- (4,2);
\draw [yshift=-10pt](4.5,5.22) -- (5.41,5.43);
\draw [yshift=-10pt](4.5,5.22) -- (2.57,5.43);
\draw [yshift=-10pt](4.5,5.22) -- (5.43,2.57);
\draw [yshift=-10pt](4.5,5.22) -- (2.57,2.57);
\draw [yshift=-10pt](4,4) circle [radius=2cm];
\filldraw [yshift=-10pt][blue] (4.5,5.22) circle [radius=2pt];

\filldraw [yshift=-10pt][red] (6,4) circle [radius=2pt];
\filldraw [yshift=-10pt][red] (4,6) circle [radius=2pt];
\filldraw [yshift=-10pt][red] (4,2) circle [radius=2pt];
\filldraw [yshift=-10pt][red] (2,4) circle [radius=2pt];
\filldraw [yshift=-10pt][red] (5.43,5.43) circle [radius=2pt];
\filldraw [yshift=-10pt][red] (2.57,5.43) circle [radius=2pt];
\filldraw [yshift=-10pt][red] (5.43,2.57) circle [radius=2pt];
\filldraw [yshift=-10pt][red] (5.43,2.57) circle [radius=2pt];
\filldraw [yshift=-10pt][red] (2.57,2.57) circle [radius=2pt];

\filldraw [yshift=-10pt][black] (4,4) circle [radius=1.33pt];

\draw [yshift=-10pt][thick, red][->](6,4) -- (5.25,4.61);
\draw [yshift=-10pt][thick, red][->](5.43,5.43) -- (4.86,5.3);
\draw [yshift=-10pt][thick, red][->](4,6) -- (4.25,5.61);
\draw [yshift=-10pt][thick, red][->](2.57,5.43) -- (2.98,5.37);
\draw [yshift=-10pt][thick, red][->](2,4) -- (2.76,4.36);
\draw [yshift=-10pt][thick, red][->](4,2) -- (4.14,2.98);
\draw [yshift=-10pt][thick, red][->](2.57,2.57) -- (3.07,3.25);
\draw [yshift=-10pt][thick, red][->](5.43,2.57) -- (5.22,3.22);

\draw [yshift=-10pt][very thick][->] (6.5,4) -- (7,4);

\draw [yshift=-10pt][thick,red][->] (4.5,5.22) -- (4,4.78);

%\draw [yshift=-10pt][->,thick, red] (4.5,6.2)arc[start angle=67.5, end angle=112.5, radius=1.5cm];
\draw [yshift=-10pt](4,7) node[text width = 1.33pt] {$II$};
% % %III
\draw [yshift=-10pt][opacity=0.3](10.5,5.22) -- (12,4);
\draw [yshift=-10pt][opacity=0.3](10.5,5.22) -- (10,6);
\draw [yshift=-10pt][opacity=0.3](10.5,5.22) -- (8,4);
\draw [yshift=-10pt][opacity=0.3](10.5,5.22) -- (10,2);
\draw [yshift=-10pt][opacity=0.3](10.5,5.22) -- (11.41,5.43);
\draw [yshift=-10pt][opacity=0.3](10.5,5.22) -- (8.57,5.43);
\draw [yshift=-10pt][opacity=0.3](10.5,5.22) -- (11.43,2.57);
\draw [yshift=-10pt][opacity=0.3](10.5,5.22) -- (8.57,2.57);
\draw [yshift=-10pt][opacity=0.3](10,4) circle [radius=2cm];
\filldraw [yshift=-10pt][opacity=0.3][blue] (10.5,5.22) circle [radius=2pt];

\filldraw [yshift=-10pt][red] (12,4) circle [radius=2pt];
\filldraw [yshift=-10pt][red] (10,6) circle [radius=2pt];
\filldraw [yshift=-10pt][red] (10,2) circle [radius=2pt];
\filldraw [yshift=-10pt][red] (8,4) circle [radius=2pt];
\filldraw [yshift=-10pt][red] (11.43,5.43) circle [radius=2pt];
\filldraw [yshift=-10pt][red] (8.57,5.43) circle [radius=2pt];
\filldraw [yshift=-10pt][red] (11.43,2.57) circle [radius=2pt];
\filldraw [yshift=-10pt][red] (11.43,2.57) circle [radius=2pt];
\filldraw [yshift=-10pt][red] (8.57,2.57) circle [radius=2pt];

\filldraw [yshift=-10pt][black][opacity=0.3] (10,4) circle [radius=1.33pt];

% % %III
%[xshift = -10pt,yshift = -10pt]
\draw [yshift = -10pt](10.2,4.78) -- (12,4);
\draw [yshift = -10pt](10.2,4.78) -- (10,6);
\draw [yshift = -10pt](10.2,4.78) -- (8,4);
\draw [yshift = -10pt](10.2,4.78) -- (10,2);
\draw [yshift = -10pt](10.2,4.78) -- (11.41,5.43);
\draw [yshift = -10pt](10.2,4.78) -- (8.57,5.43);
\draw [yshift = -10pt](10.2,4.78) -- (11.43,2.57);
\draw [yshift = -10pt](10.2,4.78) -- (8.57,2.57);
\draw [yshift = -10pt](10,4) circle [radius=2cm];
\filldraw [yshift = -10pt][blue] (10.2,4.78) circle [radius=2pt];

\draw [yshift = -10pt] (12.1,4) node[text width = 1.33pt] {$1$};
\draw [yshift = -10pt] (10,6.3) node[text width = 1.33pt] {$3$};
\draw [yshift = -10pt] (7.7,4) node[text width = 1.33pt] {$5$};
\draw [yshift = -10pt] (10,1.7) node[text width = 1.33pt] {$7$};
\draw [yshift = -10pt] (11.5,5.7) node[text width = 1.33pt] {$2$};
\draw [yshift = -10pt] (8.2,5.5) node[text width = 1.33pt] {$4$};
\draw [yshift = -10pt] (11.43,2.2) node[text width = 1.33pt] {$8$};
\draw [yshift = -10pt] (8.3,2.3) node[text width = 1.33pt] {$6$};

\filldraw [black][yshift = -10pt] (10,4) circle [radius=1.33pt];
\draw [yshift = -10pt][thick,red][->] (10.2,4.78) -- (10,4.58);

%\draw [yshift=-10pt](10,7) node[text width = 1.33pt] {$III$};
% % %II'
\draw [yshift=10pt](4.5,-2.78) -- (6,-4);
\draw [yshift=10pt](4.5,-2.78) -- (4,-2);
%\draw (4.5,-2.78) -- (2,-4);
\draw [yshift=10pt](4.5,-2.78) -- (4,-6);
\draw [yshift=10pt](4.5,-2.78) -- (5.41,-2.57);
\draw [yshift=10pt](4.5,-2.78) -- (2.57,-2.57);
\draw [yshift=10pt](4.5,-2.78) -- (5.43,-5.43);
%\draw (4.5,-2.78) -- (2.57,-5.43);
\draw [yshift=10pt](4,-4) circle [radius=2cm];
\filldraw [yshift=10pt][blue] (4.5,-2.78) circle [radius=2pt];

\filldraw [yshift=10pt][red] (6,-4) circle [radius=2pt];
\filldraw [yshift=10pt][red] (4,-6) circle [radius=2pt];
\filldraw [yshift=10pt][red] (4,-6) circle [radius=2pt];
\filldraw [yshift=10pt][gray] (2,-4) circle [radius=3pt];
\filldraw [yshift=10pt][red] (5.43,-2.57) circle [radius=2pt];
\filldraw [yshift=10pt][red] (2.57,-2.57) circle [radius=2pt];
\filldraw [yshift=10pt][red] (4,-2) circle [radius=2pt];
\filldraw [yshift=10pt][red] (5.43,-5.43) circle [radius=2pt];
\filldraw [yshift=10pt][gray] (2.57,-5.43) circle [radius=3pt];

\filldraw [yshift=10pt][black] (4,-4) circle [radius=1.33pt];

\draw [yshift=10pt][thick, red][->](6,-4) -- (5.25,-3.39);
\draw [yshift=10pt][thick, red][->](5.43,-2.57) -- (4.86,-2.7);
\draw [yshift=10pt][thick, red][->](4,-2) -- (4.25,-2.39);
\draw [yshift=10pt][thick, red][->](2.57,-2.57) -- (2.98,-2.61);
%\draw [thick, red][->](2,-4) -- (2.76,-3.64);
\draw [yshift=10pt][thick, red][->](4,-6) -- (4.14,-5.02);
%\draw [thick, red][->](2.57,-5.43) -- (3.07,-4.75);
\draw [yshift=10pt][thick, red][->](5.43,-5.43) -- (5.22,-4.78);

\draw [yshift=10pt][very thick][->] (6.5,-4) -- (7,-4);

\draw [yshift=10pt][thick,red][->] (4.5,-2.78) -- (5,-3.5);

\draw [yshift=10pt][<-,thick, red] (4.5,-1.8)arc[start angle=67.5, end angle=112.5, radius=1.5cm];
\draw [yshift=10pt](4,-7) node[text width = 1.33pt] {$II'$};
% % %III'
\draw [yshift=10pt][opacity=0.3](10.5,-2.78) -- (8.57,-2.57);
\draw [yshift=10pt][opacity=0.3](10.5,-2.78) -- (12,-4);
%\draw [opacity=0.3](10.5,-2.78) -- (8,-4);
\draw [yshift=10pt][opacity=0.3](10.5,-2.78) -- (10,-6);
\draw [yshift=10pt][opacity=0.3](10.5,-2.78) -- (11.41,-2.57);
\draw [yshift=10pt][opacity=0.3](10.5,-2.78) -- (10,-2);
\draw [yshift=10pt][opacity=0.3](10.5,-2.78) -- (11.43,-5.43);
%\draw [opacity=0.3](10.5,-2.78) -- (8.57,-5.43);
\draw [yshift=10pt][opacity=0.3](10,-4) circle [radius=2cm];
\filldraw [yshift=10pt][opacity=0.3][blue] (10.5,-2.78) circle [radius=2pt];

\filldraw [yshift=10pt][opacity=0.3][red] (12,-4) circle [radius=2pt];
\filldraw [yshift=10pt][opacity=0.3][red] (10,-2) circle [radius=2pt];
\filldraw [yshift=10pt][red][opacity=0.3] (10,-6) circle [radius=2pt];
\filldraw [yshift=10pt][gray][opacity=0.3] (8,-4) circle [radius=3pt];
\filldraw [yshift=10pt][red][opacity=0.3] (11.43,-2.57) circle [radius=2pt];
\filldraw [yshift=10pt][red][opacity=0.3] (8.57,-2.57) circle [radius=2pt];
\filldraw [yshift=10pt][red][opacity=0.3] (11.43,-5.43) circle [radius=2pt];
\filldraw [yshift=10pt][red][opacity=0.3] (11.43,-5.43) circle [radius=2pt];
\filldraw [yshift=10pt][gray][opacity=0.3] (8.57,-5.43) circle [radius=3pt];

\filldraw [yshift=10pt][black][opacity=0.3] (10,-4) circle [radius=1.33pt];

% % %III'
%[xshift = -10pt,yshift = -10pt]
\draw [xshift = 7pt,yshift = -0pt](11,-3.22) -- (11.84,-4.76);
\draw [xshift = 7pt,yshift = -0pt](11,-3.22) -- (11.84,-3.23);
%\draw [xshift = 7pt,yshift = -10pt](10.5,-2.78) -- (8,-4);
\draw [xshift = 7pt,yshift = -0pt](11,-3.22) -- (10.76,-2.15);
\draw [xshift = 7pt,yshift = -0pt](11,-3.22) -- (9.23,-2.15);
%\draw [xshift = 7pt,yshift = -10pt](11,-3.22) -- (8.15,-3.23);
%\draw [xshift = 7pt,yshift = -10pt](11,-3.22) -- (8.15,-4.76);
\draw [xshift = 7pt,yshift = -0pt](11,-3.22) -- (9.23,-5.84);
\draw [xshift = 7pt,yshift = -0pt](11,-3.22) -- (10.76,-5.84);
%\draw [opacity=0.3](10.5,-2.78) -- (8.57,-5.43);
\draw [xshift = 7pt,yshift = -0pt](10,-4) circle [radius=2cm];
\filldraw [xshift = 7pt,yshift = -0pt][blue] (11,-3.22) circle [radius=2pt];

\filldraw [xshift = 7pt,yshift = -0pt][red] (11.84,-4.76) circle [radius=2pt];
\filldraw [xshift = 7pt,yshift = -0pt][red] (11.84,-3.23) circle [radius=2pt];
\filldraw [red][xshift = 7pt,yshift = -0pt] (10.76,-2.15) circle [radius=2pt];
\filldraw [red][xshift = 7pt,yshift = -0pt] (9.23,-2.15) circle [radius=2pt];
\filldraw [gray][xshift = 7pt,yshift = -0pt] (8.15,-3.23) circle [radius=3pt];
\filldraw [red][xshift = 7pt,yshift = -0pt] (9.23,-5.84) circle [radius=2pt];
\filldraw [red][xshift = 7pt,yshift = -0pt] (10.76,-5.84) circle [radius=2pt];
%\filldraw [red][xshift = 7pt,yshift = -10pt] (11.43,-5.43) circle [radius=2pt];
\filldraw [gray][xshift = 7pt,yshift = -0pt] (8.15,-4.76) circle [radius=3pt];

\filldraw [black][xshift = 7pt,yshift = -0pt](10,-4) circle [radius=1.33pt];

\draw [yshift=10pt](12.2,-5.2) node[text width = 1pt] {$1$};
\draw [yshift=10pt](12.2,-3.3) node[text width = 1pt] {$2$};
\draw [yshift=10pt](11,-2.3) node[text width = 1pt] {$3$};
\draw [yshift=10pt](9.23,-2.3) node[text width = 1pt] {$4$};
\draw [yshift=10pt](8.15,-3.35) node[text width = 1pt] {$5$};
\draw [yshift=10pt](8,-5) node[text width = 1pt] {$6$};
\draw [yshift=10pt](9.23,-6.5) node[text width = 1pt] {$7$};
\draw [yshift=10pt](11,-6.5) node[text width = 1pt] {$8$};

\draw [xshift = 7pt,yshift = -0pt,red,->,thick](11,-3.22) -- (11.3,-3.5);

%\draw [yshift=10pt](10,-7) node[text width = 1.33pt] {$III'$};

\draw [yshift=10pt](9.5,7.4) node[text width = 1pt] {\text{$t=\tau^-$}};
\draw [yshift=10pt](3.5,7.4) node[text width = 1pt] {\text{$t=0$}};
\draw [yshift=10pt](-0.5,7.4) node[text width = 1pt] {\text{$t=0^-$}};

\draw [yshift=10pt](-1,-4) node[text width = 1pt] {\text{$\mu=1$}};
\draw [yshift=-10pt](-1,4) node[text width = 1pt] {\text{$\mu=0$}};

% % %part b)

\draw [yshift = -10cm](0.5,1.22) -- (2,0);
\draw [yshift = -10cm](0.5,1.22) -- (0,2);
\draw [yshift = -10cm](0.5,1.22) -- (-2,0);
\draw [yshift = -10cm](0.5,1.22) -- (0,-2);
\draw [yshift = -10cm](0.5,1.22) -- (1.41,1.43);
\draw [yshift = -10cm](0.5,1.22) -- (-1.43,1.43);
\draw [yshift = -10cm](0.5,1.22) -- (1.43,-1.43);
%\draw (0.5,1.22) -- (-1.43,-1.43);
\draw [yshift = -10cm][dashed] (0.5,1.22)-- (-1.43,-1.43);
\draw [yshift = -10cm](0,0) circle [radius=2cm];
\filldraw [yshift = -10cm][blue] (0.5,1.22) circle [radius=2pt];

\filldraw [yshift = -10cm][red] (2,0) circle [radius=2pt];
\filldraw [yshift = -10cm][red] (0,2) circle [radius=2pt];
\filldraw [yshift = -10cm][red] (0,-2) circle [radius=2pt];
\filldraw [yshift = -10cm][red] (-2,0) circle [radius=2pt];
\filldraw [yshift = -10cm][red] (1.43,1.43) circle [radius=2pt];
\filldraw [yshift = -10cm][red] (-1.43,1.43) circle [radius=2pt];
\filldraw [yshift = -10cm][red] (1.43,-1.43) circle [radius=2pt];
\filldraw [yshift = -10cm][red] (1.43,-1.43) circle [radius=2pt];
\filldraw [yshift = -10cm][gray] (-1.43,-1.43) circle [radius=3pt];

\filldraw [yshift = -10cm][black] (0,0) circle [radius=1.33pt];
\draw [yshift = -10cm](-0.2,-0.2) node[text width = 1.33pt] {$\mathbf{x}$};
\draw [yshift = -10cm][thick, red][->](2,0) -- (1.25,0.61);
\draw [yshift = -10cm][thick, red][->](1.43,1.43) -- (0.86,1.3);
\draw [yshift = -10cm][thick, red][->](0,2) -- (0.25,1.61);
\draw [yshift = -10cm][thick, red][->](-1.43,1.43) -- (-1.02,1.37);
\draw [yshift = -10cm][thick, red][->](-2,0) -- (-1.24,0.36);
\draw [yshift = -10cm][thick, red][->](0,-2) -- (0.14,-1.02);
%\draw [thick, red][->](-1.43,-1.43) -- (-0.93,-0.75);

\draw [yshift = -10cm][thick, red][->](1.43,-1.43) -- (1.22,-0.88);

\draw [yshift = -10cm](0.48,1.5) node[text width = 1.33pt] {$\mathbf{x}_n$};

\draw [yshift = -10cm](2.2,0) node[text width = 1.33pt] {$1$};
\draw [yshift = -10cm](1.63,1.63) node[text width = 1.33pt] {$2$};
\draw [yshift = -10cm](0,2.3) node[text width = 1.33pt] {$3$};
\draw [yshift = -10cm](-1.7,1.7) node[text width = 1.33pt] {$4$};
\draw [yshift = -10cm](-2.3,0) node[text width = 1.33pt] {$5$};
\draw [yshift = -10cm](-1.7,-1.7) node[text width = 1.33pt] {$6$};
\draw [yshift = -10cm](0,-2.3) node[text width = 1.33pt] {$7$};
\draw [yshift = -10cm](1.5,-1.7) node[text width = 1.33pt] {$8$};
\draw [yshift = -10cm][thick,black][->](3,0) -- (4,0);

\draw [yshift = -10cm,xshift = -0cm][black](5,0) -- (17,0);
\draw [yshift = -10cm,xshift = -0cm](5,-0.5) -- (5,0.5);
\draw [yshift = -10cm,xshift = -0cm](17,-0.5) -- (17,0.5);

\foreach \i in {1,...,8}
	\draw [yshift = -10cm,xshift=\i*1.5cm](5,-0.5)--(5,0.5); 
\foreach \i in {0,...,7}
	\filldraw [red,yshift = -10cm,xshift=\i*1.5cm+0.75cm](5,0) circle [radius=2pt];
\filldraw [gray,yshift = -10cm, xshift = 8.25cm] (5,0) circle [radius=3pt];
\draw [yshift = -10cm,xshift = 0.75cm] (4.9,0.5) node[text width = 1.33pt] {$1$};
\draw [yshift = -10cm,xshift = 2.25cm] (4.9,0.5) node[text width = 1.33pt] {$2$};
\draw [yshift = -10cm,xshift = 3.75cm] (4.9,0.5) node[text width = 1.33pt] {$3$};
\draw [yshift = -10cm,xshift = 5.25cm] (4.9,0.5) node[text width = 1.33pt] {$4$};
\draw [yshift = -10cm,xshift = 6.75cm] (4.9,0.5) node[text width = 1.33pt] {$5$};
\draw [yshift = -10cm,xshift = 8.25cm] (4.9,0.5) node[text width = 1.33pt] {$6$};
\draw [yshift = -10cm,xshift = 9.75cm] (4.9,0.5) node[text width = 1.33pt] {$7$};
\draw [yshift = -10cm,xshift = 11.25cm] (4.9,0.5) node[text width = 1.33pt] {$8$};

\draw [yshift = -10cm,xshift = -0cm] (4.6,-1.1) node[scale = 2, text width = 1.33pt] {$\frac{15\pi}{8}$};
\draw [yshift = -10cm,xshift = 1.5cm] (4.8,-1.1) node[scale = 2, text width = 1.33pt] {$\frac{\pi}{8}$};
\draw [yshift = -10cm,xshift = 3.0cm] (4.7,-1.1) node[scale = 2, text width = 1.33pt] {$\frac{3\pi}{8}$};
\draw [yshift = -10cm,xshift = 4.5cm] (4.7,-1.1) node[scale = 2, text width = 1.33pt] {$\frac{5\pi}{8}$};
\draw [yshift = -10cm,xshift = 6.0cm] (4.7,-1.1) node[scale = 2, text width = 1.33pt] {$\frac{7\pi}{8}$};
\draw [yshift = -10cm,xshift = 7.5cm] (4.7,-1.1) node[scale = 2, text width = 1.33pt] {$\frac{9\pi}{8}$};
\draw [yshift = -10cm,xshift = 9.0cm] (4.7,-1.1) node[scale = 2, text width = 1.33pt] {$\frac{11\pi}{8}$};
\draw [yshift = -10cm,xshift = 10.5cm] (4.7,-1.1) node[scale = 2, text width = 1.33pt] {$\frac{13\pi}{8}$};
\draw [yshift = -10cm,xshift = 12.0cm] (4.7,-1.1) node[scale = 2, text width = 1.33pt] {$\frac{15\pi}{8}$};

%\draw [yshift = -8cm,xshift =-4cm](-0.5,3) node[text width = 1pt] {\text{$t=0^-$}};

% %Second case
\draw [xshift = -10cm,yshift = -11cm](11,-3.22) -- (11.84,-4.76);
\draw [xshift = -10cm,yshift = -11cm](11,-3.22) -- (11.84,-3.23);
%\draw [xshift = 7pt,yshift = -10pt](10.5,-2.78) -- (8,-4);
\draw [xshift = -10cm,yshift = -11cm](11,-3.22) -- (10.76,-2.15);
\draw [xshift = -10cm,yshift = -11cm](11,-3.22) -- (9.23,-2.15);
%\draw [xshift = 7pt,yshift = -10pt](11,-3.22) -- (8.15,-3.23);
%\draw [xshift = 7pt,yshift = -10pt](11,-3.22) -- (8.15,-4.76);
\draw [xshift = -10cm,yshift = -11cm](11,-3.22) -- (9.23,-5.84);
\draw [xshift = -10cm,yshift = -11cm](11,-3.22) -- (10.76,-5.84);
\draw [xshift = -10cm,yshift = -11cm,dashed](8.15,-4.76)--(11,-3.22);
\draw [xshift = -10cm,yshift = -11cm,dashed](8.15,-3.23)--(11,-3.22);
\draw [xshift = -10cm,yshift = -11cm](10,-4) circle [radius=2cm];
\filldraw [xshift = -10cm,yshift = -11cm][blue] (11,-3.22) circle [radius=2pt];

\filldraw [xshift = -10cm,yshift = -11cm][red] (11.84,-4.76) circle [radius=2pt];
\filldraw [xshift = -10cm,yshift = -11cm][red] (11.84,-3.23) circle [radius=2pt];
\filldraw [red][xshift = -10cm,yshift = -11cm] (10.76,-2.15) circle [radius=2pt];
\filldraw [red][xshift = -10cm,yshift = -11cm] (9.23,-2.15) circle [radius=2pt];
\filldraw [gray][xshift = -10cm,yshift = -11cm] (8.15,-3.23) circle [radius=3pt];
\filldraw [red][xshift = -10cm,yshift = -11cm] (9.23,-5.84) circle [radius=2pt];
\filldraw [red][xshift = -10cm,yshift = -11cm] (10.76,-5.84) circle [radius=2pt];
%\filldraw [red][xshift = 7pt,yshift = -10pt] (11.43,-5.43) circle [radius=2pt];
\filldraw [gray][xshift = -10cm,yshift = -11cm] (8.15,-4.76) circle [radius=3pt];

\filldraw [black][xshift = -10cm,yshift = -11cm](10,-4) circle [radius=1.33pt];

\draw [xshift = -10.2cm,yshift=-10.7cm](12.2,-5.2) node[text width = 1.33pt] {$1$};
\draw [xshift = -10.2cm,yshift=-10.7cm](12.2,-3.3) node[text width = 1.33pt] {$2$};
\draw [xshift = -10.2cm,yshift=-10.5cm](11,-2.3) node[text width = 1.33pt] {$3$};
\draw [xshift = -10.2cm,yshift=-10.6cm](9.23,-2.3) node[text width = 1.33pt] {$4$};
\draw [xshift = -10.3cm,yshift=-10.7cm](8.15,-3.35) node[text width = 1.33pt] {$5$};
\draw [xshift = -10.2cm,yshift=-10.8cm](8,-5) node[text width = 1.33pt] {$6$};
\draw [xshift = -10.2cm,yshift=-10.7cm](9.23,-6.5) node[text width = 1.33pt] {$7$};
\draw [xshift = -10.2cm,yshift=-10.7cm](11,-6.5) node[text width = 1.33pt] {$8$};

\draw [xshift = -10cm,yshift = -11cm,red,->,thick](11,-3.22) -- (11.3,-3.5);

\draw [thick,black,yshift = -15cm][->](3,0) -- (4,0);

\draw [black,yshift = -15cm](5,0) -- (17,0);
\draw [yshift = -15cm](5,-0.5) -- (5,0.5);
\draw [yshift = -15cm](17,-0.5) -- (17,0.5);

\foreach \i in {1,...,8}
	\draw [yshift = -15cm,xshift=\i*1.5cm](5,-0.5)--(5,0.5); 
\foreach \i in {0,...,7}
	\filldraw [red,yshift = -15cm,xshift=\i*1.5cm+0.25cm](5,0) circle [radius=2pt];
\filldraw [gray, yshift = -15cm, xshift = 7.75cm] (5,0) circle [radius=3pt];
\filldraw [gray, yshift = -15cm, xshift = 6.25cm] (5,0) circle [radius=3pt];
\draw [xshift = 0.25cm, yshift = -15cm] (4.9,0.5) node[text width = 1.33pt] {$1$};
\draw [xshift = 1.75cm, yshift = -15cm] (4.9,0.5) node[text width = 1.33pt] {$2$};
\draw [xshift = 3.25cm, yshift = -15cm] (4.9,0.5) node[text width = 1.33pt] {$3$};
\draw [xshift = 4.75cm, yshift = -15cm] (4.9,0.5) node[text width = 1.33pt] {$4$};
\draw [xshift = 6.25cm, yshift = -15cm] (4.9,0.5) node[text width = 1.33pt] {$5$};
\draw [xshift = 7.75cm, yshift = -15cm] (4.9,0.5) node[text width = 1.33pt] {$6$};
\draw [xshift = 9.25cm, yshift = -15cm] (4.9,0.5) node[text width = 1.33pt] {$7$};
\draw [xshift = 10.75cm, yshift = -15cm] (4.9,0.5) node[text width = 1.33pt] {$8$};

\draw [xshift = 0cm, yshift = -15cm] (4.6,-1.1) node[scale = 2, text width = 1.33pt] {$\frac{15\pi}{8}$};
\draw [xshift = 1.5cm, yshift = -15cm] (4.8,-1.1) node[scale = 2, text width = 1.33pt] {$\frac{\pi}{8}$};
\draw [xshift = 3.0cm, yshift = -15cm] (4.7,-1.1) node[scale = 2, text width = 1.33pt] {$\frac{3\pi}{8}$};
\draw [xshift = 4.5cm, yshift = -15cm] (4.7,-1.1) node[scale = 2, text width = 1.33pt] {$\frac{5\pi}{8}$};
\draw [xshift = 6.0cm, yshift = -15cm] (4.7,-1.1) node[scale = 2, text width = 1.33pt] {$\frac{7\pi}{8}$};
\draw [xshift = 7.5cm, yshift = -15cm] (4.7,-1.1) node[scale = 2, text width = 1.33pt] {$\frac{9\pi}{8}$};
\draw [xshift = 9.0cm, yshift = -15cm] (4.7,-1.1) node[scale = 2, text width = 1.33pt] {$\frac{11\pi}{8}$};
\draw [xshift = 10.5cm, yshift = -15cm] (4.7,-1.1) node[scale = 2, text width = 1.33pt] {$\frac{13\pi}{8}$};
\draw [xshift = 12.0cm, yshift = -15cm] (4.7,-1.1) node[scale = 2, text width = 1.33pt] {$\frac{15\pi}{8}$};

\draw [yshift = -10.5cm](-4,0) node[text width = 1.33pt] {(b)};

\draw [xshift = 10.5cm, yshift = 1cm] (4.5,-1.1) node[scale = 1, text width = 1.33pt] {Bound};
\filldraw [red, yshift = 1cm, xshift = 10.5cm] (4.3,-1.1) circle [radius=2pt];
\draw [xshift = 10.5cm, yshift = 0cm] (4.5,-1.1) node[scale = 1, text width = 1.33pt] {Unbound};
\filldraw [gray, yshift = 0cm, xshift = 10.5cm] (4.3,-1.1) circle [radius=3pt];
\end{tikzpicture}
\caption{(a) Schematic representation of the migration cycle between adhesion events. Suppose that just before an event occurs at time $t=0$, the cell is in state $I$. If at time $t=0$ (de)adhesion occurs, the cell jumps into state the ($II'$)$II$ and the system evolves according to {\eqref{eq: system ODE}} until the next event occurs at time $t=\tau$, after which the cycle begins anew. The scenarios can be characterized as ``run" and ``tumble" phases in the bottom and top panels, respectively. (b) Schematic representation of the FA positions projected on cell's circumference at $t=0^-$ and $t=\tau^-$ in the top and bottom panels, respectively.}
\label{fig: motion cycle2}    
\end{figure}

\textbf{Remark.} Our assumption on constant relative distance between FA sites stems from two slightly weaker assumptions: 1) total number of adhesion sites (occupied and unoccupied) is constant; 2) there is a neighborhood around each adhesion site, in which no other site is present, and the size of this neighborhood is the same (and constant) for each site. Figure \ref{fig: motion cycle2} how it reflects on their peripheral motion. This assumption implies that in each line segment of size $2\pi/M$ (with $M=8$) there is only one FA site present, which may correspond to bound (in red) or unbound FA (in gray). See also Appendix \ref{appendix: parameters} on our discussion on the number of FA sites. 

\section{FA event model}\label{section: jump process}
In the previous section we constructed a model of cell motion between FA events. Following \cite{G92}, here we construct a stochastic model describing the random adhesion/deadhesion events and their arrival times. The discussion here differs from the standard approach of the Gillespie algorithm in \cite{G92}, as we do not assume that the propensity functions vary inappreciably between the reactions. Moreover, it provides a connection to the theory of PDMPs, as the forms of the objects, necessary to define a piecewise deterministic process (see the next section), follow from the derivations here.

\subsection{Focal adhesion events}\label{section: FA jumps}
Since there are $M$ FAs and since each FA can participate in only two reactions (binding and unbinding), then there are $2M$ total possible reactions. We adopt the following convention for enumerating reactions: reaction $j$ corresponds to a \textit{binding} reaction of the FA site $i=(j+1)/2$ if $j$ is odd; otherwise reaction $j$ corresponds to an \textit{unbinding} reaction of the FA site $i=j/2$. Let $\mathbf{Y}(t)$ be defined as before.  

Let $a_j(\mathbf{y},t)dt$ be the probability, given $\mathbf{Y}(t)=\mathbf{y}\in\left\lbrace 0,1\right\rbrace^M$, $\mu(t)$, $(\mathbf{x}(t),\mathbf{x}_n(t),\theta_1(t))$ and time $t$, that reaction $j$ will occur in the time interval $[t,t+dt)$. For clarity, we suppress here the dependence of the rate $a_j(\mathbf{y},\cdot)$ on $(\mathbf{x}(\cdot),\mathbf{x}_n(\cdot),\theta_1(\cdot))$ and $\mu(\cdot)$. We assume that the rate $a_j$ satisfies the following:
\begin{align}\label{assumption: nonzero a_j}
a_j(\mathbf{y},t) = 
\begin{cases*}
0, \text{ if $j$ is odd and $y_{(j+1)/2}=1$}\\
0, \text{ if $j$ is even and $y_{j/2}=0$}\\
\neq 0,\text{ else.}
\end{cases*}
\end{align}
That is, if the FA is (un)bound, the probability of the (un)binding reaction is zero; if the FA is (un)bound, the probability of (binding) unbinding is nonzero. This implies that $a_j(\mathbf{y},t)\neq 0$ for at least one $j\in\left\lbrace 1,\ldots,2M\right\rbrace$.\footnote{For each FA site $i\in\left\lbrace 1,\ldots,M\right\rbrace$, either $a_{2i-1}$ or $a_{2i}$ is nonzero.}

\begin{lemma}\label{lemma: no FA event}
Let $\mathbf{Y}(t)=\mathbf{y}$. Then the probability that no FA event occurs in the time interval $[t,t+dt)$ is $1-\sum_{j=1}^{2M}a_j(\mathbf{y},t)dt + o(dt)$.  
\end{lemma}
\begin{proof}
Using the definition of $a_j$, the probability that reaction $j$ does not happen is $1-a_j(\mathbf{y},t)dt$. Then, the probability that no FA reaction occurs is:
\begin{align*}
\prod_{j=1}^{2M}\left(1-a_j(\mathbf{y},t)dt\right) = 1-\sum_{j=1}^{2M}a_j(\mathbf{y},t)dt + o(dt).
\end{align*}
\end{proof}

Let $K(\tau,j|t,\mathbf{y})d\tau$ be the probability, given $\mathbf{Y}(t)=\mathbf{y}$ and $(\mathbf{x}(t),\mathbf{x}_n(t),\theta_1(t))$ at time $t$, that the \textit{next} reaction will occur in the time interval $[t+\tau,t+\tau+d\tau)$ \textit{and} will be reaction $j$. Here, again, we suppress for clarity the dependence on $\mathbf{x}$, $\mathbf{x}_n$, $\theta_1$. 

\begin{proposition}\label{proposition: next reaction probability}
Let $\tau>0$ and $\mathbf{Y}(t)=\mathbf{y}$. Then,
\begin{align*}
K(\tau,j|t,\mathbf{y})=a_j(\mathbf{y},t+\tau)\exp\left(-\int_{t}^{t+\tau}\sum_{j'=1}^{2M}a_{j'}(\mathbf{y},\tau')d\tau'\right).
\end{align*}
\end{proposition}
\begin{proof}
Let $P(\tau|t,\mathbf{y})$ denote the probability that no reaction occurs in the time interval $[t,t+\tau)$, given $\mathbf{y}$ (and $\mathbf{x}$, $\mathbf{x}_n$, $\theta_1$) at time $t$. Then, by Lemma \ref{lemma: no FA event}:
\begin{align*}
P(\tau+d\tau|t,\mathbf{y}) &= P(\tau|t,\mathbf{y})\left(1-\sum_{j=1}^{2M}a_j(\mathbf{y},t+\tau)d\tau + o(d\tau)\right)\Rightarrow\\
\frac{P(\tau+d\tau|t,\mathbf{y})-P(\tau|t,\mathbf{y})}{d\tau} &= -P(\tau|t,\mathbf{y})\sum_{j=1}^{2M}a_j(\mathbf{y},t+\tau) + P(\tau|t,\mathbf{y})\frac{o(d\tau)}{d\tau}.
\end{align*}
Letting $d\tau\rightarrow 0$ we obtain the following ODE:
\begin{align*}
\frac{d}{d\tau}P(\tau|t,\mathbf{y}) &= -P(\tau|t,\mathbf{y})\sum_{j=1}^{2M}a_j(\mathbf{y},t+\tau).
\end{align*}
Since $P(0|t,\mathbf{y})=1$, the solution $P(\tau|t,\mathbf{y})$ is given by:
\begin{align*}
P(\tau|t,\mathbf{y}) = \exp\left(-\int_{t}^{t+\tau}\sum_{j=1}^{2M}a_j(\mathbf{y},\tau')d\tau'\right).
\end{align*}
We have then:
\begin{align}\label{eq: prob until FA 1}
K(\tau,j|t,\mathbf{y}) = P(\tau|t,\mathbf{y})a_j(\mathbf{y},t+\tau) = a_j(\mathbf{y},t+\tau)\exp\left(-\int_{t}^{t+\tau}\sum_{j'=1}^{2M}a_{j'}(\mathbf{y},\tau')d\tau'\right).
\end{align}
\end{proof}

Let $K_{time}(\tau|t,\mathbf{y})d\tau$ be the probability that the next reaction will occur in the time interval $[t+\tau,t+\tau+d\tau)$, given $\mathbf{Y}(t)=\mathbf{y}$ and $(\mathbf{x}(t),\mathbf{x}_n(t),\theta_1(t))$ at time $t$. 

Let $K_{index}(j|\tau,t,\mathbf{y})$ be the probability that the index of the next reaction is $j$ given $\mathbf{Y}(t)=\mathbf{y}$, $(\mathbf{x}(t),\mathbf{x}_n(t),\theta_1(t))$ at time $t$ and given that the reaction will occur at time $t+\tau$.

By elementary probability theory (using the definition of conditional probability), we know that
\begin{align*}
K(\tau,j|t,\mathbf{y})d\tau=K_{index}(j|\tau,t,\mathbf{y})K_{time}(\tau|t,\mathbf{y})d\tau.
\end{align*}
Due to equation (\ref{eq: prob until FA 1}), we see that:
\begin{align}\label{eq: index time probability}
K_{index}(j|\tau,t,\mathbf{y}) &= \frac{a_j(\mathbf{y},t+\tau)}{a_0(\mathbf{y},t+\tau)}\nonumber\\
K_{time}(\tau|t,\mathbf{y}) &= a_0(\mathbf{y},t+\tau)\exp\left(-\int_{t}^{t+\tau}a_0(\mathbf{y},\tau')d\tau'\right),
\end{align} 
where
\begin{align*}
a_0(\mathbf{y},t) = \sum_{j=1}^{2M}a_j(\mathbf{y},t),
\end{align*}
and $a_0\neq 0$ due to equation (\ref{assumption: nonzero a_j}).

Obviously,
\begin{align*}
\sum_{j=1}^{2M}K_{index}(j|\tau,t,\mathbf{y}) &= 1\\
\int_{0}^{\infty}K_{time}(\tau|t,\mathbf{y})d\tau &= 1.
\end{align*}
Thus, if $T$ is (random) time until the next reaction, then its probability density function given by $K_{time}$, its survival function $S(s)$ is given by (without loss of generality, suppose that $t=0$):
\begin{align}\label{eq: survival function 1}
\mathbb{P}(T>s)=S(s) = \exp\left(-\int_{0}^{s}a_0(\mathbf{y},\tau')d\tau'\right),
\end{align}
and its (cumulative) distribution function is given by $1-S(s)$.\footnote{One can check this by differentiating the distribution function, given by $1-S(s)$, with respect to $s$.} Note that the distribution of a random variable is uniquely determined by its distribution function.

Using the proof of Proposition \ref{proposition: next reaction probability} one has the following:
\begin{proposition}\label{proposition: two FA events}
Let $\tau>0$ and let $\hat{K}(\tau|t,\mathbf{y})$ be the probability of more than one FA event occurring in the time interval $[t+\tau,t+\tau+d\tau)$, given the state of the system at time $t$. Then $\hat{K}(\tau|t,\mathbf{y}) = o(d\tau)$ as $d\tau\rightarrow 0$.
\end{proposition}
\begin{proof}
By the proof of Proposition \ref{proposition: next reaction probability}:
\begin{align*}
\hat{K}(\tau|t,\mathbf{y}) = P(\tau|t,\mathbf{y})o(d\tau),
\end{align*}
since, following the definition of $a_j$, the probability of more than one reaction occurring in time interval $[t,t+d\tau)$ is $o(d\tau)$.
\end{proof}
Proposition \ref{proposition: two FA events} implies that we can neglect the case when more than one FA event occurs at the event time. Thus, an FA event (binding or unbinding) unambiguously correspond to a switch in motility state. If this were not the case and the probability of two FA events at the same time were not negligible, then binding and unbinding of distinct FAs could occur simultaneously. Since the cell becomes motile after unbinding only, simultaneous events could lead to ambiguity in determining the motile state of the cell.
\subsection{Combining the cell motility and the FA event model}\label{section: summary of PDP}

With the results of the previous section we can now formally state the cyclical mesenchymal cell motility model as a stochastic process (see Figure \ref{fig: motion cycle2} and Supplementary Video 1). 

Let $t=0$, $\mathbf{x}(0)$, $\mathbf{x}_n(0)$, $\theta_1(0)$, $\mu$(0) be given and $\mathbf{Y}(0) = \mathbf{y}^0$. 
\begin{itemize}
\item The time $T_1$ of the FA event is chosen such that $\mathbb{P}(T_1>s)=S(s)$.
\item The system evolves according to (\ref{eq: system ODE}) for $t\in[0,T_1)$.
\item At time $t=\tau$, the index $j$ of the FA event is chosen with probability $K_{index}(j|T_1,0,\mathbf{y}^0)$
$\mathbf{Y}$ and $\mu$ jump to new values:
\begin{align*}
\mathbf{Y}(t=\tau) &= 
\begin{cases*}
\mathbf{y}^0+\widehat{\mathbf{e}}_i, \phantom{a}i=(j+1)/2,\text{ if $j$ is odd}\\
\mathbf{y}^0-\widehat{\mathbf{e}}_i, \phantom{a}i=j/2, \text{ else}
\end{cases*},
\\
\mu(t=\tau) &=
\begin{cases*}
0, \text{ if $j$ is odd}\\
1, \text{ else}
\end{cases*},
\end{align*}
where $\widehat{\mathbf{e}}_i\in \mathbb{R}^M$ is the standard basis vector. Note that due to  \eqref{assumption: nonzero a_j}, we always have $\mathbf{Y}(t=\tau)\in\left\lbrace0,1 \right\rbrace^M $, since the probability of (un)binding of (un)bound FA is zero.
\item The cycle starts anew with initial time $t=T_1$ and initial values of other variables at this time:
starting at $t=T_1$ we choose the time $T_2$ of the FA event such that 
\begin{align*}
\mathbb{P}(T_2>s|T_1) = \exp\left(-\int_{T_1}^{T_1+s}a_0(\mathbf{y},\tau')d\tau'\right).
\end{align*}
\item The system evolves according to (\ref{eq: system ODE}) for $t\in[T_1,T_1+T_2)$ and so on.
\end{itemize}

One sees that the cyclical process described above is a Markov process, since the evolution of the system depends only on the current state. 
This completes the formal specification of the model. In the following we will show that this process is well-defined.
\section{Piecewise deterministic process}\label{section: PDMP}
In this section we briefly overview a class of piecewise deterministic processes, first introduced by Davis \cite{davis84}. We then show how the deterministic equations, describing motion between stochastic focal adhesion events, can be combined to yield a well-defined piecewise deterministic Markov process (PDMP).

\subsection{PDMP overview}\label{PDMP overview}
Let $A$ be countable and let $\Gamma\subset\mathbb{R}^d$ be open. Let $\mathbf{X}_t\in \Gamma$ and let $\mathbf{H}_\nu:\Gamma\rightarrow \mathbb{R}^d$ for $\nu\in A$.

Let $(\Omega,\mathcal{F},(\mathcal{F}_t)_{t\geq 0}, \mathbb{P})$ be a filtered probability space, where $\Omega$ is a sample space, $\mathcal{F}$ is a $\sigma$-algebra on $\Omega$, $(\mathcal{F}_t)_{t\geq 0}$ is a (natural) filtration, and $\mathbb{P}$ is a probability measure. Let $E:=\left\lbrace (\nu,\bm{\xi}):\nu\in A, \bm{\xi}\in \Gamma\right\rbrace $ and let $(E,\mathcal{E})$ be a Borel space. For details see [Chapter 2 in \cite{davis93}].  

We can define the piecewise deterministic process on the state space $(E,\mathcal{E})$ (for a more detailed general treatment see Davis \cite{davis93}) by the following objects\footnote{Here we first provide a constructive definition. The verification of the conditions and their explicit representation corresponding to our case of cell motility is postponed for the sake of clearer exposition}:
\begin{enumerate}[I]
\item Vector fields $(\mathbf{H}_\nu,\nu\in A)$ such that for all $\nu\in A$ there exists a unique global solution $\mathbf{X}_t\in\Gamma$ to the following equation:
\begin{align}\label{eq: ODE flow}
\frac{d}{dt}\mathbf{X}_t &= \mathbf{H}_\nu(\mathbf{X}_t)\nonumber\\
\mathbf{X}_0&\in\Gamma.
\end{align} 
Let $\phi_\nu:[0,\infty)\times\Gamma\rightarrow\Gamma$ denote the flow corresponding to Equation (\ref{eq: ODE flow}), i.e.
\begin{align*}
\phi_\nu(t,\mathbf{X}_0) = \mathbf{X}_t.
\end{align*} 
\item A measurable function $a_0:E\rightarrow\mathbb{R}_+$ such that the function $s\mapsto a_0(\nu,\phi_\nu(s,\mathbf{X}_0))$ is integrable. 
\item A transition measure $Q:\mathcal{E}\times E\rightarrow[0,1]$, such that
for fixed $C\in\mathcal{E}$, $(\nu,\bm{\xi})\mapsto Q(C;(\nu,\bm{\xi}))$ is measurable for $(\nu,\bm{\xi})\in E$, and $Q(\cdot;(\nu,\bm{\xi}))$ is a probability measure  for all $(\nu,\bm{\xi})$ on $(E,\mathcal{E})$. \\
\end{enumerate}

Let $(\nu^0,\mathbf{X}^0)\in E$ at time $t=0$ be given. Let a survival function $S$ be defined similarly as in equation (\ref{eq: survival function 1}):

\begin{align}\label{eq: survival function 2}
S(t,(\nu,\mathbf{X})) := \exp\left(-\int_{0}^{t}a_0(\nu,\phi_{\nu}(s,\mathbf{X}))ds\right).
\end{align}
Let $T_1$ be the first jump time such that
\begin{align*}
\mathbb{P}(T_1>t\phantom{,}|\phantom{,}(\nu^0,\mathbf{X}^0))=S(t,(\nu^0,\mathbf{X}^0)),
\end{align*}  
and let $(\nu^1,\mathbf{X}^1)$ be distributed according to the probability law $Q(\cdot,\phi_{\nu^0}(T_1,\mathbf{X}^0))$. Then, the motion of $(\nu_t,\mathbf{X}_t)$ for $t\leq T_1$ is given by:
\begin{align*}
(\nu_t,\mathbf{X}_t) = 
\begin{cases*}
(\nu^0,\phi_{\nu_0}(t,\mathbf{X^0})), \phantom{asd} t<T_1,\\
(\nu^1,\mathbf{X}^1),\phantom{aasdasdds} t=T_1.
\end{cases*} 
\end{align*}
At time $t=T_1$ the next jump time $T_2$ is distributed such that
\begin{align*}
\mathbb{P}(T_2-T_1>s\phantom{,}|\phantom{,}(\nu_{T_1},\mathbf{X}_{T_1}))=S(s,(\nu_{T_1},\mathbf{X}_{T_1})).
\end{align*}  
The value of the process at the jump time $T_2$ is determined by the measure $Q(\cdot,\phi_{\nu_{T_1}}(T_2,\mathbf{X}_{T_1}))$ and the process continues in a similar way. Thus, we have a well-defined piecewise deterministic process \cite{davis84}. 

\begin{theorem}[\cite{davis84}]
The process $(\nu_t,\mathbf{X}_t)_{t\geq 0}$ is a homogeneous Markov process.
\end{theorem}

\subsection{Cell motility and PDMP}\label{section: cell motility and PDMP}
In this section we show that the cyclical cell motility model described in Section \ref{section: summary of PDP} is a well-defined PDMP.

One can show that $\mathbf{F}_i(\mathbf{x}_n,\theta_1)$ satisfies the Lipschitz condition for $(\mathbf{x}_n,\theta_1)\in\Omega_{cell}\cup[0,2\pi):=D_{cell}\footnote{The restriction to the interval $[0,2\pi)$ is due to the periodic dependence on $\theta_1$ in the definition of $\mathbf{x}_i$. See Section \ref{section: The Model}}$. Furthermore, one can show that $\beta_{cell}^{-1}\mathbf{F}(\mathbf{x}_n,\theta_1,\mathbf{Y})$ and $\mu\beta_{rot}^{-1}\lVert\mathbf{x}_n\rVert\mathbf{F}(\mathbf{x}_n,\theta_1,\mathbf{Y})\cdot\hat{\boldsymbol{\varphi}}(\mathbf{x}_n)$, given by 
\begin{align*}
\beta_{cell}^{-1}\mathbf{F}(\mathbf{x}_n,\theta_1,\mathbf{Y}) &=- \beta_{cell}^{-1}\sum_{i=1}^{M}Y_i\mathbf{F}_i(\mathbf{x}_n,\theta_1)\\
\mu\beta_{rot}^{-1}\lVert\mathbf{x}_n\rVert\mathbf{F}(\mathbf{x}_n,\theta_1,\mathbf{Y})\cdot\hat{\boldsymbol{\varphi}}(\mathbf{x}_n) &= - \mu\beta_{rot}^{-1}\lVert\mathbf{x}_n\rVert\sum_{i=1}^{M}Y_i\mathbf{F}_i(\mathbf{x}_n,\theta_1)\cdot\hat{\boldsymbol{\varphi}}(\mathbf{x}_n)
\end{align*}
also satisfy the Lipschitz condition for $(\mathbf{x}_n,\theta_1)\in D_{cell}$ and arbitrary $\mu\in\left\lbrace0,1 \right\rbrace$, $\mathbf{Y}\in\left\lbrace 0,1\right\rbrace^M $. 

\begin{proposition}\label{proposition: ODE existence}
Let $\mathbf{x}(0)=\mathbf{x}_0$, $(\mathbf{x}_n(0),\theta_1(0))\in D_{cell}$. Let $\mu\in\left\lbrace0,1 \right\rbrace$, $\mathbf{Y}\in\left\lbrace 0,1\right\rbrace^M $. Then there exists a unique solution of the system 
\begin{align}\label{eq: existence prop}
\dot{\mathbf{x}} &=\mu\beta_{ECM}^{-1}\mathbf{F}(\mathbf{x}_n,\theta_1,\mathbf{Y})\cdot\hat{\mathbf{r}}(\mathbf{x}_n)\hat{\mathbf{r}}(\mathbf{x}_n)\nonumber \\ 
\dot{\mathbf{x}}_n &= \beta_{cell}^{-1}\mathbf{F}(\mathbf{x}_n,\theta_1,\mathbf{Y})\nonumber\\
\dot{\theta}_1 &=  \mu\beta_{rot}^{-1}\lVert\mathbf{x}_n\rVert\mathbf{F}(\mathbf{x}_n,\theta_1,\mathbf{Y})\cdot\hat{\boldsymbol{\varphi}}(\mathbf{x}_n),\nonumber\\
\end{align}
for $t>0$.
\end{proposition}
\begin{proof}
Note that since the evolution of $\mathbf{x}$ is decoupled from the other two equations, it is sufficient to prove the claim for the following subsystem:
\begin{align}\label{eq: mod small system}
\dot{\mathbf{x}}_n &= \beta_{cell}^{-1}\mathbf{F}(\mathbf{x}_n,\theta_1,\mathbf{Y})  \nonumber\\
\dot{\theta}_1 &=  \mu\beta_{rot}^{-1}\lVert\mathbf{x}_n\rVert\mathbf{F}(\mathbf{x}_n,\theta_1,\mathbf{Y})\cdot\hat{\boldsymbol{\varphi}}(\mathbf{x}_n) 
\end{align}  
Since the right hand side of this system is Lipschitz on $D_{cell}$ and $(\mathbf{x}_n(0),\theta_1(0))\in D_{cell}$, then there exists a unique solution of the subsystem (\ref{eq: mod small system}) for time $t\leq t_{D_{cell}}$, where $t_{D_{cell}}=\inf \left\lbrace t^*>0\phantom{,}|\phantom{,}\mathbf{x}_n(t^*)\notin\Omega_{cell}\right\rbrace $ is the exit time from $D_{cell}$. By Corollary \ref{corollay: xn exit}, we see that $t_{D_{cell}}=\infty$.
 
\end{proof}
Let $A:=\left\lbrace 1,2,\ldots,2^{M+1}\right\rbrace$ and let $\bm{\alpha}:A\rightarrow\left\lbrace 0,1\right\rbrace\times\left\lbrace0,1 \right\rbrace^M$ be a bijection. This is simply a mapping such that $\bm{\alpha}(\nu)=(\mu,\mathbf{Y})\in\left\lbrace 0,1\right\rbrace\times\left\lbrace0,1 \right\rbrace^M$  corresponds to a particular cell motion and FA states (recall that the former can either be moving or stationary).

Let $(\mathbf{x}(0),\mathbf{x}_n(0),\theta_1(0))\in\Gamma:=\mathbb{R}^2\times\Omega_{cell}\times [0,2\pi)$ and denote $\mathbf{X}_t=(\mathbf{x}(t),\mathbf{x}_n(t),\theta_1(t))$. Moreover, let   $\mathbf{H}_\nu:\Gamma\rightarrow\mathbb{R}^5$ be such that
\begin{align}
\mathbf{H}_\nu(\mathbf{X}) :=
\begin{pmatrix*}
\alpha_1(\nu)\beta_{ECM}^{-1}\mathbf{F}(\mathbf{x}_n,\bm{\alpha}_2(\nu),\theta_1)\cdot\hat{\mathbf{r}}(\mathbf{x}_n)\hat{\mathbf{r}}(\mathbf{x}_n)\\
\beta_{cell}^{-1}\mathbf{F}(\mathbf{x}_n,\bm{\alpha}_2(\nu),\theta_1)\\
\alpha_1(\nu)\beta_{rot}^{-1}\lVert\mathbf{x}_n\rVert\mathbf{F}(\mathbf{x}_n,\bm{\alpha}_2(\nu),\theta_1)\cdot\hat{\boldsymbol{\varphi}}(\mathbf{x}_n)
\end{pmatrix*},
\end{align} 
where $\alpha_1$ and $\bm{\alpha}_2$ are, respectively, the first and the second components of $\bm{\alpha}$\footnote{In this work, the subscript index in $\bm{\alpha}_2$ always denotes the second component of $\bm{\alpha}$.}. 

Let the probability $(\Omega,\mathcal{F},(\mathcal{F}_t)_{t\geq 0}, \mathbb{P})$ and state space $(E,\mathcal{E})$ be defined as in the previous section. 

We now specify the objects (I,II,III) described in Section \ref{PDMP overview}. 

\begin{enumerate}[I]
\item By Proposition \ref{proposition: ODE existence} we see that for all $\nu\in A$, there exists a unique global solution to (\ref{eq: ODE flow}).
\item 
Note that in our case the rate function $a_0$ is given by (recalling Section \ref{section: FA jumps})\footnote{We abuse the notation introduced in Section \ref{section: FA jumps}: $a_j(\bm{\alpha}_2(\nu),t)=a_j(\bm{\alpha}_2(\nu),\mathbf{X}_t)=a_j(\nu,\mathbf{X}_t)$ for $j=0,\ldots,2M$. }:
\begin{align*}
a_0(\nu,\mathbf{X}_t) =a_0(\bm{\alpha}_2(\nu),\mathbf{X}_t)= \sum_{j=1}^{2M}a_j(\bm{\alpha}_2(\nu), \mathbf{X}_t).
\end{align*}
Thus, for the integrability condition to be satisfied, we assume that each probability rate function $a_j$ is integrable along the solution of equation (\ref{eq: ODE flow}). An exact form of the rates $a_j$ satisfying this condition will be given in the subsequent section. Note that $a_0$ is nonzero, which follows from \eqref{assumption: nonzero a_j}.

\item In our case, the measure $Q(\cdot;(\nu,\bm{\xi}))$ is given by (recalling Section \ref{section: jump process}):
\begin{align}\label{eq: transition measure}
Q(\left\lbrace \eta\right\rbrace\times d\bm{\xi}';(\nu,\bm{\xi}))=\delta_{\bm{\xi}}(d\bm{\xi}')\sum_{j=1}^{M}&\delta_{\alpha_1(\eta),0}\frac{a^+_j(\bm{\alpha}_2(\nu),\bm{\xi})}{a_0(\bm{\alpha}_2(\nu),\bm{\xi})}\delta_{\bm{\alpha}_2(\eta)_j,1}\prod_{i\neq j}^{M}\delta_{\bm{\alpha}_2(\eta)_i,\bm{\alpha}_2(\nu)_i}\nonumber\\
+&\delta_{\alpha_1(\eta),1}\frac{a^-_j(\bm{\alpha}_2(\nu),\bm{\xi})}{a_0(\bm{\alpha}_2(\nu),\bm{\xi})}\delta_{\bm{\alpha}_2(\eta)_j,0}\prod_{i\neq j}^{M}\delta_{\bm{\alpha}_2(\eta)_i,\bm{\alpha}_2(\nu)_i},
\end{align}
where $\delta$ is the Kronecker delta function, $\delta_{\bm{\xi}}(\cdot)$ is the Dirac measure at $\bm{\xi}$, $a_j^+=a_{2j-1}$ and $a_j^-=a_{2j}$ correspond to, respectively, the binding and unbinding probability rates at FA site $j$, and $\bm{\alpha}_2(\cdot)_i$ is the $i^{\text{th}}$ component of the vector $\bm{\alpha}_2(\cdot)$.
\end{enumerate}
The justification for choosing the functions above stems from our deductions in Section \ref{section: FA jumps}. In particular, the rate function $a_0$ in (II) is due to (\ref{eq: index time probability}): the probability density function of the jump time $T_{k+1}$, given that $T_{k}\leq t<T_{k+1}$, is given by: 
\begin{align*}
K_{time}(\cdot|t,\bm{\alpha}_2(\nu_t)) &= a_0(\bm{\alpha}_2(\nu_t),\mathbf{X}_t)\exp\left(-\int_{0}^{\cdot}a_0(\nu_t,\phi_{\nu_t}(s,\mathbf{X}_t))ds\right)\nonumber\\
&=a_0(\bm{\alpha}_2(\nu_{T_k}),\mathbf{X}_t)\exp\left(-\int_{0}^{\cdot}a_0(\nu_{T_k},\phi_{\nu_{T_k}}(s,\mathbf{X}_t))ds\right),
\end{align*}
which corresponds to the survival function given by (\ref{eq: survival function 2}).

We now turn our attention to the measure $Q$ in (\ref{eq: transition measure}).
The components of $\mathbf{X}_t$ do not jump, and vary continuously in time, i.e. if $T_k$ is the jump time, then $\mathbf{X}_{T^-_k} =\mathbf{X}_{T_k}$ (see Section \ref{section: summary of PDP}), hence the Dirac measure $\delta_{\bm{\xi}}(\cdot)$ at $\bm{\xi}$ in (\ref{eq: transition measure}). Clearly, such transition of the continuous component $\mathbf{X}_t$ of the PDMP at a jump time is probabilistically independent of the transition of the discrete component $\nu$. Hence we have the product of the Dirac measure with the sum, which is a discrete measure for the transition of the discrete component.

By Proposition \ref{proposition: two FA events}, there is only one FA event at a jump time. Hence, for the probability of transition $\nu\rightarrow \eta$ to be nonzero, the vectors of FA states $\bm{\alpha}_2(\nu)$ and $\bm{\alpha}_2(\eta)$ must differ only by one component. Consider the following example to illustrate the jump mechanism.

\textit{Example.} Let $M=4$, $T_k\leq t<T_{k+1}$ and suppose $T_{k+1}$, $\mathbf{X}_t$ are given. Let $\nu,\eta\in A$ be such that $\bm{\alpha}(\nu) = (\mu_\nu,\mathbf{Y}_{\nu})$ and $\bm{\alpha}(\eta) = (\mu_\eta,\mathbf{Y}_{\eta})$, where
\begin{align*}
\mathbf{Y}_{\nu} = 
\begin{bmatrix*}
0\\0\\0\\1
\end{bmatrix*},
\mathbf{Y}_{\eta} = 
\begin{bmatrix*}
0\\1\\0\\1
\end{bmatrix*}.
\end{align*}
Then, by equation (\ref{eq: index time probability}), the probability of the transition $\nu\rightarrow\eta$ is given by:
\begin{align*}
\delta_{\mu_\eta,0}K_{index}(3|T_{k+1}-t,t,\mathbf{Y}_{\nu}) &= \delta_{\mu_\eta,0}\frac{a_3(\bm{\alpha}_2(\nu),\phi_\nu(T_{k+1}-t,\mathbf{X}_t))}{a_0(\bm{\alpha}_2(\nu),\phi_\nu(T_{k+1}-t,\mathbf{X}_t))} \\
&=\delta_{\alpha_1(\eta),0}\frac{a_3(\bm{\alpha}_2(\nu),\mathbf{X}_{T_{k+1}})}{a_0(\bm{\alpha}_2(\nu),\mathbf{X}_{T_{k+1}})}\\
&=\delta_{\alpha_1(\eta),0}\frac{a^+_2(\bm{\alpha}_2(\nu),\mathbf{X}_{T_{k+1}})}{a_0(\bm{\alpha}_2(\nu),\mathbf{X}_{T_{k+1}})}.
\end{align*}
Clearly, the transition $\mathbf{Y}_\nu\rightarrow\mathbf{Y}_{\eta}$ corresponds to the binding event at FA site 2, explaining the Kronecker delta term (see Sections \ref{section: kinematics summary}, \ref{section: summary of PDP}). Now, consider the sum in (\ref{eq: transition measure}) for this example. We see that  
\begin{align*}
\prod_{i\neq j}^{M}\delta_{\bm{\alpha}_2(\eta)_i,\bm{\alpha}_2(\nu)_i} \neq 0,\text{  for $j=2$ only}.
\end{align*}
We therefore obtain
\begin{align*}
Q(\left\lbrace \eta\right\rbrace\times d\bm{\xi}';(\nu,\mathbf{X}_{T_{k+1}})) = \delta_{\mathbf{X}_{T_{k+1}}}(d\bm{\xi}')\delta_{\alpha_1(\eta),0}\frac{a^+_2(\bm{\alpha}_2(\nu),\mathbf{X}_{T_{k+1}})}{a_0(\bm{\alpha}_2(\nu),\mathbf{X}_{T_{k+1}})}\delta_{\bm{\alpha}_2(\eta)_2,1}.
\end{align*}
Note that if at time $t$ the vector of FA states is given by $\mathbf{Y}_\nu$, then there are $M$ possible FA state vectors into which a transition can occur with nonzero probability:
\begin{align*}
\left\lbrace
\begin{bmatrix*}
0\\0\\0\\0
\end{bmatrix*},
\begin{bmatrix*}
1\\0\\0\\1
\end{bmatrix*},
\begin{bmatrix*}
0\\1\\0\\1
\end{bmatrix*},
\begin{bmatrix*}
0\\0\\1\\1
\end{bmatrix*}
 \right\rbrace.
\end{align*}  
Similarly as with the rate function $a_0$, we can derive equation (\ref{eq: transition measure}) from the principles we established before.
\begin{proposition}\label{proposition: transition measure form}
The transition probability measure $Q(\cdot,(\nu,\bm{\xi}))$ is given by equation (\ref{eq: transition measure}) for each $(\nu,\bm{\xi})\in E$.
\end{proposition} 
\begin{proof}
Let $(\nu,\bm{\xi})\in E$, $\left\lbrace \eta\right\rbrace\times d\bm{\xi}'\in\mathcal{E}$. Let $(N,\bm{\Xi})$ and $(N_-,\bm{\Xi}_-)$ be $E$-valued random variables before and after the jumps. Then,
\begin{align*}
Q(\left\lbrace \eta\right\rbrace\times d\bm{\xi}';(\nu,\bm{\xi})) &= \mathbb{P}\left((N,\bm{\Xi})\in\left\lbrace \eta\right\rbrace\times d\bm{\xi}'|(N_-,\bm{\Xi}_-)=(\nu,\bm{\xi}) \right) \\
&=\mathbb{P}\left(\left\lbrace \eta\right\rbrace\times d\bm{\xi}'|(\nu,\bm{\xi}) \right),
\end{align*} 
where we omitted the random variables for notational convenience. Then we have:
\begin{align*}
\mathbb{P}\left(\left\lbrace \eta\right\rbrace\times d\bm{\xi}'|(\nu,\bm{\xi}) \right) = \mathbb{P}(d\bm{\xi}'|\left\lbrace \eta\right\rbrace,(\nu,\bm{\xi}) )\mathbb{P}(\left\lbrace\eta\right\rbrace|(\nu,\bm{\xi})) = \delta_{\bm{\xi}}(d\bm{\xi}')\mathbb{P}(\eta|(\nu,\bm{\xi})),
\end{align*}
since $\bm{\Xi}=\bm{\Xi}_-$ a.s., by construction of the process. Since $\bm{\alpha}$ is a bijection, we have\footnote{If $(H,\bm{Z})$ is a random variable, then, due to $\bm{\alpha}$ being a bijection:
\begin{align*}
\left\lbrace H=\eta\right\rbrace=\left\lbrace \bm{\alpha}^{-1}(H)=\bm{\alpha}^{-1}(\eta)\right\rbrace=\left\lbrace \bm{\alpha}(H)=\bm{\alpha}(\eta)\right\rbrace.
\end{align*}}
\begin{align*}
\mathbb{P}(\eta|(\nu,\bm{\xi})) &= \mathbb{P}\left((\alpha_1(\eta),\bm{\alpha}_2(\eta))|(\alpha_1(\nu),\bm{\alpha}_2(\nu)),\bm{\xi}\right)\\
&=\mathbb{P}\left(\alpha_1(\eta),\bm{\alpha}_2(\eta)|\bm{\alpha}_2(\nu),\bm{\xi}\right),
\end{align*}
since, by construction of the cell motility process, the new FA state $\bm{\alpha}_2(\eta)$ is determined independently of whether a cell was moving or not (represented by $\alpha_1(\nu)\in\left\lbrace 0,1\right\rbrace $) and the new motility state $\alpha_1(\eta)$ is determined only by which FA event took place (binding or unbinding), regardless of whether a cell was previously moving or not.

Note that when a jump occurs, then, by Proposition $\ref{proposition: two FA events}$, only one of $j=1,\ldots,2M$ possible (binding and unbinding) FA events occurs. Thus, for $j,j'\in\left\lbrace 1,\ldots,2M\right\rbrace$ and $j\neq j'$ the events ``reaction $j$ occurs" and ``reaction $j'$ occurs" are mutually exclusive. We then have, by the definition of conditional probability:
\begin{align*}
\mathbb{P}\left(\alpha_1(\eta),\bm{\alpha}_2(\eta)|\bm{\alpha}_2(\nu),\bm{\xi}\right) = \sum_{j=1}^{2M}\mathbb{P}\left(\alpha_1(\eta),\bm{\alpha}_2(\eta)|j,\bm{\alpha}_2(\nu),\bm{\xi}\right)\mathbb{P}\left(j|\bm{\alpha}_2(\nu),\bm{\xi}\right),
\end{align*}
where 
\begin{itemize}
\item $\mathbb{P}\left(j|\bm{\alpha}_2(\nu),\bm{\xi}\right)$ is the probability that the FA event $j$ occurs, given $\bm{\alpha}_2(\nu)$ and $\bm{\xi}$. 
\item $\mathbb{P}\left(\alpha_1(\eta),\bm{\alpha}_2(\eta)|j,\bm{\alpha}_2(\nu),\bm{\xi}\right)$ is the probability of a jump into cell state $\left(\alpha_1(\eta),\bm{\alpha}_2(\eta)\right)$, given $\bm{\alpha}_2(\nu)$ and $\bm{\xi}$, and that the FA event $j$ occurred.
\end{itemize}
Let $j\in\left\lbrace 1,\ldots,M\right\rbrace $ and $j^+=2j-1$ and $j^-=2j$. Due to (\ref{eq: index time probability}) we have:
\begin{align}\label{eq: probability of the next reaction j}
\mathbb{P}\left(j^{\pm}|\bm{\alpha}_2(\nu),\bm{\xi}\right) &= \frac{a^{\pm}_j(\bm{\alpha}_2(\nu),\bm{\xi})}{a_0(\bm{\alpha}_2(\nu),\bm{\xi})}.
\end{align}
Furthermore,
\begin{align*}
\mathbb{P}\left(\alpha_1(\eta),\bm{\alpha}_2(\eta)|j^+,\bm{\alpha}_2(\nu),\bm{\xi}\right)&=\delta_{\alpha_1(\eta),0}\delta_{\bm{\alpha}_2(\eta)_j,1}\prod_{i\neq j}^{M}\delta_{\bm{\alpha}_2(\eta)_i,\bm{\alpha}_2(\nu)_i}\\
\mathbb{P}\left(\alpha_1(\eta),\bm{\alpha}_2(\eta)|j^-,\bm{\alpha}_2(\nu),\bm{\xi}\right)&=\delta_{\alpha_1(\eta),1}\delta_{\bm{\alpha}_2(\eta)_j,0}\prod_{i\neq j}^{M}\delta_{\bm{\alpha}_2(\eta)_i,\bm{\alpha}_2(\nu)_i}.
\end{align*}

Therefore,
\begin{align*}
\sum_{j=1}^{2M}\mathbb{P}\left(\alpha_1(\eta),\bm{\alpha}_2(\eta)|j,\bm{\alpha}_2(\nu),\bm{\xi}\right)&\mathbb{P}\left(j|\bm{\alpha}_2(\nu),\bm{\xi}\right)\\
&=\sum_{j=1}^{M}\mathbb{P}\left(\alpha_1(\eta),\bm{\alpha}_2(\eta)|j^+,\bm{\alpha}_2(\nu),\bm{\xi}\right)\mathbb{P}\left(j^+|\bm{\alpha}_2(\nu),\bm{\xi}\right)\\
&\phantom{\sum_{j=1}^{M}}+\mathbb{P}\left(\alpha_1(\eta),\bm{\alpha}_2(\eta)|j^-,\bm{\alpha}_2(\nu),\bm{\xi}\right)\mathbb{P}\left(j^-|\bm{\alpha}_2(\nu),\bm{\xi}\right)\\
&=\sum_{j=1}^{M}\delta_{\alpha_1(\eta),0}\frac{a^+_j(\bm{\alpha}_2(\nu),\bm{\xi})}{a_0(\bm{\alpha}_2(\nu),\bm{\xi})}\delta_{\bm{\alpha}_2(\eta)_j,1}\prod_{i\neq j}^{M}\delta_{\bm{\alpha}_2(\eta)_i,\bm{\alpha}_2(\nu)_i}\nonumber\\
&\phantom{\sum_{j=1}^{M}}+\delta_{\alpha_1(\eta),1}\frac{a^-_j(\bm{\alpha}_2(\nu),\bm{\xi})}{a_0(\bm{\alpha}_2(\nu),\bm{\xi})}\delta_{\bm{\alpha}_2(\eta)_j,0}\prod_{i\neq j}^{M}\delta_{\bm{\alpha}_2(\eta)_i,\bm{\alpha}_2(\nu)_i}.
\end{align*}        
\end{proof}

\section{Adhesion kinetics}\label{section: adhesion kinetics}
While we introduced the probability rates of binding and unbinding events, we have not yet fully specified them. Here we elaborate on how such quantities unambiguously correspond to the relevant subcell dynamics by providing the precise functional forms of the propensity functions.

\subsection{Unbinding rate}
Consider the unbinding rate $a^-_j$ of FA adhesion site $j\in\left\lbrace 1,\ldots, M\right\rbrace $ and let $\mathbf{y}\in\left\lbrace 0,1\right\rbrace^M$, $\bm{\xi}=(\mathbf{x},\mathbf{x}_n,\theta_1)\in\Gamma$. Following Bell \cite{Bell78}, the bond disassociation rate under applied force is given by:
\begin{align}\label{eq: force unbinding rate}
a^-_j(\mathbf{y},\bm{\xi}) = k^0_{off}e^{\lVert \mathbf{F}_j(\mathbf{x}_n,\theta_1)\rVert/F_b}y_j, 
\end{align}    
where $k^0_{off}$ is the FA disassociation rate under no load, $\mathbf{F}_i$ is the force applied at the FA, given by equation \eqref{eq: Fi}, and $F_b$ is a characteristic force scale. The last factor $y_j$ simply indicates that only bound FAs can unbind (thus satisfying equation \eqref{assumption: nonzero a_j}). Clearly, the function in \eqref{eq: force unbinding rate} is integrable. Here we neglect the fact that the force may be applied at the FA (and consequently at the transmembrane receptors) at an angle and assume for tractability of the model that it is applied normally to the FA (hence consider only magnitude).

\textbf{Remark.} In the context of cell migration and within the framework of our model, we only consider FA disassembly on the cell periphery (including the lamellae). The primary cause of such FA unbinding has mechanical, rather than biochemical nature due to the cell contractile mechanism applying load to the adhesions. Although it is known that the Rho family of GTPases (in particular its member RhoA) mediates disassembly of FAs, their effect is indirect: the activity of myosin motors, which generate contractile forces in SFs, is regulated by RhoA \cite{Ridley2001}. Hence the force dependence of the unbinding rate $a_j^-$.  Recalling the definition of $\mathbf{F}_i$ in equation \eqref{eq: Fi}, we can include such indirect biochemical mediation by considering mediators of the force $T_i$. In order to keep the model tractable, we omit the interaction between RhoA and myosin motors.

\subsection{Binding dynamics}  
Consider the binding probability rate $a^+_j$ of the FA adhesion site $j\in\left\lbrace 1,\ldots,M\right\rbrace $ and let $\mathbf{y}\in\left\lbrace 0,1\right\rbrace ^M$, $\bm{\xi}=(\mathbf{x},\mathbf{x}_n,\theta_1)\in\Gamma$. The probability rate $a^+_j$ is given by:
\begin{align*}
a^+_j(\mathbf{y},\bm{\xi}) = k_{on,j}(\bm{\xi})(1-y_j), 
\end{align*}
where $k_{on,j}:\Gamma\rightarrow\mathbb{R}^+$ is the effective binding rate at FA site $j$. The last term $(1-y_j)$ simply indicates that only unbound FAs can bind. Whereas unbinding can be viewed effectively as a bond rupturing under applied tension, a binding reaction, or focal adhesion assembly and maturation, is a highly regulated process. Due to the complexity of the FA assembly process, we focus on two major mediators of FA formation: Rac activity and contractile forces. 
\subsubsection{Rac dependence}
For simplicity, we assume that the probability of FA formation is directly proportional to local Rac concentration. Consider the case of chemotactic cell migration. Leading edge protrusions preferentially form in the direction of a chemoattractant. Since Rac is required for formation of lamellipodium and FA formation \cite{Ridley2001}, then local Rac activity correlates with the concentration of the chemical cues. Conversely, local Rac activity negatively correlates with chemorepellent. 

Let $Q_{cue}:\mathbb{R}^2\rightarrow\mathbb{R}^+$ denote the concentration of a cue in the spatial domain and let $q:\mathbb{R}^+\rightarrow\mathbb{R}^+$ denote the $Q_{cue}$ dependent concentration of Rac. Clearly, $q$ is an increasing function for the case of an attractant and a decreasing function for a repellent. Then, 
\begin{align*}
k_{on,j}(\bm{\xi}) \propto q(Q_{cue}(\mathbf{x}+\mathbf{x}_j)),
\end{align*}
where we recall that $\mathbf{x}_j$ is the position of the $j^{\text{th}}$ FA site. 

For example, we can take $Q_{cue}(\mathbf{x})$ to be the density of the ECM (or chemoattractant) at $\mathbf{x}\in\mathbb{R}^2$ and take $q(x)=x$. Then, the probability of binding is simply proportional to the ECM (or chemoattractant) density.\footnote{This corresponds to a bi-molecular reaction rate, which depends on the \textit{number} of one of the reactants (FA) and on the \textit{concentration} (ECM or a chemoattractant) of the other.} Although a more complex function $q$ can be considered, such as those in \cite{giverso2018mechanical}, \cite{holmes2012modelling}, \cite{narang2006spontaneous}, in order to keep the minimal character of the model, we opted for simple linear relation. Moreover, following Model 4 in \cite{holmes2012modelling} and assuming no feedback with phosphoinositides, then in a steady state we have:
\begin{align*}
R &= \bar{R}\frac{\alpha}{\delta_R} C + \bar{R}(\hat{I}_R+Q_{cue})\nonumber\\
\rho &= \frac{\bar{\rho}}{\delta_{\rho}}\frac{\hat{I}_{\rho}}{1+\left(\frac{R}{a_2}\right)^n}\nonumber\\
C &= \frac{\bar{C}}{\delta_C}\frac{\hat{I}_C}{1+\left(\frac{R}{a_1}\right)^n},
%{1+\left(\frac{R}{a_1}\right)^n}
\end{align*}
where $R,\rho,C$ denote the concentrations of Rac, RhoA, Cdc42, and the rest are constant parameters (see \cite{holmes2012modelling} for details). That is, there is a linear dependence of Rac on an external cue. 
\subsubsection{Force dependence}\label{section: adhesion force dependence}
\begin{wrapfigure}{hr}[0.1cm]{0.4\textwidth}
\vspace{-0.5cm}
\includegraphics[width=0.9\textwidth,height=0.9\textwidth]{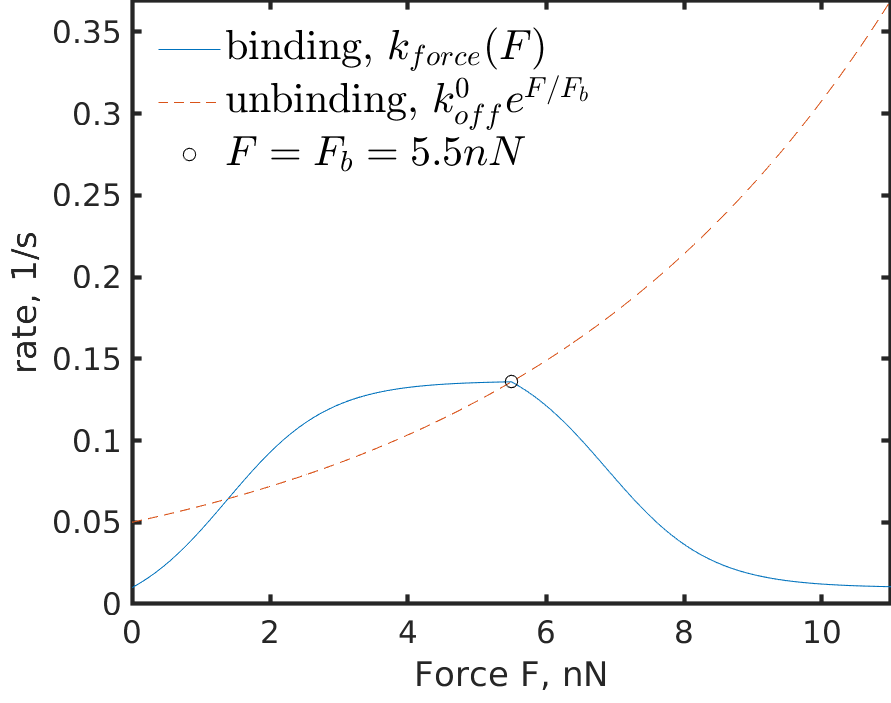}
\caption{Force dependence of unbinding and binding rates.}
\label{figure: force dependence}
\end{wrapfigure}

Note that the enlargement of nascent adhesions is concurrent with their maturation into focal adhesions. Thus, enlargement and maturation are synonymous. While the initial stage of adhesion growth is force-independent \cite{CVZWMH08}, maturation occurs in a force-dependent manner \cite{gardel2010mechanical}, \cite{stricker2011spatiotemporal}, \cite{Zaidel-Bar04}. 

However, during such a force-dependent maturation, the positive correlation between the adhesion size and the applied tension exists only in the initial stages of maturation. As FAs increase in size, the force dependence plateaus \cite{stricker2011spatiotemporal}. 

That is, the study by Stricker et. al. \cite{stricker2011spatiotemporal} showed that for mature FAs there is no correlation between applied force and FA size. One can consider an adhesion site as mature when its size reaches $\sim 1\mu m^2$ (see e.g. \cite{Balaban01}, \cite{OakesMBTF15}).

Choi et. al. \cite{CVZWMH08} showed that nascent adhesions assemble at a rate of $\sim 1.3\text{min}^{-1}=0.021s^{-1}$ reaching a size of $\sim0.2\mu m^2$. Furthermore, it was shown that the formation of these adhesions is independent of fibronectin density \cite{CVZWMH08}, stiffness \cite{CVZWMH08} and myosin II activity \cite{alexandrova2008comparative},\cite{CVZWMH08},\cite{vicente2007regulation}.

Let $k^{0}_{on}$ be the force-independent FA maturation rate. Since FAs are larger in size than nascent adhesions, which assemble at a rate of $0.021s^{-1}$, then their assembly is slower and hence we take $k^0_{on}=0.01s^{-1}$. We now want to find a function that could represent force dependence of FA maturation rate. Denote this function $k_{force}:\mathbb{R}^+\rightarrow\mathbb{R}^+$.  

It satisfies the following:
\begin{itemize}
\item $k_{force}(0) = k^0_{on}$, i.e. if there is no force applied, the rate is $k^0_{on}$.
\item If the applied force is below the characteristic force $F_b$, then $k_{force}$ is greater than the unbinding rate, i.e. it is more likely that an FA increases in size than that it ruptures. 
\item If the characteristic force $F_b$ is applied, the rate $k_{force}$ should equal the unbinding rate, i.e. we assume that there is a dynamic equilibrium in some sense. 
\item If the applied force is larger than $F_b$, then the unbinding rate dominates the binding rate. Note that as FA increases in size, the force dependence plateaus \cite{stricker2011spatiotemporal}. Thus, $k_{force}$ should plateau around $F_b$. We also assume that for large applied forces $k_{force}$ plateaus at $k^0_{on}$, since exceeding loads rupture integrin bonds frequently and impede stable maturation. 
\end{itemize}
The following form of $k_{force}$ satisfies the conditions above:
\begin{align}\label{eq: force binding rate}
k_{force}(F) = 
\begin{cases*}
\frac{k^0_{off}e+k^0_{on}}{1+\exp\left(-\gamma_1(F-F^*_1)/F_b\right)}+k^0_{on}-\epsilon,\phantom{a} F\leq F_b\\
\frac{k^0_{off}e+k^0_{on}}{1+\exp\left(\gamma_2(F-F^*_2)/F_b\right)}+k^0_{on},\phantom{a} \text{else}
\end{cases*},
\end{align}
where $F_1^*=F_b/4$ and $F_2^*=5F_b/4$ are the midpoints of the sigmoid functions (see Figure \ref{figure: force dependence}). To find the values of $\gamma_1,\gamma_2$ and $\epsilon$, see Appendix {\ref{appendix: parameters}}. 
 
\textbf{Remark.} Chan and Odde \cite{CO08} showed that adhesion site dynamics depends on substrate stiffness. In particular, they showed that for a stiff substrate the transmembrane bonds rupture more frequently, compared to the case with softer substrate under the same load, since the critical load is reached faster. This mechanism provides means for a cell to assess the surrounding rheology. Within the context of our model, this means that the force $F_b$ is smaller for the stiffer substrate, thus increasing the disassociation rate for the same load (see \eqref{eq: force unbinding rate}). Consequently, the force dependent binding rate $k_{force}$ also changes for the stiffer ECM. In this case, the curves in Figure \ref{figure: force dependence} will shift to the left. Therefore, our model provides an opportunity to include mechanosensitivity of migrating cells by considering the relevant dynamics for individual FAs.  

Therefore, the binding propensity rate $a_j^+$ of an adhesion $j\in\left\lbrace 1,\ldots,M\right\rbrace $ is given by:
\begin{align}\label{equation: binding propensity}
a^+_j(\mathbf{y},\bm{\xi}) = q(Q_{cue}(\mathbf{x}+\mathbf{x}_j))k_{force}(\lVert \mathbf{F}_j(\mathbf{x}_n,\theta_1)\rVert)(1-y_j).
\end{align}

\section{Numerical Simulations}\label{section: numerical simulations}
%section: next event index simulation
Here we show the results of simulating cell trajectories under different scenarios, which represent various experimental settings, namely:
\begin{enumerate}
\item Uniform environment with no cues.
\item Non-uniform environment with external cue gradient and uneven myosin motor activity within a cell. 
\item Striped ECM architecture.
\end{enumerate}

Note that the total number of adhesion sites $M$ is a free parameter, which differs from cell to cell. Nevertheless, it is a crucial parameter, determining whether the motility type is amoeboid or mesenchymal. Amoeboid motility is characterized by a large number of weak adhesions, high turnover, and more contractile cell body. On the other hand, mesenchymal migration relies on fewer, but stronger focal adhesions with slower turnover and lower overall contractility. The cells with the former motility type are faster and have more diffusive motion \cite{liu2015confinement}, \cite{pavnkova2010molecular}. Note that if $a_j^{\pm}\sim O(1)$, then the rate function is $a_0\sim O(M)$. Therefore, adhesion events occur more frequently for increasing $M$, implying higher adhesion turnover. Thus, by varying $M$, our model is capable of explaining this particular aspect of the difference between the two migration types.  

For all scenarios we employ the same initial conditions for all cells. Namely, at $t=0$ the conditions are: 
\begin{itemize}
\item $\mathbf{x}$ is at the origin, $\mathbf{x}_n$ is uniformly distributed on a circle with radius $R_{cell}$, and $\theta=0$.
\item Each FA is in (un)bound state  and each cell is in moving state with probability $1/2$. 
\end{itemize}
We simulate  trajectories of $n_{cell}:=56$ cells for time $t_{end}:=600min$. We divided the time interval into $n_{time}:=5000$ intervals, at the end of which we recorded the cell centroid positions $\mathbf{x}$. For details on the parameters and numerical methods used for simulations, see Appendix \ref{appendix: parameters} and \ref{appendix: simulation of the pdmp}, respectively. For details on the data analysis, see Appendix \ref{appendix: data analysis}.
\subsection{Uniform environment}
The results of the simulation with uniform spatial cue $Q_{cue}$ are presented in Figure \ref{figure: uniform results}. Due to the absence of spatial variation of $Q_{cue}$, we take $q=1$ in equation \eqref{equation: binding propensity}. 
             
\begin{figure}
%\vspace{0mm}
\centering
%First row
\subfloat[]
{
	\includegraphics[width=40mm,height=37mm]{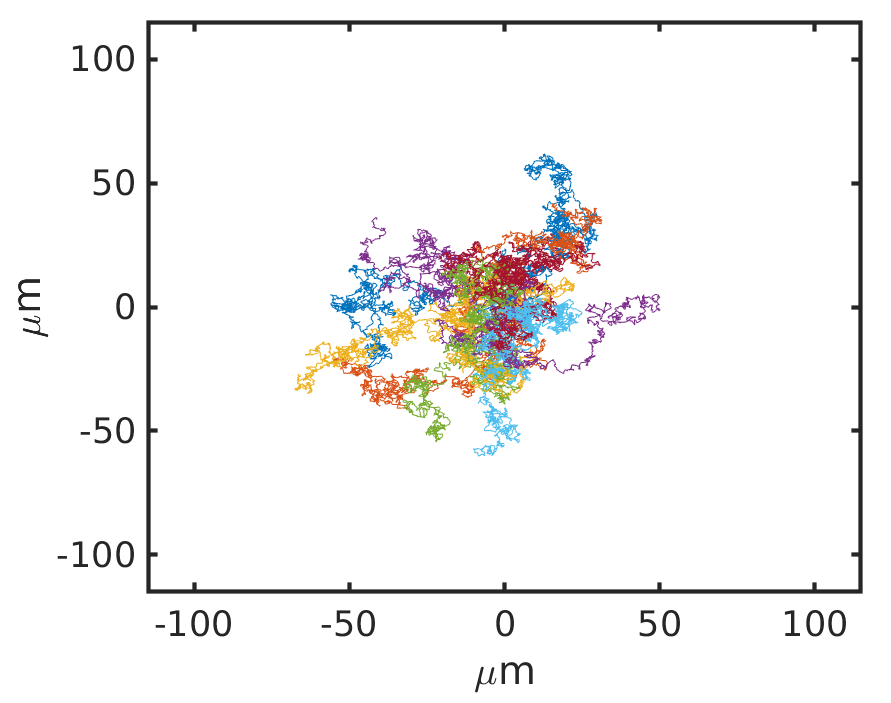}
}
\subfloat[]
{
	\includegraphics[width=40mm,height=37mm]{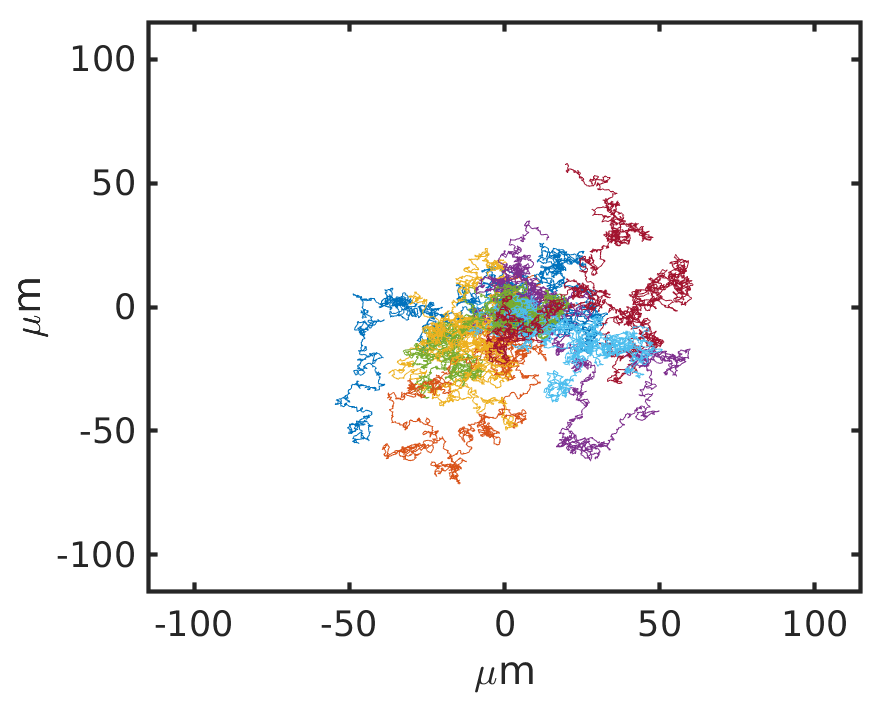}
}
\subfloat[]
{
	\includegraphics[width=40mm,height=37mm]{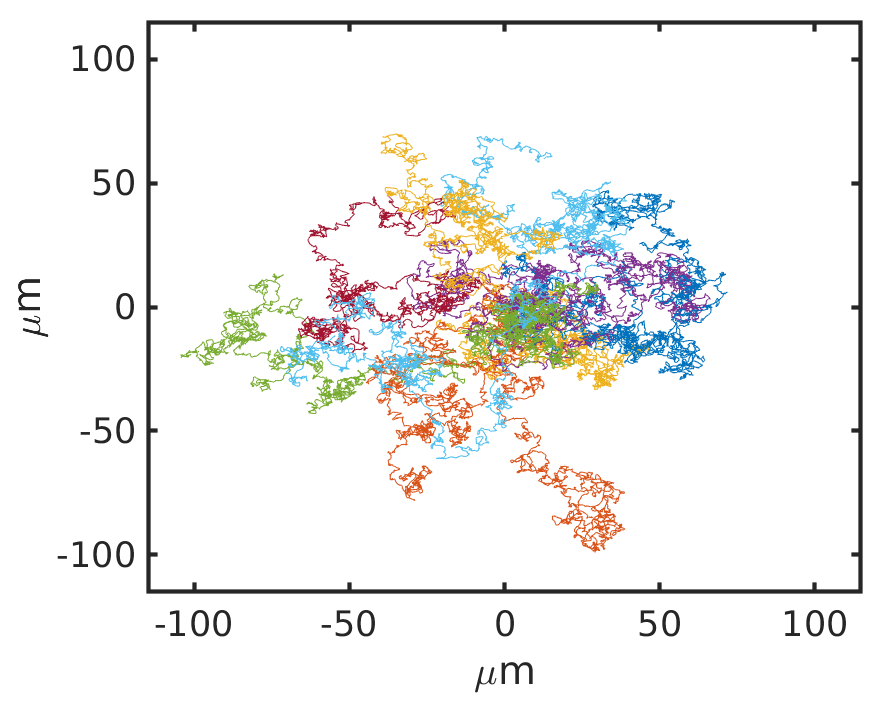}
}
\hspace{0mm}
%Second row
\subfloat[]
{
	\includegraphics[width=40mm,height=37mm]{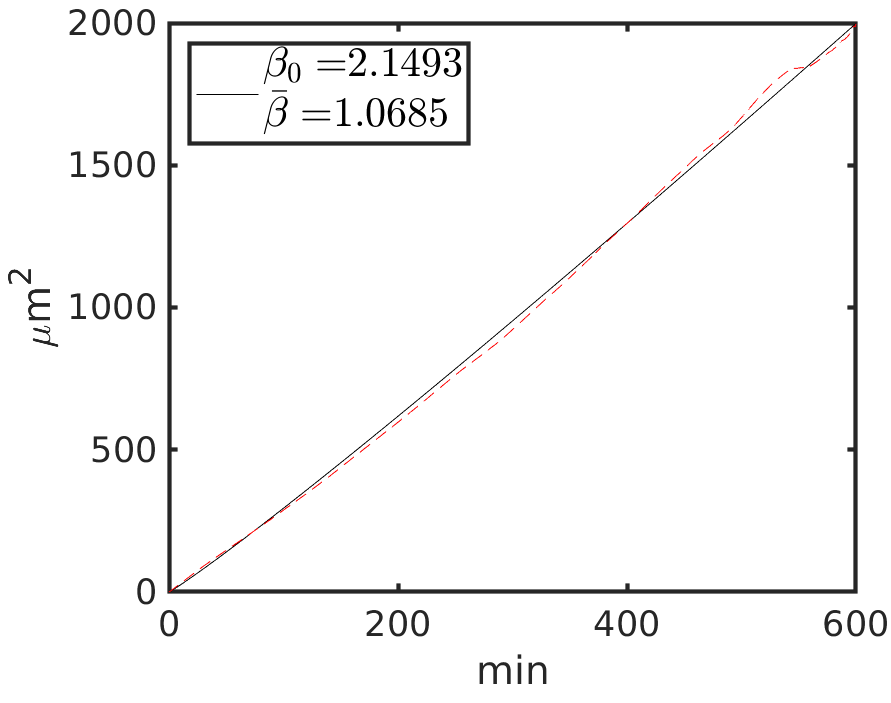}
}
\subfloat[]
{
	\includegraphics[width=40mm,height=37mm]{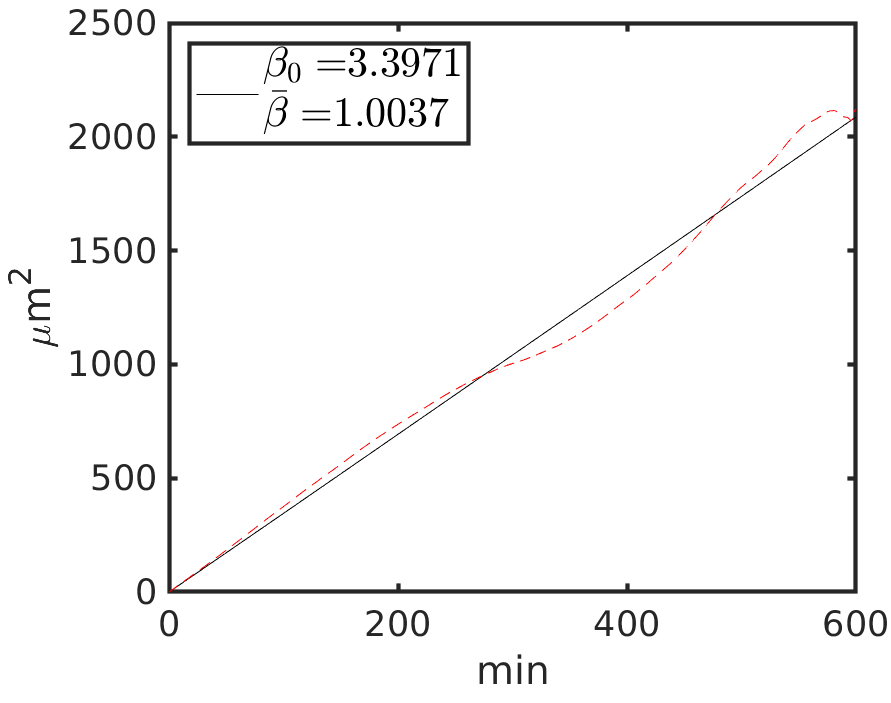}
}
\subfloat[]
{
	\includegraphics[width=40mm,height=37mm]{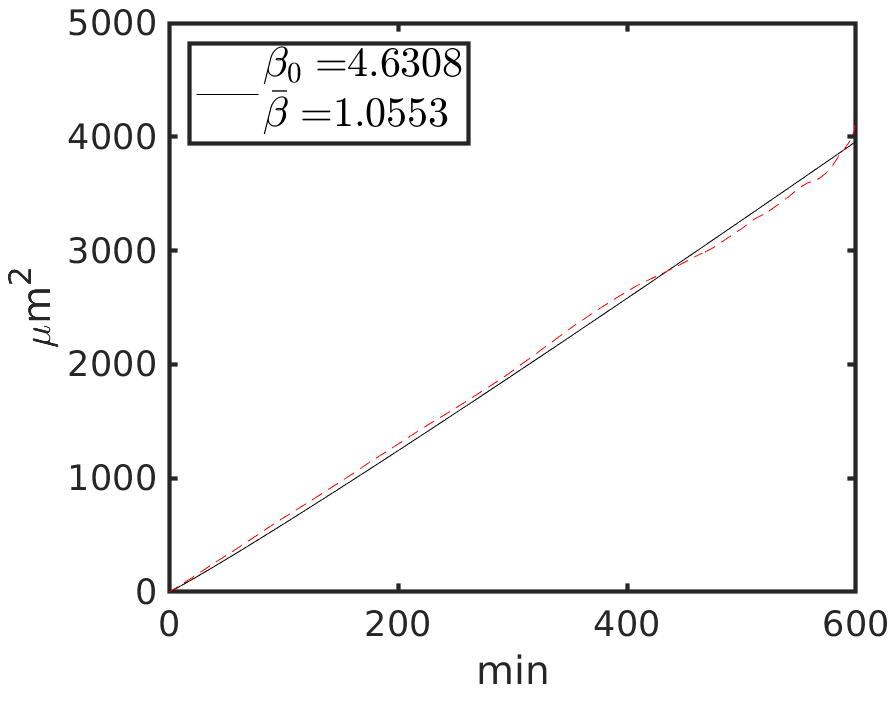}
}
\hspace{0mm}
%Thrid row
\subfloat[]
{
	\includegraphics[width=40mm,height=37mm]{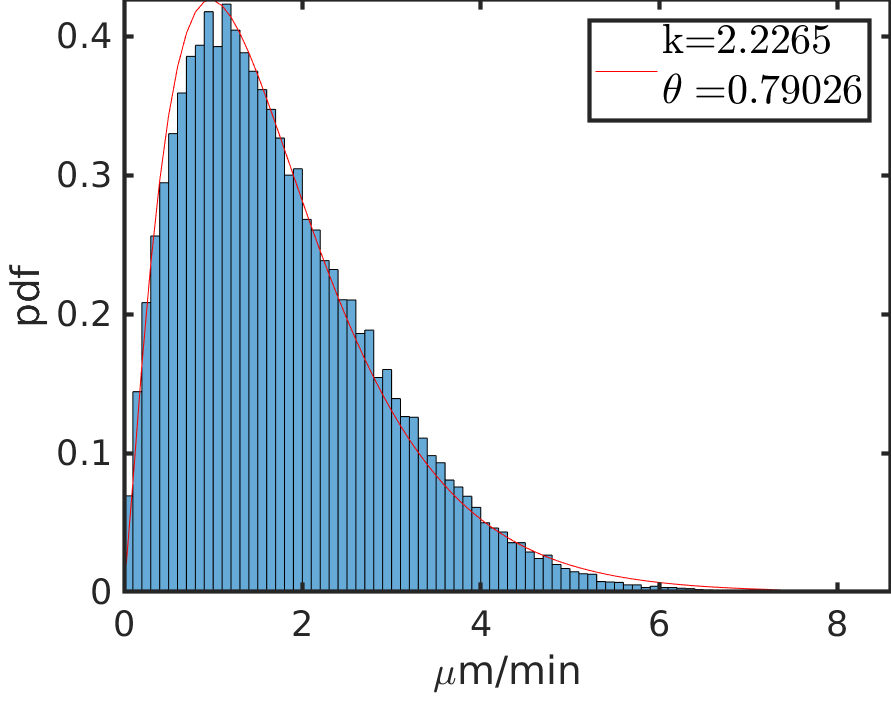}
}
\subfloat[]
{
	\includegraphics[width=40mm,height=37mm]{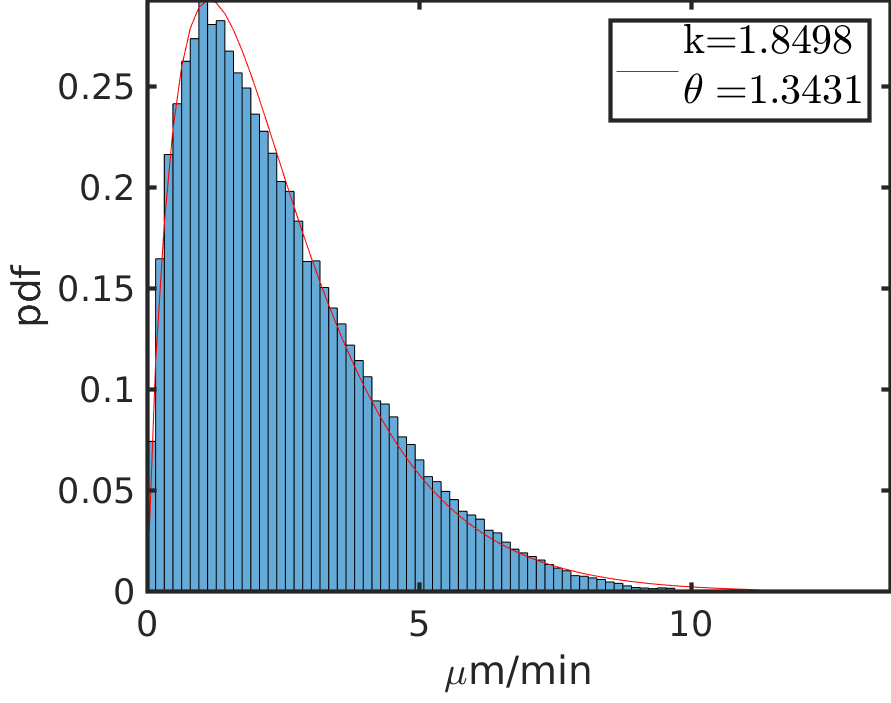}
}
\subfloat[]
{
	\includegraphics[width=40mm,height=37mm]{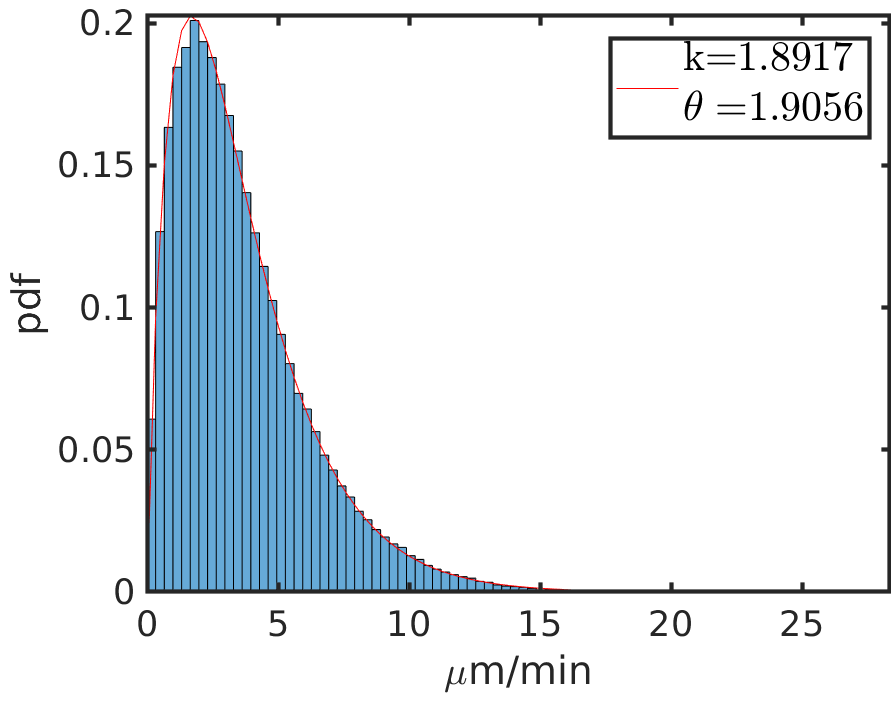}
}
\hspace{0mm}
%Fourth row
\subfloat[]
{
	\includegraphics[width=40mm,height=37mm]{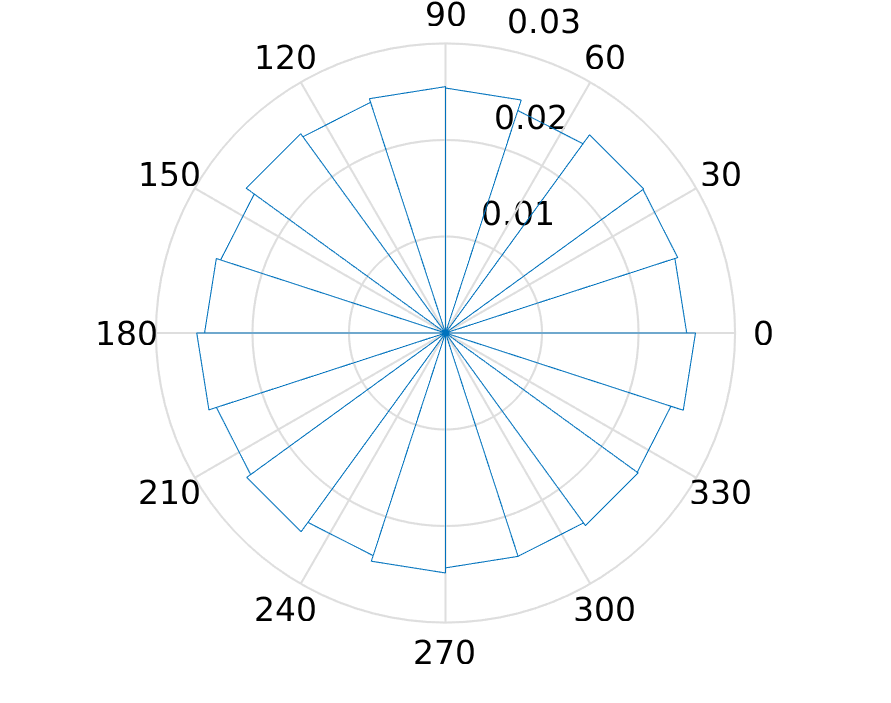}
}
\subfloat[]
{
	\includegraphics[width=40mm,height=37mm]{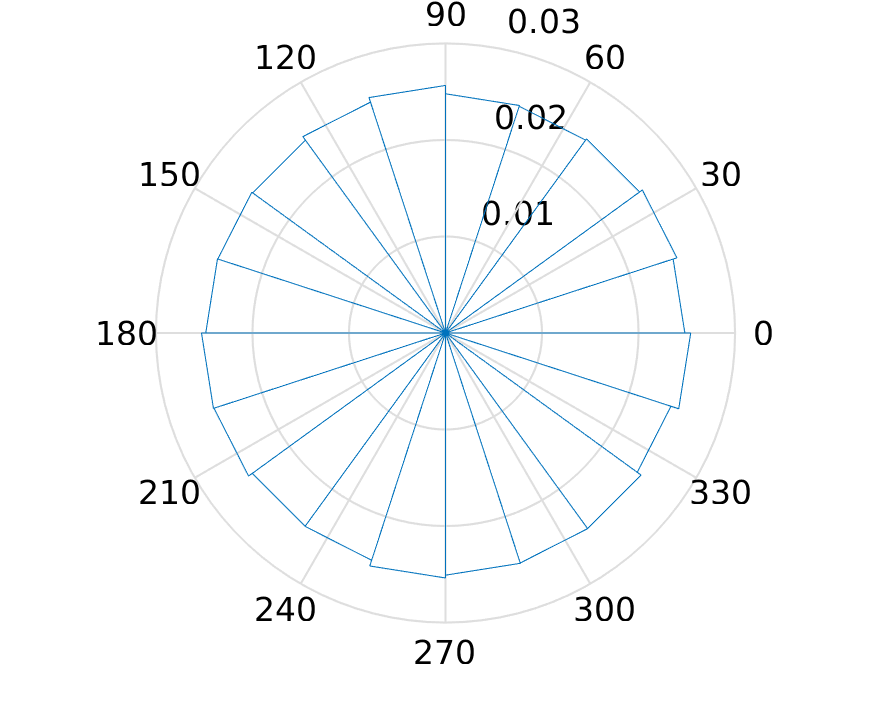}
}
\subfloat[]
{
	\includegraphics[width=40mm,height=37mm]{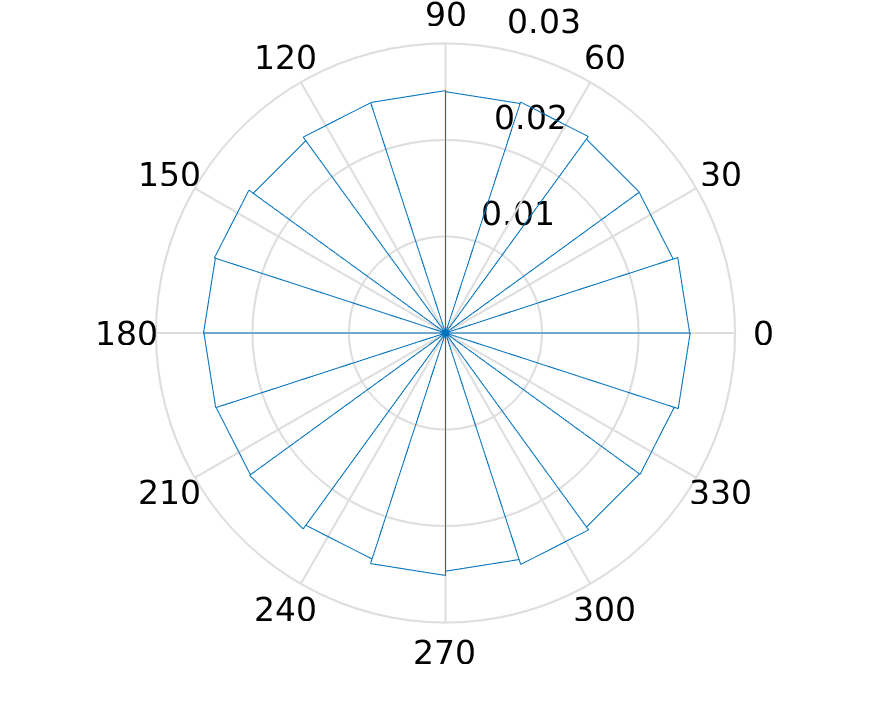}
}
\caption{Simulation results with $M=8,16,32$ adhesions in the first, second, and third columns respectively. (a-c) Trajectories of 13 cell centroids $\mathbf{x}(t)$. (d-f) Mean-squared displacements $msd(t)$ (red, dash) and fitted $\widehat{msd}(t)$ (black, solid) with parameters $\beta_0$ and $\bar{\beta}$ (see text for details). The unit of $\beta_0$ is $\mu m^2/min^{\bar{\beta}}$. (g-i) Histograms of speed probability density functions and fitted density function of gamma distribution (red) with parameters $k$ and $\theta$. (j-l) Relative frequency histogram of normalized velocities.}
\label{figure: uniform results}
\end{figure}
The cell trajectories with varying number of adhesion sites, depicted in Figure \ref{figure: uniform results} (a-c), show no clear trend and resemble those of a Brownian motion. Indeed, fitting the mean-squared displacement $msd(t)$ to the curve $\widehat{msd}(t) = \beta_0t^{\bar{\beta}}$ (see Appendix \ref{appendix: data analysis} for details), we see that the exponent is $\bar{\beta}\approx1$ for the three cases (see Figure \ref{figure: uniform results} (d-f)). This suggests that the cell motion has diffusive characteristics in this scenario. In Figure \ref{figure: uniform results} (g-i) we see the simulated distribution of speeds. The average speeds $s_{av}$ and the parameters $\beta_0$, $\bar{\beta}$ are shown in Table \ref{table: uniform parameters}. We see that as $M$ increases, the cell motion becomes progressively faster and more diffusive\footnote{Since $\bar{\beta}\sim1$, the slope $\beta_0$ is a measure of diffusivity. See Appendix \ref{appendix: data analysis}.}, which is expected for a dominantly amoeboid type of motility. Because $\bar{\beta}\sim1$, we can estimate the diffusion coefficient $D=\beta_0/4$. The obtained values are lower, but within an order of magnitude estimated by Liu et. al. \cite{liu2015confinement}, who found that $D\approx2.7\mu m^2/min$. Interestingly, the gamma distribution gives a very good fit to the simulated data for various values of $M$, suggesting that cell speeds are gamma distributed. Moreover, the obtained values of the average speed $s_{av}$ fall in the range reported by Liu et. al. \cite{liu2015confinement}, who found the speeds to be in the range from $0.2\mu m/min$ to $7\mu m/min$. Although there are very high speeds observed in Figure \ref{figure: uniform results} (i), which seem to be outliers, speeds as high as $25\mu m/min$ have been observed \cite{friedl1994locomotor}. As expected, rose plots of normalized velocities show no bias in any particular direction in Figure \ref{figure: uniform results} (j-l). Along with time scaling of the squared displacement, persistence of motion can be measured by directionality ratio (distance between cell centroids divided by path length) and velocity autocorrelation \cite{gorelik2014quantitative}. Expectedly, Figure \ref{figure: persistence uniform} illustrates that the cells strongly deviate from straight-path migration (Figure \ref{figure: persistence uniform} (left); see also time average of the directionality ratio $\bar{r}$ in Table \ref{table: uniform parameters}) and the deviation directions are uncorrelated (Figure \ref{figure: persistence uniform} (right)). The rapid decay in Figure \ref{figure: persistence uniform} (right) also suggests that correlations in time become stationary very fast. The increased oscillations in Figure {\ref{figure: persistence uniform}} (right) towards the end of the observation window are due to decreased number of observations (see Appendix \ref{appendix: data analysis}).  

\begin{figure}[h]
\captionsetup[subfigure]{labelformat=empty}
\subfloat[]
{
	\includegraphics[width=55mm,height=50mm]{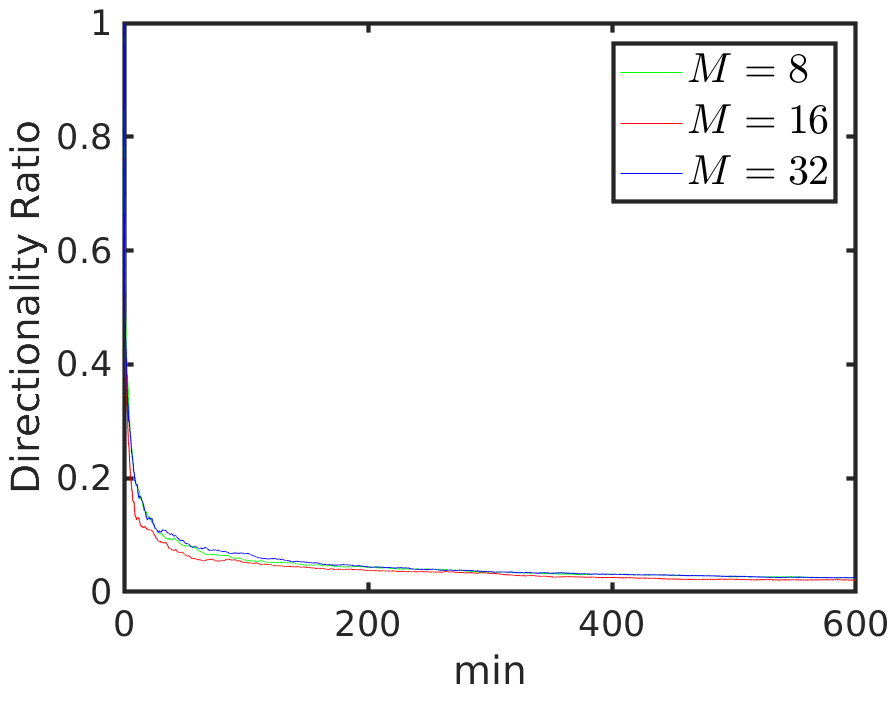}
}
\subfloat[]
{
	\includegraphics[width=55mm,height=50mm]{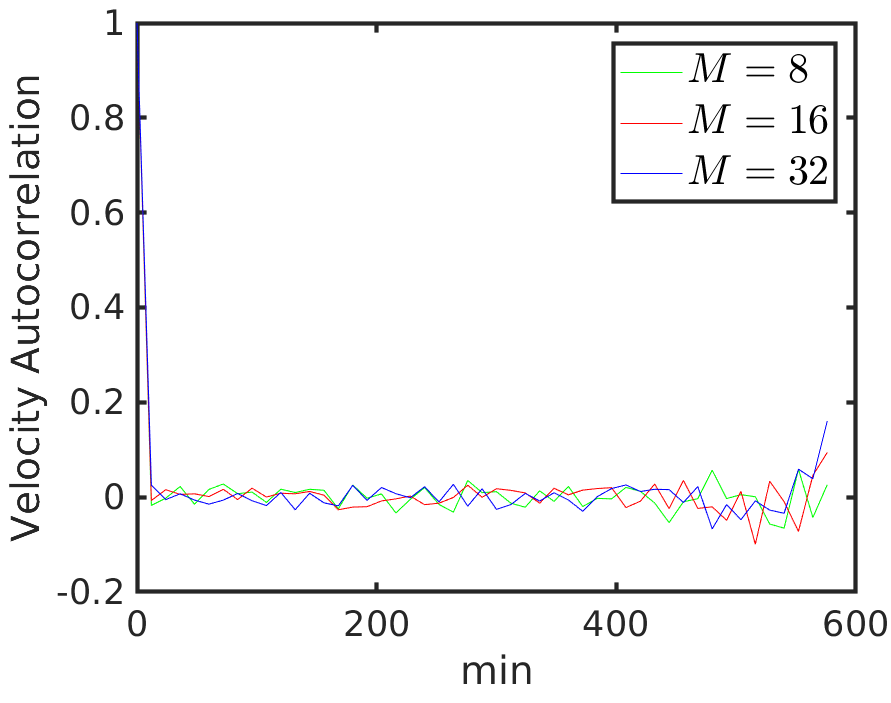}
}
\caption{Directionality ratio (left) and velocity autocorrelation (right) for $M=8$ (green), $16$ (red), $32$ (blue).}
\label{figure: persistence uniform}
\end{figure}

\begin{table}[h]
\centering
\def\arraystretch{1.5}
\begin{tabular}{|c|c c c|}
\hline
M & 8 & 16 & 32\\
\hline
$s_{av}$, $\mu m/min$ & 1.7595 & 2.4845 & 3.6047\\
$\bar{\beta}$, 1 & 1.0683 & 1.0035 & 1.0552\\
$\beta_0$, $\mu m^2/min^{\bar{\beta}}$ & 2.1493 & 3.3971 & 4.6308\\
$\bar{r}$, 1 & 0.0483 & 0.0408 & 0.0497\\
\hline

\end{tabular}
 \caption{Parameters obtained from the simulations.}
 \label{table: uniform parameters}
\end{table}

Although our results show in Figure \ref{figure: uniform results} (d-f) the mean-squared displacement scales diffusively (i.e. linearly) with time, this is not consistent with the reported results. For example, Dietrich et. al. \cite{dieterich2008anomalous}, Liang et. al. \cite{liang2008persistent}, and Liu et. al. \cite{liu2015confinement} showed that the displacement scales superdiffusively. In these studies it was experimentally found that the time scaling went as $\sim t^{\bar{\beta}}$, where $\bar{\beta}\approx1.2-1.5$. The primary reason why, in our case, we have diffusive behavior is that, given a certain state of adhesion sites, there is a complete circular symmetry of the probability rates $a_j^\pm$ with respect to $\mathbf{x}_n$ variable. Due to this symmetry, then, the probability of moving in one direction is exactly the same as the probability of moving in the opposite direction if we rotate $\mathbf{x}_n$ by $\pi$ radians. Hence, somewhat reminiscent of a random walk, we obtain a diffusive time scaling of the squared displacement. Moreover, this symmetry of the probability rates is somewhat idealistic, since it implies that the signaling activity relevant for adhesion dynamics is homogeneous within a cell. One of the ways to break this symmetry, is to multiply each binding probability rate $a_j^+$ by $1+u$, where $u\sim U(-\delta,\delta)$ is uniformly distributed on the interval $(-\delta,\delta)$ with $\delta\in(0,1)$. Then, on average, the rates are unmodified\footnote{The multiplication factor $1+u$ for each $j=1,\ldots,M$ of every cell is computed at the beginning of simulations and is held fixed thereafter.}. This way, we not only simulate a non-homogeneous binding rate (and hence, for example, non-homogeneous Rac activity) within a cell, but also simulate otherwise completely identical copies of cells. Such a modification, where we do not explicitly apply a directed, predefined bias can be referred to as chemokinesis \cite{PDY09}. 

\begin{figure}
%\vspace{0mm}
\centering
%First row
\subfloat[]
{
	\includegraphics[width=40mm,height=37mm]{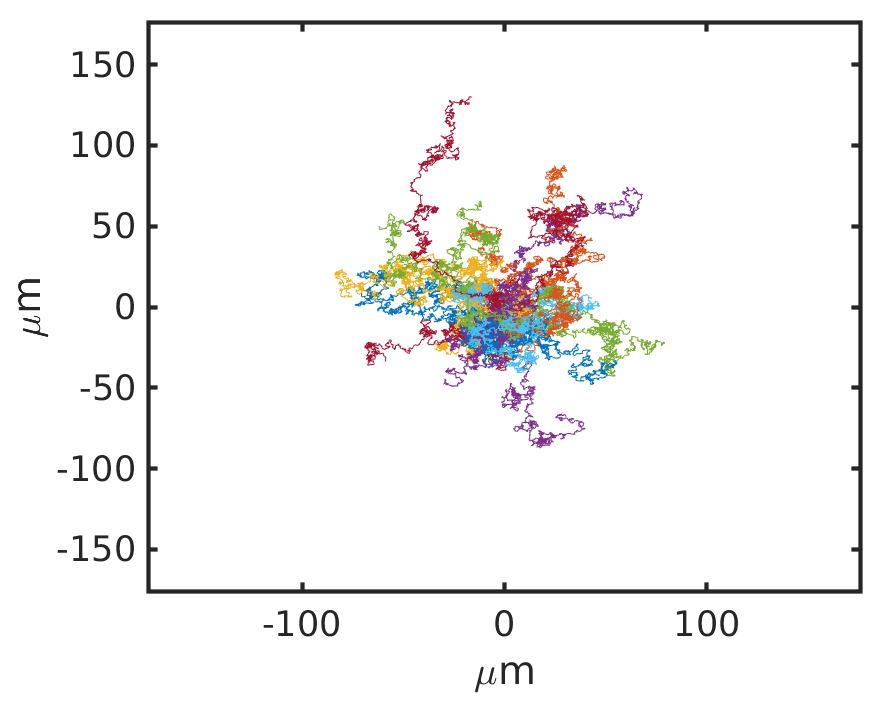}
}
\subfloat[]
{
	\includegraphics[width=40mm,height=37mm]{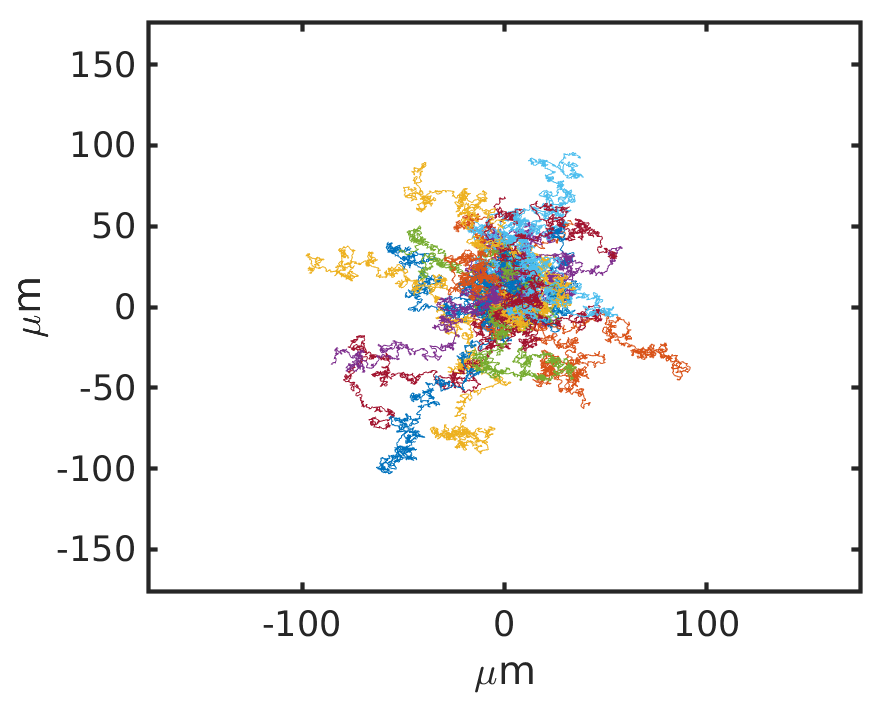}
}
\subfloat[]
{
	\includegraphics[width=40mm,height=37mm]{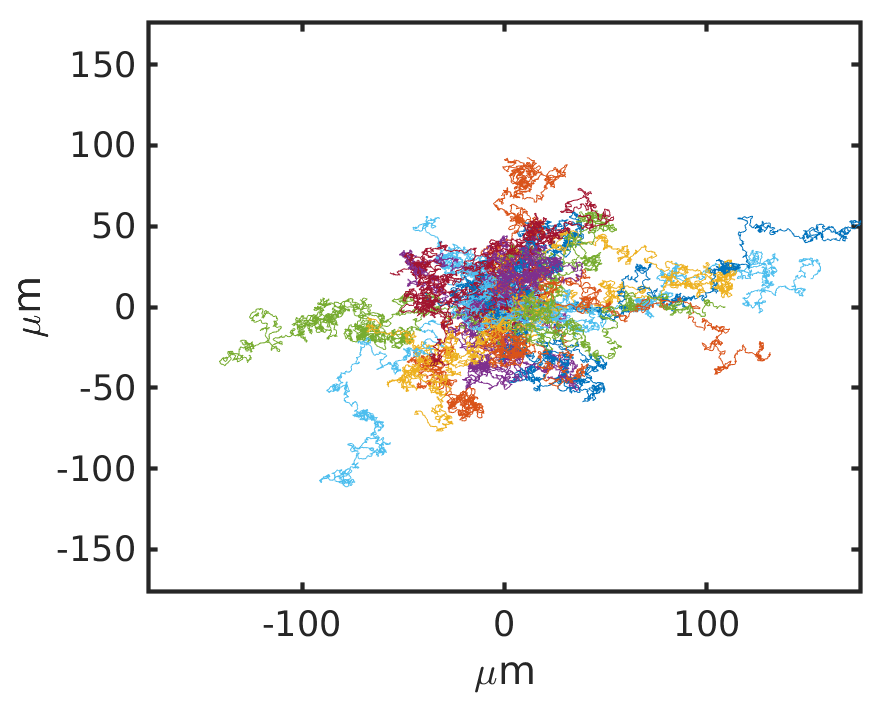}
}
\hspace{0mm}
%Second row
\subfloat[]
{
	\includegraphics[width=40mm,height=37mm]{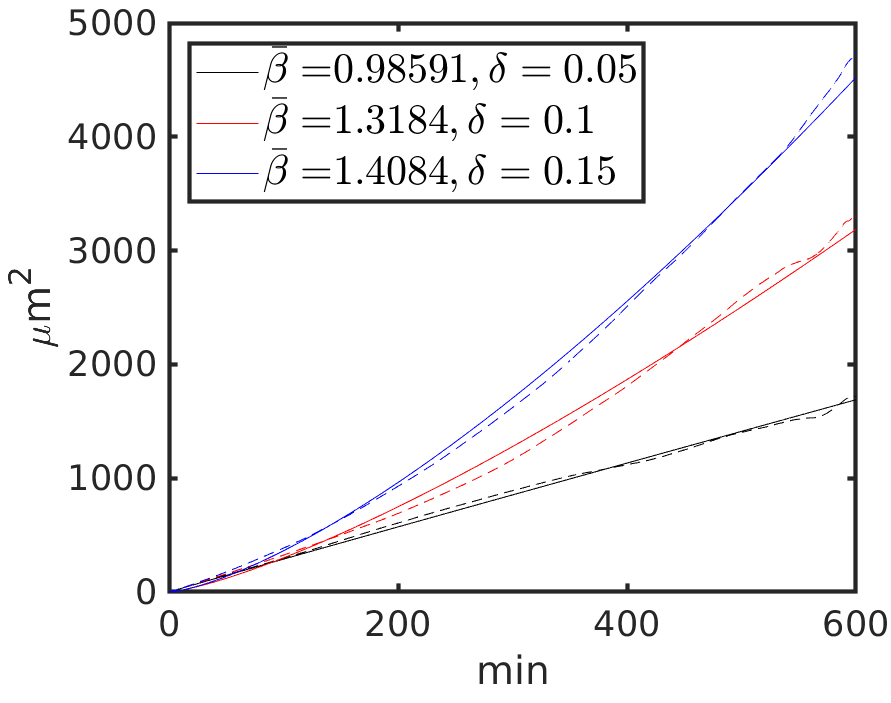}
}
\subfloat[]
{
	\includegraphics[width=40mm,height=37mm]{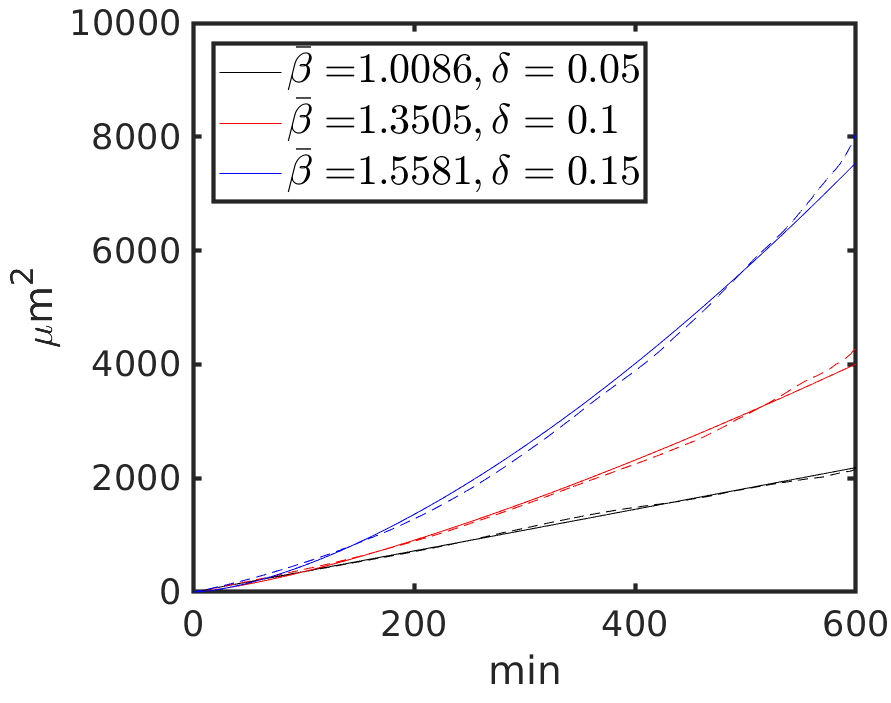}
}
\subfloat[]
{
	\includegraphics[width=40mm,height=37mm]{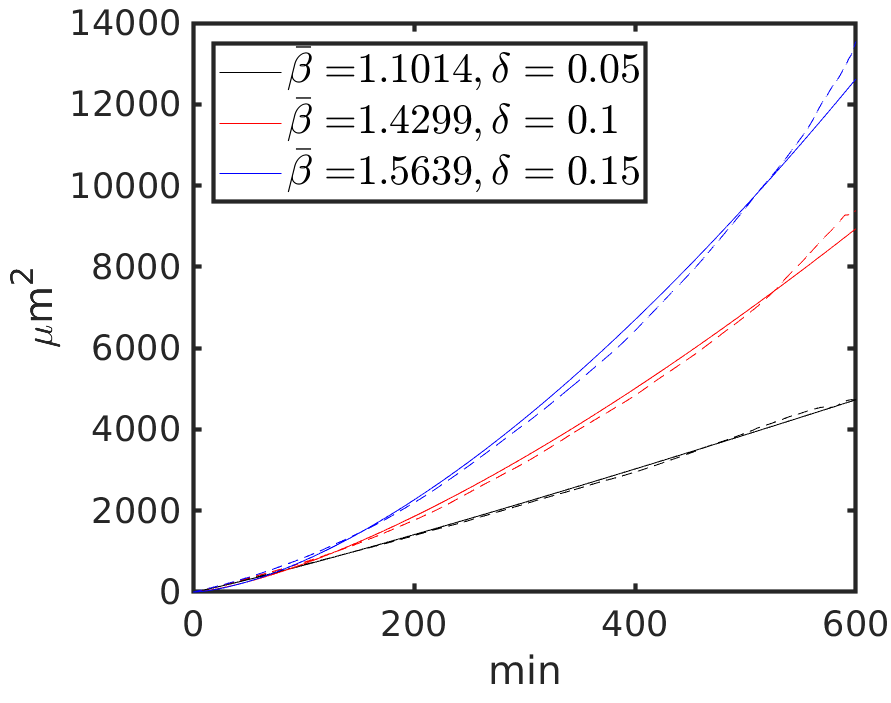}
}
\hspace{0mm}
%Thrid row
\subfloat[]
{
	\includegraphics[width=40mm,height=37mm]{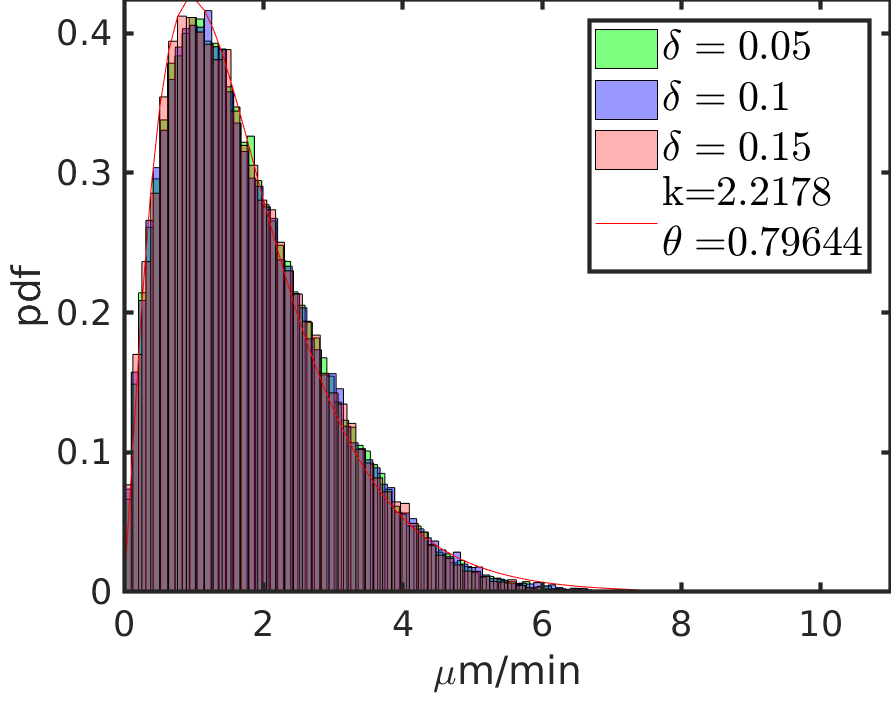}
}
\subfloat[]
{
	\includegraphics[width=40mm,height=37mm]{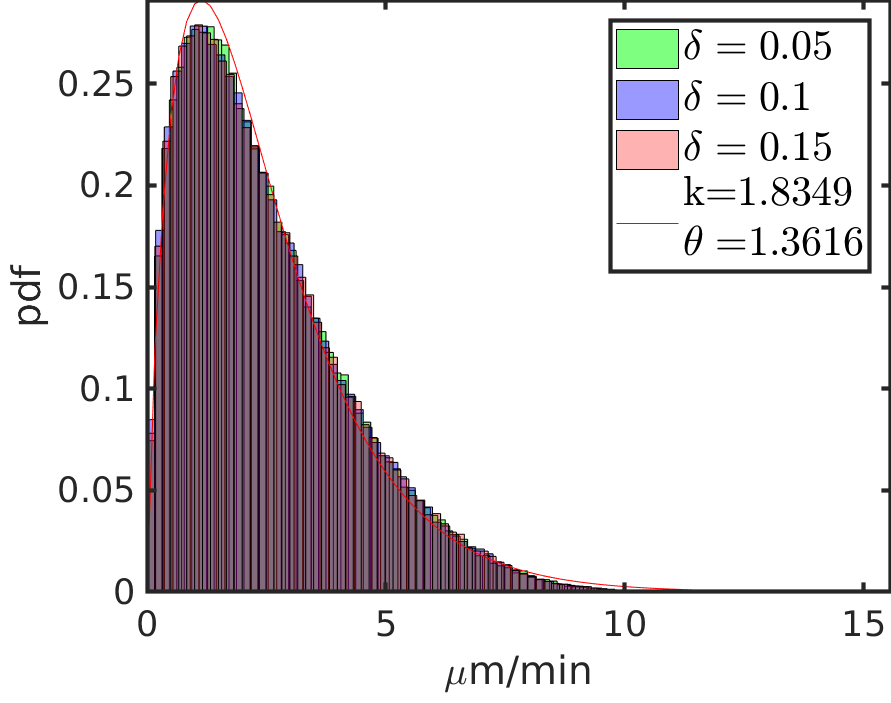}
}
\subfloat[]
{
	\includegraphics[width=40mm,height=37mm]{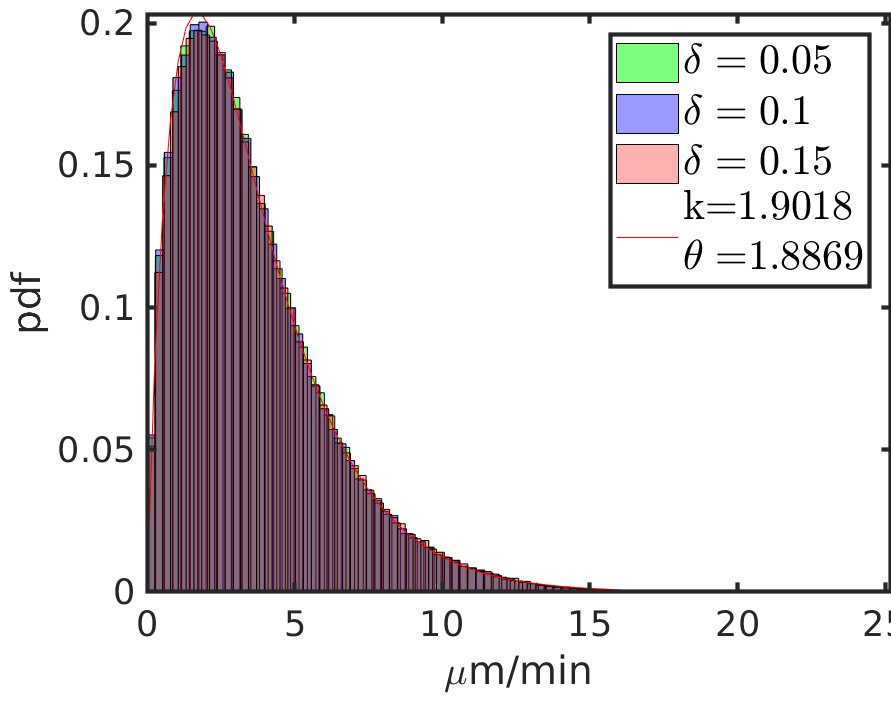}
}
\hspace{0mm}
%Fourth row
\subfloat[]
{
	\includegraphics[width=40mm,height=37mm]{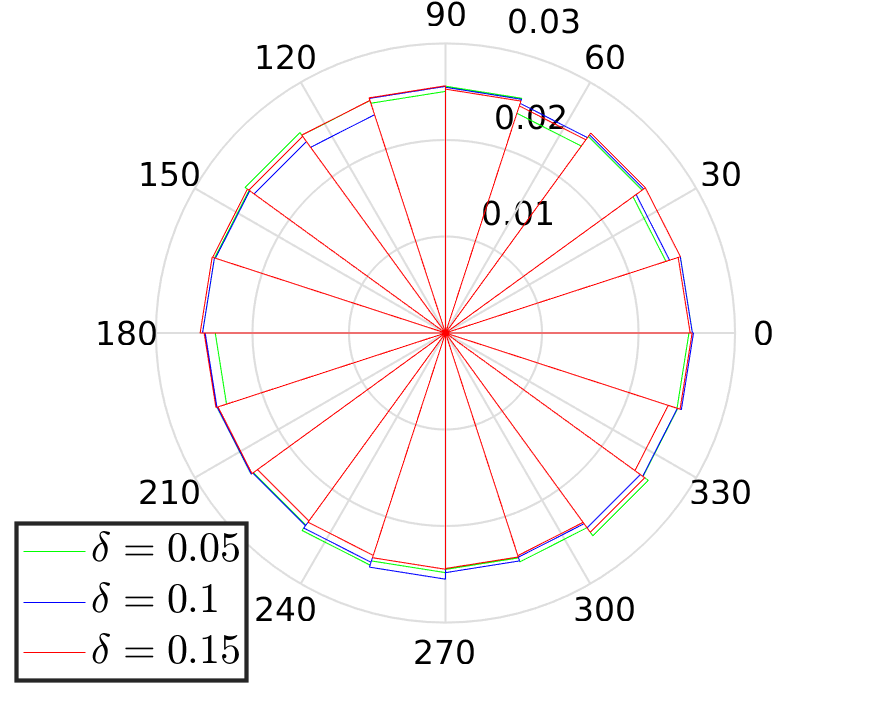}
}
\subfloat[]
{
	\includegraphics[width=40mm,height=37mm]{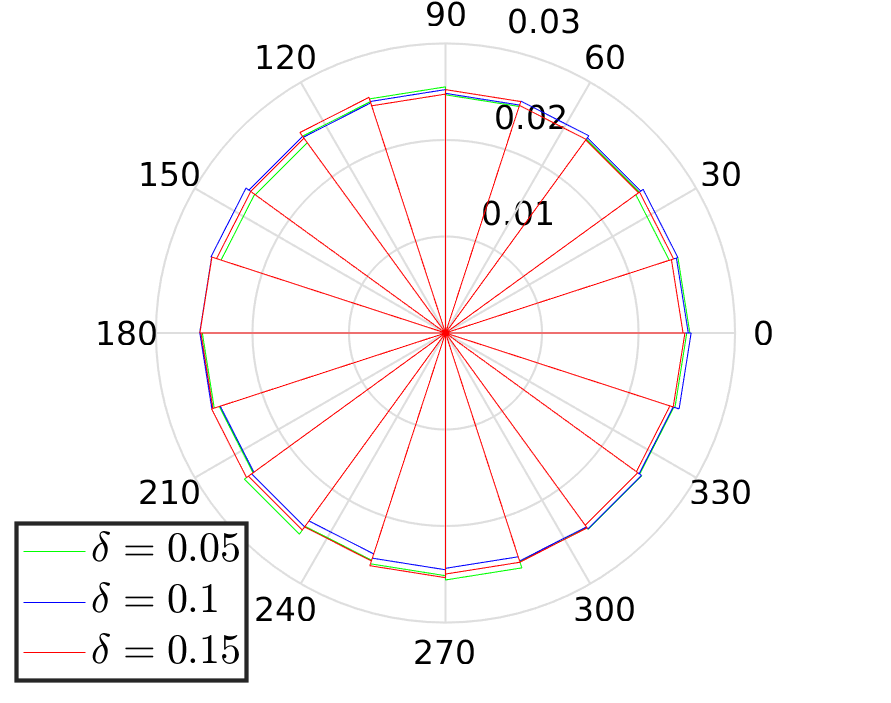}
}
\subfloat[]
{
	\includegraphics[width=40mm,height=37mm]{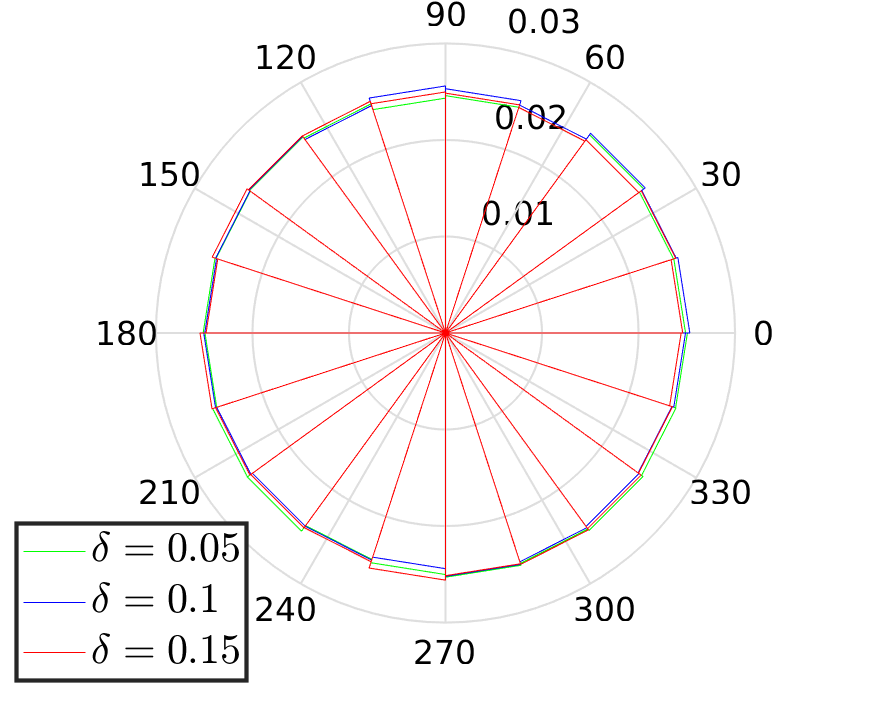}
}
\caption{Simulation results with $M=8,16,32$ adhesions in the first, second, and third columns respectively, and with various values of $\delta$. (a-c) Trajectories of 27 cell centroids $\mathbf{x}(t)$ with $\delta=0.1$. (d-f) Mean-squared displacements $msd(t)$ (solid) and fitted $\widehat{msd}(t)$ (dash) with $\delta=0.05$ (black), $0.1$ (red), $0.15$ (blue). (g-i) Superimposed histograms of speed probability density functions and fitted density function of gamma distribution (solid red) with average parameters $k$ and $\theta$ (see text for details). (j-l) Superimposed histogram of relative frequency of normalized velocities. }
\label{figure: uniform bias results}
\end{figure}
\begin{table}[h]
\centering
\def\arraystretch{1.5}
\setlength\tabcolsep{5.5pt}
\begin{tabular}{|c|c c c|c c c|c c c|}
\hline

M & \multicolumn{3}{c|}{8} & \multicolumn{3}{c|}{16} & \multicolumn{3}{c|}{32}\\
\hline
$\delta$ & 0.05 & 0.1 & 0.15 & 0.05 & 0.1 & 0.15 & 0.05 & 0.1 & 0.15\\
\hline
$\bar{\beta}$, 1 & 0.9859 & 1.3184 & 1.4084 & 1.0086 & 1.3505 & 1.5581 & 1.1014 & 1.4299 & 1.5639\\
\hline
\begin{tabular}{@{}c@{}}$s_{av}$, \\ $\mu m/min$\end{tabular}  & 1.7656 & 1.7768 & 1.7557 & 2.4918 & 2.5009 & 2.5021 & 3.5818 & 3.5735 & 3.6104\\
\hline
\begin{tabular}{@{}c@{}}$\beta_0$, \\ $\mu m^2/min^{\bar{\beta}}$\end{tabular}& 3.0846 & 0.6934 & 0.5534 & 3.4517 & 0.7103 & 0.3543 & 4.1257 & 0.9536 & 0.5716\\
\hline
$\bar{r}$, 1 & 0.0452 & 0.0519 & 0.0597 & 0.0440 & 0.0513 & 0.0587 & 0.0522 & 0.05 & 0.0623\\ 
\hline
\end{tabular}
 \caption{Parameters obtained from the simulations with varying $\delta$.}
 \label{table: uniform bias parameters}
\end{table}

The effect of modifying the rates $a_j^+$ with $\delta=0.05,0.1,0.15$ can be seen in Figure \ref{figure: uniform bias results}. The cell trajectories, depicted in Figure \ref{figure: uniform bias results} (a-c),  show that the motion consists of periods with relatively regular path intermingled with highly irregular and random movement. In Figure \ref{figure: uniform bias results} (d-f) we see that the rate modification leads to a superdiffusive time scaling of the mean-squared displacement, as the exponent $\bar{\beta}$ becomes larger than one and falls within the experimentally observed range of values \cite{dieterich2008anomalous}, \cite{liang2008persistent}, \cite{liu2015confinement}. Moreover, we see that as $\delta$ increases, so does $\bar{\beta}$, and the increase of the latter is more pronounced for a larger number of adhesion sites $M$ (see also Table \ref{table: uniform bias parameters}). This is due to the fact that as each adhesion site is modified independently, the variance of the modified rates of a cell grows with the number of FAs, which corresponds to increased cell polarization, and hence more prominent persistent motion resulting in higher values of $\bar{\beta}$. However, the distribution of speeds for the corresponding values of $M$ is virtually identical to the case with the unmodified probability rates (Figure \ref{figure: uniform bias results} (g-i) and Table \ref{table: uniform bias parameters}). The uniform distribution of normalized velocities also remained unchanged (Figure \ref{figure: uniform bias results} (j-l)). These results suggest that in the absence of spatial cues, the distribution of speeds for a given adhesiveness (represented by the total number of adhesions $M$) remains invariant under symmetry breaking of adhesion binding, while the diffusion type (normal vs. anomalous) does not. Thus, the adhesion number and its turnover is a major determinant of the cell speed, which is consistent with  \cite{pavnkova2010molecular}.  

\begin{figure}[h]
%\vspace{0mm}
\centering
%First row
\subfloat[]
{
	\includegraphics[width=40mm,height=37mm]{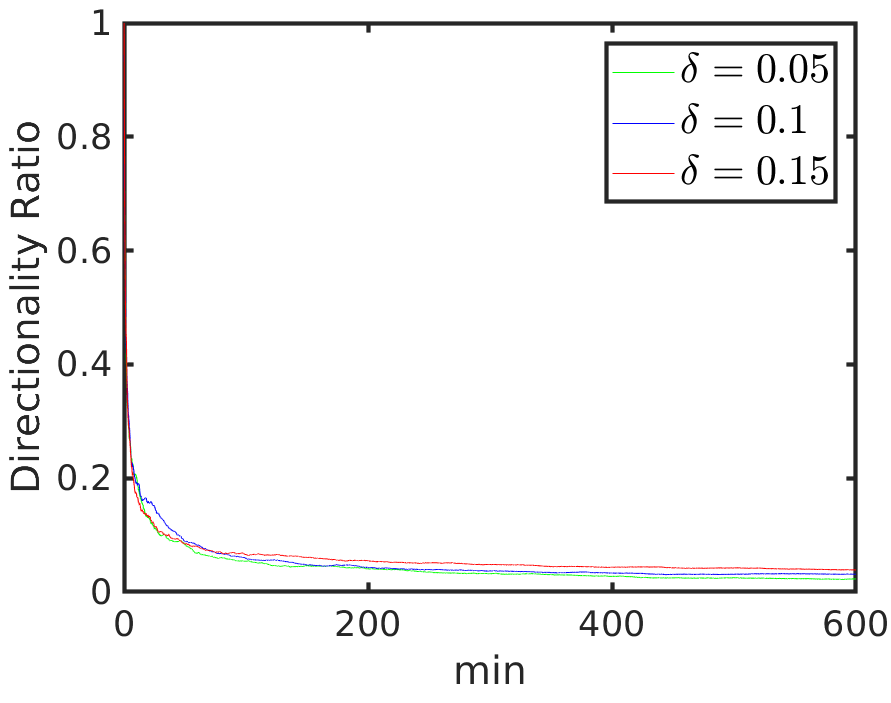}
}
\subfloat[]
{
	\includegraphics[width=40mm,height=37mm]{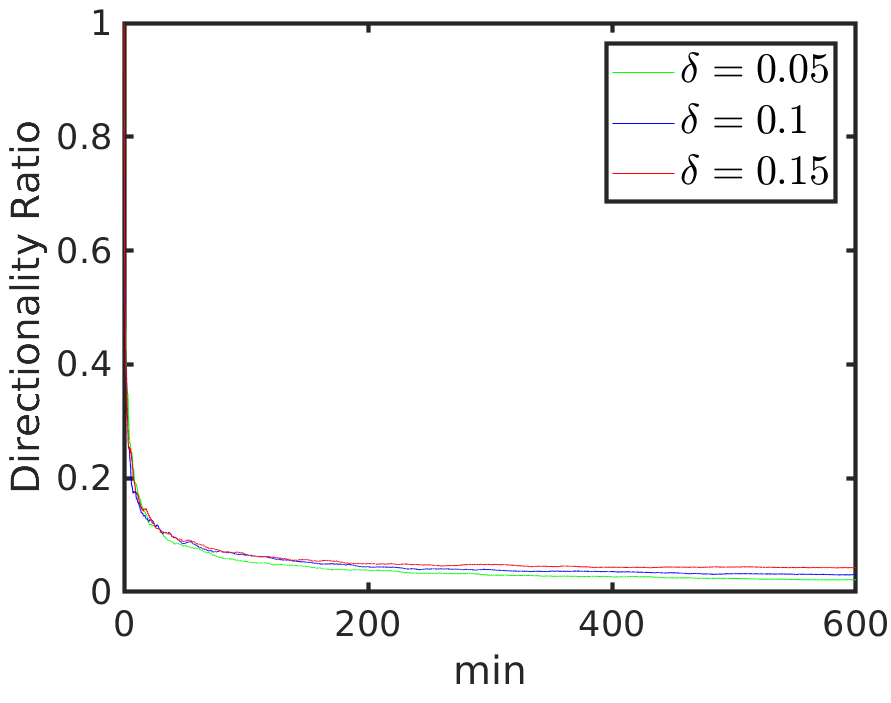}
}
\subfloat[]
{
	\includegraphics[width=40mm,height=37mm]{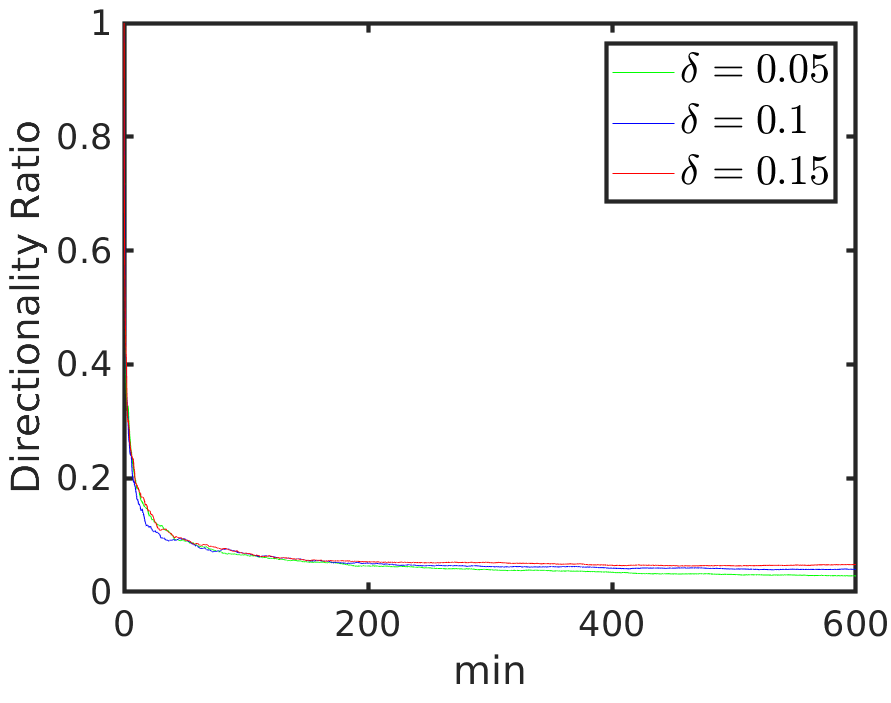}
}
\hspace{0mm}
\subfloat[]
{
	\includegraphics[width=40mm,height=37mm]{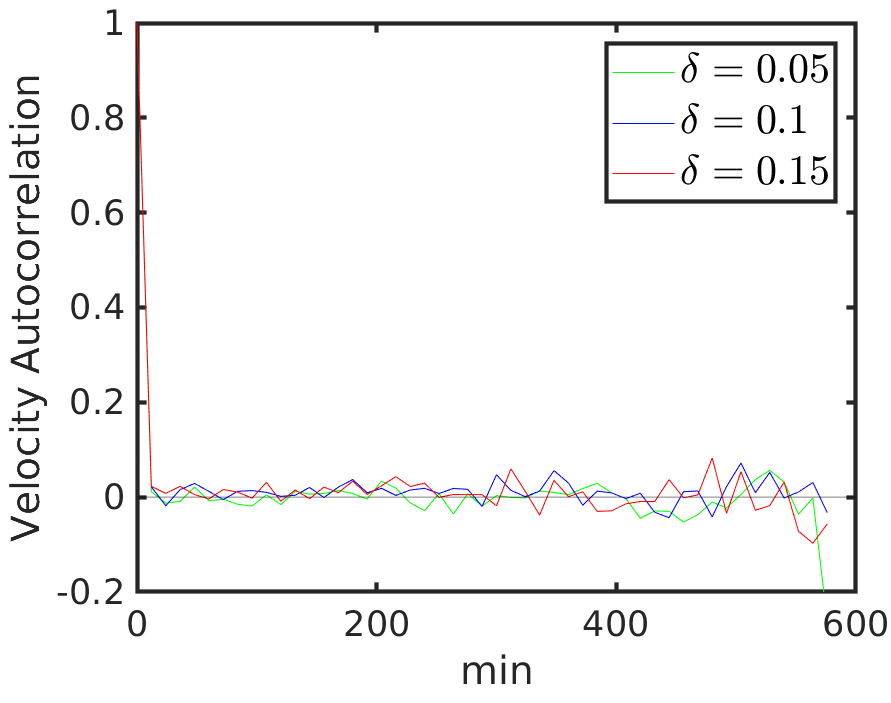}
}
\subfloat[]
{
	\includegraphics[width=40mm,height=37mm]{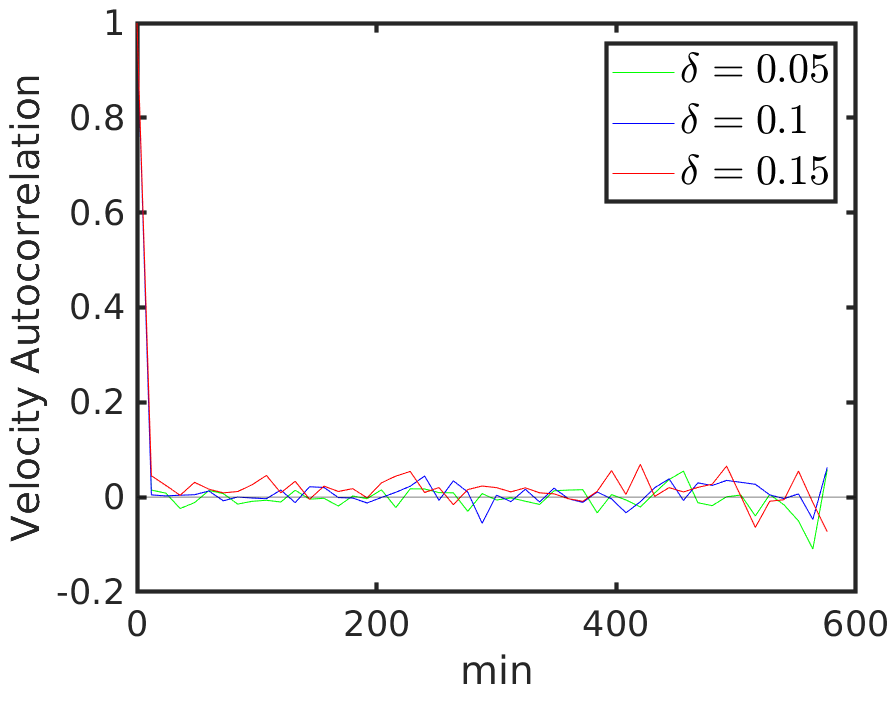}
}
\subfloat[]
{
	\includegraphics[width=40mm,height=37mm]{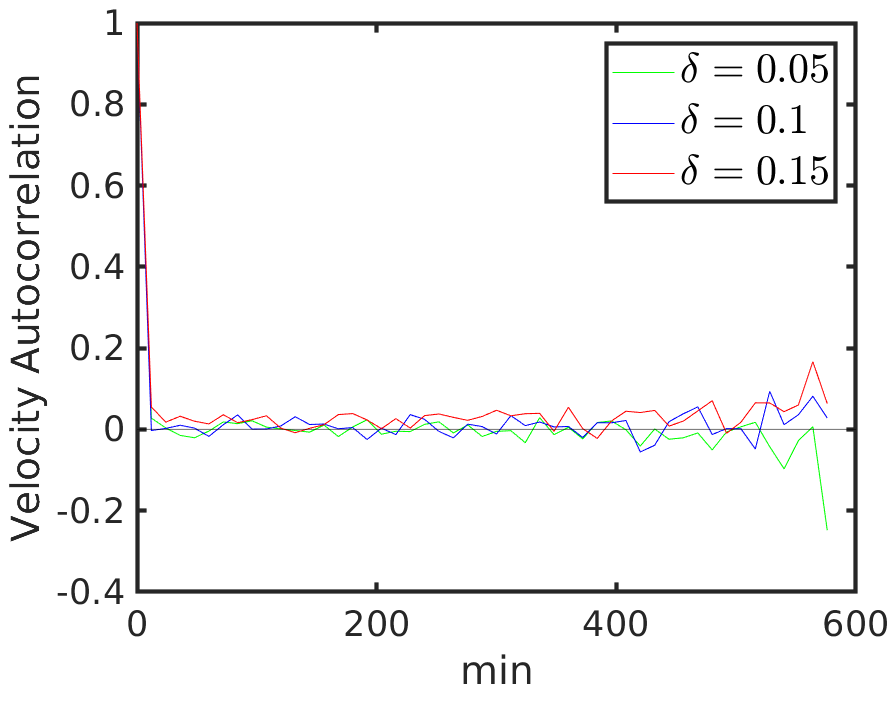}
}
\caption{Persistence of motion for cells with $M=8,16,32$ adhesions in the first, second, and third columns, respectively, and with $\delta=0.05$ (green), $0.1$ (blue), $0.15$ (red). (a-c) Directionality ratio. (d-f) Velocity autocorrelation.}
\label{figure: persistence uniform bias}
\end{figure}

Note that the increased values of $\bar{\beta}$ indicate that the cells explore a larger surface area \cite{gorelik2014quantitative}. However, other indicators of motion persistence are not affected significantly (Figure \ref{figure: persistence uniform bias}), although migration paths become slightly straighter, as indicated by increased values of $\bar{r}$ (Table \ref{table: uniform bias parameters}). These results suggest that symmetry breaking of adhesion binding may allow cells to explore larger area without introducing velocity correlations (Figure \ref{figure: persistence uniform bias}(d-f)).

As cell polarization is required for migration even in the absence of external signals, it is not surprising that our results show that an imbalance of adhesion formation within a cell leads to experimentally observed superdiffusive scaling of the squared displacement \cite{dieterich2008anomalous},\cite{liang2008persistent}, \cite{liu2015confinement}. Nevertheless, this highlights a potential mechanism of anomalous diffusion. In the following, we examine whether our model gives biologically consistent results in the case of externally induced polarization.

\subsection{External cue gradient}\label{section: external cue}
We first investigate how cell trajectories are varied in the presence of an external cue gradient. If a cue, for example, is a chemoattractant, then it is well known that adhesion formation in a cell is biased in the direction of the attractant. Thus, to simulate such biased migration, we let the functions $Q_{cue}$ and $q$ to have the following form (recall equation \eqref{equation: binding propensity}):
\begin{align*}
Q_{cue}(\mathbf{x})&=
\begin{cases*}
1+\delta_{E}x_2,\text{ if } x_2\geq0\\
1,\phantom{+\delta_{E}x_2a} \text{ else} 
\end{cases*}\\
q(Q_{cue}(\mathbf{x})) &= Q_{cue}(\mathbf{x}),
\end{align*}
where $\delta_E$ represents the gradient magnitude and $x_2$ is the second component of $\mathbf{x}$. Here, for simplicity we took the identity function for $q$ and a linear cue gradient in the $y$ coordinate. This cue can represent, for example, density of ECM or concentration of a chemoattractant. Thus, we simulate, respectively, hapto- or chemotactic migration.      

\begin{figure}[!h]
%\vspace{0mm}
\centering
%First row
\subfloat[]
{
	\includegraphics[width=40mm,height=37mm]{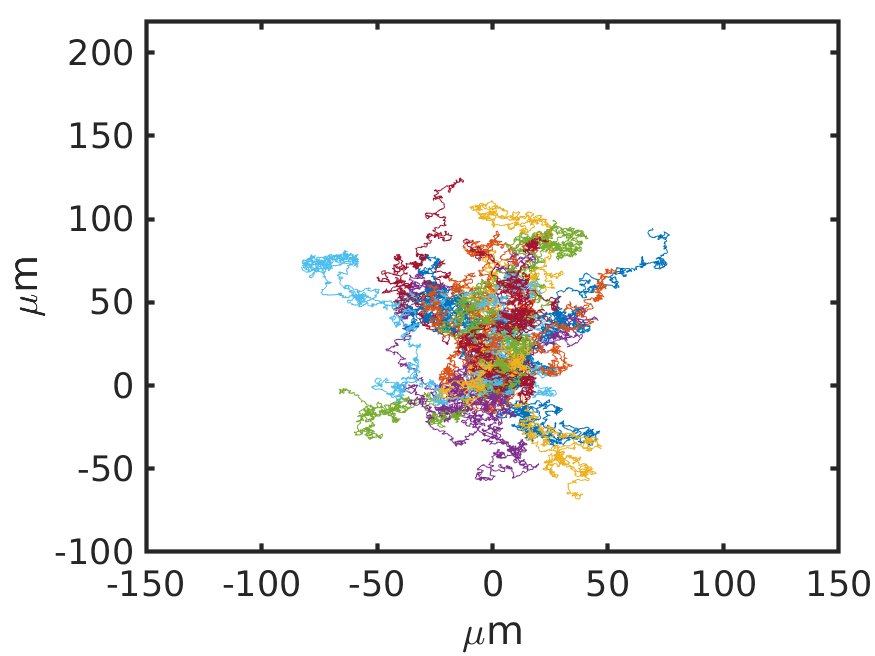}
}
\subfloat[]
{
	\includegraphics[width=40mm,height=37mm]{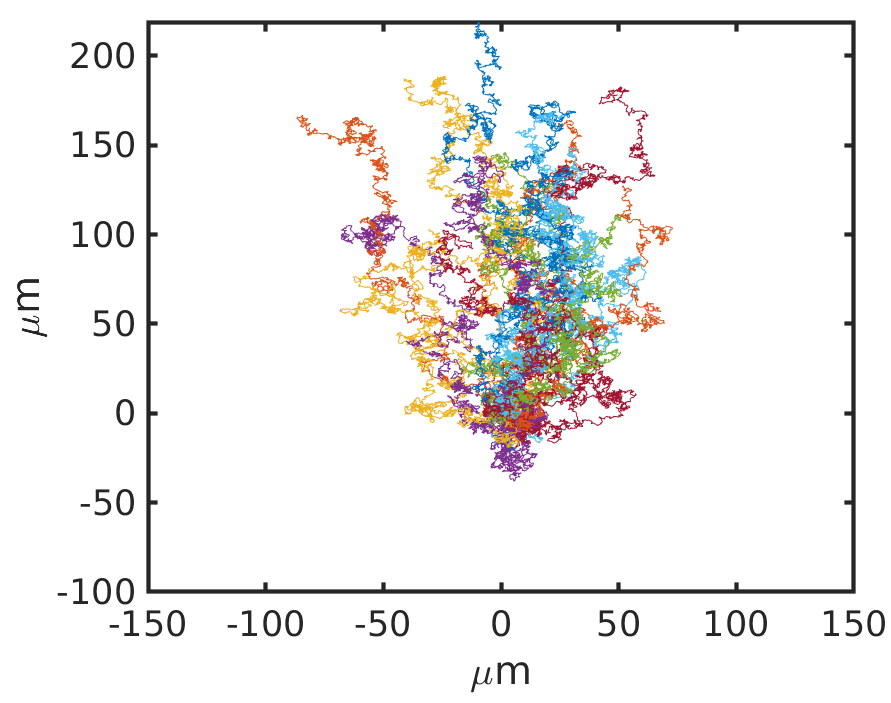}
}
\subfloat[]
{
	\includegraphics[width=40mm,height=37mm]{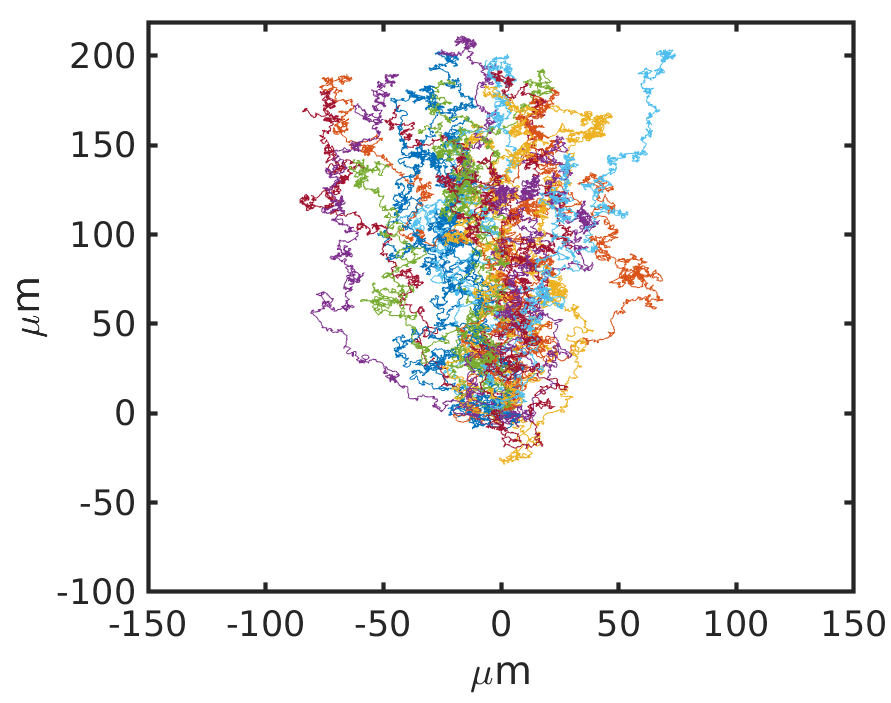}
}
\hspace{0mm}
%Second row
\subfloat[]
{
	\includegraphics[width=40mm,height=37mm]{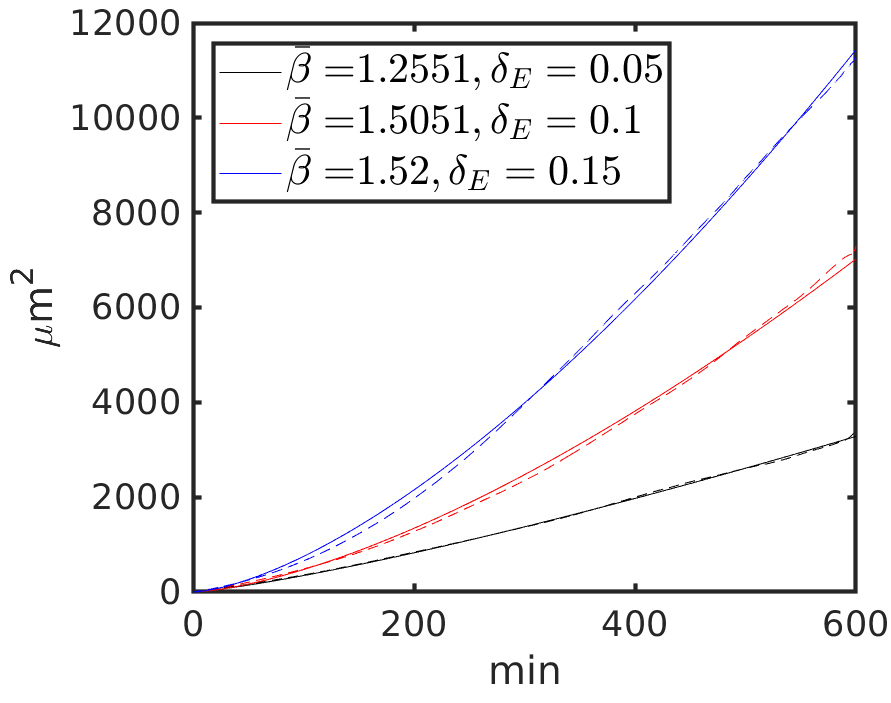}
}
\subfloat[]
{
	\includegraphics[width=40mm,height=37mm]{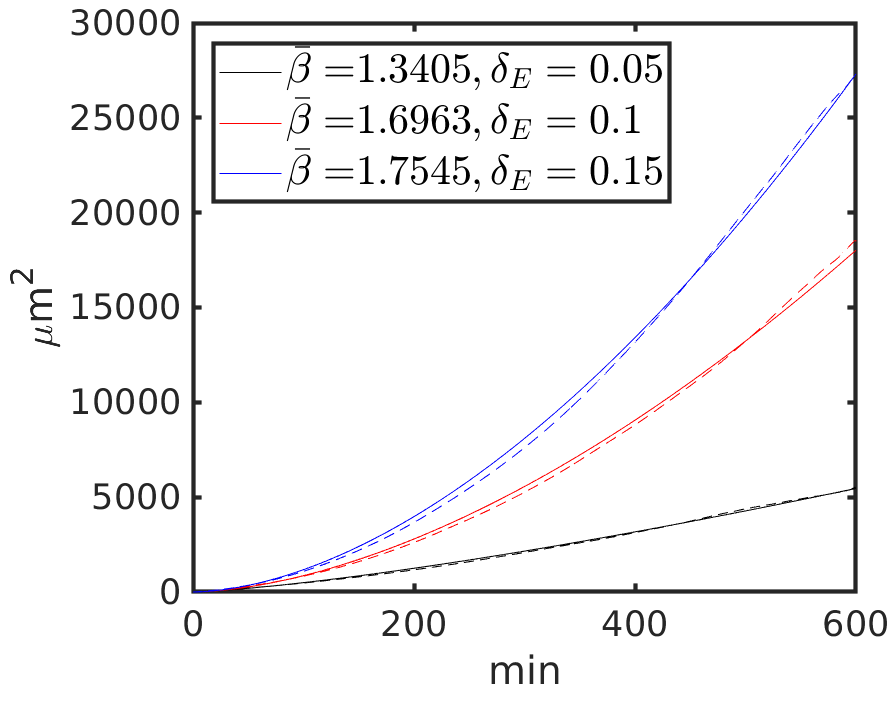}
}
\subfloat[]
{
	\includegraphics[width=40mm,height=37mm]{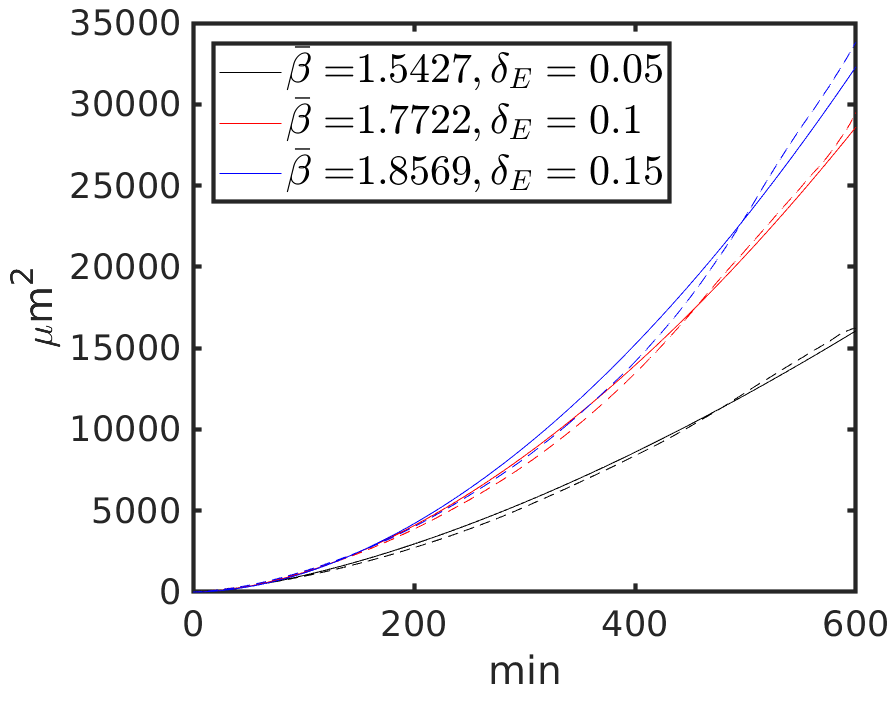}
}
\hspace{0mm}
%Thrid row
\subfloat[]
{
	\includegraphics[width=40mm,height=37mm]{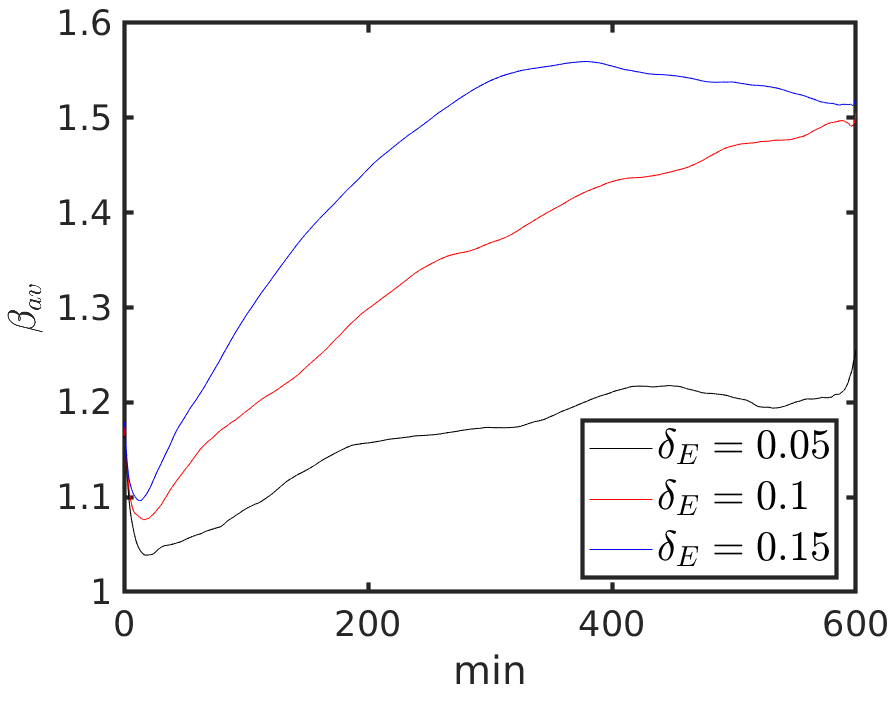}
}
\subfloat[]
{
	\includegraphics[width=40mm,height=37mm]{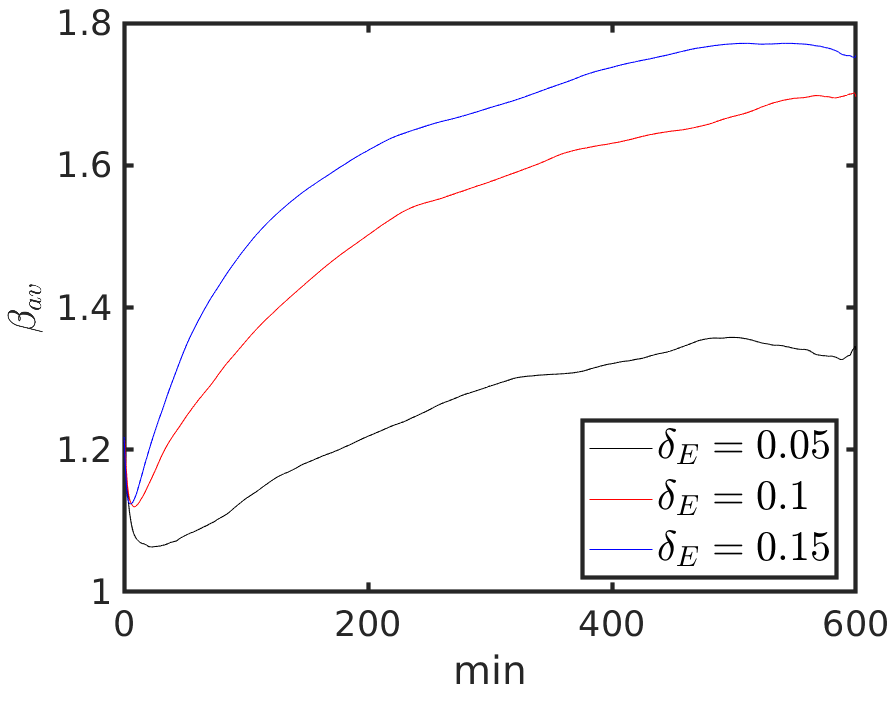}
}
\subfloat[]
{
	\includegraphics[width=40mm,height=37mm]{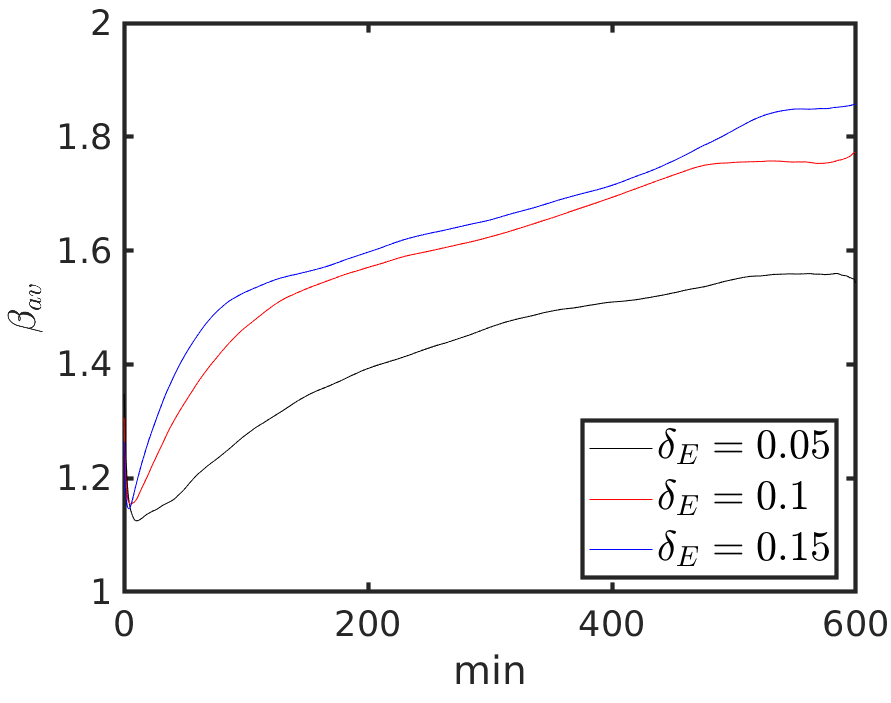}
}

\caption{Simulation results with $M=8,16,32$ adhesions in the first, second, and third columns respectively, and with various values of $\delta_{E}$. (a-c) Trajectories of 27 cell centroids $\mathbf{x}(t)$ with $\delta_{E}=0.1$. (d-f) Mean-squared displacements $msd(t)$ (solid) and fitted $\widehat{msd}(t)$ (dash) with $\delta_E=0.05$ (black), $0.1$ (red), $0.15$ (blue). (g-i) Time-averaged exponents $\beta_{av}(t)$ }
\label{figure: ECM bias results}
\end{figure}

\begin{figure}[!h]
%\vspace{0mm}
\centering
%First row
\subfloat[]
{
	\includegraphics[width=40mm,height=37mm]{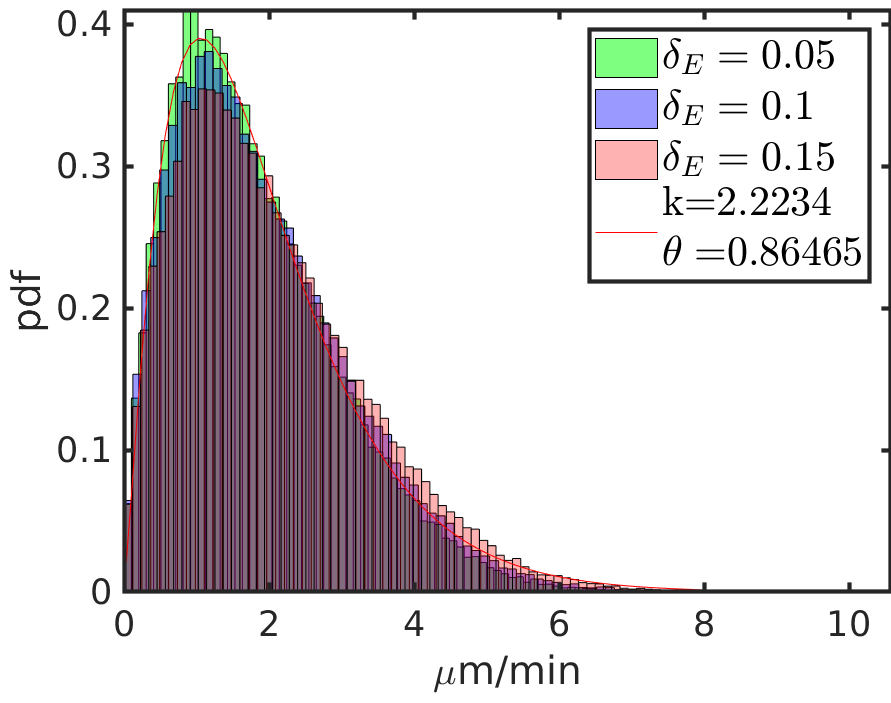}
}
\subfloat[]
{
	\includegraphics[width=40mm,height=37mm]{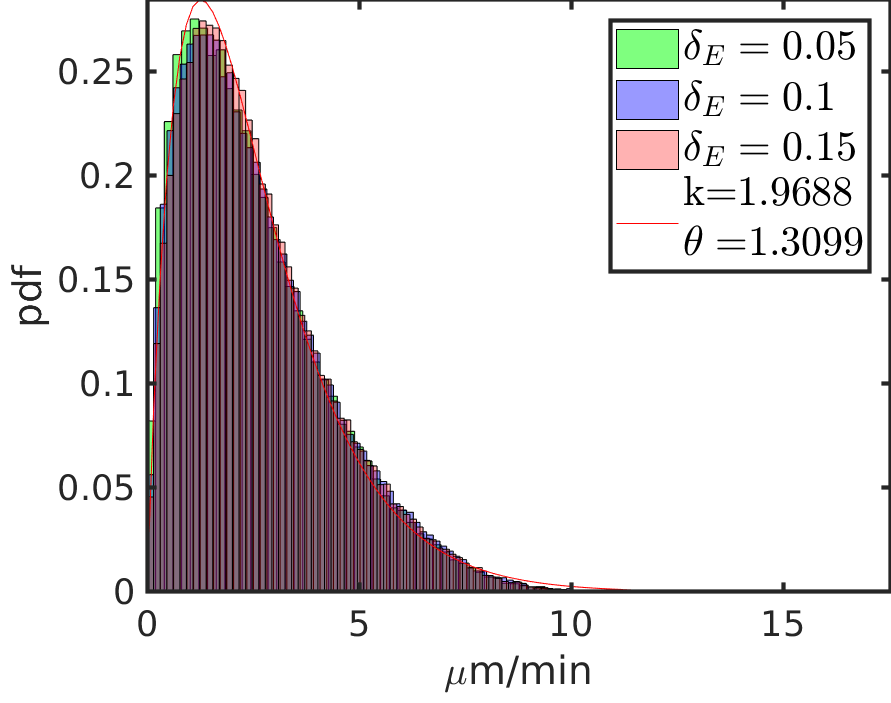}
}
\subfloat[]
{
	\includegraphics[width=40mm,height=37mm]{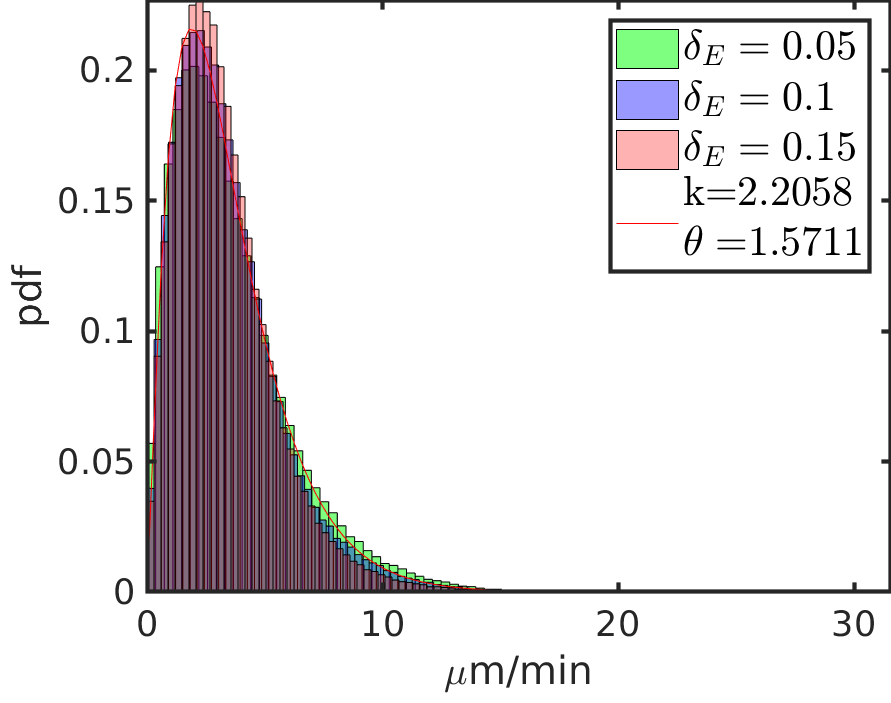}
}
\hspace{0mm}
%Second row
\subfloat[]
{
	\includegraphics[width=40mm,height=37mm]{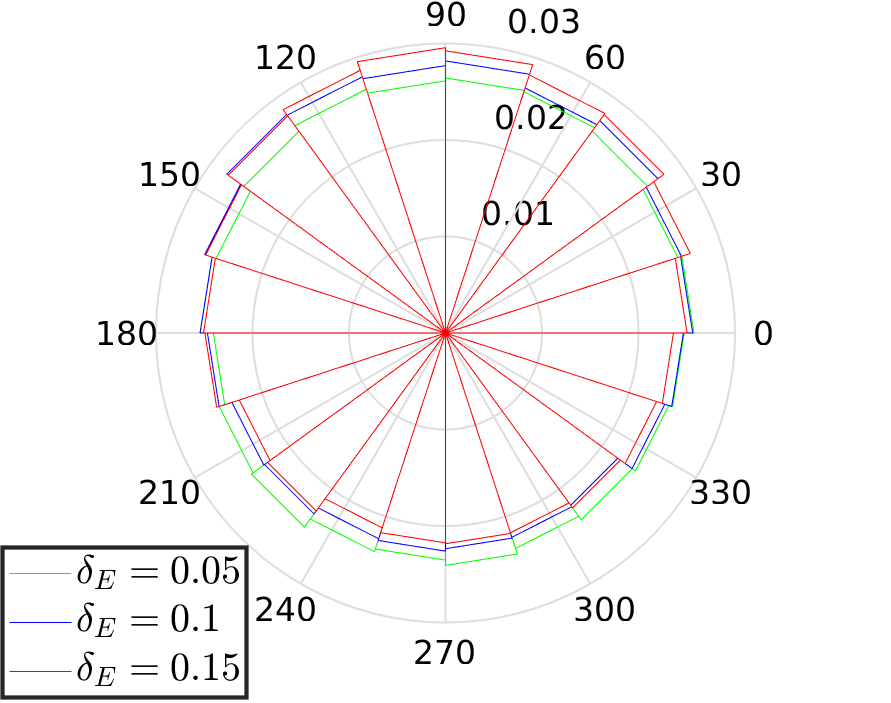}
}
\subfloat[]
{
	\includegraphics[width=40mm,height=37mm]{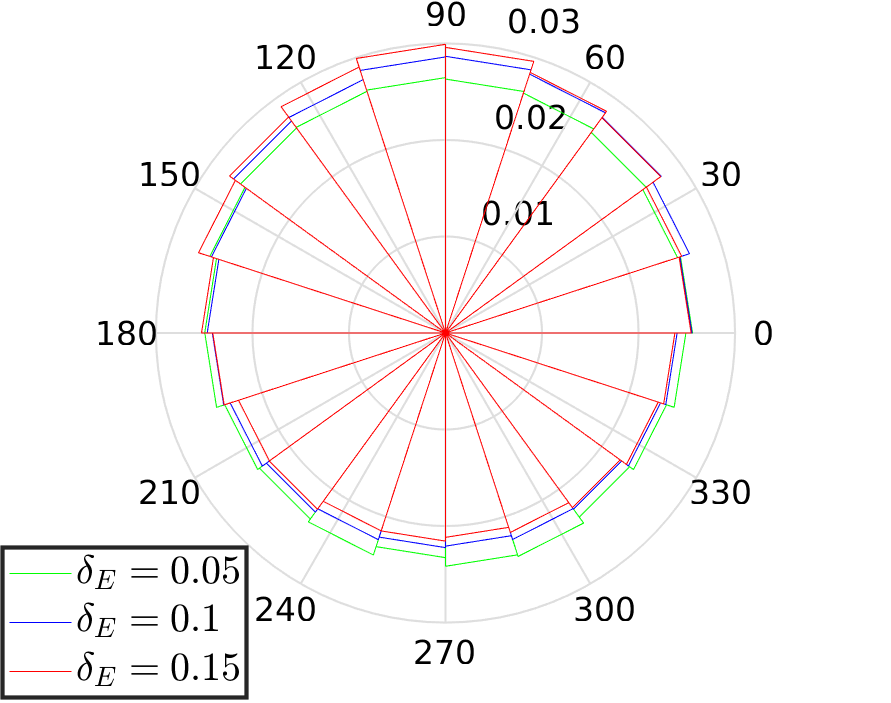}
}
\subfloat[]
{
	\includegraphics[width=40mm,height=37mm]{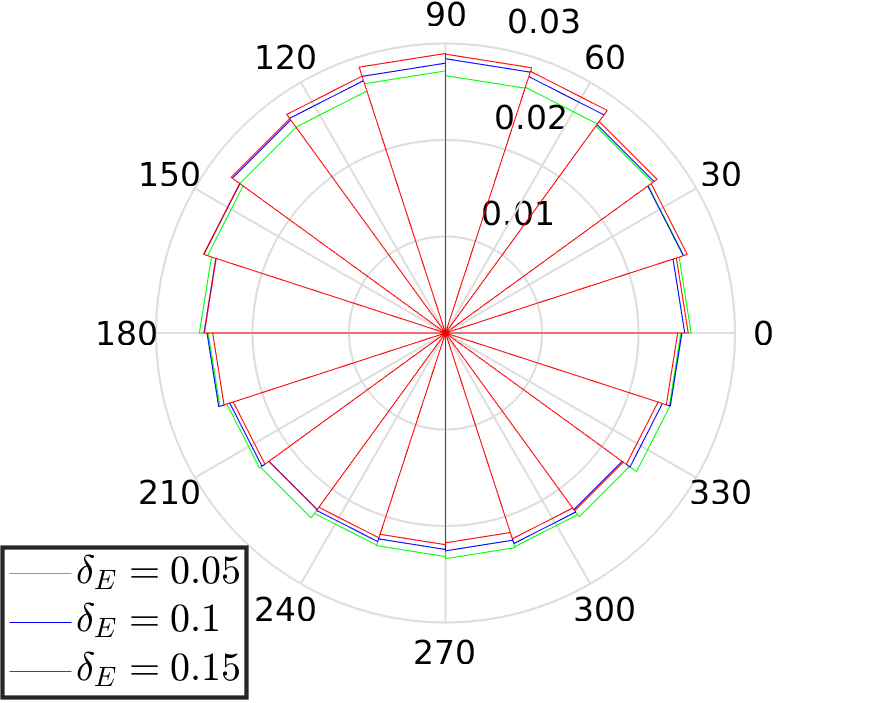}
}
\hspace{0mm}
%Third row

%Fourth row
\subfloat[]
{
	\includegraphics[width=40mm,height=37mm]{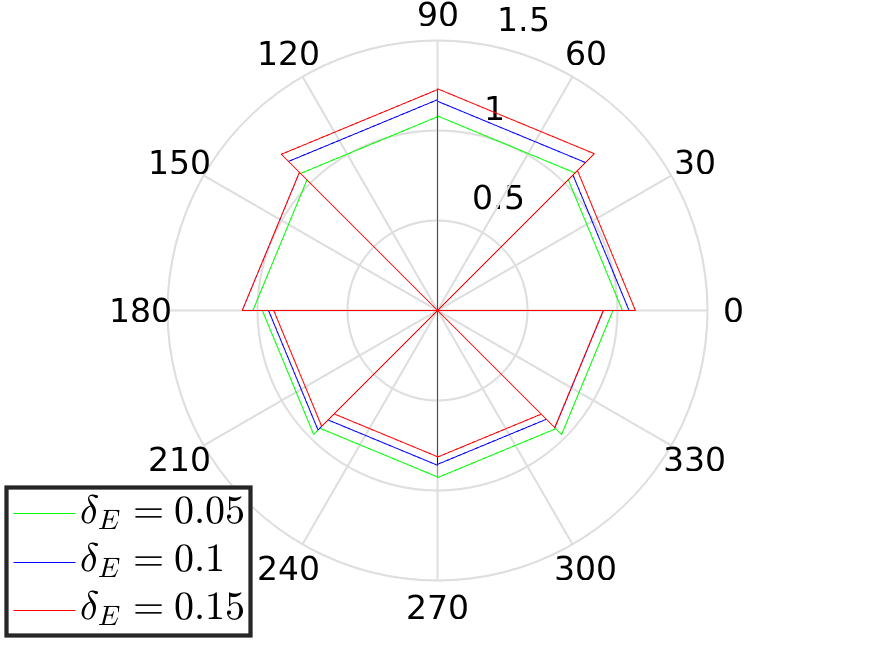}
}
\subfloat[]
{
	\includegraphics[width=40mm,height=37mm]{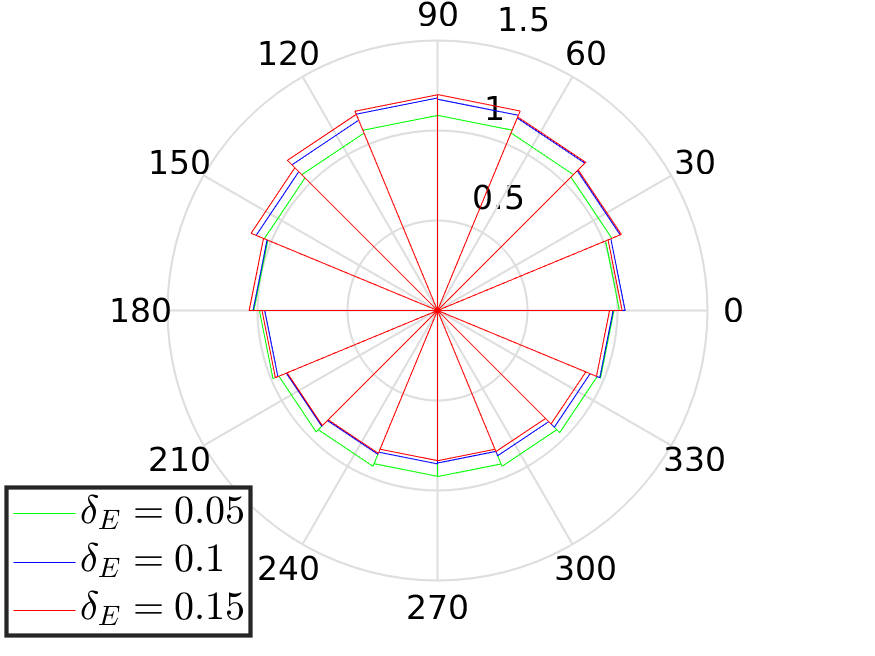}
}
\subfloat[]
{
	\includegraphics[width=40mm,height=37mm]{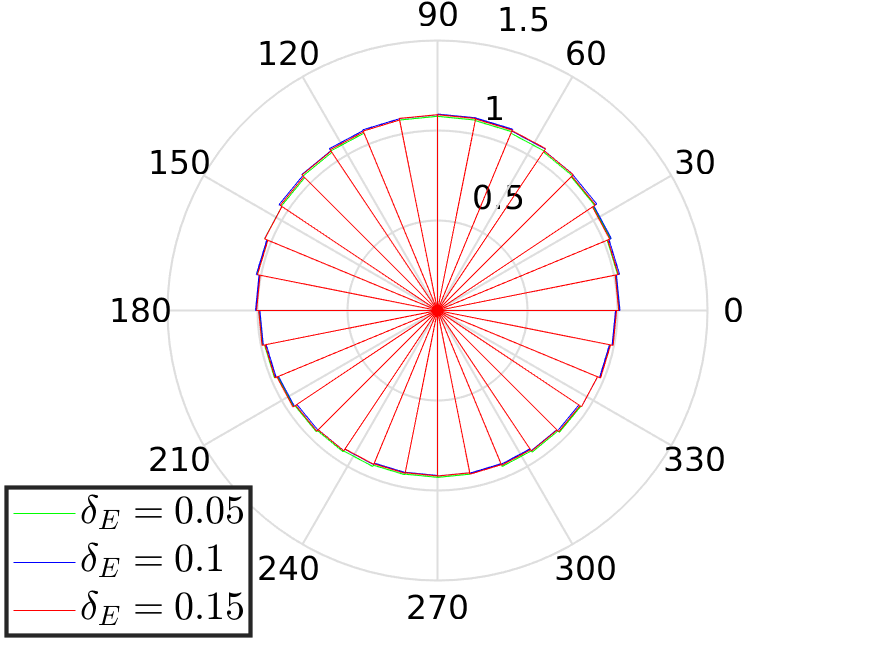}
}
\caption{Superimposed histograms of speeds, velocities and adhesion events with $M=8,16,32$ adhesions in the first, second, and third columns respectively, and with various values of $\delta_{E}$. (a-c) Speed probability density functions and fitted density function of gamma distribution (solid red) with average parameters $k$ and $\theta$ (see text for details). (d-f) Relative frequency of normalized velocities.
(g-i) Ratio of the number of binding to unbinding events in each sector, such that any given time, only one adhesion site is in each sector}
% (g-i) Ratio of the number of binding events to total number of events. (j-l) Ratio of the number of binding to unbinding events in each adhesion site.}
\label{figure: ECM bias speeds and events}
\end{figure}

\begin{table}[!h]
\centering
\def\arraystretch{1.5}
\setlength\tabcolsep{5.5pt}
\begin{tabular}{|c|c c c|c c c|c c c|}
\hline

M & \multicolumn{3}{c|}{8} & \multicolumn{3}{c|}{16} & \multicolumn{3}{c|}{32}\\
\hline
$\delta_{E}$ & 0.05 & 0.1 & 0.15 & 0.05 & 0.1 & 0.15 & 0.05 & 0.1 & 0.15\\
\hline
$\bar{\beta}$, 1 & 1.2551 & 1.5051 & 1.52 & 1.3405 & 1.6963 & 1.7545 & 1.5427 & 1.7722 & 1.8569\\
\hline
\begin{tabular}{@{}c@{}}$s_{av}$, \\ $\mu m/min$\end{tabular} & 1.8136 & 1.9133 & 2.0365 & 2.5235 & 2.5972 & 2.6089 & 3.5819 & 3.4218 & 3.3074\\
\hline
\begin{tabular}{@{}c@{}}$\beta_0$, \\ $\mu m^2/min^{\bar{\beta}}$\end{tabular} & 1.0697 & 0.4625 & 0.6845 & 1.0312 & 0.3496 & 0.3654 & 0.8319 & 0.3412 & 0.2242\\
\hline
$\bar{r}$, 1 & 0.0523 & 0.0607 & 0.0693 & 0.053 & 0.08 & 0.097 & 0.0726 & 0.1 &0.1223\\
\hline
\end{tabular}
 \caption{Parameters obtained from the simulations with varying $\delta_{E}$.}
 \label{table: ECM no bias parameters}
\end{table}
In the presence of a cue gradient, we see that the cell trajectories, shown in Figure \ref{figure: ECM bias results} (a-c), exhibit a clear trend in the direction of an increasing concentration. The corresponding plots of the mean-squared displacements show the superdiffusive time scaling in Figure \ref{figure: ECM bias results} (d-f), with the exponent $\bar{\beta}>1$ for all cases. Notice that as the number of adhesion sites $M$ increases, so does $\bar{\beta}$ for the same $\delta_E$ (see Table \ref{table: ECM no bias parameters}). Together with the trajectory plots in Figure \ref{figure: ECM bias results}, our results suggest that in the presence of an external gradient, the taxis becomes more prominent and a cell more sensitive to a cue for increasing number of FAs. Moreover, comparing with the case of a uniform environment, we see that although the amoeboid motility is more diffusive in the absence of external cues, it is also more regular and directed when a cue gradient is present (see Tables \ref{table: uniform parameters}, \ref{table: uniform bias parameters} vs. Table \ref{table: ECM no bias parameters} and Figures \ref{figure: uniform results}, \ref{figure: uniform bias results}(a-c) vs. \ref{figure: ECM bias results}(a-c)). In Figure \ref{figure: ECM bias results} (g-i) we see that the evolution of time-averaged exponents $\beta_{av}(t)$ (see Appendix \ref{appendix: data analysis}) have three phases. Following the rapid increase in the first phase, there is a gradual decrease in the rate of change in the second phase, followed by stabilization of $\beta_{av}(t)$ at $\bar{\beta}$. Curiously, a similar behavior has also been observed by Dieterich et al. \cite{dieterich2008anomalous}.

The distribution of speeds again remained invariant and the average speeds are very close to the cases with no external cues (see Table \ref{table: ECM no bias parameters}). However, the frequency of normalized velocities (see Figure \ref{figure: ECM bias speeds and events} (d-f)) show, as expected, that the cell velocities are aligned with the cue gradient. Accordingly, we see that persistent motion emerges: directionality ratio increases compared to unbiased migration (Table \ref{table: ECM no bias parameters}) and the velocities become correlated (Figure \ref{figure: persistence ECM}(d-f)). We also observe that an external signal has a stronger impact on motion persistence for higher number of adhesions due to relative increases of $\bar{r}$ and the degree of velocity autocorrelation. Recall that in the presence of, for example, a chemotactic cue, a cell polarizes so that its adhesion dynamics is aligned with the gradient. In particular, adhesions are preferentially formed at the front (where the chemoattractant concentration is larger), and preferentially ruptured at the back. We can see in Figure \ref{figure: ECM bias speeds and events} (g-i), that our simulation results reproduce such polarized dynamics: the ratio of binding to unbinding events is larger (smaller) than unity in the northern (southern) part of the cells, where the cue is stronger (weaker) relative to the cell centroid. Also, for a smaller number of adhesion sites, the effects of increasing the cue gradient have more noticeable effect on the ratios of events (see \ref{figure: ECM bias speeds and events} (g-i)). This is simply due to the reduced density of adhesion sites, which leads to larger relative difference in the concentration of the cue between them. From Figure \ref{figure: Qcue linear gradient} we can asses the effect of an external cue $Q_{cue}$ on the the binding rate $a^+_i$ (omitting the force dependence for clarity), since the rate is proportional to $Q_{cue}$.   

\begin{figure}[h]
%\vspace{0mm}
\centering
%First row
\subfloat[]
{
	\includegraphics[width=40mm,height=37mm]{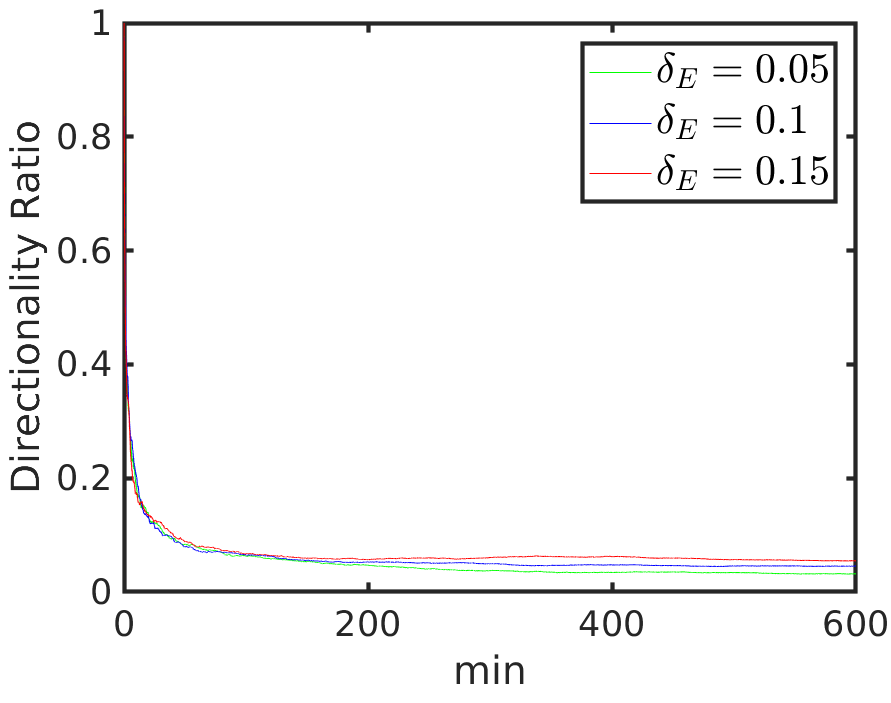}
}
\subfloat[]
{
	\includegraphics[width=40mm,height=37mm]{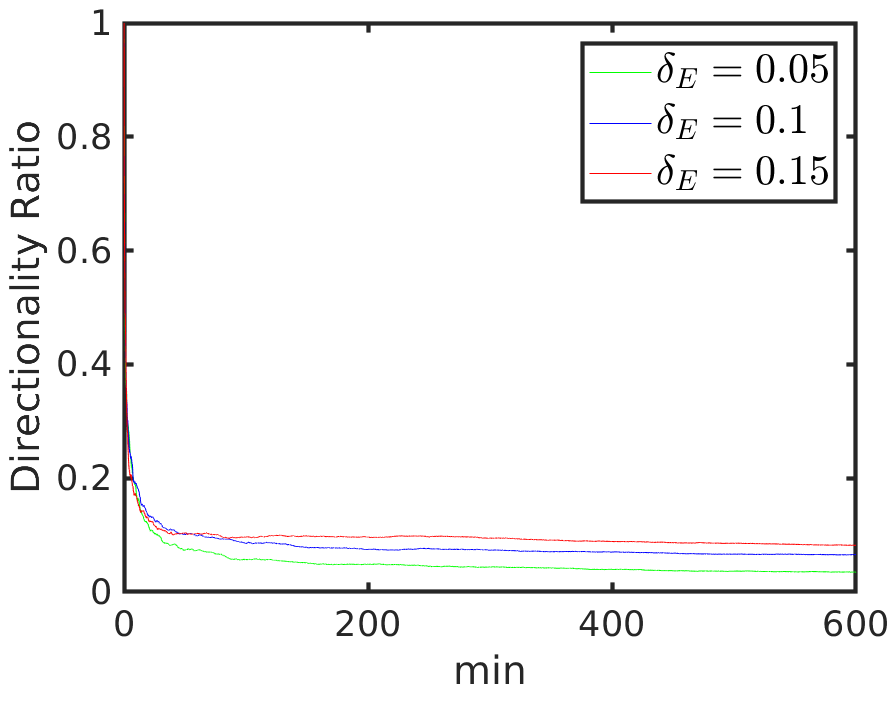}
}
\subfloat[]
{
	\includegraphics[width=40mm,height=37mm]{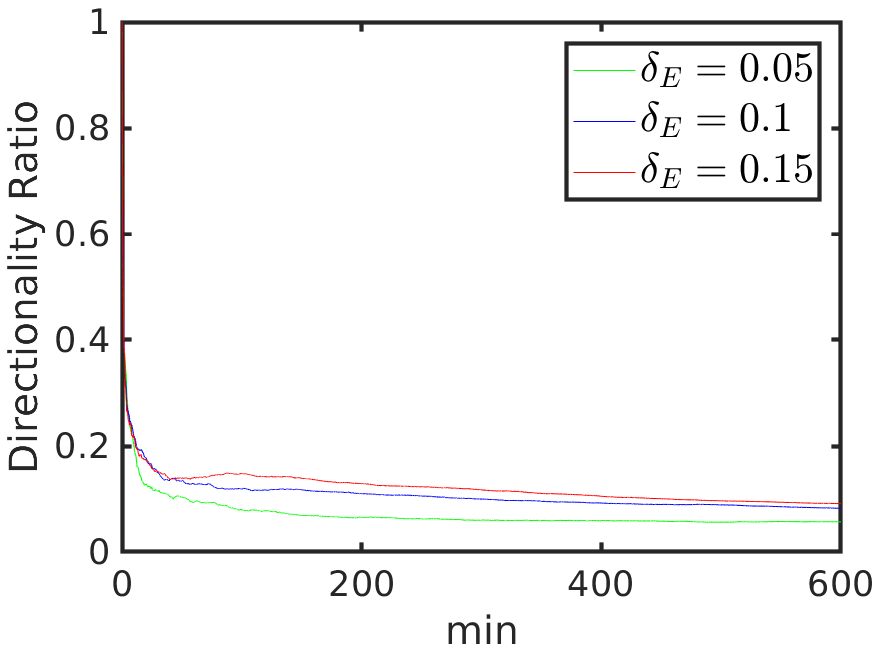}
}
\hspace{0mm}
\subfloat[]
{
	\includegraphics[width=40mm,height=37mm]{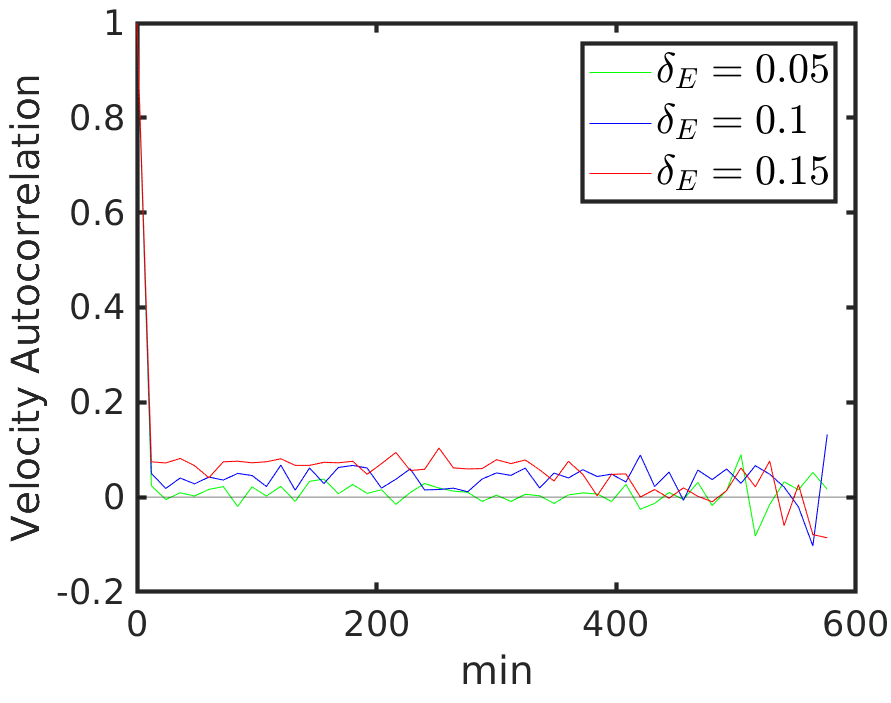}
}
\subfloat[]
{
	\includegraphics[width=40mm,height=37mm]{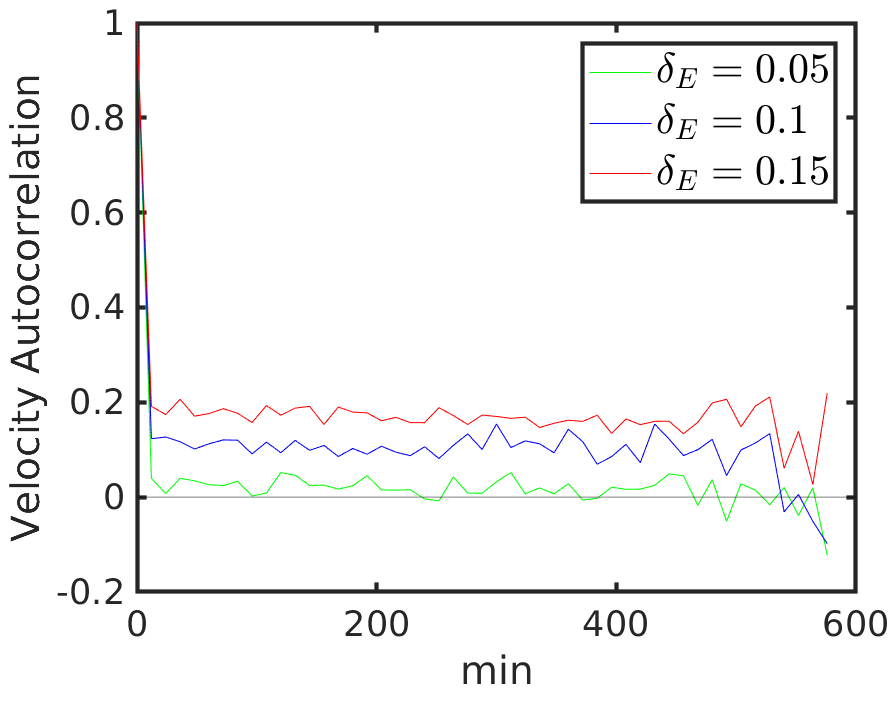}
}
\subfloat[]
{
	\includegraphics[width=40mm,height=37mm]{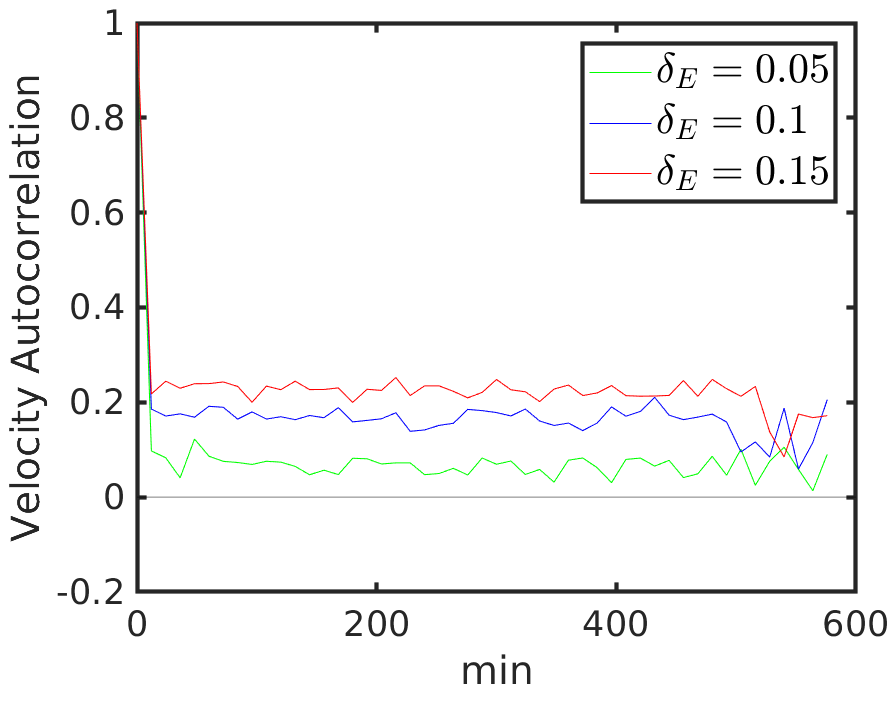}
}
\caption{Persistence of motion for cells with $M=8,16,32$ adhesions in the first, second, and third columns, respectively, and with $\delta_E=0.05$ (green), $0.1$ (blue), $0.15$ (red). (a-c) Directionality ratio. (d-f) Velocity autocorrelation.}
\label{figure: persistence ECM}
\end{figure}
Together with Figure \ref{figure: Qcue linear gradient}, the simulations illustrate that directed tactic migration, resulting from biased adhesion formation, follows from the \textit{local} information about the external cue. That is, the spatial dependence of the FA binding rate is solely due to the local concentration of an external cue (see \eqref{equation: binding propensity}) and no central mechanism for gradient determination was utilized to bias adhesion formation. Consequently, migration along the gradient of an external cue is achieved without its explicit ``computation" by the cell.  

Along with external cue, force dependence of the binding rate is also important for directed migration and without it, the cells do not exhibit biased migration (data not shown). Figure \ref{figure: force dependence migration cycle} illustrates how the dependence fits into the migration cycle (recall Figure \ref{fig: four cycle}). For the directed migration to occur, at the time of FA disassociation $\mathbf{x}_n$ must be preferentially in the rear (Figure \ref{figure: force dependence migration cycle} step 2). After FA unbinding the increased force at the rear FAs due to extended SFs promotes binding there (Figure \ref{figure: force dependence migration cycle} step 3). Note that since cell body translocation occurs only after an unbinding event, formation of new FA in the prospective rear of the cell does not lead to backwards movement. Also, due to the external signal more FAs tend to be at the front than at the rear. 
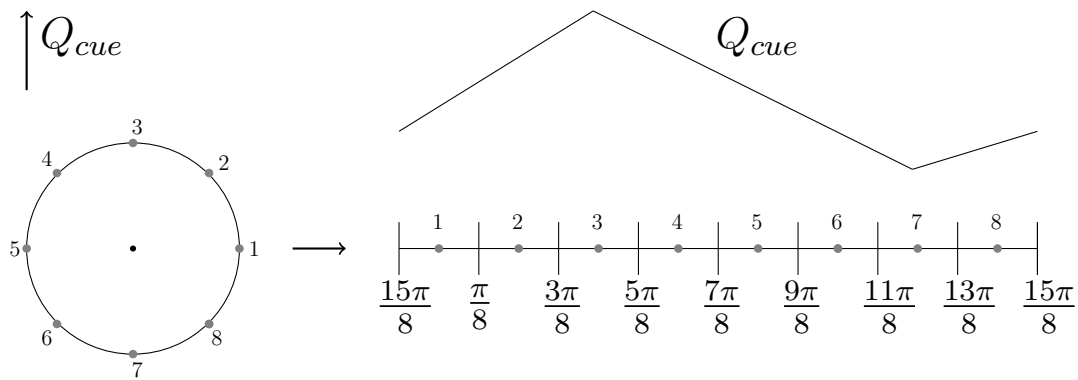
\begin{figure}[H]
\begin{tikzpicture}[scale=0.7,transform shape,every node/.style={scale=1}]

\draw (0,0) circle [radius=2cm];

\filldraw [gray] (2,0) circle [radius=2pt];
\filldraw [gray] (0,2) circle [radius=2pt];
\filldraw [gray] (0,-2) circle [radius=2pt];
\filldraw [gray] (-2,0) circle [radius=2pt];
\filldraw [gray] (1.43,1.43) circle [radius=2pt];
\filldraw [gray] (-1.43,1.43) circle [radius=2pt];
\filldraw [gray] (1.43,-1.43) circle [radius=2pt];
\filldraw [gray] (1.43,-1.43) circle [radius=2pt];
\filldraw [gray] (-1.43,-1.43) circle [radius=2pt];

\filldraw [black] (0,0) circle [radius=1.33pt];

\draw (2.2,0) node[text width = 1.33pt] {$1$};
\draw (1.63,1.63) node[text width = 1.33pt] {$2$};
\draw (0,2.3) node[text width = 1.33pt] {$3$};
\draw (-1.7,1.7) node[text width = 1.33pt] {$4$};
\draw (-2.3,0) node[text width = 1.33pt] {$5$};
\draw (-1.7,-1.7) node[text width = 1.33pt] {$6$};
\draw (0,-2.3) node[text width = 1.33pt] {$7$};
\draw (1.5,-1.7) node[text width = 1.33pt] {$8$};

\draw [thick,black][->](3,0) -- (4,0);

\draw [black](5,0) -- (17,0);
\draw (5,-0.5) -- (5,0.5);
\draw (17,-0.5) -- (17,0.5);

\foreach \i in {1,...,8}
	\draw [xshift=\i*1.5cm](5,-0.5)--(5,0.5); 
\foreach \i in {0,...,7}
	\filldraw [gray,xshift=\i*1.5cm+0.75cm](5,0) circle [radius=2pt];
\filldraw [gray, xshift = 8.25cm] (5,0) circle [radius=2pt];
\draw [xshift = 0.75cm] (4.9,0.5) node[text width = 1.33pt] {$1$};
\draw [xshift = 2.25cm] (4.9,0.5) node[text width = 1.33pt] {$2$};
\draw [xshift = 3.75cm] (4.9,0.5) node[text width = 1.33pt] {$3$};
\draw [xshift = 5.25cm] (4.9,0.5) node[text width = 1.33pt] {$4$};
\draw [xshift = 6.75cm] (4.9,0.5) node[text width = 1.33pt] {$5$};
\draw [xshift = 8.25cm] (4.9,0.5) node[text width = 1.33pt] {$6$};
\draw [xshift = 9.75cm] (4.9,0.5) node[text width = 1.33pt] {$7$};
\draw [xshift = 11.25cm] (4.9,0.5) node[text width = 1.33pt] {$8$};

\draw [xshift = 0cm] (4.6,-1.1) node[scale = 2, text width = 1.33pt] {$\frac{15\pi}{8}$};
\draw [xshift = 1.5cm] (4.8,-1.1) node[scale = 2, text width = 1.33pt] {$\frac{\pi}{8}$};
\draw [xshift = 3.0cm] (4.7,-1.1) node[scale = 2, text width = 1.33pt] {$\frac{3\pi}{8}$};
\draw [xshift = 4.5cm] (4.7,-1.1) node[scale = 2, text width = 1.33pt] {$\frac{5\pi}{8}$};
\draw [xshift = 6.0cm] (4.7,-1.1) node[scale = 2, text width = 1.33pt] {$\frac{7\pi}{8}$};
\draw [xshift = 7.5cm] (4.7,-1.1) node[scale = 2, text width = 1.33pt] {$\frac{9\pi}{8}$};
\draw [xshift = 9.0cm] (4.7,-1.1) node[scale = 2, text width = 1.33pt] {$\frac{11\pi}{8}$};
\draw [xshift = 10.5cm] (4.7,-1.1) node[scale = 2, text width = 1.33pt] {$\frac{13\pi}{8}$};
\draw [xshift = 12.0cm] (4.7,-1.1) node[scale = 2, text width = 1.33pt] {$\frac{15\pi}{8}$};

\draw [xshift = 3.75cm] (4.9,4.5) -- (10.9,1.5);
\draw [xshift = 3.75cm] (10.9,1.5) -- (13.25,2.22);
\draw (5,2.22)--(8.65,4.5);

\draw [thick,black,->](-2,3)--(-2,4.5);
\draw (-1.7,4) node[scale = 2, text width = 1.33pt] {$Q_{cue}$};
\draw (11,4) node[scale = 2, text width = 1.33pt] {$Q_{cue}$};
\end{tikzpicture}
\caption{Concentration of an external cue projected on the cell's circumference. Gray bullets represent FA sites.} 
\label{figure: Qcue linear gradient}
\end{figure}
Thus, the pulling force exerted by the front on the rear tends to be larger than the opposite and hence the cell moves preferentially in the direction of the gradient. Without the signal, of course, movement becomes unbiased, as shown in the previous section. This suggests that the SF length dependence of the forces (see \eqref{eq: Fi}) and the force dependence of the FA binding rate (see \eqref{eq: force binding rate}) are necessary for directed migration resulting from biased adhesion formation in the presence of an external signal.
\begin{figure}[H]
\includegraphics[width = 0.5\textwidth]{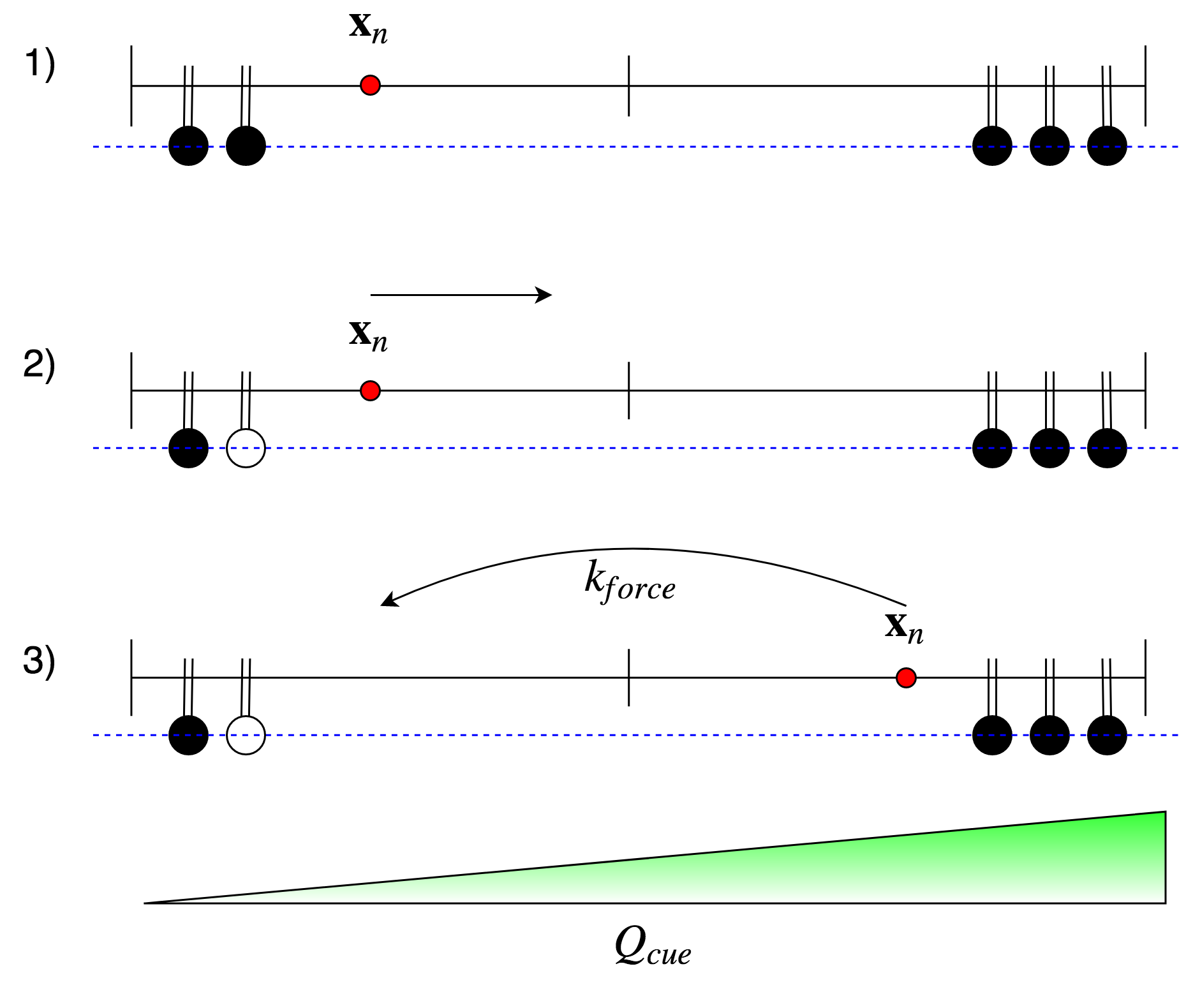}
\caption{The force dependence of the binding rate and the biased adhesion formation during the migration cycle. Side view schematic of the cell is illustrated, where (un)bound FAs are shown as (white)black circles. 1) Initial configuration. 2) Unbinding leads to cell translocation and motion of $\mathbf{x}_n$ within the cell. 3) Increased force on the cell rear (due to its dependence on SF extension) promotes FA association due to force dependence $k_{force}$ of the binding rate (see Section \ref{section: adhesion force dependence}), after which the cycle begins anew.}
\label{figure: force dependence migration cycle}
\end{figure}

\subsection{Fibrillar architecture of ECM}
The ECM topography is another important determinant of directed cell migration. In particular, the spatial distribution of the ECM fibers guides the motility by inducing cell shape alignment along the adhesive cues, resulting in characteristic directed movement along the fiber tracts \cite{PDY09}. Such guided migration is called contact guidance \cite{PDY09}, \cite{ROG17}. Ramirez-San Juan et. al. \cite{ROG17} showed that contact guidance can be modulated by micrometer scale variations of interfiber spacing. Inspired by this study, we simulate how subcellular scale fiber spacing influences cell motility, and whether such ECM architecture yields migration patterns characteristic of contact guidance.

\begin{figure}[!h]
\captionsetup[subfigure]{labelformat=empty}
\centering
%First row
\subfloat[]
{
	\includegraphics[width=40mm,height=37mm]{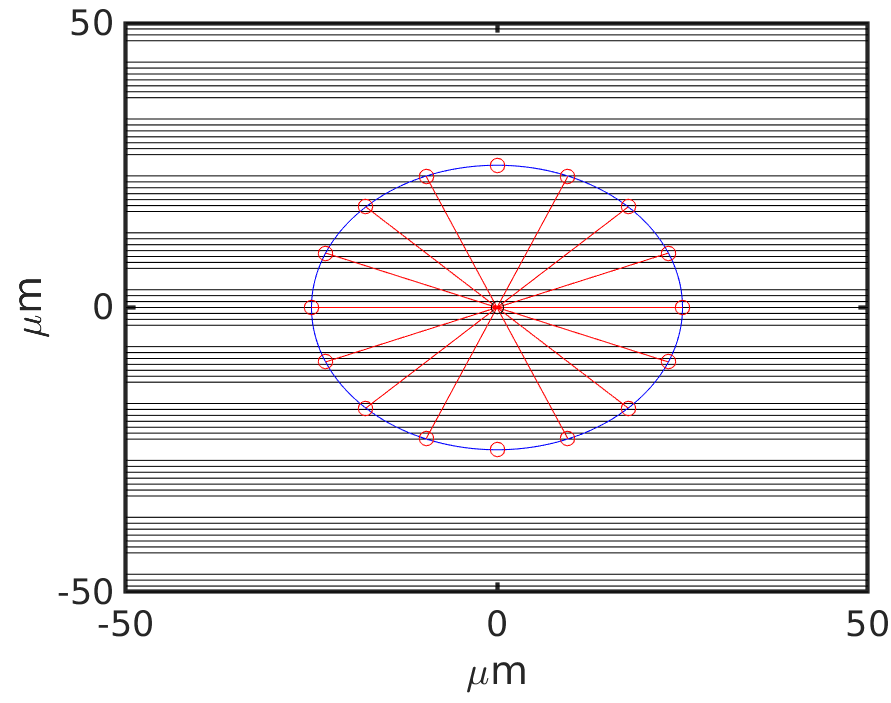}
}
\subfloat[]
{
	\includegraphics[width=40mm,height=37mm]{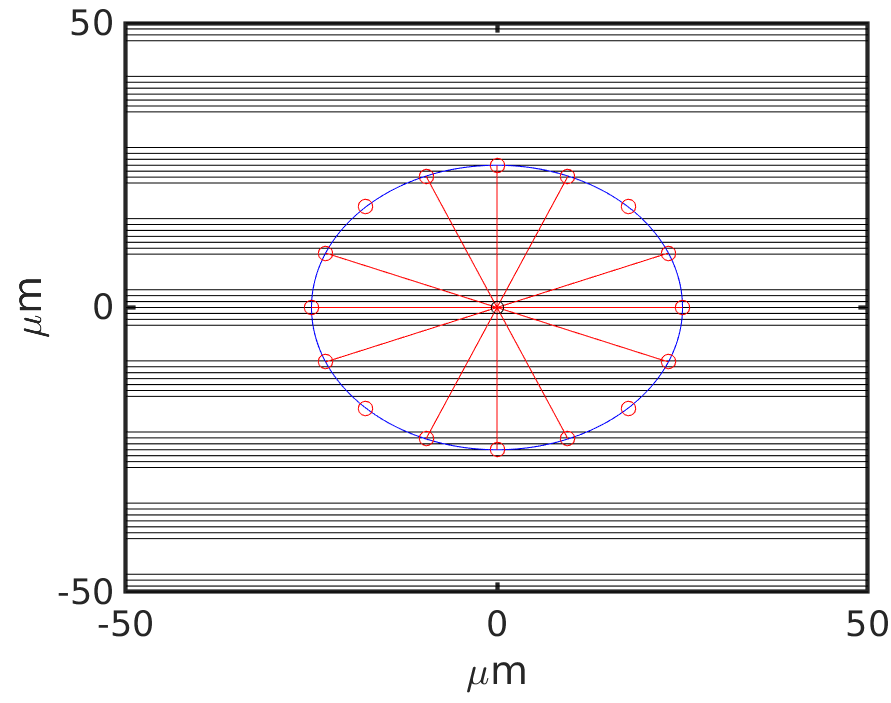}
}
\subfloat[]
{
	\includegraphics[width=40mm,height=37mm]{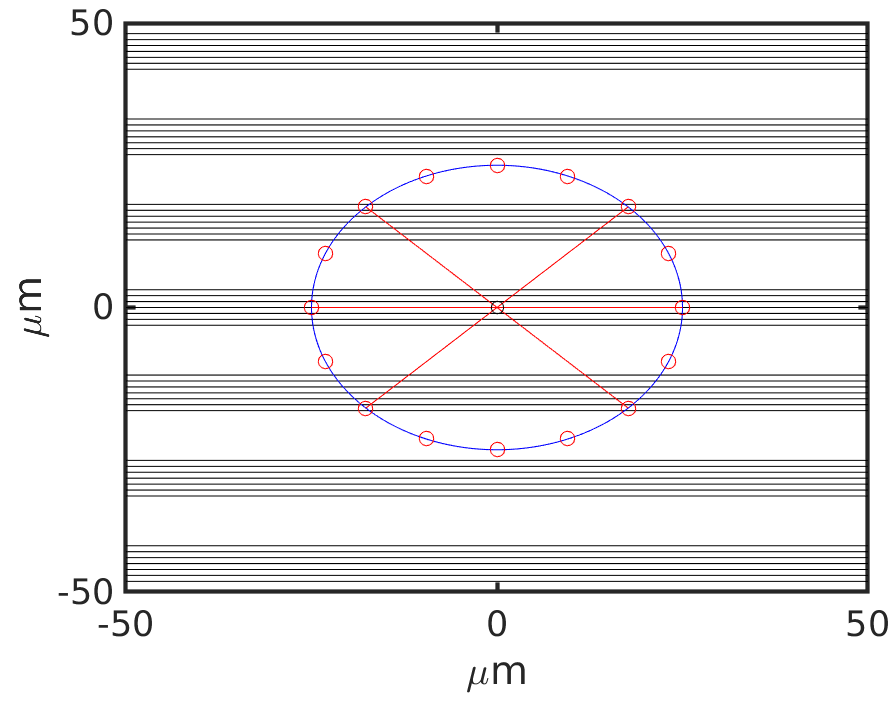}
}
\hspace{0mm}

\caption{Stripe pattern with $\delta_G=0.15,0.25,0.35$ on, respectively, left, middle, and right plots. A cell is illustrated such that each FA on a stripe is bound and cell center $\mathbf{x}$ coincides with $\mathbf{x}_n$  at the origin.}
\label{figure: ECM stripes}
\end{figure}
Similar to the case with an external cue gradient, the functions $Q_{cue}$ and $q$ have the following form:
\begin{align*}
Q_{cue}(\mathbf{x})&=
\begin{cases*}
1,\phantom{aaa}\text{  if } \mathbf{x}\in \Omega_{\delta_{G}}\\
0.01, \phantom{,}\text{ else} 
\end{cases*}\\
q(Q_{cue}(\mathbf{x})) &= Q_{cue}(\mathbf{x}),
\end{align*}
where $\Omega_{\delta_{G}}$ represents the stripe pattern, $\delta_{G} = 0.15, 0.25, 0.35$ represents the spacing between stripes such that the distances between them is $\delta_GR_{cell}$ (Figure \ref{figure: ECM stripes}). The stripe width is taken to be $0.25R_{cell}$. Similarly as in \cite{ROG17}, these dimensions are chosen so that a cell is spread on multiple stripes.

The simulation results, shown in Figure \ref{figure: contact guidance plots}, indicate that the cell motility has characteristics of contact guidance. Namely, the trajectories show preferential horizontal cell movement (Figure \ref{figure: contact guidance plots} (a-c)), and the displacements are aligned with the fiber pattern (Figure \ref{figure: contact guidance plots} (d-f)). However, increasing the spacing does not simply lead to a greater adhesion alignment along the horizontal direction, as can be observed in Figure \ref{figure: contact guidance plots} (g-i). Rather, it is the combination of the ECM pattern and the radial position of FAs that gives rise to, for example, definite x-shaped adhesion binding patterns (Figure \ref{figure: contact guidance plots} (h)).         
     
\begin{figure}[H]
%\vspace{0mm}
\centering
%First row
\subfloat[]
{	
	\includegraphics[width=40mm,height=37mm]{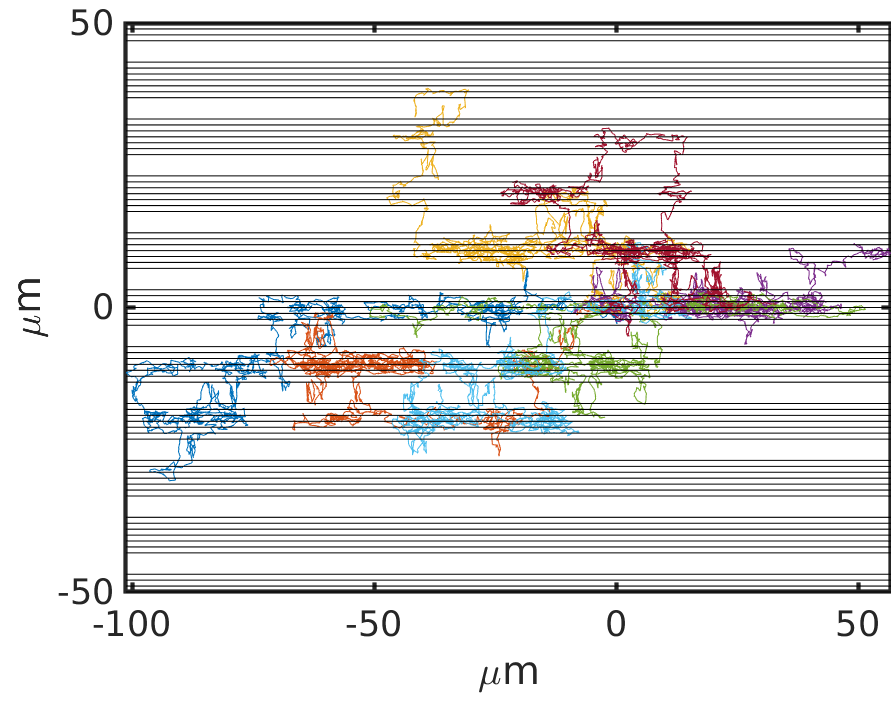}
}
\subfloat[]
{
	\includegraphics[width=40mm,height=37mm]{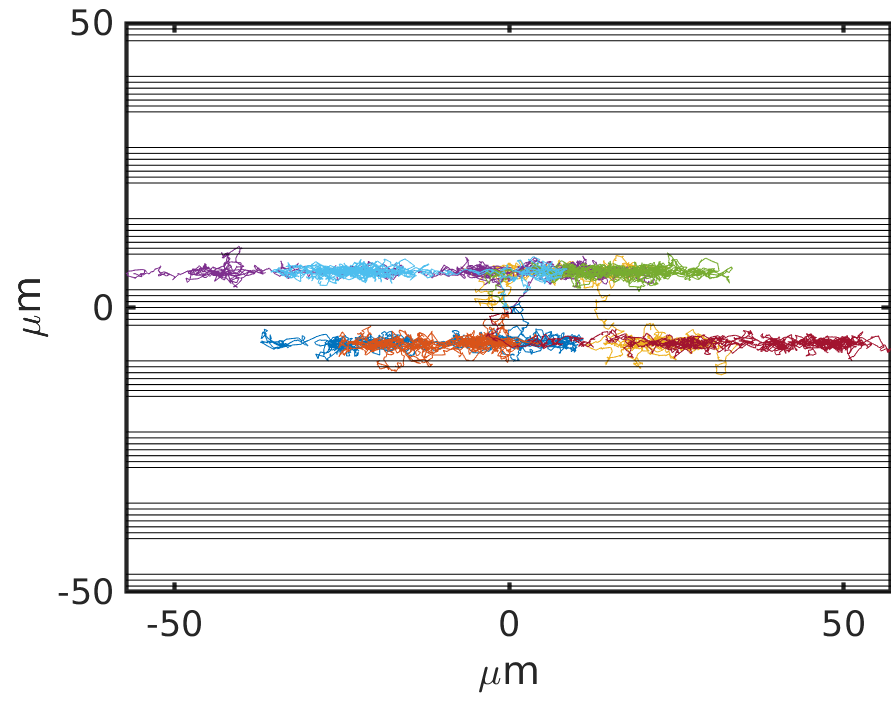}
}
\subfloat[]
{
	\includegraphics[width=40mm,height=37mm]{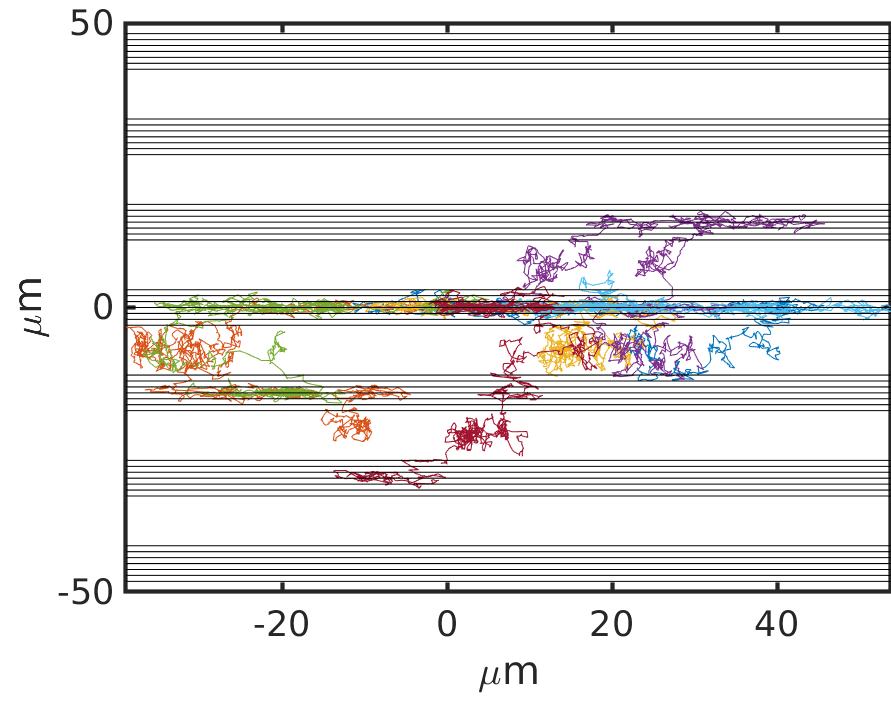}
}
\hspace{0mm}
%Second row
\subfloat[]
{
	\includegraphics[width=40mm,height=37mm]{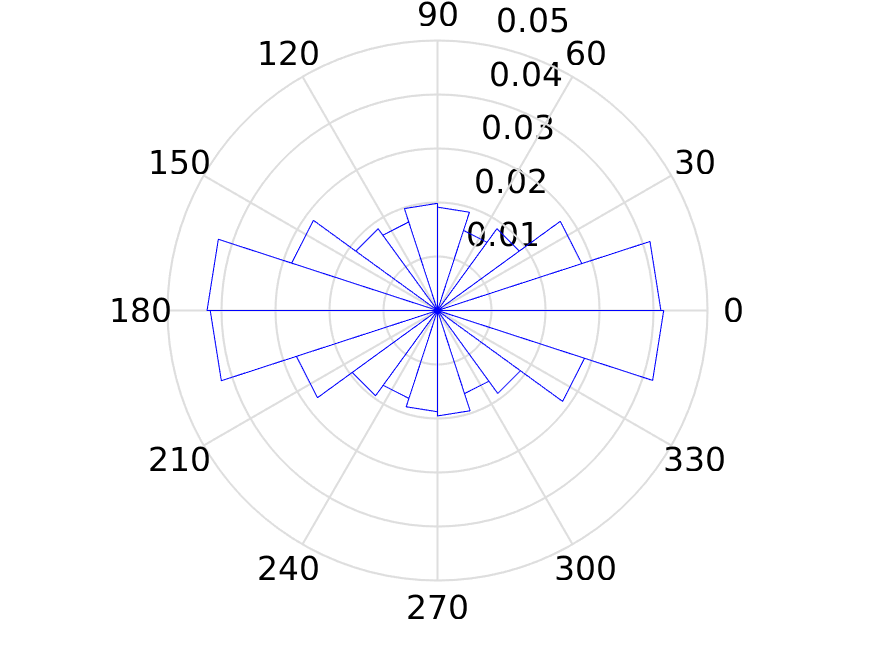}
}
\subfloat[]
{
	\includegraphics[width=40mm,height=37mm]{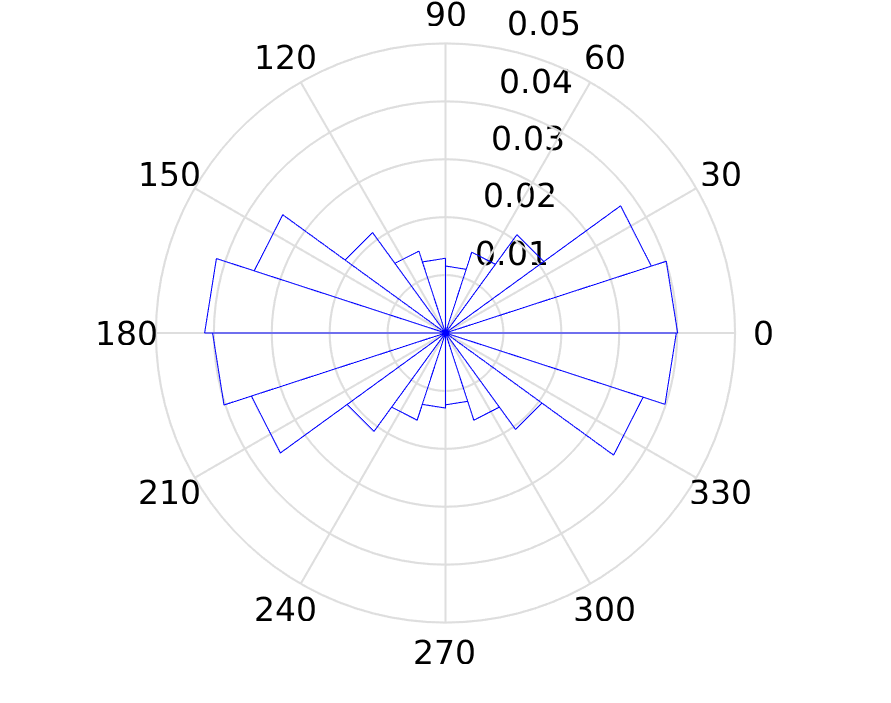}
}
\subfloat[]
{
	\includegraphics[width=40mm,height=37mm]{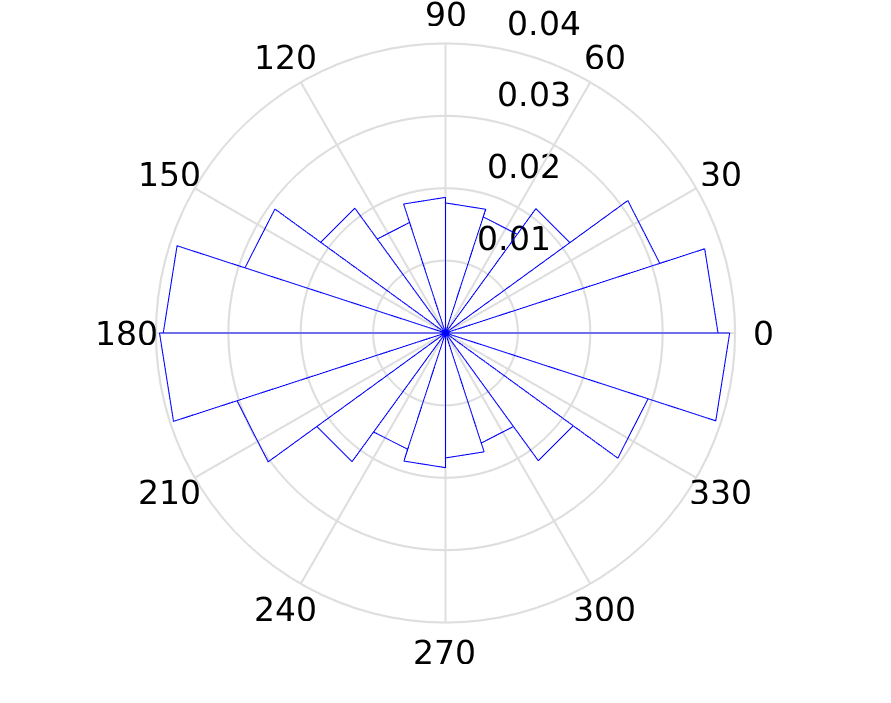}
}
\hspace{0mm}
\subfloat[]
{
	\includegraphics[width=40mm,height=37mm]{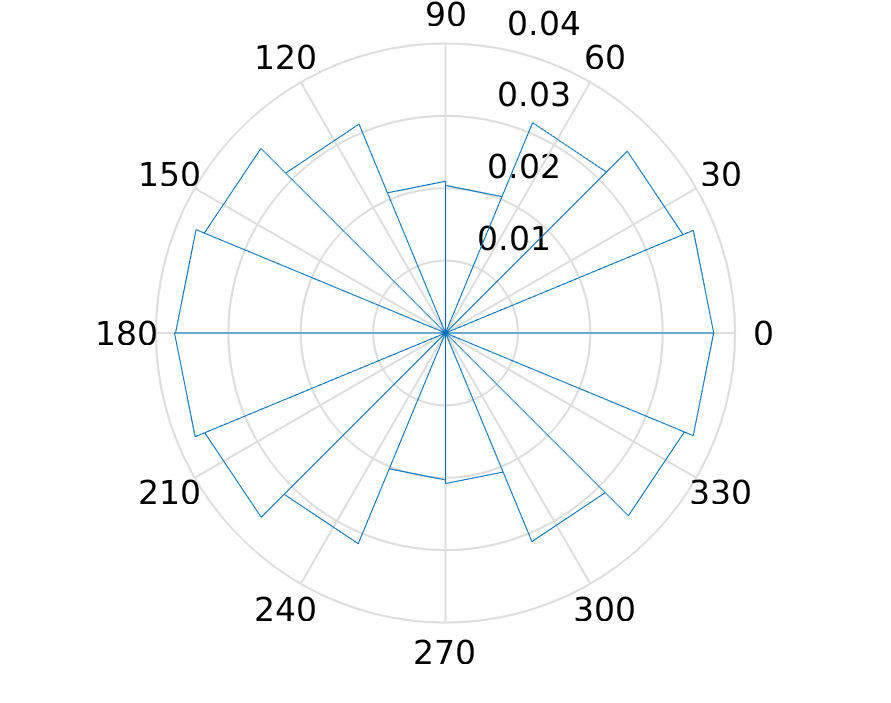}
}
\subfloat[]
{
	\includegraphics[width=40mm,height=37mm]{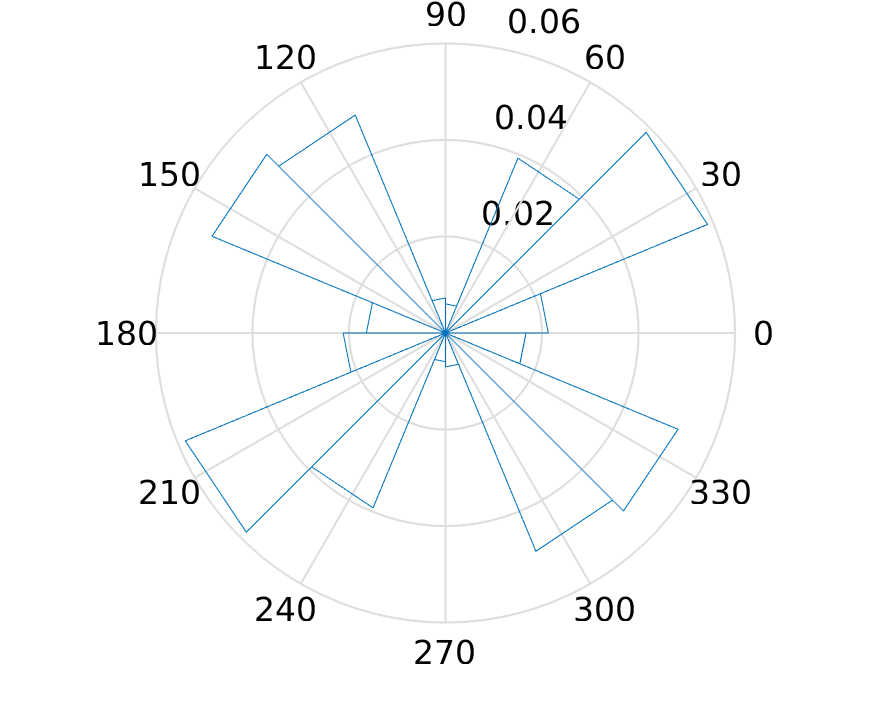}
}
\subfloat[]
{
	\includegraphics[width=40mm,height=37mm]{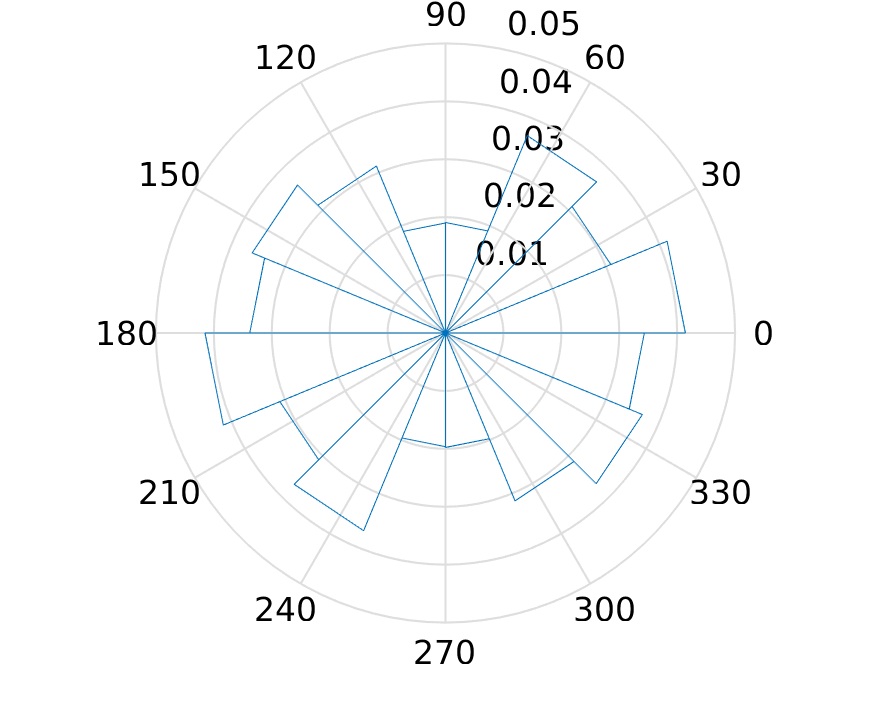}
}
\caption{Simulation results with $M=16$, and $\delta_G=0.15,0.25,0.35$ in first, second, and third columns, respectively. (a-c) Trajectories of 7 cells and the striped ECM pattern. (d-f) Relative frequencies of normalized velocities. (g-i) Relative frequencies of binding events in each of the 16 cell sectors.}
\label{figure: contact guidance plots}
\end{figure} 
Such binding (and unbinding) pattern leads to fluctuating movement along northwest-southeast and northeast-southwest axis, with the resulting net migration pattern shown in Figure \ref{figure: contact guidance plots} (b). Similarly, the binding pattern shown in Figure \ref{figure: contact guidance plots} (i) with more frequent events along the equator corresponds to a mixture of diagonal and horizontal movements (Figure \ref{figure: contact guidance plots} (c)), as larger interfiber spacing precludes FA binding at the poles and facilitates adhesion along, as well as across the stripes in x-shaped pattern (see also Figure \ref{figure: ECM stripes} (right) for illustration of a characteristic FA configuration). On the other hand, smaller spacing also leads to horizontal movement, but with more frequent vertical displacement across the stripes (Figure \ref{figure: contact guidance plots}). These results are also in line with conclusions made by in \cite{kubow2017contact}, where it was found that adhesion alignment determines contact guidance  We also found that the average speeds were lower than in previous scenarios (Tables \ref{table: uniform bias parameters}, \ref{table: ECM no bias parameters}): $1.52\mu m/min$, $0.94\mu m/min$, and $0.87\mu m/min$ corresponding to, respectively, $\delta_G = 0.15, 0.25, 0.35$. Interestingly, the average speeds reported in \cite{ROG17} were $\sim0.6\mu m/min$, although in that study the speeds were nearly constant for varying fiber pattern.  

\begin{figure}[H]
\begin{tikzpicture}[scale=0.7,transform shape,every node/.style={scale=1}]

\draw (0,0) circle [radius=2cm];
\foreach \i in {0,...,15}
	\filldraw [gray]({2*cos(\i*360/16)},{2*sin(\i*360/16)}) circle [radius=2pt];
\foreach \i in {-2,...,2}
	\draw [yshift=\i*1.2cm+0.25cm] (-2.4,0)--(2.4,0);
\foreach \i in {-2,...,2}
	\draw [yshift=\i*1.2cm-0.25cm] (-2.4,0)--(2.4,0);	
\foreach \i in {-2,...,2}
	\foreach \j in {0,...,20}{
		\draw [xshift = \j*0.22cm,yshift=\i*1.2cm-0.25cm] (-2.2,0)--(-2,0.5);
	}
\filldraw [red] (2,0) circle [radius = 3pt];
\filldraw [red] ({2*cos(2*360/16)},{2*sin(2*360/16)}) circle [radius = 3pt];
\filldraw [red] ({2*cos(6*360/16)},{2*sin(6*360/16)}) circle [radius = 3pt];
\filldraw [red] ({2*cos(8*360/16)},{2*sin(8*360/16)}) circle [radius = 3pt];		
\filldraw [red] ({2*cos(10*360/16)},{2*sin(10*360/16)}) circle [radius = 3pt];
\filldraw [red] ({2*cos(14*360/16)},{2*sin(14*360/16)}) circle [radius = 3pt];				
\filldraw [black] (0,0) circle [radius=1.33pt];

\draw (2.5,0) node[scale = 1.2, text width = 1.33pt] {$1$};
\draw (1.7,1.7) node[scale = 1.2, text width = 1.33pt] {$3$};
\draw (-1.7,1.7) node[scale = 1.2, text width = 1.33pt] {$7$};
\draw (-2.5,0) node[scale = 1.2, text width = 1.33pt] {$9$};
\draw (-1.8,-1.8) node[scale = 1.2, text width = 1.33pt] {$11$};
\draw (1.6,-1.8) node[scale = 1.2, text width = 1.33pt] {$15$};

\draw [thick,red] ({2*cos(2*360/16)},{2*sin(2*360/16)}) -- (0,0);
\draw [thick,red] ({2*cos(6*360/16)},{2*sin(6*360/16)}) -- (0,0);
\draw [thick,red] ({2*cos(8*360/16)},{2*sin(8*360/16)}) -- (0,0);
\draw [thick,red] ({2*cos(10*360/16)},{2*sin(10*360/16)}) -- (0,0);
\draw [thick,red] ({2*cos(14*360/16)},{2*sin(14*360/16)}) -- (0,0);
\draw [thick,red] ({2*cos(16*360/16)},{2*sin(16*360/16)}) -- (0,0);

\draw [thick,black][->](3,0) -- (4,0);

\draw [yshift = -2cm][black](5,0) -- (17,0);
\draw [yshift = -2cm](5,-0.5) -- (5,0.5);
\draw [yshift = -2cm](17,-0.5) -- (17,0.5);

\foreach \i in {1,...,16}
	\draw [yshift = -2cm][xshift=\i*0.75cm](5,-0.25)--(5,0.25); 
\foreach \i in {0,...,15}
	\filldraw [yshift = -2cm][gray,xshift=\i*0.75cm+0.375cm](5,0) circle [radius=2pt];
\foreach \i in {1,...,16}
	\draw [yshift = -2cm][xshift = \i*0.75cm-0.375cm] (4.9,0.5) node[text width = 1.33pt] {\i};	
\draw [yshift = -2cm](5,4.5)--(5.75,4.5);
\draw [yshift = -2cm](5.75,4.5) -- (5.75,1.5);
\draw [yshift = -2cm](5.75,1.5) -- (6.5,1.5);
\draw [yshift = -2cm](6.5,1.5) -- (6.5,4.5);
\draw [yshift = -2cm](6.5,4.5) -- (7.25,4.5);
\draw [yshift = -2cm](7.25,4.5) -- (7.25,1.5);
\draw [yshift = -2cm](7.25,1.5) -- (9.5,1.5);
\draw [yshift = -2cm](9.5,1.5) -- (9.5,4.5);
\draw [yshift = -2cm](9.5,4.5) -- (10.25,4.5);
\draw [yshift = -2cm](10.25,4.5) -- (10.25,1.5);
\draw [yshift = -2cm](10.25,1.5) -- (11,1.5);
\draw [yshift = -2cm](11,1.5) -- (11,4.5);
\draw [yshift = -2cm](11.0,4.5) -- (11.75,4.5);
\draw [yshift = -2cm](11.75,4.5) -- (11.75,1.5);
\draw [yshift = -2cm](11.75,1.5) -- (12.5,1.5);
\draw [yshift = -2cm](12.5,1.5) -- (12.5,4.5);
\draw [yshift = -2cm](12.5,4.5) -- (13.25,4.5);
\draw [yshift = -2cm](13.25,4.5) -- (13.25,1.5);
\draw [yshift = -2cm](13.25,1.5) -- (15.5,1.5);
\draw [yshift = -2cm](15.5,1.5) -- (15.5,4.5);
\draw [yshift = -2cm](15.5,4.5) -- (16.25,4.5);
\draw [yshift = -2cm](16.25,4.5) -- (16.25,1.5);
\draw [yshift = -2cm](16.25,1.5) -- (17,1.5);
\draw [yshift = -2cm](17,1.5) -- (17,4.5);
\draw [yshift = -2cm](11,6) node[scale = 2, text width = 1.33pt] {$Q_{cue}$};

\filldraw [yshift = -2cm][red,xshift = 0.375cm] (5.0,0) circle [radius = 3pt];
\filldraw [yshift = -2cm][red,xshift = 1.875cm] (5.0,0) circle [radius = 3pt];
\filldraw [yshift = -2cm][red,xshift = 4.875cm] (5.0,0) circle [radius = 3pt];
\filldraw [yshift = -2cm][red,xshift = 6.375cm] (5.0,0) circle [radius = 3pt];
\filldraw [yshift = -2cm][red,xshift = 7.875cm] (5.0,0) circle [radius = 3pt];
\filldraw [yshift = -2cm][red,xshift = 10.875cm] (5.0,0) circle [radius = 3pt];

\end{tikzpicture}
\end{figure}
\begin{figure}[H]
\begin{tikzpicture}[scale=0.7,transform shape,every node/.style={scale=1}]

\draw [yshift = 0.6cm](0,0) circle [radius=2cm];
\foreach \i in {0,...,15}
	\filldraw [yshift = 0.6cm][gray]({2*cos(\i*360/16)},{2*sin(\i*360/16)}) circle [radius=2pt];
\foreach \i in {-2,...,2}
	\draw [yshift=\i*1.2cm+0.25cm] (-2.4,0)--(2.4,0);
\foreach \i in {-2,...,2}
	\draw [yshift=\i*1.2cm-0.25cm] (-2.4,0)--(2.4,0);	
\foreach \i in {-2,...,2}
	\foreach \j in {0,...,20}{
		\draw [xshift = \j*0.22cm,yshift=\i*1.2cm-0.25cm] (-2.2,0)--(-2,0.5);
	}
\filldraw [yshift = 0.6cm][red] ({2*cos(1*360/16)},{2*sin(1*360/16)}) circle [radius = 3pt];
\filldraw [yshift = 0.6cm][red] ({2*cos(3*360/16)},{2*sin(3*360/16)}) circle [radius = 3pt];
\filldraw [yshift = 0.6cm][red] ({2*cos(4*360/16)},{2*sin(4*360/16)}) circle [radius = 3pt];
\filldraw [yshift = 0.6cm][red] ({2*cos(5*360/16)},{2*sin(5*360/16)}) circle [radius = 3pt];		
\filldraw [yshift = 0.6cm][red] ({2*cos(7*360/16)},{2*sin(7*360/16)}) circle [radius = 3pt];
\filldraw [yshift = 0.6cm][red] ({2*cos(9*360/16)},{2*sin(9*360/16)}) circle [radius = 3pt];
\filldraw [yshift = 0.6cm][red] ({2*cos(11*360/16)},{2*sin(11*360/16)}) circle [radius = 3pt];
\filldraw [yshift = 0.6cm][red] ({2*cos(12*360/16)},{2*sin(12*360/16)}) circle [radius = 3pt];		
\filldraw [yshift = 0.6cm][red] ({2*cos(13*360/16)},{2*sin(13*360/16)}) circle [radius = 3pt];
\filldraw [yshift = 0.6cm][red] ({2*cos(15*360/16)},{2*sin(15*360/16)}) circle [radius = 3pt];

\draw [yshift = 0.6cm](2,1.1) node[scale = 1.2, text width = 1.33pt] {$2$};
\draw [yshift = 0.6cm](0.75,2.3) node[scale = 1.2, text width = 1.33pt] {$4$};
\draw [yshift = 0.6cm](-0.1,2.5) node[scale = 1.2, text width = 1.33pt] {$5$};
\draw [yshift = 0.6cm](-0.95,2.3) node[scale = 1.2, text width = 1.33pt] {$6$};
\draw [yshift = 0.6cm](-2.1,1.1) node[scale = 1.2, text width = 1.33pt] {$8$};
\draw [yshift = 0.6cm](-2.35,-1.2) node[scale = 1.2, text width = 1.33pt] {$10$};
\draw [yshift = 0.6cm](0.85,-2.4) node[scale = 1.2, text width = 1.33pt] {$14$};
\draw [yshift = 0.6cm](-0.2,-2.5) node[scale = 1.2, text width = 1.33pt] {$13$};
\draw [yshift = 0.6cm](-1.1,-2.4) node[scale = 1.2, text width = 1.33pt] {$12$};
				
\filldraw [yshift = 0.6cm][black] (0,0) circle [radius=2pt];

\draw [yshift = 0.6cm][thick,red] ({2*cos(1*360/16)},{2*sin(1*360/16)}) -- (0,0);
\draw [yshift = 0.6cm][thick,red] ({2*cos(3*360/16)},{2*sin(3*360/16)}) -- (0,0);
\draw [yshift = 0.6cm][thick,red] ({2*cos(4*360/16)},{2*sin(4*360/16)}) -- (0,0);
\draw [yshift = 0.6cm][thick,red] ({2*cos(5*360/16)},{2*sin(5*360/16)}) -- (0,0);
\draw [yshift = 0.6cm][thick,red] ({2*cos(7*360/16)},{2*sin(7*360/16)}) -- (0,0);
\draw [yshift = 0.6cm][thick,red] ({2*cos(9*360/16)},{2*sin(9*360/16)}) -- (0,0);
\draw [yshift = 0.6cm][thick,red] ({2*cos(11*360/16)},{2*sin(11*360/16)}) -- (0,0);
\draw [yshift = 0.6cm][thick,red] ({2*cos(12*360/16)},{2*sin(12*360/16)}) -- (0,0);
\draw [yshift = 0.6cm][thick,red] ({2*cos(13*360/16)},{2*sin(13*360/16)}) -- (0,0);
\draw [yshift = 0.6cm][thick,red] ({2*cos(15*360/16)},{2*sin(15*360/16)}) -- (0,0);

\draw [thick,black][->](3,0) -- (4,0);

\draw [yshift = -2cm][black](5,0) -- (17,0);
\draw [yshift = -2cm](5,-0.5) -- (5,0.5);
\draw [yshift = -2cm](17,-0.5) -- (17,0.5);

\foreach \i in {1,...,16}
	\draw [yshift = -2cm][xshift=\i*0.75cm](5,-0.25)--(5,0.25); 
\foreach \i in {0,...,15}
	\filldraw [yshift = -2cm][gray,xshift=\i*0.75cm+0.375cm](5,0) circle [radius=2pt];
\foreach \i in {1,...,16}
	\draw [yshift = -2cm][xshift = \i*0.75cm-0.375cm] (4.9,0.5) node[text width = 1.33pt] {\i};	
\draw [yshift = -2cm](5,1.5)--(5.75,1.5);
\draw [yshift = -2cm](5.75,1.5) -- (5.75,4.5);
\draw [yshift = -2cm](5.75,4.5) -- (6.5,4.5);
\draw [yshift = -2cm](6.5,4.5) -- (6.5,1.5);
\draw [yshift = -2cm](6.5,1.5) -- (7.25,1.5);
\draw [yshift = -2cm](7.25,1.5) -- (7.25,4.5);
\draw [yshift = -2cm](7.25,4.5) -- (9.5,4.5);
\draw [yshift = -2cm](9.5,4.5) -- (9.5,1.5);
\draw [yshift = -2cm](9.5,1.5) -- (10.25,1.5);
\draw [yshift = -2cm](10.25,1.5) -- (10.25,4.5);
\draw [yshift = -2cm](10.25,4.5) -- (11,4.5);
\draw [yshift = -2cm](11,4.5) -- (11,1.5);
\draw [yshift = -2cm](11,1.5) -- (11.75,1.5);
\draw [yshift = -2cm](11.75,1.5) -- (11.75,4.5);
\draw [yshift = -2cm](11.75,4.5) -- (12.5,4.5);
\draw [yshift = -2cm](12.5,4.5) -- (12.5,1.5);
\draw [yshift = -2cm](12.5,1.5) -- (13.25,1.5);
\draw [yshift = -2cm](13.25,1.5) -- (13.25,4.5);
\draw [yshift = -2cm](13.25,4.5) -- (15.5,4.5);
\draw [yshift = -2cm](15.5,4.5) -- (15.5,1.5);
\draw [yshift = -2cm](15.5,1.5) -- (16.25,1.5);
\draw [yshift = -2cm](16.25,1.5) -- (16.25,4.5);
\draw [yshift = -2cm](16.25,4.5) -- (17,4.5);
\draw [yshift = -2cm](17,4.5) -- (17,1.5);
%\draw [yshift = -2cm](17,1.5) -- (17,4.5);
%

\filldraw [yshift = -2cm][red,xshift = 1.125cm] (5.0,0) circle [radius = 3pt];
\filldraw [yshift = -2cm][red,xshift = 2.625cm] (5.0,0) circle [radius = 3pt];
\filldraw [yshift = -2cm][red,xshift = 3.375cm] (5.0,0) circle [radius = 3pt];
\filldraw [yshift = -2cm][red,xshift = 4.125cm] (5.0,0) circle [radius = 3pt];
\filldraw [yshift = -2cm][red,xshift = 5.625cm] (5.0,0) circle [radius = 3pt];
\filldraw [yshift = -2cm][red,xshift = 7.125cm] (5.0,0) circle [radius = 3pt];
\filldraw [yshift = -2cm][red,xshift = 8.625cm] (5.0,0) circle [radius = 3pt];
\filldraw [yshift = -2cm][red,xshift = 9.375cm] (5.0,0) circle [radius = 3pt];
\filldraw [yshift = -2cm][red,xshift = 10.125cm] (5.0,0) circle [radius = 3pt];
\filldraw [yshift = -2cm][red,xshift = 11.625cm] (5.0,0) circle [radius = 3pt];

\end{tikzpicture}
\caption{Profiles of an adhesion pattern and an external cue, projected on cell's circumference. Bound and unbound focal adhesions are depicted as red and gray circles, respectively. Stress fibers are also colored in red.} 
\label{figure: fibrillar FA pattern}
\end{figure}
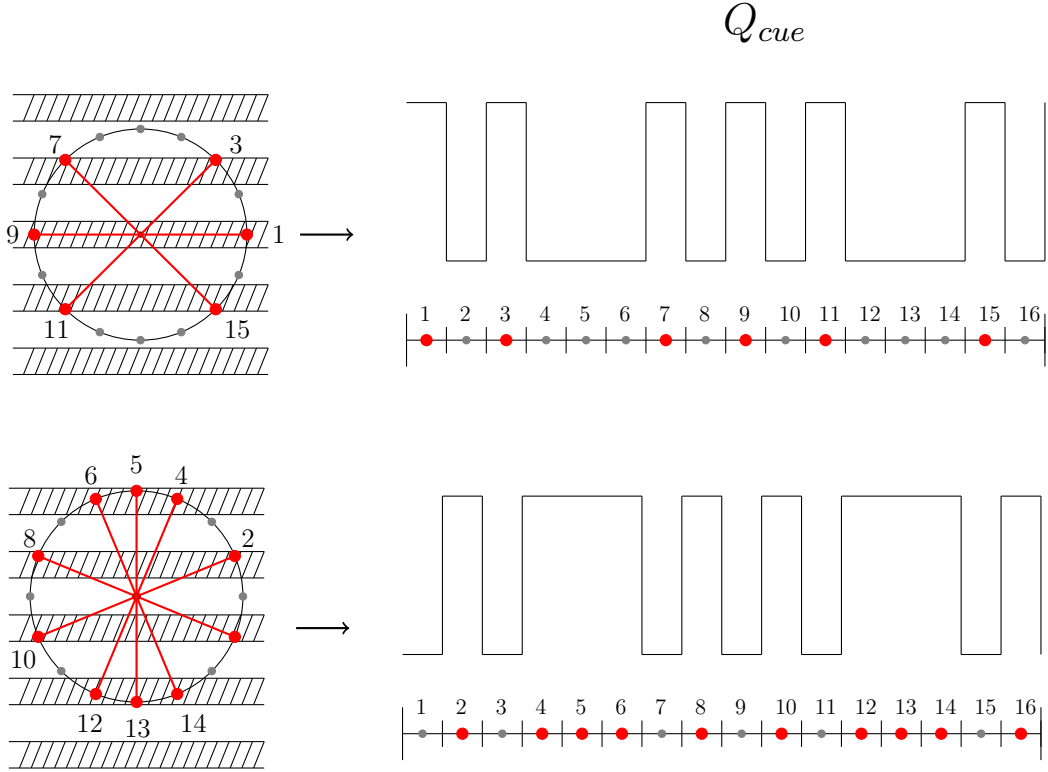
In Figure \ref{figure: fibrillar FA pattern} we illustrate the characteristic adhesions pattern and the profiles of the FA binding rate corresponding to ECM architecture in Figure \ref{figure: ECM stripes} (right). Assuming that there is a mechanical equilibrium for simplicity, we see that the adhesion pattern on the cell's periphery reflect the structure of the environment, since low values of $Q_{cue}$ translate into low probability of focal adhesion binding. Alternatively, if the cell is positioned as in Figure \ref{figure: fibrillar FA pattern} (bottom, left), then the adhesion pattern is modified accordingly. Thus, we see that our assumption about constant relative distance of FAs does not preclude the characteristic cell adhesion patterns to reflect environmental inhomogeneities (see also Figure \ref{figure: Qcue linear gradient}). 

Altogether, our simulations of contact guidance are, for the most part, consistent with the observations reported in the literature. In particular, we obtain the expected guidance of cell movement (Figure \ref{figure: contact guidance plots} (a-f)) and the geometric constraint of adhesion sites (Figure \ref{figure: contact guidance plots} (g-i)) by the fibrillar ECM pattern, in agreement with \cite{PDY09}, \cite{ROG17}. Nevertheless, since our model does not explicitly take into account morphological changes in cell shape (recall that in our model cell shape is normalized to a circle; see Section \ref{section: The Model}) and since cell shape control is essential to contact guidance \cite{PDY09}, \cite{ROG17}, increasing the interfiber distance does not necessarily lead to greater alignment of cell migration along the ECM fibers in our simulations\footnote{The values of the guidance parameter $G$, defined as in \cite{ROG17} (See Appendix \ref{appendix: data analysis} for details), were found to be $0.64,0.70,0.64$, corresponding to, respectively, $\delta_G=0.15,0.25,0.35$.}. Moreover, in the case when the total number of adhesion sites is very low, the stripes are too narrow, and the separation between them is large, then it might occur that all adhesion sites ``miss" the stripes, although the cell is spread over multiple stripes. In this case, the probability that any adhesion binds to the substrate is low, which is not biologically consistent. To remedy these shortcomings, the model needs to be extended in order to accommodate strong changes to cell morphology. 

\subsection{Asymmetric contractility}
We now investigate how cell motility is influenced by asymmetrical contractile forces in a cell. Along with preferential adhesion formation, due to, for example, a chemotactic gradient, formation of cell rear by increased actomyosin contractile activity serves as an alternative mechanism by which a directed migration can be induced in the absence of such gradient \cite{Cramer2010}. In particular, local stimulation of contractility leads to directed motility in the direction opposite to the stimulated area, even in the absence of response to chemotactic stimuli \cite{yam2007actin}. Here we show that our model is also capable of capturing such directed movement, triggered by breaking myosin mediated contractile symmetry. 

Recall that $T_i$ in equation \eqref{eq: Fi} denotes the force generated by myosin motors at an adhesion site $i$. Instead of taking it constant, we let it vary with the radial position of an FA. Namely, let $T_i:[0,2\pi)\rightarrow \mathbb{R}_{+}$ be defined as:
\begin{align*}
T_i(\theta) = 
\begin{cases*}
(1+\delta_{myo})T_{i_0},\text{ if } \pi<\theta+(i-1)\frac{2\pi}{M}<2\pi\\
T_{i_0}, \phantom{aasdaasad} \text{else}
\end{cases*},
\end{align*}
where $\delta_{myo} = 0.35,0.40,0.45$, $T_{i_0}$ is the constant value used in previous simulations (see Appendix \ref{appendix: parameters}), and $\theta+(i-1)\frac{2\pi}{M}$ is the radial position of the $i^{\text{th}}$ FA. Thus, contractile forces south of cell equator are larger by $35\%,40\%, 45\%$ for corresponding values of $\delta_{myo}$. We should, therefore, expect in our simulations that the northern part of a cell becomes the front due to the imposed contractile symmetry breaking, and that cells will move accordingly (see Figure \ref{figure: asymmetric contractility} for illustration).

\begin{figure}[!h]
%\vspace*{-3cm}
\begin{tikzpicture}

%Top row
% %Left
\draw (0,0) circle [radius=1.5cm];
\draw (0,-0.4) -- (1.5,0);
\draw (0,-0.4) -- (1.0607,1.0607);
\draw (0,-0.4) -- (0,1.5);
\draw (0,-0.4) -- (-1.0607,1.0607);
\draw (0,-0.4) -- (-1.5,0);
\draw (0,-0.4) -- (-1.0607,-1.0607);
\draw (0,-0.4) -- (0,-1.5);
%\draw (0,-0.4) -- (1.0607,-1.0607);
\draw [dashed](0,-0.4) -- (1.0607,-1.0607);

\filldraw [red] (1.5,0) circle [radius=1.5pt];
\filldraw [red] (1.0607,1.0607) circle [radius=1.5pt];
\filldraw [red] (0,1.5) circle [radius=1.5pt];
\filldraw [red] (-1.0607,1.0607) circle [radius=1.5pt];
\filldraw [red] (-1.5,0) circle [radius=1.5pt];
\filldraw [red] (-1.0607,-1.0607) circle [radius=1.5pt];
\filldraw [red] (0,-1.5) circle [radius=1.5pt];
\filldraw [red] (1.0607,-1.0607) circle [radius=1.5pt];

\draw [thick,red][->](1.5,0) -- (0.8667,-0.14);
\draw [thick,red][->](1.0607,1.0607) -- (0.5303,0.3303);
\draw [thick,red][->](0,1.5) -- (0,0.5);
\draw [thick,red][->](-1.0607,1.0607) -- (-0.5303,0.3303);
\draw [thick,red][->](-1.5,0) -- (-0.8667,-0.14);
\draw [thick,red][->](-1.0607,-1.0607) -- (-0.3979,-0.64);
\draw [thick,red][->](0,-1.5) -- (0,-0.5);
%\draw [thick,red][->](1.0607,-1.0607) -- (0.3979,-0.64);

%\draw [xshift=5cm,yshift = -6.3cm][very thick,black][->](0,0) --(0,0.5);
\filldraw [blue] (0,-0.4) circle [radius=1.5pt];

%\draw [xshift = -0.2cm][yshift=0.2cm][very thick, blue][->] (-1.0607,1.0607)--(-1.4,1.4);
%\draw [very thick, blue][->] (-1.4782,0.6123)--(-1.9215,0.7959);
\draw [very thick, blue][->] (0,-0.4)--(-0.6929,0.1);
\filldraw [blue] (0,-0.4) circle [radius=1.5pt];

% %Mid
\draw [xshift=5cm](0,0) circle [radius=1.5cm];
\draw [xshift=5cm](0,-0.4) -- (1.5,0);
\draw [xshift=5cm](0,-0.4) -- (1.0607,1.0607);
\draw [xshift=5cm](0,-0.4) -- (0,1.5);
\draw [xshift=5cm](0,-0.4) -- (-1.0607,1.0607);
\draw [xshift=5cm](0,-0.4) -- (-1.5,0);
\draw [xshift=5cm](0,-0.4) -- (-1.0607,-1.0607);
%\draw [xshift=5cm](0,-0.4) -- (0,-1.5);
\draw [xshift=5cm](0,-0.4) -- (1.0607,-1.0607);
\draw [xshift=5cm][dashed](0,-0.4) -- (0,-1.5);

\filldraw [xshift=5cm][red] (1.5,0) circle [radius=1.5pt];
\filldraw [xshift=5cm][red] (1.0607,1.0607) circle [radius=1.5pt];
\filldraw [xshift=5cm][red] (0,1.5) circle [radius=1.5pt];
\filldraw [xshift=5cm][red] (-1.0607,1.0607) circle [radius=1.5pt];
\filldraw [xshift=5cm][red] (-1.5,0) circle [radius=1.5pt];
\filldraw [xshift=5cm][red] (-1.0607,-1.0607) circle [radius=1.5pt];
\filldraw [xshift=5cm][red] (0,-1.5) circle [radius=1.5pt];
\filldraw [xshift=5cm][red] (1.0607,-1.0607) circle [radius=1.5pt];

\draw [xshift=5cm][thick,red][->](1.5,0) -- (0.8667,-0.14);
\draw [xshift=5cm][thick,red][->](1.0607,1.0607) -- (0.5303,0.3303);
\draw [xshift=5cm][thick,red][->](0,1.5) -- (0,0.5);
\draw [xshift=5cm][thick,red][->](-1.0607,1.0607) -- (-0.5303,0.3303);
\draw [xshift=5cm][thick,red][->](-1.5,0) -- (-0.8667,-0.14);
\draw [xshift=5cm][thick,red][->](-1.0607,-1.0607) -- (-0.3979,-0.64);
%\draw [xshift=5cm][thick,red][->](0,-1.5) -- (0,-0.5);
\draw [xshift=5cm][thick,red][->](1.0607,-1.0607) -- (0.3979,-0.64);

%\draw [xshift=5cm,yshift = -6.3cm][very thick,black][->](0,0) --(0,0.5);
\filldraw [xshift=5cm][blue] (0,-0.4) circle [radius=1.5pt];

%\draw [xshift = -0.2cm][yshift=0.2cm][very thick, blue][->] (-1.0607,1.0607)--(-1.4,1.4);
%\draw [very thick, blue][->] (-1.4782,0.6123)--(-1.9215,0.7959);
\draw [xshift=5cm][very thick, blue][->] (0,-0.4)--(0,0.4);
\filldraw [xshift=5cm][blue] (0,-0.4) circle [radius=1.5pt];
% %Right
\draw [xshift=10cm](0,0) circle [radius=1.5cm];
\draw [xshift=10cm](0,-0.4) -- (1.5,0);
\draw [xshift=10cm](0,-0.4) -- (1.0607,1.0607);
\draw [xshift=10cm](0,-0.4) -- (0,1.5);
\draw [xshift=10cm](0,-0.4) -- (-1.0607,1.0607);
\draw [xshift=10cm](0,-0.4) -- (-1.5,0);
%\draw [xshift=10cm](0,-0.4) -- (-1.0607,-1.0607);
\draw [xshift=10cm](0,-0.4) -- (0,-1.5);
\draw [xshift=10cm](0,-0.4) -- (1.0607,-1.0607);
\draw [xshift=10cm][dashed](0,-0.4) -- (-1.0607,-1.0607);

\filldraw [xshift=10cm][red] (1.5,0) circle [radius=1.5pt];
\filldraw [xshift=10cm][red] (1.0607,1.0607) circle [radius=1.5pt];
\filldraw [xshift=10cm][red] (0,1.5) circle [radius=1.5pt];
\filldraw [xshift=10cm][red] (-1.0607,1.0607) circle [radius=1.5pt];
\filldraw [xshift=10cm][red] (-1.5,0) circle [radius=1.5pt];
\filldraw [xshift=10cm][red] (-1.0607,-1.0607) circle [radius=1.5pt];
\filldraw [xshift=10cm][red] (0,-1.5) circle [radius=1.5pt];
\filldraw [xshift=10cm][red] (1.0607,-1.0607) circle [radius=1.5pt];

\draw [xshift=10cm][thick,red][->](1.5,0) -- (0.8667,-0.14);
\draw [xshift=10cm][thick,red][->](1.0607,1.0607) -- (0.5303,0.3303);
\draw [xshift=10cm][thick,red][->](0,1.5) -- (0,0.5);
\draw [xshift=10cm][thick,red][->](-1.0607,1.0607) -- (-0.5303,0.3303);
\draw [xshift=10cm][thick,red][->](-1.5,0) -- (-0.8667,-0.14);
%\draw [xshift=10cm][thick,red][->](-1.0607,-1.0607) -- (-0.3979,-0.64);
\draw [xshift=10cm][thick,red][->](0,-1.5) -- (0,-0.5);
\draw [xshift=10cm][thick,red][->](1.0607,-1.0607) -- (0.3979,-0.64);

%\draw [xshift=5cm,yshift = -6.3cm][very thick,black][->](0,0) --(0,0.5);
\filldraw [xshift=10cm][blue] (0,-0.4) circle [radius=1.5pt];

%\draw [xshift = -0.2cm][yshift=0.2cm][very thick, blue][->] (-1.0607,1.0607)--(-1.4,1.4);
%\draw [very thick, blue][->] (-1.4782,0.6123)--(-1.9215,0.7959);
\draw [xshift=10cm][very thick, blue][->] (0,-0.4)--(0.6929,0.1);
\filldraw [xshift=10cm][blue] (0,-0.4) circle [radius=1.5pt];

%Mid row

%Bot row
\draw [xshift=5cm,yshift = -4cm](0,0) circle [radius=1.5cm];
\draw [xshift=5cm,yshift = -4cm](0,-0.4) -- (1.5,0);
\draw [xshift=5cm,yshift = -4cm](0,-0.4) -- (1.0607,1.0607);
\draw [xshift=5cm,yshift = -4cm](0,-0.4) -- (0,1.5);
\draw [xshift=5cm,yshift = -4cm](0,-0.4) -- (-1.0607,1.0607);
\draw [xshift=5cm,yshift = -4cm](0,-0.4) -- (-1.5,0);
\draw [xshift=5cm,yshift = -4cm](0,-0.4) -- (-1.0607,-1.0607);
\draw [xshift=5cm,yshift = -4cm](0,-0.4) -- (0,-1.5);
\draw [xshift=5cm,yshift = -4cm](0,-0.4) -- (1.0607,-1.0607);

\filldraw [xshift=5cm,yshift = -4cm][red] (1.5,0) circle [radius=1.5pt];
\filldraw [xshift=5cm,yshift = -4cm][red] (1.0607,1.0607) circle [radius=1.5pt];
\filldraw [xshift=5cm,yshift = -4cm][red] (0,1.5) circle [radius=1.5pt];
\filldraw [xshift=5cm,yshift = -4cm][red] (-1.0607,1.0607) circle [radius=1.5pt];
\filldraw [xshift=5cm,yshift = -4cm][red] (-1.5,0) circle [radius=1.5pt];
\filldraw [xshift=5cm,yshift = -4cm][red] (-1.0607,-1.0607) circle [radius=1.5pt];
\filldraw [xshift=5cm,yshift = -4cm][red] (0,-1.5) circle [radius=1.5pt];
\filldraw [xshift=5cm,yshift = -4cm][red] (1.0607,-1.0607) circle [radius=1.5pt];

\draw [xshift=5.2cm,yshift = -4cm] (1.5,0) node[scale = 1, text width = 1.33pt] {$1$};
\draw [xshift=5.2cm,yshift = -3.8cm] (1.0607,1.0607) node[scale = 1, text width = 1.33pt] {$2$};
\draw [xshift=5.3cm,yshift = -3.8cm] (0,1.5) node[scale = 1, text width = 1.33pt] {$3$};
\draw [xshift=4.7cm,yshift = -3.8cm] (-1.0607,1.0607) node[scale = 1, text width = 1.33pt] {$4$};
\draw [xshift=4.7cm,yshift = -4cm] (-1.5,0) node[scale = 1, text width = 1.33pt] {$5$};
\draw [xshift=4.7cm,yshift = -4.2cm] (-1.0607,-1.0607) node[scale = 1, text width = 1.33pt] {$6$};
\draw [xshift=5cm,yshift = -4.3cm] (0,-1.5) node[scale = 1, text width = 1.33pt] {$7$};
\draw [xshift=5.2cm,yshift = -4.2cm] (1.0607,-1.0607) node[scale = 1, text width = 1.33pt] {$8$};

\draw [xshift=5cm,yshift = -4cm][thick,red][->](1.5,0) -- (0.8667,-0.14);
\draw [xshift=5cm,yshift = -4cm][thick,red][->](1.0607,1.0607) -- (0.5303,0.3303);
\draw [xshift=5cm,yshift = -4cm][thick,red][->](0,1.5) -- (0,0.5);
\draw [xshift=5cm,yshift = -4cm][thick,red][->](-1.0607,1.0607) -- (-0.5303,0.3303);
\draw [xshift=5cm,yshift = -4cm][thick,red][->](-1.5,0) -- (-0.8667,-0.14);
\draw [xshift=5cm,yshift = -4cm][thick,red][->](-1.0607,-1.0607) -- (-0.3979,-0.64);
\draw [xshift=5cm,yshift = -4cm][thick,red][->](0,-1.5) -- (0,-0.5);
\draw [xshift=5cm,yshift = -4cm][thick,red][->](1.0607,-1.0607) -- (0.3979,-0.64);

%\draw [xshift=5cm,yshift = -6.3cm][very thick,black][->](0,0) --(0,0.5);
\filldraw [xshift=5cm,yshift = -4cm][blue] (0,-0.4) circle [radius=1.5pt];

\draw [xshift=6cm][yshift=-3.6cm][very thick, black][->] (1.0607,1.0607)--(1.6,1.6);
\draw [xshift=4cm][yshift=-3.6cm][very thick, black][->] (-1.0607,1.0607)--(-1.6,1.6);
\draw [xshift=5cm][yshift=-2.3cm][very thick, black][->] (0,0)--(0,0.5);

\end{tikzpicture}
\caption{A schematic representation of asymmetric contractility. (Bottom row) Increased contractility causes $\mathbf{x}_n$ (blue circle) to shift south of the otherwise equilibrium point in the center. (Top row) Preferential unbinding of FAs south of equator leads to directed movement indicated by blue arrows. See also Figure \ref{fig: motion cycle2} (II') for schematic representation of cell motility in case of an unbinding event.}

\label{figure: asymmetric contractility}
\end{figure}
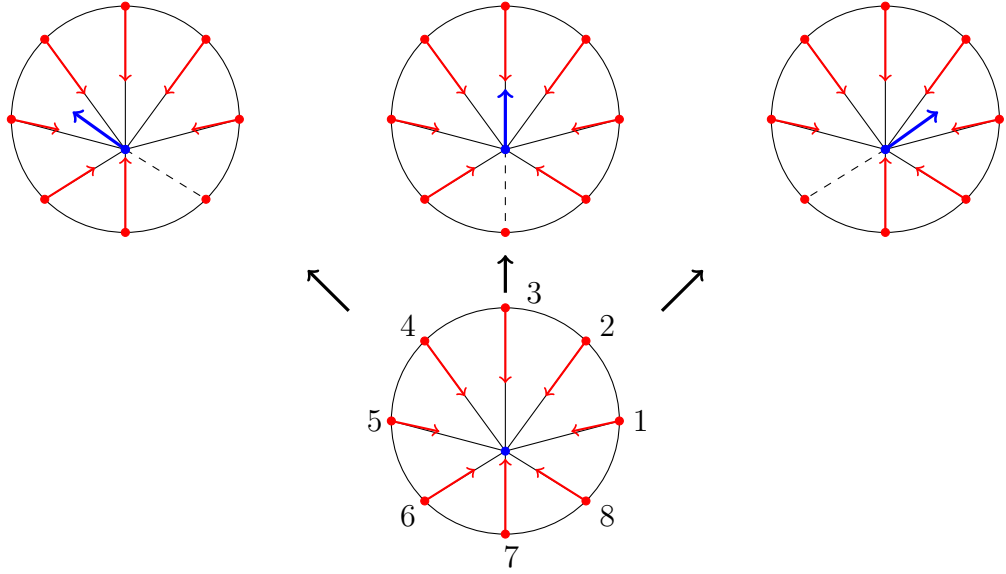
Indeed, Figure \ref{figure: myosin plots} (a-c) shows, as expected, the trajectories  of cells maintaining north-south polarity corresponding to, respectively, front and rear. Since the asymmetry of myosin forces remained during the simulations, the cell's north-south polarity also persisted, resulting in the cell movement that was highly directed along this axis, consistent with \cite{yam2007actin}. Consequently, we obtain higher values of $\beta_{av}(t)$, as shown in Figure \ref{figure: myosin plots} (d-f). In particular, for $\delta_{myo}=0.45$, we see that the time scaling of the mean squared displacement is close to ballistic (see Table \ref{table: myosin parameters} for values of $\bar{\beta}$). 
\begin{figure}[!h]
\vspace{0mm}
\centering
%First row
\subfloat[]
{
	\includegraphics[width=40mm,height=37mm]{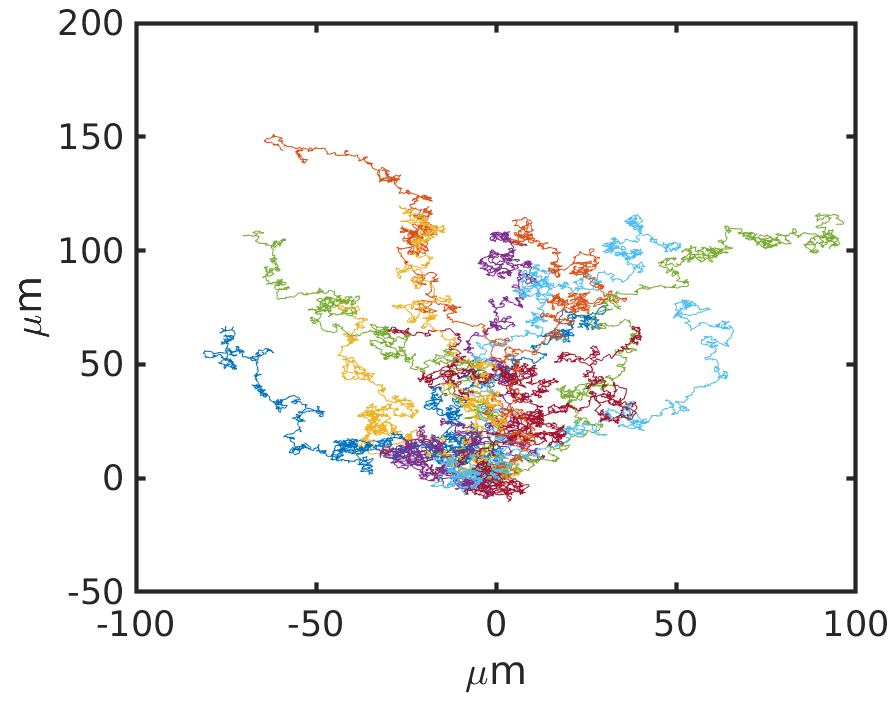}
}
\subfloat[]
{
	\includegraphics[width=40mm,height=37mm]{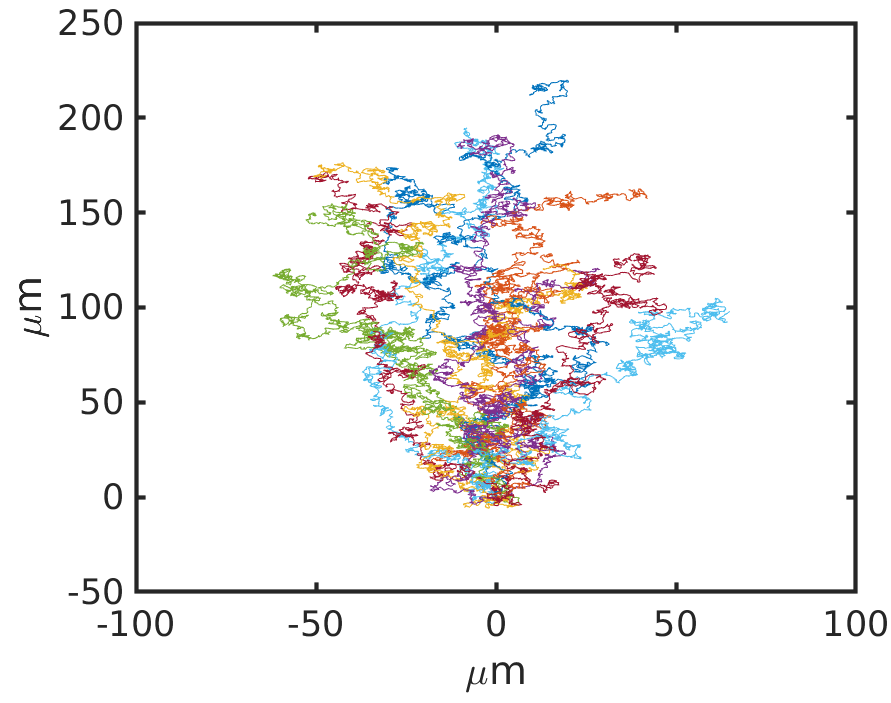}
}
\subfloat[]
{
	\includegraphics[width=40mm,height=37mm]{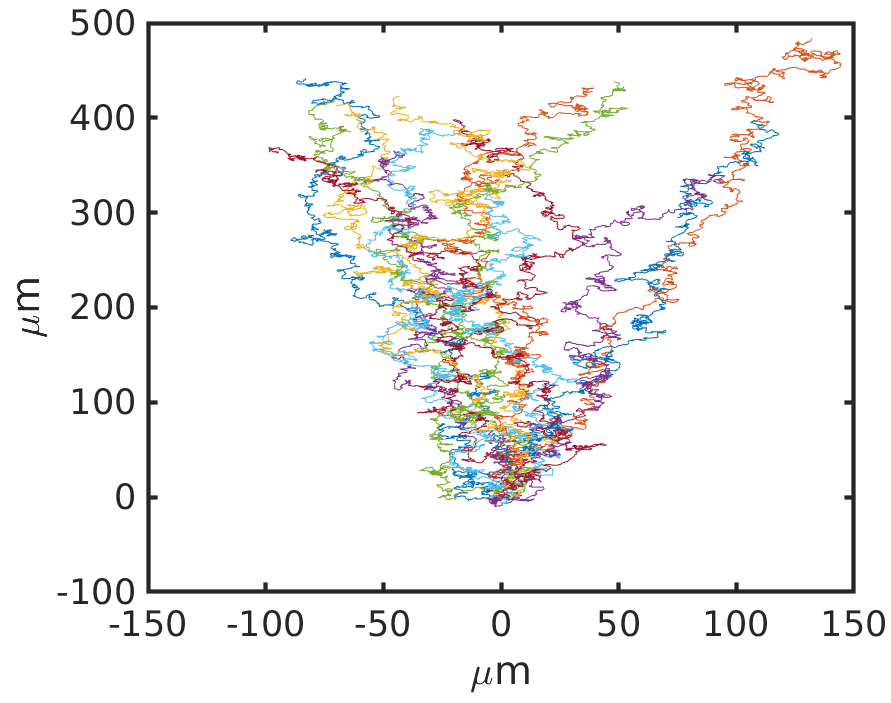}
}
\hspace{0mm}
%Second row
\subfloat[]
{
	\includegraphics[width=40mm,height=37mm]{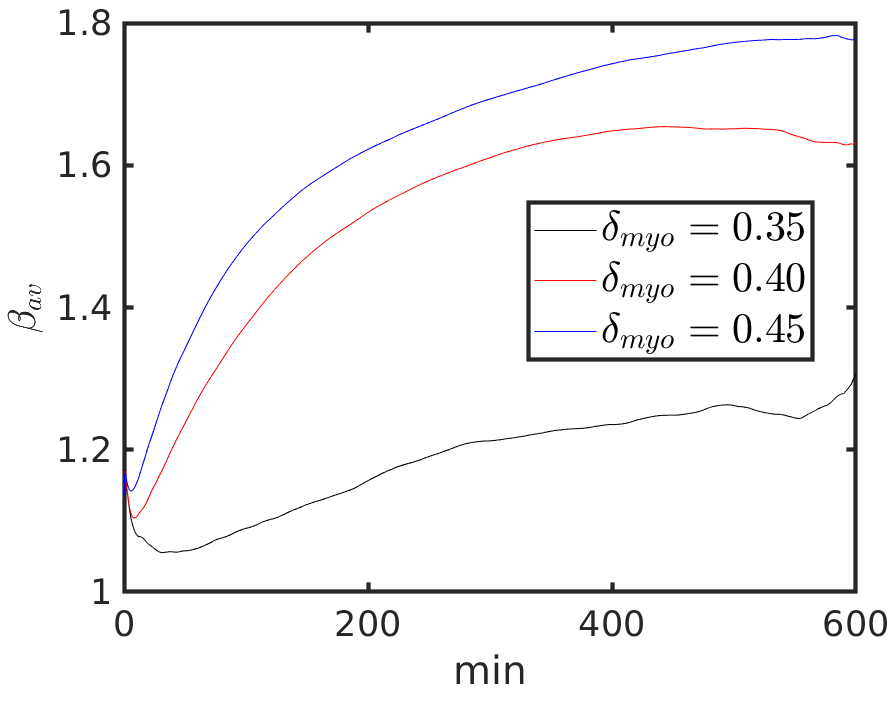}
}
\subfloat[]
{
	\includegraphics[width=40mm,height=37mm]{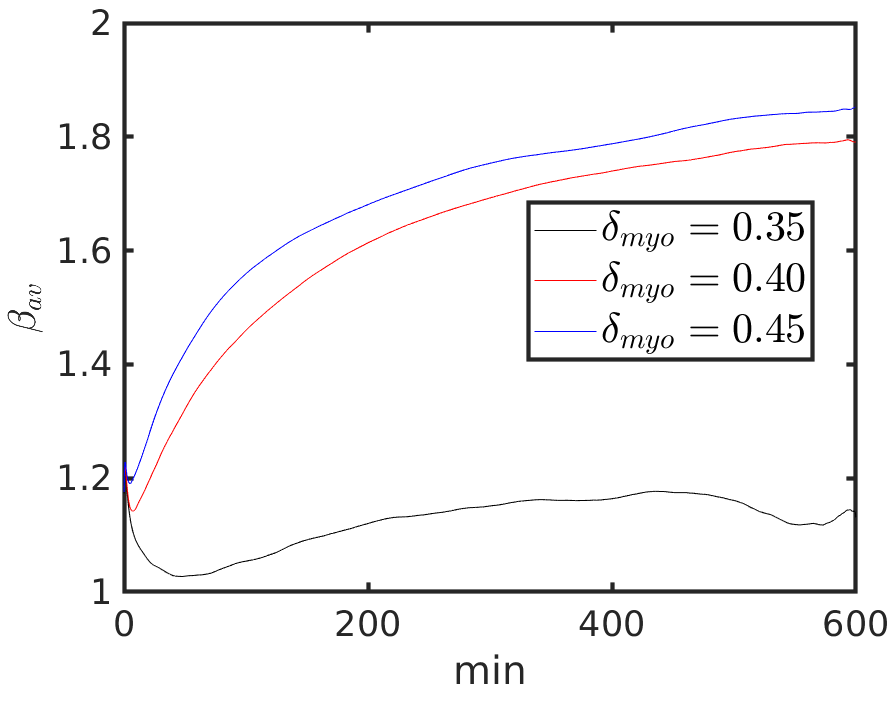}
}
\subfloat[]
{
	\includegraphics[width=40mm,height=37mm]{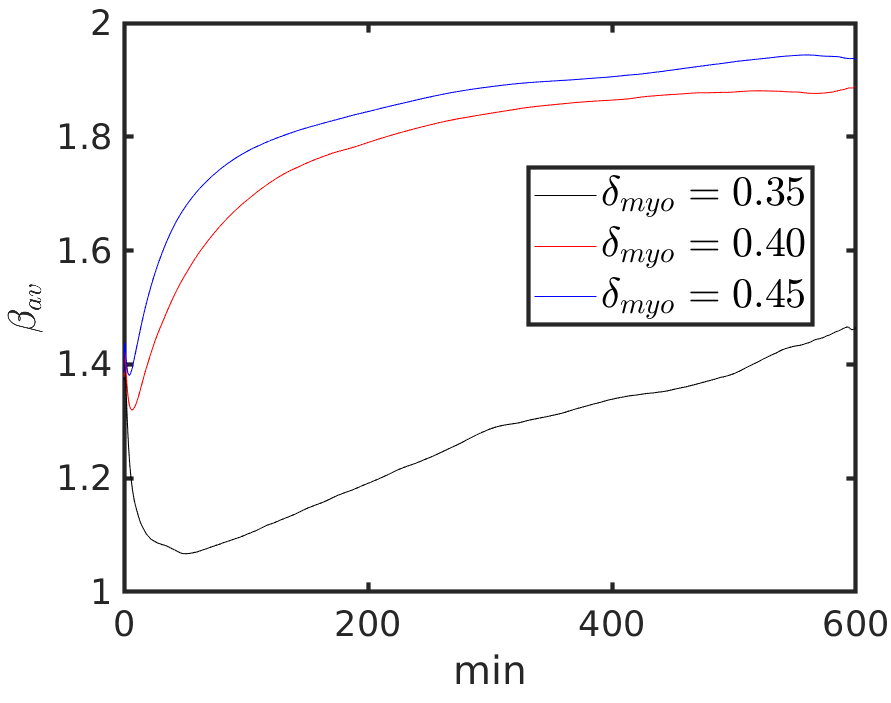}
}
\hspace{0mm}
%%Third row

%%Fourth row
\subfloat[]
{
	\includegraphics[width=40mm,height=37mm]{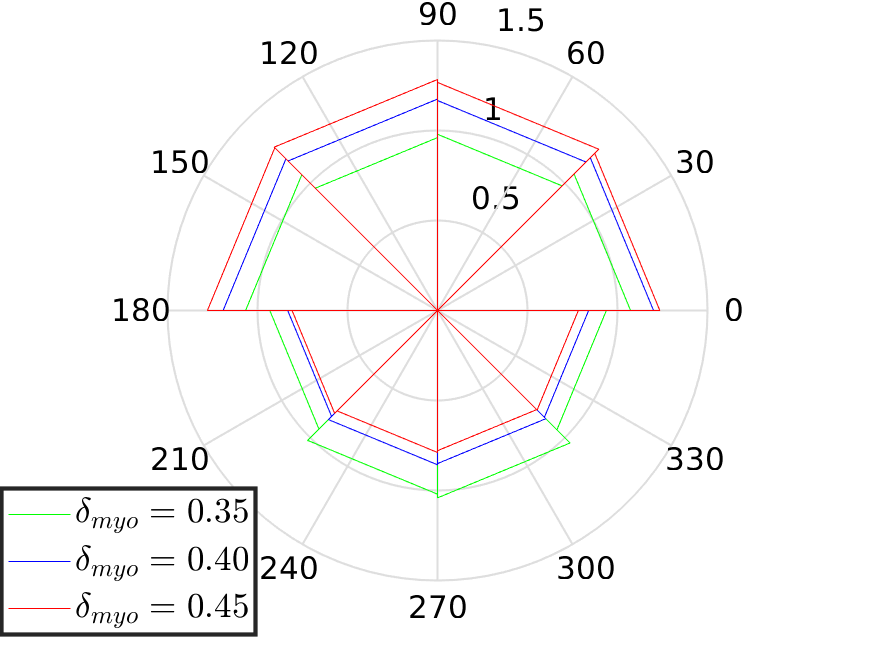}
}
\subfloat[]
{
	\includegraphics[width=40mm,height=37mm]{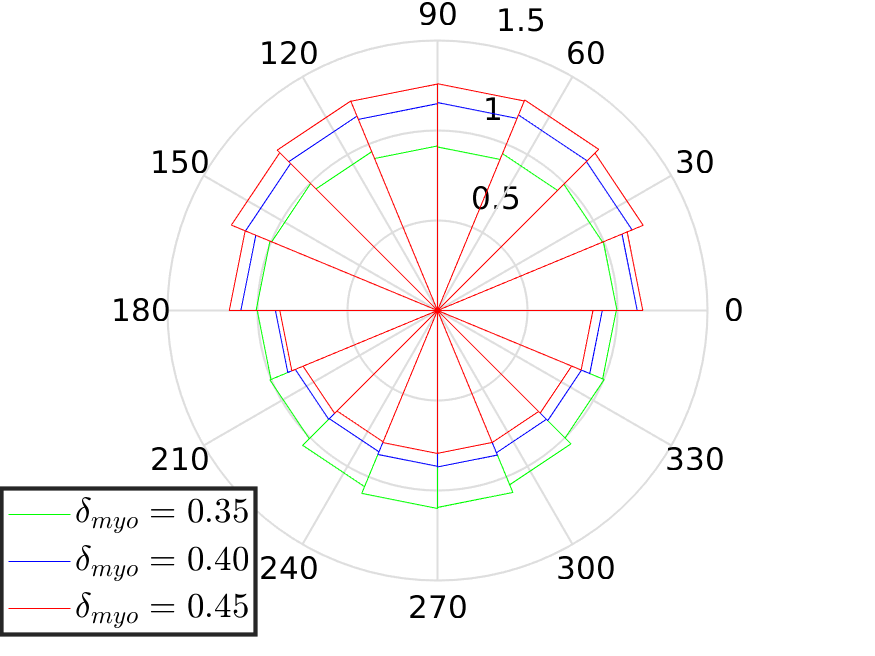}
}
\subfloat[]
{
	\includegraphics[width=40mm,height=37mm]{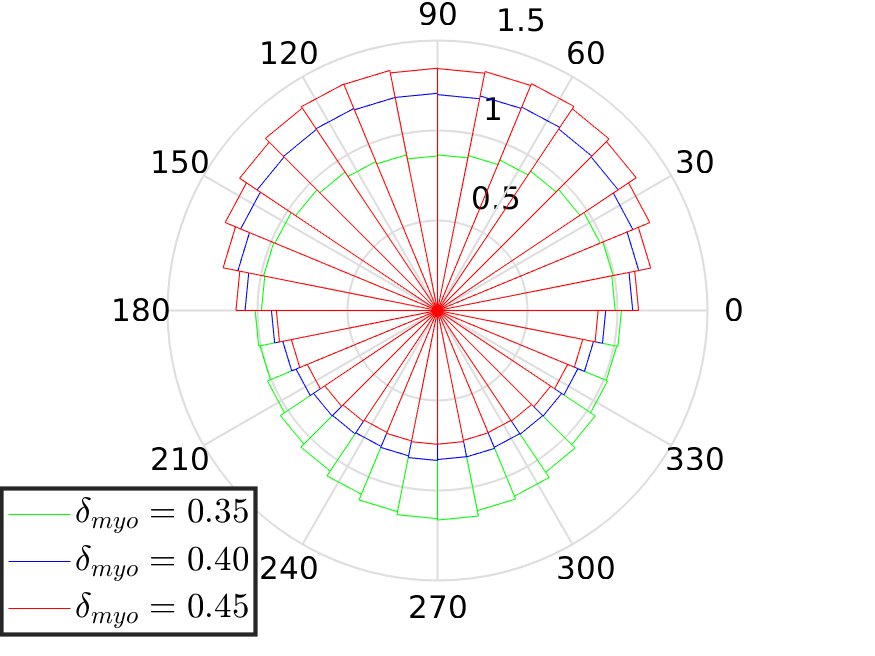}
}

\caption{Simulation results with $M=8,16,32$ adhesions in the first, second, and third columns respectively, and with various values of $\delta_{myo}$. (a-c) Trajectories of 13 cells with $\delta_{myo}=0.4$. (d-f) Time-average exponents $\beta_{av}$. (g-i) Ratio of the number of binding to unbinding events in each sector. }
\label{figure: myosin plots}
\end{figure}
\begin{table}[H]
\centering
\def\arraystretch{1.5}
\setlength\tabcolsep{5.5pt}
\begin{tabular}{|c|c c c|c c c|c c c|}
\hline

M & \multicolumn{3}{c|}{8} & \multicolumn{3}{c|}{16} & \multicolumn{3}{c|}{32}\\
\hline
$\delta_{myo}$ & 0.35 & 0.40 & 0.45 & 0.35 & 0.40 & 0.45 & 0.35 & 0.4 & 0.45\\
\hline
$\bar{\beta}$, 1 & 1.3072 & 1.6325 & 1.7759 & 1.1311 & 1.7905 & 1.8524 & 1.4650 & 1.8892 & 1.9353\\
\hline
\begin{tabular}{@{}c@{}}$s_{av}$, \\ $\mu m/min$\end{tabular} & 1.6230 & 1.5639 & 1.5103 & 2.3742 & 2.3798 & 2.3516 & 3.5861 & 3.7414 & 3.7701\\
\hline
\begin{tabular}{@{}c@{}}$\beta_0$, \\ $\mu m^2/min^{\bar{\beta}}$\end{tabular} & 0.7767 & 0.3659 & 0.2493 & 2.3088 & 0.2893 & 0.3037 & 1.0713 & 0.9894 & 1.2787\\
\hline
\end{tabular}
 \caption{Parameters obtained from the simulations with varying $\delta_{myo}$.}
 \label{table: myosin parameters}
\end{table}

Moreover, increasing the number of FAs leads to more polarized, directed migration. As in the previous cases, neither speed averages (Table \ref{table: myosin parameters}) nor their distribution (data not shown) changed significantly for a given number of focal adhesions. Interestingly, for $\delta_{myo}=0.35$, the binding is relatively more frequent in the rear (i.e. south of equator) than in the front, and unbinding is relatively more frequent in the front (i.e. north of equator) than in the back (Figure \ref{figure: myosin plots} (g-i)). This suggests, then, that cells were preferentially moving in the southern direction. However, as can be seen in Figure \ref{figure: myosin 35p plots}, this is not the case. Although movements southwards are more frequent in this situation (due to the above-mentioned event frequencies), the speeds are lower than northward movements: the ratios of the average speeds directed north to the average speeds directed south were found to be $1.0165$, $1.0181$, and $1.0858$ corresponding to, respectively, $M=8,16,32$. The net effect is northward movement. For higher values of $\delta_{myo}$, we see that the unbinding is, expectedly, more frequent in the rear, while binding is preferentially in the front.

\begin{figure}[H]
\captionsetup[subfigure]{labelformat=empty}
%\vspace{0mm}
\centering
%First row
\subfloat[]
{
	\includegraphics[width=40mm,height=37mm]{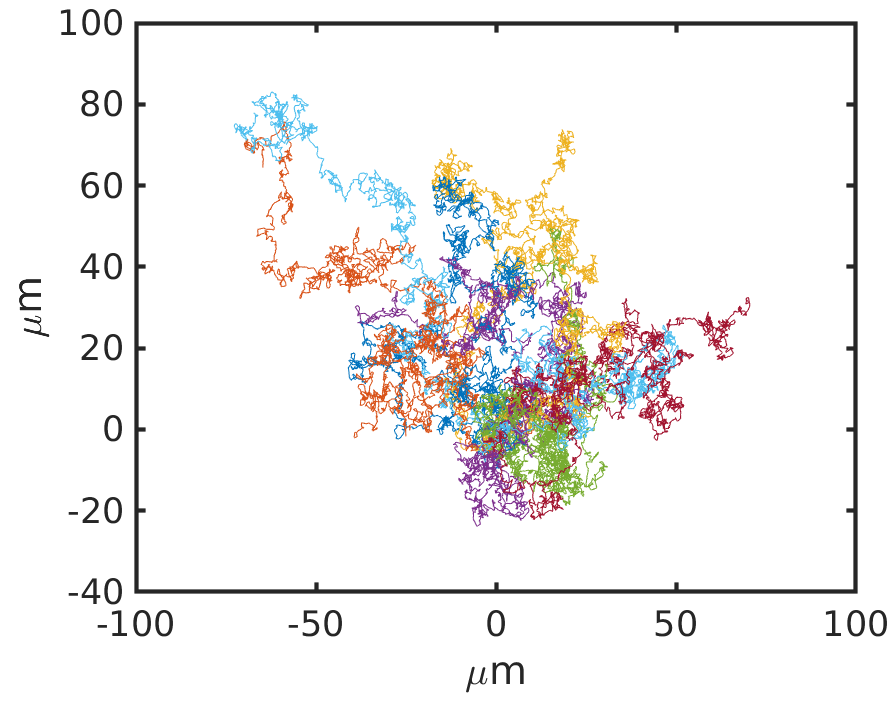}
}
\subfloat[]
{
	\includegraphics[width=40mm,height=37mm]{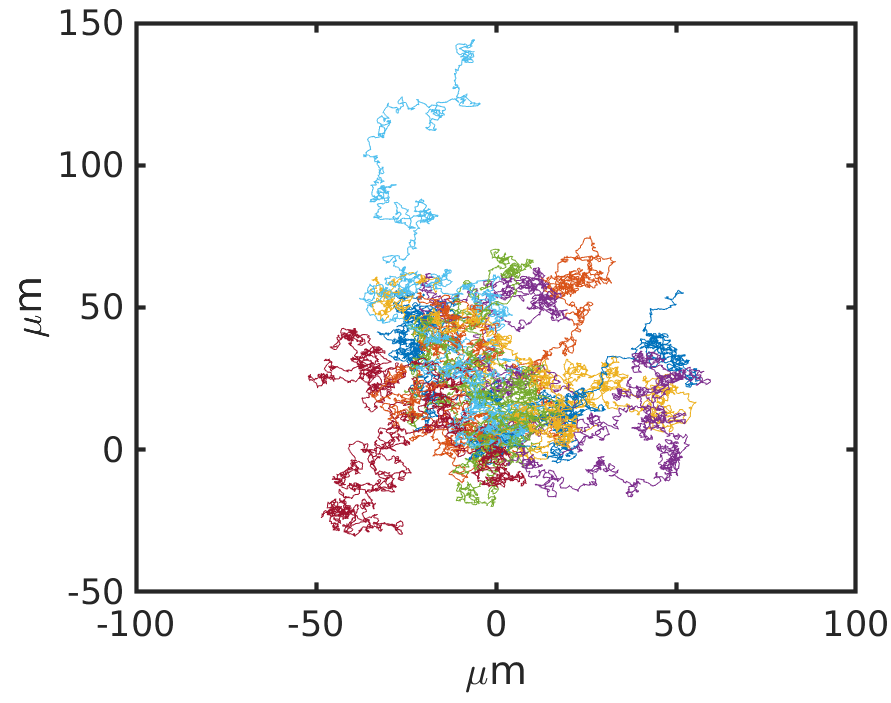}
}
\subfloat[]
{
	\includegraphics[width=40mm,height=37mm]{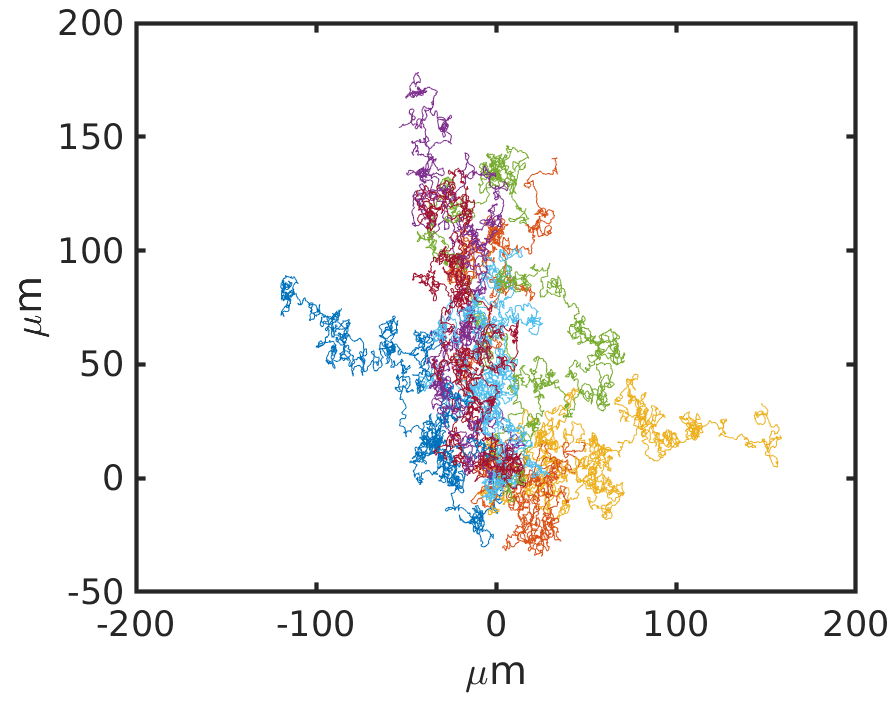}
}
\caption{Simulated trajectories with $M=8,16,32$ adhesions with $\delta_{myo} = 0.35$ on left, middle, and right plots, respectively.}
\label{figure: myosin 35p plots}
\end{figure} 

These adhesion frequency patterns also illustrate the significance of the force dependence of the FA binding rate. Recalling Figure \ref{figure: force dependence}, we see that, for $\delta_{myo}=0.4,0.45$ (corresponding to $T_i=1.018F_b, 1.054F_b$), the binding rate dominates unbinding north of equator due to greater SF extension (see Figure \ref{figure: asymmetric contractility} for an illustration) leading to increased contractile force. Since the expected adhesion pattern is reversed for $\delta_{myo}=0.35$ (corresponding to $T_i=0.981F_b$) and yet the cells migrate northwards, it may suggest that there is a threshold value of $\delta_{myo}$, above which cells can migrate in a certain direction solely by asymmetric contractility, and/or below which cells must additionally bias adhesion formation to do so.        

\begin{figure}[H]
%\captionsetup[subfigure]{labelformat=empty}
%\vspace{0mm}
\centering
%First row
\subfloat[$\delta_{EA}=0$]
{
	\includegraphics[width=40mm,height=37mm]{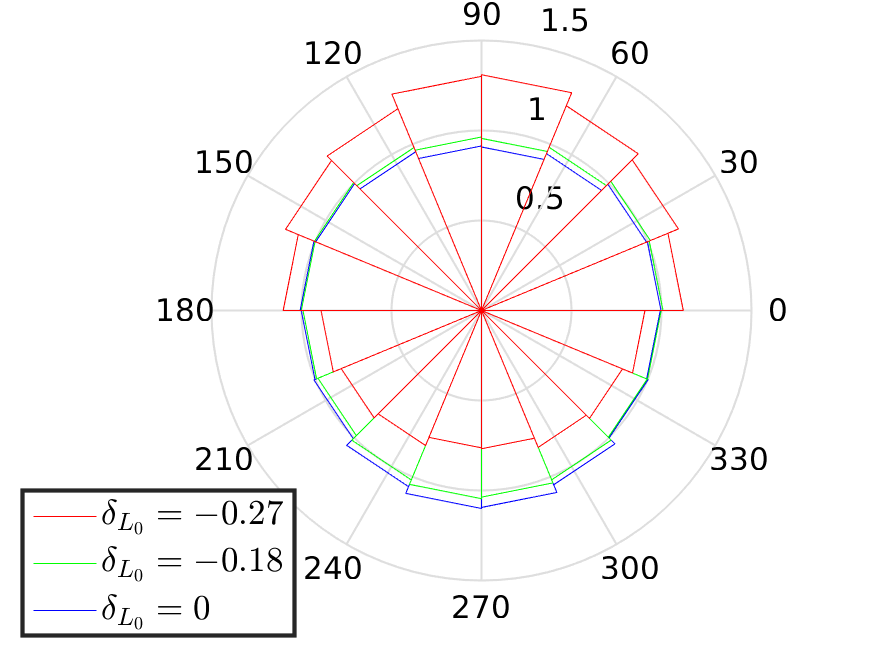}
}
\subfloat[$\delta_{L_0}=-0.18$]
{
	\includegraphics[width=40mm,height=37mm]{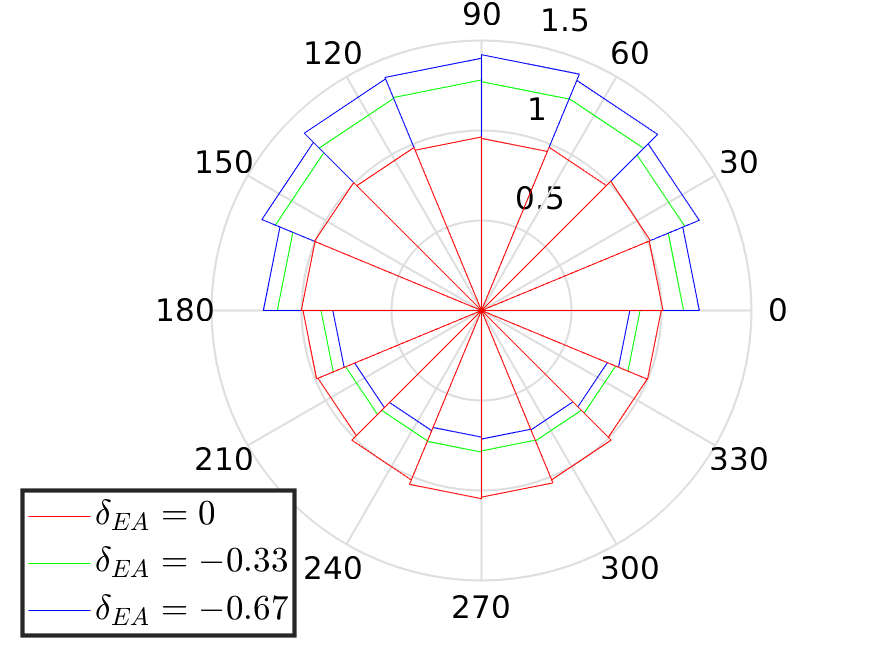}
}
\subfloat[$\delta_{L_0}=-0.27$]
{
	\includegraphics[width=40mm,height=37mm]{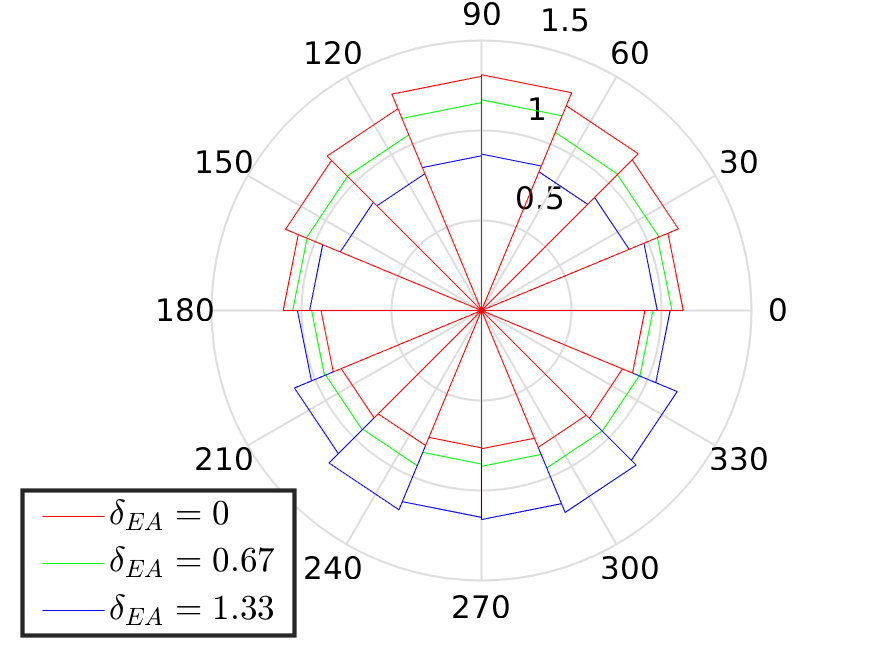}
}
\caption{Ratios of the number of binding to unbinding events in each sector with varying buckling length and stiffness of SFs. (a) The effect of reducing the buckling length $L_0$ with fixed and stiffness value $EA$. (b-c) The effects of varying stiffness $EA$ for reduced buckling length.}
\label{figure: myosin 35p force FAs}
\end{figure} 
This prompted us to investigate whether varying mechanical properties of SFs can yield the expected adhesion pattern for lower degree of asymmetry, corresponding to $\delta_{myo}=0.35$. Specifically, we varied the buckling length $L_0$ and stiffness $EA$ such that  $x=(1+\delta_{x})x^0$ corresponds to the modified value of the parameter $x\in\lbrace L_0, EA\rbrace$, where $x^0$ corresponds to the default value given in Appendix \ref{appendix: parameters}. In Figure \ref{figure: myosin 35p force FAs}(a) we see that reducing the buckling length $L_0$ by 27\% leads to the expected adhesion pattern, while reducing it by 18\% leaves it largely unchanged. However, decreasing and increasing stiffness when $\delta_{L_0}=-0.18,-0.27$, respectively, leads to the opposite results (Figure \ref{figure: myosin 35p force FAs}(b,c)). This suggests that if SFs are less prone to buckling and less stiff, lower degree of myosin induced contractile asymmetry may be required to drive directed migration.

\textbf{Remark.} Another way to induce contractile asymmetry is, for example, to decrease the myosin force $T_i$ north of the cell's equator. Then, again, the south of the cell equator is more contractile. However, the simulated trajectories show southward directed movement (data not shown), contrary to what we should expect. Therefore, merely inducing contractile asymmetry is not sufficient. For the expected directed migration to occur, there must be a local increase of contractile forces above some critical level in the prospective cell rear, rather than a local decrease of contractility in the prospective front. Interestingly, Yam et al. \cite{yam2007actin} were able to initiate directed movement by increasing local actomyosin contraction, while locally decreasing the contractile activity did not lead to migration initiation. More recently, Shellard et al. \cite{Shellard339} showed that directed collective cell migration of neural crest cells requires greater contractility at the rear of the clump.
\section{Discussion and outlook}\label{section: discussion}
In this paper we constructed a stochastic model of cell migration using a minimal representation of cellular structures, essential for crawling, such as stress fibers and focal adhesions. Using this representation, and observing that FA assembly and disassembly events of the migration cycle lead to different migratory outcomes, we obtained the equations describing deterministic cell motion between the random occurrence of FA events. After introducing the rates of FA binding and unbinding, we obtained the remaining necessary objects to define a piecewise deterministic Markov process: the distribution of interarrival times and of the next FA event. Note that the forms of these distributions have been derived, rather than simply postulated. Having specified the coupling between SFs and FAs, as well as between the cellular environment and FAs, we performed numerical simulations. We showed that our model is able to reproduce experimental observations, such as: superdiffusive scaling of the mean-squared displacement \cite{dieterich2008anomalous}, \cite{liang2008persistent}, \cite{liu2015confinement} (Figure \ref{figure: uniform bias results}); biased motility in the presence of external cue (Figure \ref{figure: ECM bias results}); contact guidance \cite{ROG17} (Figure \ref{figure: contact guidance plots}). In these cases, the obtained results followed solely due to asymmetric, dynamic instability of FAs in direct response to environmental stimuli. Specifically, it is only the biased FA assembly rate that drives biased cell motility along the cue gradient or the fiber tracts (Figures \ref{figure: ECM bias speeds and events} and  \ref{figure: contact guidance plots} (d-i)). That is, preferred velocities were not imposed or chosen in any way, but simply followed as a consequence of front-rear polarity, as the cell front is characterized by preferential FA binding and the rear by unbinding.         

Another characteristic of directed migration is the asymmetric contraction of actomyosin bundles. By increasing the force generation of myosin motors in the prospective rear, we obtained directed movement in the opposite direction (Figure \ref{figure: myosin 35p plots}). Here asymmetric FA dynamics (and so front-rear polarity) was also obtained, but as a consequence of locally induced contractile activity, consistent with \cite{yam2007actin}.

Our simulation results in various settings suggest that the cell speeds follow a gamma distribution (Figures \ref{figure: uniform results} (g-i), \ref{figure: uniform bias results} (g-i), \ref{figure: ECM bias speeds and events} (a-c)). Furthermore, the number of adhesion sites seems to be a determinant of the gamma distribution, as its parameters are similar under different settings and given number of FAs. These results suggest that cell speeds are independent of biased FA formation, i.e. the bias only alters the directionality and not the speed. It is also interesting to see a correlation between the number of adhesion sites and diffusivity (Table \ref{table: uniform parameters}), as well as average speed (Tables \ref{table: uniform parameters}-\ref{table: myosin parameters}). Note that faster and diffusive amoeboid movement is characterized by an increased number of weaker adhesions with high turnover and contractility \cite{pavnkova2010molecular}. Thus, the aforementioned correlation is also consistent with experimental observations. We note that our model is not fit to take into account motility strongly relying on cell shape control, which is required, for example, in highly mobile cells. However, the simulations reproduce migration along fiber tracts, where cell reshaping takes place \cite{ROG17}. Our results suggest, then, that adhesion along the tracts is sufficient to produce such migration patterns.   

Although the model of the internal contractile machinery driving cell locomotion and cytoskeletal remodeling is simple, the resulting numerical simulations explain several experimental observations. Moreover, the cyclical nature of cell motility is captured with our piecewise deterministic model. While migration of a crawling cell is accompanied by changes in its shape, dynamic coupling of cell-substrate adhesions and contractile machinery, i.e. focal adhesions and stress fibers, represent another side of the coin. Numerous sophisticated phase-field or free boundary models that produce realistic morphology of motile cells, often do not emphasize this coupling (and stochasticity) during the migration cycle. We attempted to remedy this issue in our model, and showed with our simulations that numerous aspects of cell migration can be explained without detailed account of cell shape changes. Nevertheless, shape control is essential for a more complete description of the phenomena. We believe this can be done with the framework provided by vertex-based models \cite{FLETCHER20142291}: a more complex contractility apparatus can be described via active cable network model \cite{TBS12} and a more detailed account of mechanical forces (e.g. protrusions due to actin polymerization) can be done as in \cite{COPOS20172672}. Together with models of RhoGTPases signaling pathways \cite{holmes2012comparison}, \cite{holmes2012modelling}, \cite{mori2008wave}, the most significant drawbacks of our approach (including rigid rotation of the SF structure) can be overcome. The presented framework of piecewise deterministic motility process can also be extended to three-dimensional setting, as neither the event interarrival time distribution, nor the transition measure rely on the particular features of migration in a plane. 

However, given the relative simplicity of the stochastic model and its ability to explain a handful of the experimental observations, it is possible to extend the model to include cell-cell collisions in the context of contact inhibition of locomotion (CIL). Here, the collisions lead to cessation of locomotion and to repolarization, such that the formation of new adhesions at the site of contact is inhibited, while contractility is stimulated \cite{roycroft2016molecular}. Within the framework of our model this can be implemented in a straightforward manner: collisions cause a switch to a non-moving state and the FA probability rates are modified according to contact location, as was done in our work in \cite{uatay2019mathematical} Yet another extension is obtained by treating cells as particles and using kinetic theory, yielding equations governing population migration. Thus, we can achieve a multiscale description of cell motility. Due to limitations in size and scope, cell-cell collisions and population migration will be treated in forthcoming works.

\section*{Acknowledgement}
The author acknowledges support of the German Academic Exchange Service (DAAD).
%\clearpage
\begin{appendices}

\section{Equations of cell motion}\label{appendix: equations of motion}
In our model, using common, lab's reference frame will yield the same governing relations, because the involved forces are determined by relative position of cellular structures.
Below, we show why this is the case and provide a more detailed explanation regarding the equation of motions for $\mathbf{x},\mathbf{x}_n,\theta_1$ presented in Sections 2.1-2.2. 

Let $\mathbf{x}'_n=\mathbf{x}+\mathbf{x}_n$ and $\mathbf{x}'_i = \mathbf{x}+\mathbf{x}_i$, where primes indicate the corresponding variables in the lab's reference frame (recall that $\mathbf{x}_i$ is the position of the $i^\text{th}$ FA in cell's reference frame). Then, in this frame, the length of the $i^{\text{th}}$ SF $L'_i$ and the unit vector $\mathbf{e}'_i$ along it are given by 
\begin{align*}
L'_i &= \lVert\mathbf{x}_n'-\mathbf{x}'_i\rVert=\lVert\mathbf{x}_n-\mathbf{x}_i\rVert=L_i\\
\mathbf{e}'_i&= \frac{\mathbf{x}'_n-\mathbf{x}'_i}{L'_i}=\frac{\mathbf{x}_n-\mathbf{x}_i}{L_i}=\mathbf{e}_i,
\end{align*}
respectively.
Thus, $\mathbf{F}'_i(\mathbf{x}'_n,\theta_1)=\mathbf{F}_i(\mathbf{x}_n,\theta_1)$, where $\mathbf{F}'_i$ is the force applied by the $i^{\text{th}}$ SF at the $i^{\text{th}}$ FA.  Note that the force at $\mathbf{x}'_n$ (or $\mathbf{x}_n$) due to the $i^{\text{th}}$ SF is $-\mathbf{F}'_i(\mathbf{x}'_n,\theta_1)$ (or $-\mathbf{F}_i(\mathbf{x}_n,\theta_1)$) by action-reaction principle. Therefore, net force $\mathbf{F}'$ at $\mathbf{x}'_n$ is $\mathbf{F}'(\mathbf{x}'_n,\theta_1,\mathbf{Y})=-\sum_{i=1}^{M}Y_i\mathbf{F}'_i=-\sum_{i=1}^{M}Y_i\mathbf{F}_i=\mathbf{F}(\mathbf{x}_n,\theta_1,\mathbf{Y})$. Neglecting inertia, we have 
\begin{align*}
\beta_{cell}\dot{\mathbf{x}}_n = \mathbf{F}(\mathbf{x}_n,\theta_1,\mathbf{Y}) = \mathbf{F}'(\mathbf{x}'_n,\theta_1,\mathbf{Y}).
\end{align*}
Now, let us examine the equations of motion after FA unbinding, stated in Section 2.2.2, but in lab's reference frame. In this frame, the radial unit vector $\hat{\mathbf{r}}'(\mathbf{x}'_n)$ from the cell center $\mathbf{x}$ is given by (see Figure \ref{fig: five} in the manuscript for illustration) 
\begin{align*}
\hat{\mathbf{r}}'(\mathbf{x}'_n) = \frac{\mathbf{x}'_n-\mathbf{x}}{\lVert\mathbf{x}'_n-\mathbf{x}\rVert} = \frac{\mathbf{x}_n}{\lVert \mathbf{x}_n\rVert} = \hat{\mathbf{r}}(\mathbf{x}_n).
\end{align*}
Analogously, the tangential unit vector $\hat{\bm{\varphi}}'(\mathbf{x}'_n)$ is given by
\begin{align*}
\hat{\bm{\varphi}}'(\mathbf{x}'_n) = \left(-\frac{x'_{n,2}-x_2}{\lVert\mathbf{x}'_n-\mathbf{x}\rVert},\frac{x'_{n,1}-x_1}{\lVert\mathbf{x}'_n-\mathbf{x}\rVert}\right)^T = \left(-\frac{x_{n,2}}{\lVert\mathbf{x}_n\rVert},\frac{x_{n,1}}{\lVert\mathbf{x}_n\rVert}\right)^T = \hat{\bm{\varphi}}(\mathbf{x}_n).
\end{align*}
Note that the tangential component $F'_{\varphi}$ of the force $\mathbf{F}'$ at $\mathbf{x}'_n$ induces rotational motion, while the radial component $F'_{r}$ of the force $\mathbf{F}'$ at $\mathbf{x}'_n$ induces translational motion. These components are given by:
\begin{align*}
F'_{\varphi} &= \mathbf{F}'(\mathbf{x}'_n,\theta_1,\mathbf{Y})\cdot\hat{\bm{\varphi}}'(\mathbf{x}'_n) = \mathbf{F}(\mathbf{x}_n,\theta_1,\mathbf{Y})\cdot\hat{\bm{\varphi}}(\mathbf{x}_n) = F_{\varphi}\\
F'_{r} &= \mathbf{F}'(\mathbf{x}'_n,\theta_1,\mathbf{Y})\cdot\hat{\mathbf{r}}'(\mathbf{x}'_n) = \mathbf{F}(\mathbf{x}_n,\theta_1,\mathbf{Y})\cdot\hat{\mathbf{r}}(\mathbf{x}_n) = F_r.
\end{align*}
Neglecting rotational inertia, we then have 
\begin{align*}
\beta_{rot}\dot{\theta}_1 &= \lVert\mathbf{x}'_n-\mathbf{x}\rVert F'_{\varphi}(\mathbf{x}'_n,\theta_1,\mathbf{Y})\\
&=\lVert\mathbf{x}_n\rVert F_{\varphi}(\mathbf{x}_n,\theta_1,\mathbf{Y})
\end{align*}
where the right hand side in the first (second) line is the torque due to tangential component of the force $\mathbf{F}'$ ($\mathbf{F}$) at $\mathbf{x}'_n$ ($\mathbf{x}_n$). Because of low Reynolds number, we also have
\begin{align*}
\beta_{ECM}\dot{\mathbf{x}} &= F'_r(\mathbf{x}'_n,\theta_1,\mathbf{Y})\hat{\mathbf{r}}'(\mathbf{x}'_n)\\
&=F_r(\mathbf{x}_n,\theta_1,\mathbf{Y})\hat{\mathbf{r}}(\mathbf{x}_n),
\end{align*}
due to translational motion induced by the radial component of the force $\mathbf{F}'$ at $\mathbf{x}'_n$.

In the common reference frame, the following system of ODEs holds after unbinding (using the definition of $\mathbf{x}'_n$):
\begin{align*}
\beta_{ECM}\dot{\mathbf{x}} &= F'_r(\mathbf{x}'_n,\theta_1,\mathbf{Y})\hat{\mathbf{r}}'(\mathbf{x}'_n)\\
\beta_{cell}\dot{\mathbf{x}}'_n &= \beta_{cell}\dot{\mathbf{x}}+\mathbf{F}'(\mathbf{x}'_n,\theta_1,\mathbf{Y})\\
\beta_{rot}\dot{\theta}_1 &= \lVert\mathbf{x}'_n-\mathbf{x}\rVert F'_{\varphi}(\mathbf{x}'_n,\theta_1,\mathbf{Y}),
\end{align*}
which is equivalent to \eqref{eq: 3}. Using the common reference frame becomes even less convenient when we formulate and analyze our stochastic process of cell motility. Moreover, our approach in the main text does not contradict the formulation with the single reference frame, and is equivalent to it.

\section{Parameter assessment}\label{appendix: parameters}
Note that the length of myosin mini-filaments is $\sim0.3\mu m$ \cite{svitkina1989minifils} and the interfilament distance is $\sim1\mu m$ in an uncontracted fiber \cite{aratynschaus2011dynamic}. Assuming vanishing interfilament distance (see Figure \ref{figure: stress fiber contraction 1} for illustration and \cite{Murrell2015} for a review on actomyosin contraction mechanism), then the proportion of the minifilaments to the initial, uncontracted SF length is $\frac{0.3}{1+0.3}=0.23$. Kassianidou et. al. \cite{kassianidou2017geometry} showed that the retraction length scales linearly with the initial length. If the interfilament distance vanishes, the front myosin motor cannot perform a power stroke and step forward, which renders the single motor and the entire minifilament unable to apply contractile stress. Therefore, taking the initial length to be $\sim R_{cell}$, we estimate the critical length $L_c=0.2R_{cell}$. Interestingly, Deguchi et. al. \cite{deguchi2006stress} found that stress fibers contract, on average, to $20\%$ of their original length. As stress fibers generally span more than half of a cell, and since it was found that there is a preexisting strain \cite{deguchi2006stress} in them, we estimate $L_0=1.1R_{cell}$. 

If we take $R_{cell}=25\mu m$, and assuming a SF has at rest the length of $R_{cell}$, we can estimate the number of myosin minifilaments in a fiber to be $25\mu m/1.3\mu m\approx20$. It was estimated in \cite{Verkhovsky637} that there are $10$-$30$ myosin motors in each minifilament. As each motor produces a force of $2-10pN$ \cite{Finer1994}, \cite{MBKJTW95}, \cite{Murrell2015}, we then estimate $T_i=4nN$. Balaban et. al. \cite{Balaban01} found that focal adhesions apply a constant stress of $\sim5.5nN/\mu m^2$ over an area of $\sim 1\mu m^2$ on an elastic substrate. Thus, we take $F_b=5.5nN$. Assuming a preexisting strain was $0.1$ when these measurements were taken, and since $T_i=4nN$, we then estimate the one dimensional Young's modulus $EA=15nN$.

Using Stokes' Law for drag in the low Reynolds number regime, the drag coefficient $\beta_{ECM}$ can be estimated as:
\begin{align*}
\beta_{ECM} = 6\pi\eta_{ECM}R_{cell},
\end{align*}
where $\eta_{ECM}$ is the dynamic viscosity of the environment. Assuming that the viscosity $\eta_{ECM}$ is higher than that of water, and taking into account that the contact between cell surface and the substrate increase the effective friction, we estimate $\beta_{ECM}\approx 10-10^3\frac{N\cdot s}{m^2}\times R_{cell}$. 
Similarly, due to the low Reynolds number, the rotational drag coefficient $\beta_{rot}$ is given by:
\begin{align*}
\beta_{rot} = 8\pi\eta_{ECM}R^3_{cell}\approx 10-10^3\frac{N\cdot s}{m^2}\times R^3_{cell}.
\end{align*}
In order to obtain estimates for the drag coefficient $\beta_{cell}$ one needs to have an estimate of the cytoplasm viscosity. Assessing the effective cytoplasmic viscosity of migrating cells is a challenging task, since the viscoelastic properties of the cytoskeleton (which, among other things, consists of polymer networks) are highly dynamic due to constant remodeling and spatiotemporal mediation of the rheology by various signaling pathways. Particularly, the actin network bundle size and cross-linkers influence the viscoelastic properties \cite{Gardel2003rheology}. Furthermore, the effective viscosity experienced by an experimental probe (or a protein) in polymer solutions depends not only on the type of material properties of the fluid, but also on the size of a probe \footnote{Using a small molecule as a probe, the cytoplasm viscosity was found to be $\simeq2-3 \times 10^{-4}Pa\cdot s$ \cite{Mastro01061984}. With larger probes, the viscosity was found to be $\simeq 2-4\times 10^{-2}Pa\cdot s$ \cite{Kalwarczyk2011viscosity} and $\simeq 5\times 10^{-2}Pa\cdot s$\cite{Arcizet2008viscosity}} \cite{Kalwarczyk2011viscosity}. Inferring that the body being dragged in the cell is the nucleus with radius $R_{nucleus}$, we estimate:
\begin{align*}
\beta_{cell} = 6\pi\eta_{cell}R_{nucleus} \approx 10-10^2\frac{N\cdot s}{m^2}\times R_{nucleus},
\end{align*}
where $\eta_{cell}$ is the cytoplasm viscosity.

\begin{table}[h]
\centering
\def\arraystretch{1.5}
\begin{tabular}{ |c|>{\centering\arraybackslash}m{3.5cm}|c||c|>{\centering\arraybackslash}m{3.5cm}|c|}
\hline
 Parameter& Value  &Value  & Parameter & Value &Value\\
 \hline
 $T_i$ & $4nN$ & 0.72 & $R_{cell}$& $25\mu m$  &1 \\
 \hline
 $EA$& $15nN$ &2.72  & $L_0$ & $27.5\mu m$ &1.1\\
 \hline
 $F_b$ & $5.5 nN$ &1 & $L_c$ & $5\mu m$ &0.2\\
 \hline
 $k^0_{off}$ & $0.05s^{-1}$&1   & $\beta_{rot}$ & $1.56\times10^{-11}N\cdot s\cdot m$ &5.68\\
 \hline
 $k^0_{on}$ & $0.01s^{-1}$&0.2 & $\beta_{ECM}$ & $5\times10^{-4}\frac{N\cdot s}{m}$ &0.11\\
 \hline 
 $R_{nucleus}$ & $5\mu m$&0.2  & $\beta_{cell}$ & $5\times10^{-3}\frac{N\cdot s}{m}$ &1.14\\
 \hline 
\end{tabular}
 \caption{Parameters used for simulation and their relative size with respect to spatial, temporal, and force scales. See Section \ref{section: kinematics summary} for details.}
 \label{table: parameter1}
\end{table}

Note that a focal adhesion is a cluster of transmembrane receptors (integrins) linking the substrate with the cytoskeleton, which is always under tension. The cluster also includes adapter proteins, which interlinks these receptors with the cytoskeleton (see e.g. reviews \cite{BCW1996}, \cite{geiger2009environmental}). Thus, if there is no load on a focal adhesion, then, since the cytoskeleton is always under tension, such a focal adhesion can be treated as a complex of independent bonds to the substrate without a link to SF. In the absence of a load, the average cluster lifetime $T_{life}$ is given by \cite{ES04}: 
\begin{align*}
T_{life} = \frac{1}{k^1_0+k^1_{on}}\left[H_N+\sum_{n=1}^{N}\left(\frac{k^1_{on}}{k^1_0}\right)^n\frac{1}{n}{{N}\choose{n}}\right],
\end{align*} 
where $N$ is the number of bonds in a cluster, $H_N$ is the $N$th harmonic number, $k_0^1$ and $k^1_{on}$ are, respectively, unbinding and binding rates of integrins under no load. Note that in the absence of a load, (re-)binding of individual integrins is an independent event, which bears no relation to the FA, since tension is required  for an FA to form and sustain itself. We then estimate:

%a focal adhesion is functionally defined as an adhesion complex transmitting internally generated traction forces
\begin{align*}
k^0_{off} = 1/T_{life}\lvert_{k^1_{on}=0}=k^1_0/H_N.
\end{align*}  
Li et al.\cite{LREM03} found that $k^1_{0} = 0.012s^{-1}$ for $\alpha_5\beta_1$-integrin binding to fibronectin. For $N=10^3-10^4$ we estimate that $k^0_{off}=0.05s^{-1}-0.1s^{-1}$.

The rationale for simulating $M=8,16,32$ FAs is the following. Note that the number of cell-substrate adhesions is higher than the number of focal adhesions we chose for our simulations. However, not all adhesions are directly involved in translocating the cell body, during which large traction forces are applied to the substrate through focal adhesions (which are fewer in number than immature, less stable focal complexes/points and nascent adhesions). Moreover, detachment of focal adhesions that leads to translocation, is primarily the result of contractile tension applied by ventral stress fibers, as opposed to transverse arcs and dorsal stress fibers \cite{kassianidou2015biomechanical}. The latter two have primarily structural role, while the former is fundamental to rear retraction \cite{Chen187}, \cite{kassianidou2015biomechanical}. Thus, the number of focal adhesions that are directly involved in cell body translocation is controlled by the number of the associated ventral stress fibers, which are also the most significant source of traction force applied to the substrate due to large tension within them \cite{kassianidou2015biomechanical}, \cite{OakesMBTF15}. Although reports of ventral stress fiber numbers are elusive, visual inspection of the fluorescence images in, for example, \cite{Hotulainen383}, \cite{kuragano2018different}, \cite{OakesMBTF15} (or any other appropriate study) suggests that simulations with the chosen number of (ventral) fibers (and focal adhesions) is realistic. Moreover, diameter of focal adhesions $d\sim 1-5\mu m$ \cite{gardel2010mechanical}. Assuming that the separation between focal adhesions is comparable to their size, and taking the cell radius to be $R_{cell} = 25\mu m$ (as in our simulations), we see that the upper range of possible number of adhesions on the cell circumference is $2\pi R_{cell}/2d\approx 16-80$. We reiterate that this number is an estimate of focal adhesions attached to ventral stress fibers, and it underestimates the \textit{total} number of focal adhesions that a cell employs, since significant number of them are attached to other types of stress fibers and may also be present within the cell body and not on its periphery. We performed simulations with $M=64$ focal adhesions and did not find any added insight. 

The values of $\gamma_1,\gamma_2$ and $\epsilon$, mentioned in \ref{section: adhesion force dependence}, can be found as follows: 

Suppose $F\leq F_b$. Then,
\begin{align*}
\gamma_1 = -\frac{F_b}{F-F_1^*}\log\left(\frac{k^0_{off}e+k^0_{on}}{k_{force}(F)-k^0_{on}+\epsilon}-1\right).
\end{align*}
Since $k_{force}(0) = k^0_{on}$ and $k_{force}(F_b) = k^0_{off}e$, then:
\begin{align*}
\gamma_1 &= \frac{F_b}{F_1^*}\log\left(\frac{k^0_{off}e+k^0_{on}}{\epsilon}-1\right)\\
\gamma_1 &= -\frac{F_b}{F_b-F_1^*}\log\left(\frac{k^0_{off}e+k^0_{on}}{k^0_{off}e-k^0_{on}+\epsilon}-1\right).
\end{align*}
It follows that $\epsilon$ is given as the solution of the following equation:
\begin{align*}
\frac{F_b}{F_1^*}\log\left(\frac{k^0_{off}e+k^0_{on}}{\epsilon}-1\right)+\frac{F_b}{F_b-F_1^*}\log\left(\frac{k^0_{off}e+k^0_{on}}{k^0_{off}e-k^0_{on}+\epsilon}-1\right) = 0.
\end{align*}
Similarly, since $k_{force}(F_b) = k^0_{off}e$, we find:
\begin{align*}
\gamma_2 = \frac{F_b}{F_b-F^*_2}\log\left(\frac{k_{off}^0e+k^0_{on}}{k^0_{off}e-k_{on}^0}-1\right).
\end{align*}
The values of $\gamma_1,\gamma_2,$ and $\epsilon$ are fixed for a value of $F_b$.

\section{Data analysis}\label{appendix: data analysis}
Given that the time interval $[0,t_{end}]$ is divided into $n_{time}$ subintervals of equal length $\Delta t$ and given the positions $\mathbf{x}^i(t_j)$ of cell $i$ at the time points $t_j:=j\Delta t$, $j=0,\ldots,n_{time}$, the mean-squared displacement $msd_i(t_j)$ of cell $i\in\left\lbrace 1,\ldots,n_{cells}\right\rbrace $ over a time interval of length $t_j$ is given by:
\begin{align}\label{eq: time average msd}
msd_i(t_j) := \frac{1}{n_{time}-j}\sum_{k=1}^{n_{time}-j}\lVert\mathbf{x}^i(t_{j+k})-\mathbf{x}^i(t_{k})\rVert^2,
\end{align}
where $j=1,\ldots,n_{time}-1$ and $n_{cells}$ is the total number of cells. Then, the mean-squared displacement $msd(t)$ of an ensemble of cells over time interval of length $t_j$ is defined by:
\begin{align}\label{eq: ensemble average msd}
msd(t):=\frac{1}{n_{cells}}\sum_{i=1}^{n_{cells}}msd_i(t_j). 
\end{align} 
\textbf{Remark.} In general, the (time-averaged) mean-squared displacement $<d^2(t,T)>$ of a particle trajectory $\mathbf{x}(t)$ at the time $t$, time endpoint $T$ is formally defined as:
\begin{align}\label{eq: time msd integral}
<d^2(t,T)> = \frac{1}{T-t}\int_{0}^{T-t}\lVert \mathbf{x}(s+t)-\mathbf{x}(t)\rVert^2ds.
\end{align}
For an ergodic process, we have:
\begin{align*}
\lim_{T\rightarrow\infty}<d^2(t,T)> = <\mathbf{x}^2(t)>,
\end{align*}
where $<\mathbf{x}^2(t)>$ is formally defined as:
\begin{align*}
<\mathbf{x}^2(t)> = \int\mathbf{x}^2P_t(d\mathbf{x}),
\end{align*}
and $P_t(d\mathbf{x})$ is the probability measure of the underlying stochastic process at time $t$. That is, for an ergodic process, and for sufficiently long times, the time average equals the phase space average. However, our cell motility process need not be ergodic and hence, using a quadrature to evaluate the integral in equation \eqref{eq: time msd integral}, we obtain time average displacement in equation \eqref{eq: time average msd}. To smooth out trajectory-to-trajectory fluctuations, we then average the displacements over all trajectories in equation \eqref{eq: ensemble average msd}. 

For a diffusive motion we expect that $msd(t)\sim t^{\beta(t)}$ with $\beta(t)\approx1$, while for a ballistic motion $\beta(t)\approx2$. Since $msd(0)=0$, we can estimate the exponent $\beta(t)$ for $t\in[\Delta t, t_{end}-\Delta t]$ from the simulated data as:
\begin{align*}
\beta(t)=\frac{d\ln{msd(t)}}{d\ln t}.
\end{align*}    

Although averaging reduces fluctuations, it does not eliminate them completely. Thus, computing the derivative above will yield a result that may oscillate wildly, which we want to avoid. Then, in order to investigate how $\beta$ varies with time, we define the time average $\beta_{av}(t)$ over the interval $[\Delta t, t]$ as:
\begin{align*}
\beta_{av}(t):=\frac{1}{t-\Delta t}\int_{\Delta t}^{t}\beta(s)ds = \frac{1}{t-\Delta t}\left(s\ln{msd(s)}\bigg\rvert _{\Delta t}^{t}-\int_{\Delta t}^{t}\ln{msd(s)}ds\right),
\end{align*}
where $t\in[2\Delta t, t_{end}-\Delta t]$, and we used integration by substitution and by parts. Then, $\bar{\beta}:=\beta_{av}(t_{end}-\Delta t)$ estimates the time scaling of $msd$ over the whole time interval. To asses how well $\bar{\beta}$ reflects the scaling of $msd$, we define the following function $\widehat{msd}(t):=\beta_0t^{\bar{\beta}}$, where $\beta_0$ is found by minimizing the square error:
\begin{align*}
\min_{\beta_0}\frac{1}{2}&\sum_{j=1}^{n_{time}-1}\left(\beta_0t_j^{\bar{\beta}}-msd(t_j)\right)^2\Rightarrow\\
&\beta_0 = \frac{\sum_{j=1}^{n_{time}-1}msd(t_j)t_j^{\bar{\beta}}}{\sum_{j=1}^{n_{time}-1}t_j^{2\bar{\beta}}}.
\end{align*}       
Letting $\bar{\beta}=\beta_{av}(t-\Delta t)$ to asses time scaling of $msd\sim t^{\beta(t)}$ is more accurate than the standard methods used for Brownian motion, since it also takes into account time dependence of the exponent. Also, our stochastic model has no Gaussian component. We refer to Section \ref{section: numerical simulations} for comparisons between $msd$ and $\widehat{msd}$, which show that the former well approximates the latter.

Note that because binding events can occur, a cell need not to have moved between the two time points $t_j$ and $t_{j+1}$. Thus, the speed between the consecutive time points may be zero for many time points, which would give an inaccurate statistical assessment of cell speeds. In order to estimate the speeds of a cell $i$ we use the following procedure: 

First, we find $l_i$, given by:
\begin{align*}
l_i:=\min\left\lbrace l\in\mathbb{N}:\mathbf{x}^i(t_{l+k})\neq\mathbf{x}^i(t_k), 0\leq k<n_{time}, l+k\leq n_{time}\right\rbrace.
\end{align*}       
Then we find the set of speeds $S_i$ as:
\begin{align*}
S_i := \left\lbrace s\in\mathbb{R}^+: s=\frac{\lVert\mathbf{x}^i(t_{(k+1)l_i})-\mathbf{x}^i(t_{kl_i})\rVert}{l_i\Delta t}, k\in\mathbb{N},(k+1)l_i\leq n_{time} \right\rbrace. 
\end{align*} 
This simply means that to compute speeds we only use a (minimal) time interval, such that a cell $i$ is guaranteed to change its position. The total set of speeds for $n_{cells}$ is $S:=\cup_{i=1}^{n_{cells}}S_i$. The average speed $s_{av}$ is then defined as arithmetic average:
\begin{align*}
s_{av}:=\frac{1}{\lvert S\rvert}\sum_{s\in S}s.
\end{align*}

The set of normalized velocities $V_i$ (or, alternatively, displacements) of cell $i$ is given by:
\begin{align*}
V_i := \left\lbrace \mathbf{v}\in\mathbb{R}^2: \mathbf{v}=\frac{\mathbf{x}^i(t_{(k+1)l_i})-\mathbf{x}^i(t_{kl_i})}{\lVert\mathbf{x}^i(t_{(k+1)l_i})-\mathbf{x}^i(t_{kl_i})\rVert}, k\in\mathbb{N},(k+1)l_i\leq n_{time} \right\rbrace,
\end{align*}
and the total set of normalized velocities $V$ is given by $V:=\cup_{i=1}^{n_{cells}}V_i$.

The directionality ratio $r_i(t_j)$ of cell $i$ over a time interval of length $t_j$ is given by:
\begin{align*}
r_i(t_j) := \frac{\sum_{k=1}^{j}\lVert\mathbf{x}^i(t_k)-\mathbf{x}^i(t_{k-1})\rVert}{\lVert\mathbf{x}^i(t_j)-\mathbf{x}^i(t_{0})\rVert}.
\end{align*}
The population and the time averages of the directionality ratio are given by, respectively:
\begin{align*}
r(t_j)& = \frac{1}{n_{cell}}\sum_{i=1}^{n_{cell}}r_i(t_j)\\
\bar{r} &= \frac{1}{n_{time}}\sum_{j=1}^{n_{time}}r(t_j). 
\end{align*}

Velocity autocorrelation $v^i_{ac}(t_j)$ of cell $i$ over a time interval of length $t_j$ is given by:
\begin{align*}
v^i_{ac}(t_j):=\frac{1}{n_{time}-j}\sum_{k=1}^{n_{time}-j}\frac{\mathbf{v}^i(t_{j+k})\cdot \mathbf{v}^i(t_{k})}{\lVert \mathbf{v}^i(t_{j+k})\rVert\phantom{.} \lVert\mathbf{v}^i(t_{k})\rVert},
\end{align*} 
where $j=1,\ldots, n_{time}-1$ and $\mathbf{v}^i(t_{k}) = \left(\mathbf{x}^i(t_{k})-\mathbf{x}^i(t_{k-1})\right)/\Delta t$. Velocity autocorrelation of the population $v_{ac}(t)$ is simply the arithmetic average of each cell's velocity autocorrelation. To compute $v_{ac}$ we used the time step of $12min$, whereas for all other quantities involving time dependence (e.g. $msd$) we used the time step of $0.12min$.  

We define the guidance parameter $G\in[0,1]$ similarly as in \cite{ROG17}:
\begin{align*}
G:=\frac{1}{\lvert\Theta\rvert}\sum_{\theta\in\Theta}g(\theta), 
\end{align*}
where $\Theta:=\cup_{i=1}^{n_{cells}}\Theta_i$. The set $\Theta_i$ of angles between a displacement vector of cell $i$ and the ECM fibers is defined by 
\begin{align*}
\Theta_i := \left\lbrace \theta\in[-\frac{\pi}{2},\frac{\pi}{2}]: 
\theta=arcsin\left( \frac{x^i_{2}(t_{(k+1)l_i})-x^i_{2}(t_{kl_i})}{\lVert\mathbf{x}^i(t_{(k+1)l_i})-\mathbf{x}^i(t_{kl_i})\rVert}\right), k\in\mathbb{N},(k+1)l_i\leq n_{time} \right\rbrace,
\end{align*}
where $x^i_{2}$ is the $y$-component of $\mathbf{x}^i$. The function $g:[-\frac{\pi}{2},\frac{\pi}{2}]\rightarrow\left\lbrace 0,1\right\rbrace$ is given by
\begin{align*}
g(\theta) = 
\begin{cases*}
1,\text{ if } \lvert\theta\rvert\leq\pi/4\\
0, \text{ else}
\end{cases*}.
\end{align*}
Thus, $G$ increases when displacements are aligned with the horizontal axis.
\section{Simulation of the PDMP}\label{appendix: simulation of the pdmp}
In order to compute the trajectories of the cell motility process, the following algorithm is used:
\begin{algorithm}[H]
\begin{enumerate}
\item Set $(\nu_0,\mathbf{X}_0)\in A\times\Gamma$ and $t=T_0=0$
\item For $k=0,1,\ldots$\\
Generate interarrival time $\Delta_k=T_{k+1}-T_{k}$, whose distribution function is given by:
\begin{align}\label{eq: tau distribution function}
\mathbb{P}\left(\Delta_k\leq\tau\right) = 1-\exp\left(-\int_{t}^{t+\tau}a_0(\nu_t,\phi_\nu(s,\mathbf{X}_t))ds\right)
\end{align}
 Compute $\mathbf{X}_{t+s^-}:=\phi_\nu(s,\mathbf{X}_t)$\\
 Set $t=T_{k+1}=T_k+\Delta_k$\\
 Generate $(\nu_t,\mathbf{X}_t)\sim Q(\cdot;(\nu_{t^-},\mathbf{X}_{t^-}))$  
\end{enumerate}
\caption{Simulation of the PDMP}
\end{algorithm}

To generate the interarrival time $\Delta_k$, we need to solve for $\tau$ in the following equation:
\begin{align}\label{eq: tau random}
f(\tau) := \int_{t}^{t+\tau}a_0(\nu_t,\phi_{\nu_t}(s,\mathbf{X}_t))ds + \ln(1-u)=0,
\end{align} 
where $u$ is uniformly distributed on the interval $(0,1)$.
Notice that the evaluation of the integral by a quadrature rule requires computing the solution $\mathbf{X}_{t+s}=\phi_{\nu_t}(s,\mathbf{X}_t)$ up to time $s$, where $s$ is a quadrature point. Moreover, using an iterative method to solve \eqref{eq: tau random} requires computing the integral at each iteration. Therefore, it is imperative to devise an efficient method to sample from distribution \eqref{eq: tau distribution function}. In the following, we propose a general method to generate the next event time.      
\subsection{Simulation of the next event time}\label{section: next event time simulation}

Let $T_k\leq t< T_{k+1}$ and let $\mathbf{G}(\cdot;h):\Gamma\rightarrow\Gamma$ represent a numerical method to solve the ODE system 
\begin{align*}
\frac{d}{dt}\mathbf{X}_t = \mathbf{H}_{\nu_t}(\mathbf{X}_t)
\end{align*}
for one time step $h$. That is, $\mathbf{X}_{t+h}=\mathbf{G}(\mathbf{X}_t;h)$ is the numerical solution of the above ODE system at time $t+h$. 

Let $[T_k,T_{k+1})^2\ni(s',t')\mapsto A_0(t',s'):=\int_{s'}^{t'}a_0(\nu_t,\phi_{\nu_t}(u,\mathbf{X}_t))du$  denote the integrated rate function.

The method to find the root $\tau$ of equation \eqref{eq: tau random} is given in Algorithm \ref{alg: Event time computation}. First, in steps (1-22), we find the upper bound $\tau_{max}$ by solving the ODE system for $n$ steps with step size $h$ and store the solution, the computed rate $a_0$, and the integrated rate $A_0$ at these time steps. Then, for any $\tau\leq\tau_{max}$ we can compute $A_0(t+\tau,t)$ upon using the stored value of $A_0$ at time $t+\tau_i$, where $i=\left\lfloor{\frac{\tau}{h}}\right\rfloor$ and $\tau_i=ih$:
\begin{align*}
A_0(t+\tau,t) &= \int_{t}^{t+\tau}a_0(\nu_t,\phi_{\nu_t}(s,\mathbf{X}_t))ds\\
&=\int_{t}^{t+\tau_i}a_0(\nu_t,\mathbf{X}_{t+s})ds+\int_{t+\tau_i}^{t+\tau_i+\tau-\tau_i}a_0(\nu_t,\phi_{\nu_t}(s,\mathbf{X}_t))ds\\
&=A_0(t+\tau_i,t)+\int_{t+\tau_i}^{t+\tau_i+h_i}a_0(\nu_t,\phi_{\nu_t}(s,\mathbf{X}_t))ds.
\end{align*}   

\begin{algorithm}[H]
\small
\caption{Event time computation}
\label{alg: Event time computation}
\begin{algorithmic}[1]
\Procedure{Initialization}{}
\State \textbf{Input}: Time $t=T_{k}$; $(\nu_t,\mathbf{X}_t)$; time step $h$ and ODE method $\mathbf{G}$; $n,m\in\mathbb{N}$. 
\State Set $s_0:=0,n_0:=0$, create $List_{a_0}$, $List_{A_0}$, and $List_{\mathbf{X}}$.
\State Append $List_{a_0}\leftarrow a_0(\nu_t,\mathbf{X}_t)$, $List_{A_0}\leftarrow 0$, $List_{\mathbf{X}}\leftarrow \mathbf{X}_t$.
\State Set $\tau_{max}:=nh$, $s_0:=n_0h$\label{begin}.
%\State Solve $\dot{\mathbf{X}}_s=\mathbf{H}_\nu(\mathbf{X}_s)$ for $s\in[t+s_0,t+s_0+\tau_{max}]$ with time step $h$.
\State Set initial condition $\mathbf{X}_s\leftarrow List_{\mathbf{X}}[last]$.
\State Set $A_0^0:=List_{A_0}[last]$.
\For{$i=1$ to $n$}
\State $s_i:=t+s_0+ih$.
\State Compute $\mathbf{X}_{s_i} := \mathbf{G}(\mathbf{X}_{s};h)$ and $a_0(\nu_t,\mathbf{X}_{s_i})$.
\State Compute $A_0(s_i; t+s_0)$ with quadrature points $s_j$, $j=0,\ldots,i$.
\State $A_0(s_i;t) := A_0^0+A_0(s_i; t+s_0)$.
\State Append $List_{a_0}\leftarrow a_0(\nu_t,\mathbf{X}_{s_i})$, $List_{A_0}\leftarrow A_0(s_i;t)$, $List_{\mathbf{X}}\leftarrow \mathbf{X}_{s_i}$.  
\EndFor
\State Generate $u\sim U(0,1)$.
\If{$List_{A_0}[last]<-ln(1-u)$}
\State $n:=n+m$.
\State $n_0 = n_0+m$.
\State go to \ref{begin}.
\EndIf
\State \textbf{Output:} $\tau_{max}$, $List_{a_0}$, $List_{A_0}$, and $List_{\mathbf{X}}$.
\EndProcedure
\Procedure{Evalution of $f$}{}\label{function eval}
\State \textbf{Input}: $\tau$; time step $h$; $List_{a_0}$, $List_{A_0}$; Integrated interpolation method $I$.
\State Set $i:=\left\lfloor{\frac{\tau}{h}}\right\rfloor$ and $h_i=\tau-ih$.
\State \textbf{Output}: $f(\tau)=List_{A_0}[i]+I(List_{a_0}[i],List_{a_0}[i+1];h_i,h)+\ln(1-u)$
\EndProcedure
\Procedure{Event time}{}
\State Find the root $\tau$ of $f(\tau)=0$ using \ref{function eval} and a root finding method.
\State $i:=\left\lfloor{\frac{\tau}{h}}\right\rfloor$.
\State Compute $\mathbf{X}_{t+\tau}:=\mathbf{G}(List_{\mathbf{X}}[i];\tau-ih)$.
\State \textbf{Output}: $\tau$, $\mathbf{X}_{t+\tau}$ 
\EndProcedure
\end{algorithmic}
\end{algorithm}

To compute the last integral in the expression above, we interpolate the integrand using the stored values $a_0(\nu_t,\mathbf{X}_{t+\tau_i})$ and $a_0(\nu_t,\mathbf{X}_{t+\tau_{i+1}})$. Let
\begin{align*}
I\left(a_0(\nu_t,\mathbf{X}_{t+\tau_i}),a_0(\nu_t,\mathbf{X}_{t+\tau_{i+1}}); h_i,h\right):= \int_{t+\tau_i}^{t+\tau_i+h_i}a_0(\nu_t,\phi_{\nu_t}(s,\mathbf{X}_t))
\end{align*}
denote the approximation of the integral using interpolated integrand.
We can use the following interpolation methods for $t+\tau_i<s<t+\tau_{i+1}$:
\begin{enumerate}
\item Piecewise constant\\
Forward: $a_0(\nu_t,\phi_{\nu_t}(s,\mathbf{X}_t))=a_0(\nu_t,\phi_{\nu_t}(\tau_{i},\mathbf{X}_t))=a_0(\nu_t,\mathbf{X}_{t+\tau_i})$.
\begin{align}\label{eq: forward interpol}
I\left(a_0(\nu_t,\mathbf{X}_{t+\tau_i}),a_0(\nu_t,\mathbf{X}_{t+\tau_{i+1}}); h_i,h\right) = h_ia_0(\nu_t,\mathbf{X}_{t+\tau_i}).
\end{align}
Backward:
$a_0(\nu_t,\phi_{\nu_t}(s,\mathbf{X}_t))=a_0(\nu_t,\phi_{\nu_t}(\tau_{i+1},\mathbf{X}_t))=a_0(\nu_t,\mathbf{X}_{t+\tau_{i+1}})$
\begin{align}\label{eq: backward interpol}
I\left(a_0(\nu_t,\mathbf{X}_{t+\tau_i}),a_0(\nu_t,\mathbf{X}_{t+\tau_{i+1}}); h_i,h\right)  = h_ia_0(\nu_t,\mathbf{X}_{t+\tau_{i+1}}).
\end{align}
\item Average: $a_0(\nu_t,\phi_{\nu_t}(s,\mathbf{X}_t)) = \frac{1}{2}\left(a_0(\nu_t,\mathbf{X}_{t+\tau_i})+a_0(\nu_t,\mathbf{X}_{t+\tau_{i+1}})\right)$.\\
\begin{align}\label{eq: average interpol}
I\left(a_0(\nu_t,\mathbf{X}_{t+\tau_i}),a_0(\nu_t,\mathbf{X}_{t+\tau_{i+1}}); h_i,h\right) = \frac{1}{2}h_i\left(a_0(\nu_t,\mathbf{X}_{t+\tau_i})+a_0(\nu_t,\mathbf{X}_{t+\tau_{i+1}})\right).
\end{align}
\item Piecewise linear: 
\begin{align*}
a_0(\nu_t,\phi_{\nu_t}(s,\mathbf{X}_t)) = (s-t-\tau_i)\frac{a_0(\nu_t,\mathbf{X}_{t+\tau_{i+1}})-a_0(\nu_t,\mathbf{X}_{t+\tau_{i}})}{h}+a_0(\nu_t,\mathbf{X}_{t+\tau_i}).
\end{align*}
\begin{align}\label{eq: linear interpol}
I\big(a_0(\nu_t,\mathbf{X}_{t+\tau_i}), a_0(\nu_t,&\mathbf{X}_{t+\tau_{i+1}}); h_i,h\big).\nonumber\\ 
&= h_ia_0(\nu_t,\mathbf{X}_{t+\tau_i})\left(1-\frac{h_i}{h}\right)+\frac{h^2_i}{2h}a_0(\nu_t,\mathbf{X}_{t+\tau_{i+1}}) .
\end{align}
\end{enumerate}
Thus, $f(\tau)$ can be evaluated using equations \eqref{eq: forward interpol}-\eqref{eq: linear interpol}: 
\begin{align*}
f(\tau) = A_0(t+\tau_i,t)+I\left(a_0(\nu_t,\mathbf{X}_{t+\tau_i}),a_0(\nu_t,\mathbf{X}_{t+\tau_{i+1}}); h_i,h\right) + \ln(1-u).
\end{align*}

 Using the interpolations above, we can now employ, for example, Newton's method to find the root of equation \eqref{eq: tau random}:
 \begin{align*}
 \tau_{l+1} = \tau_l - \frac{f(\tau_l)}{a_0(\nu_t,\phi_{\nu_t}(\tau_l,\mathbf{X}_t))},
 \end{align*}
 or any other root-finding method (note that $a_0>0$).
 
 Once the root is found, we simply advance the ODE system for one time step as described in Steps (28-33) of the Algorithm \ref{alg: Event time computation}. 
 
 Note that we solve the ODE system for $n+1=\tau_{max}/h+1$ steps and we also solve for the interarrival time $\tau$ primarily by using a look-up table. Moreover, obtaining a relatively sharp upper bound $\tau_{max}$ does not yield a large computational overhead, since one simply can start the Algorithm \ref{alg: Event time computation} with a small $n,m$. Consequently, choosing an initial guess close to the sharp upper bound for Newton's method results in a faster convergence. In case of thinning methods (see \cite{lewis1979simulation} or \cite{2015arXiv151002451B} for adaptive method), after each rejection one needs to solve the ODE system for time period that is, on average, approximately the same as $\tau_{max}$ (in the best case scenarios for both methods, i.e. when the bound $\tau_{max}$ and the bound for the rate function in thinning methods are sharp). Of course, these arguments hold when the computational cost of solving the ODE system is relatively large.

\subsection{Sampling from the transition measure}\label{section: next event index simulation}
%section: summary of PDP eq: probability of the next reaction j
Given the time $t$ of the next event and $\mathbf{X}_t$ we need to sample from the transition distribution $Q(\cdot,(\nu_{t^-},\mathbf{X}_{t^-}))$. Recalling Section \ref{section: summary of PDP} and the proof of Proposition \ref{proposition: transition measure form}, in order to sample from the transition measure it is sufficient to simulate the index $j\in\left\lbrace1,\ldots,2M \right\rbrace $ of the next reaction, since the continuous component of the process does not jump. The discrete distribution of the next reaction index is given by equation \eqref{eq: probability of the next reaction j}:
\begin{align*}
\mathbb{P}(j|\bm{\alpha}_2(\nu_{t^-}),\mathbf{X}_{t^-}) = \frac{a_j(\bm{\alpha}_2(\nu_{t^-}),\mathbf{X}_{t^-})}{a_0(\bm{\alpha}_2(\nu_{t^-}),\mathbf{X}_{t^-})}.
\end{align*}
To simulate from the discrete distribution one can use the fairly efficient Vose Alias Method \cite{vose1991random}.

\end{appendices}
%	\clearpage
%	\newpage
\providecommand{\noopsort}[1]{}

	\bibliographystyle{abbrv}
	
\end{document}